\documentclass{article}
\usepackage[a4paper,left=25mm,top=25mm,right=25mm,bottom=25mm]{geometry}

\usepackage{amsmath, amsfonts, amssymb, amsthm}
\usepackage{mathtools}
\usepackage{enumerate}
\usepackage{graphicx}
\usepackage{mathptmx}
\usepackage[mathscr]{euscript}
\usepackage[normalem]{ulem}
\usepackage[rightcaption]{sidecap}

\usepackage{csquotes} 

\usepackage[square,sort,comma,numbers]{natbib}
\usepackage{authblk}
\usepackage[font=footnotesize,width=.95	\textwidth,labelfont=bf]{caption}

\usepackage{color}
\usepackage[dvipsnames,svgnames]{xcolor}
\usepackage[colorlinks]{hyperref}
\hypersetup{
	citecolor=blue,
   linkcolor=blue,
}

\newtheorem{theorem}{Theorem}[section]
\newtheorem{lemma}[theorem]{Lemma}
\newtheorem{proposition}[theorem]{Proposition}
\newtheorem{definition}[theorem]{Definition}
\newtheorem{corollary}[theorem]{Corollary}

\newtheorem*{remark}{Remark}
\newtheorem{fact}[theorem]{Observation}

\makeatletter
\def\moverlay{\mathpalette\mov@rlay}
\def\mov@rlay#1#2{\leavevmode\vtop{%
    \baselineskip\z@skip \lineskiplimit-\maxdimen
    \ialign{\hfil$\m@th#1##$\hfil\cr#2\crcr}}}
\newcommand{\charfusion}[3][\mathord]{
  #1{\ifx#1\mathop\vphantom{#2}\fi
    \mathpalette\mov@rlay{#2\cr#3}
  }
  \ifx#1\mathop\expandafter\displaylimits\fi}
\DeclareRobustCommand\bigop[1]{%
  \mathop{\vphantom{\sum}\mathpalette\bigop@{#1}}\slimits@
}
\newcommand{\bigop@}[2]{%
  \vcenter{%
    \sbox\z@{$#1\sum$}%
    \hbox{\resizebox{\ifx#1\displaystyle.9\fi\dimexpr\ht\z@+\dp\z@}{!}{$\m@th#2$
}}%
  }%
}
\makeatother

\newcommand{\cupdot}{\charfusion[\mathbin]{\cup}{\cdot}}
\DeclareMathOperator{\bigcupdot}{\charfusion[\mathop]{\bigcup}{\cdot}}

\DeclareMathOperator{\child}{child}
\DeclareMathOperator{\parent}{par}
\DeclareMathOperator{\nss}{nss}
\DeclareMathOperator{\lca}{lca}
\DeclareMathOperator{\LCA}{LCA}

\newcommand{\mscr}{\mathscr}
\newcommand{\mc}{\mathcal}

\newcommand{\relV}{\mathrm{rel}V}
\newcommand{\relN}{\mathrm{relNgbh}}
\newcommand{\partN}{\partial N}
\newcommand{\partT}{\partial T}
\newcommand{\Tcat}{\mathrm{T}_n^{cat}}
\newcommand{\Tstar}{\mathrm{T}^{star}_n}
\newcommand{\TstarTwo}{\mathrm{T}^{star}_2}
\newcommand{\NC}{\mathrm{N}^{cres}}

\newcommand{\Tmod}{\mathrm{T}_n^{MS}}
\newcommand{\Tpush}{\mathrm{T}_n^{PS}}
\newcommand{\TmodTWO}{\mathrm{T}_2^{MS}}
\newcommand{\TpushTWO}{\mathrm{T}_2^{PS}}
\newcommand{\Nmb}{\mscr{N}_n^{MB}}

\newcommand{\BinLevelOneN}{\mscr{N}_n^{\mathrm{L1,bin}}}
\newcommand{\BinLevelOneNvar}[1]{\mscr{N}_{#1}^{\mathrm{L1,bin}}}

\newcommand{\HybdidDegTwoN}{\mscr{N}_n^{\mathrm{L1,in}\leq 2}}
\newcommand{\HybdidDegTwoNvar}[1]{\mscr{N}_{#1}^{\mathrm{L1,in}\leq 2}}

\DeclareMathOperator{\indeg}{indeg}
\DeclareMathOperator{\outdeg}{outdeg}

\newtheorem{innercustomthm}{Theorem}
\newenvironment{ctheorem}[1]
  {\renewcommand\theinnercustomthm{#1}\innercustomthm}
  {\endinnercustomthm}

\newtheorem{innercustomlem}{Lemma}
\newenvironment{clemma}[1]
  {\renewcommand\theinnercustomlem{#1}\innercustomlem}
  {\endinnercustomlem}

\newtheorem{innercustomprop}{Proposition}
\newenvironment{cproposition}[1]
  {\renewcommand\theinnercustomprop{#1}\innercustomprop}
  {\endinnercustomprop}

\newtheorem{innercustomdef}{Definition}
\newenvironment{cdefinition}[1]
  {\renewcommand\theinnercustomdef{#1}\innercustomdef}
  {\endinnercustomdef}

\newtheorem{innercustomcor}{Corollary}
\newenvironment{ccorollary}[1]
  {\renewcommand\theinnercustomcor{#1}\innercustomcor}
  {\endinnercustomcor}

\newtheorem{innercustomfact}{Observation}
\newenvironment{cfact}[1]
  {\renewcommand\theinnercustomfact{#1}\innercustomfact}
  {\endinnercustomfact}

\usepackage{pdfpages}
  
\begin{document}

\date{}

\title{The weighted total cophenetic index: A novel balance index for phylogenetic networks}
\author[1]{Linda Kn{\"u}ver}
\author[1]{Mareike Fischer}
\author[2,*]{Marc Hellmuth}
\author[3]{Kristina Wicke}

\affil[1]{Institute of Mathematics and Computer Science, University of Greifswald, Walther-Rathenau-Str. 47, 17487 Greifswald, Germany}
\affil[2]{Department of Mathematics, Faculty of Science, Stockholm University, SE - 106 91 Stockholm,   Sweden }
\affil[3]{Department of Mathematical Sciences, New Jersey Institute of Technology, Newark, NJ 07102, USA}

\affil[*]{corresponding author; \texttt{marc.hellmuth@math.su.se}}

\maketitle

\begin{abstract}
Phylogenetic networks play an important role in evolutionary biology as, other
than phylogenetic trees, they can be used to accommodate reticulate evolutionary
events such as horizontal gene transfer and hybridization. Recent research has
provided a lot of progress concerning the reconstruction of such networks from
data as well as insight into their graph theoretical properties. However, 
methods and tools to quantify structural properties of networks or differences 
between them are still very limited. For example, for phylogenetic trees, it is 
common to use balance indices to draw conclusions concerning the underlying evolutionary 
model, and more than twenty such indices have been proposed and are used for different
purposes. One of the most frequently used balance index for trees is the
so-called total cophenetic index, which has several mathematically and
biologically desirable properties. For networks, on the other hand, balance
indices are to-date still scarce.

In this contribution, we introduce the \textit{weighted} total cophenetic index as
a generalization of the total cophenetic index for trees to make it applicable
to general phylogenetic networks. As we shall see, this index can be determined
efficiently and behaves in a mathematical sound way, i.e., it satisfies
so-called locality and recursiveness conditions. In addition, we analyze its
extremal properties and, in particular, we investigate its maxima and minima as
well as the structure of networks that achieve these values within the space of
so-called level-$1$ networks. We finally briefly compare this novel index to the
two other network balance indices available so-far.
\end{abstract}

\sloppy

\section{Introduction}
\label{sec:intro}

Phylogenetic networks are used to represent evolutionary relationships that cannot be adequately
described by a single tree due to reticulation events such as recombination, horizontal gene transfer, or hybridization
\cite{nakleh:04,Gusfield2003,Gontier2015}. Hence, phylogenetic
networks play a crucial role in evolutionary research as well as in biomathematics \cite{Huson2010}, 
and understanding and investigating the structure of phylogenetic networks is 
essential in order to understand basic principles in evolution.

The study of evolutionary histories in terms of trees has a rather long tradition, see
\cite{sem-ste-03a} for an excellent overview. 
One way to investigate the structure of a tree is to measure its \emph{balance}.
Balance of trees plays an essential role in many different research areas, but
it is particularly relevant in theoretical computer science, e.g., to measure
the balance of search trees \cite{Knuth3,leiserson1994introduction}, and
evolutionary research, where it is used in a variety of contexts, including (but
not limited to) testing evolutionary models (cf.~\cite{Aldous_stochastic_2001,
Blum2005, Kirkpatrick1993, Mooers1997}), assessing the impact of fertility
inheritance and selection (cf.~\cite{Blum2006b, Maia2004, Verboom2020}), or to
study tumor evolution (cf.~\citep{DiNardo2019,Scott2019}).
For phylogenetic trees there are currently more than 20 
indices available that allow for measuring the extent to which a tree is
balanced. We refer to \cite{Fischer2021Book} for a comprehensive overview on tree balance indices available in the literature and their extremal properties. 
It is well
known that different balance indices have different merits and are useful for
different purposes: While some balance indices provide total enumerations of the
tree space, others have a more limited range but are better at determining underlying evolutionary
models, like for instance the so-called Yule model \cite{Yule1925}. Some balance
indices have unique extremal trees even in the binary case (in which every
vertex is only allowed to have two children), whereas other indices allow for ties in such
cases. 

For phylogenetic networks, on the other hand, there are currently only two
balance indices available: Zhang
\cite{Zhang2022} generalized the Sackin index \cite{Sackin1972, Fischer2021}
(one of the earliest balance indices for trees), while Bienvenu et al.\
\cite{Bienvenu2021} generalized a lesser known
tree balance index, namely the so-called $B_2$ index \cite{Shao1990}. However, one of the most
frequently used balance indices for phylogenetic trees, namely the total
cophenetic index \cite{Mir2013}, has not yet been made available for networks,
despite the fact that it has some quite intriguing
properties. For example, the most balanced binary tree according to the total
cophenetic index is unique for all leaf numbers $n$, which is not the case for
many other indices (including Sackin and $B_2$), and concerning trees, it has a
larger range and a lower probability of ties than, say, the Sackin index 
(see \cite{Fischer2021Book} for details).

We thus aim to generalize the concept of the (total) cophenetic index from
phylogenetic trees to networks and to analyze its extremal properties. In
particular, we introduce a weighted version of the total cophenetic index, which --
when restricted to trees -- coincides with the classic definition of the total
cophenetic index \cite{Mir2013} (up to an added constant).
 In particular, this weighted version involves a parameter $\epsilon\geq 0$. As this
parameter can be arbitrarily adjusted, this generalized
notion of the cophenetic index can actually be considered a class of indices
for networks. We use weights to quantify the balance and distribution of the
blocks (i.e., the maximal biconnected subgraphs) and the \enquote{tree-likeness}
of the given network. 
In particular, each block $B$ obtains a well-defined weight $\omega(B)$. In addition,
each vertex $v$ of the underlying network $N$ obtains a weight $\phi(v)$ which
may involve the weights of some blocks as well the parameter $\epsilon$. The
weights $\phi(v)$ allow us to distinguish between substructures that are trees
and substructures that are \enquote{very close} to trees but that contain cycles. A
well-defined subset of vertices in $N$ (so-called relevant vertices) is then
used to determine the balance of $N$.  We analyze here several properties of
this new index and characterize networks with extremal indices.

The weighted total cophenetic index is suitable for general networks. 
However, both the reconstruction of networks from data as well as their mathematical
analyses are challenging and often more intricate than for trees (see, \ e.g., 
\cite{TanLongLiao,gambette:17,FischerMPnetworks,JinNaklehSnirTuller}). To
analyze the properties of the weighted total cophenetic index, we therefore focus on so-called \emph{level-$1$ networks}, i.e., on networks in which
each non-trivial block has precisely one reticulation vertex which are those
vertices of in-degree larger than 1. Level-$1$ networks play an important
role in evolutionary biology, e.g.,  in the analysis of virus evolution
\cite{huber:11}. Although level-$1$ networks can be considered as one of the
combinatorially simplest types of phylogenetic networks, this simplicity has
proven to be deceptive, as the combinatorial structure of such networks has
turned out to be more complicated than originally thought
\cite{GH:11,Huber2013,gambette:17,Huber2017,huber:11}). Nevertheless, subclasses
of level-$1$ networks (e.g., so-called galled trees \cite{Gusfield2003}
or binary level-1 networks) are in many cases more easily accessible from biological data than other types of
networks (see e.g.,  \cite{Jansson2016,ROSSELLO200954,gambette:17}).  Moreover, restricting the level of a
network is currently a common way to obtain scalable network inference
software, see e.g.,  PhyloNetworks \cite{solis-lemus2016, solis-lemus2017},
PhyNEST \cite{kong2022inference}, or NANUQ \cite{allman2019nanuq}. We emphasize
that level-$1$ networks considered here are not restricted to be binary (i.e., we do not restrict to
those networks where each reticulation vertex has in-degree 2 and out-degree 1
and where every other non-leaf vertex has out-degree 2). In fact, we conduct our
extremal analyses both for the binary and non-binary cases. 

The remainder of this manuscript is organized as follows. We begin by defining
phylogenetic networks and related concepts in Section \ref{sec:prelim}. In Section \ref{sec:wtci}, we then formalize the concept of network balance, before introducing the central concept of this manuscript,
namely the weighted total cophenetic index. Subsequently, in Section
\ref{sec:otherproperties} we state some general properties of the weighted total
cophenetic index, such as its locality and recursiveness, before finally turning
to its extremal properties. Here, we begin by analyzing the structure of
extremal blocks in Section \ref{sec:structure-blocks}, before deriving the
extremal values of the weighted total cophenetic index and characterizing the
networks achieving them in Section \ref{sec:extrema}. We conclude our study with
a brief comparison of the weighted total cophenetic index to the above mentioned
two other indices available for phylogenetic networks and indicate some
directions for future research in Section \ref{sec:discussion}.

Since the material is extensive and very technical, we subdivide our
presentation into a main part
(Sections~\ref{sec:intro}--\ref{sec:discussion}), which can be viewed as more \enquote{narrative}, and a technical part
(Sections~\ref{APPX:prelim}--\ref{APPX:sec:structure-blocks}), which contains all
proofs and additional material in full detail. Together with the
definitions and preliminaries in Section~\ref{sec:prelim}, the technical part
is self-contained. Definitions and results appearing in the narrative part
are therefore restated. The order of the material in the two parts is
slightly different. 

\section{Preliminaries}
\label{sec:prelim}

We mainly follow here the terminology from \cite{HSS:22}.

\paragraph{\bf Graphs, paths, and connectedness}
We consider graphs $G=(V(G),E(G))$ with vertex set $V\coloneqq V(G)$ and edge set
$E\coloneqq E(G)$. A graph $G$ is \emph{undirected} if $E\subseteq
\{\{u,v\}\colon u,v\in V, \ u \neq v\}$, i.e., if $E$ is a subset of the set of two-element (unordered) subsets of $V$. $G$ is \emph{directed} (a
di-graph, for short) if $E\subseteq \{(u,v)\colon u,v\in V,u \neq v\}$, i.e., if $E$
is a subset of \emph{ordered} pairs of the elements of $V$. Thus, edges $e\in E$
in undirected graphs are of the form $e=\{x,y\}$ and in di-graphs of the
form $e=(x,y)$ with $x,y\in V$ being distinct. The \emph{degree} of a vertex
$v\in V$ in an undirected or directed graph $G$, denoted by $\deg_G(v)$, is the
number of edges that are incident with $v$. If $G$ is directed, we furthermore
distinguish the in-degree $\indeg_{G}(v) = \vert \{u \mid (u,v)\in E\}\vert$ and
the out-degree $\outdeg_{G}(v) = \vert \{u \mid (v,u)\in E\}\vert$. Whenever
there is no ambiguity, we omit the subscript and simply write $\deg(v),
\indeg(v)$, and $\outdeg(v)$. We write $H \subseteq G$ if $H$ is a subgraph of
$G$ and use $G[W]$ to denote the subgraph of $G$ that is induced by some subset
of vertices $W \subseteq V$. Two di-graphs $G$ and $H$ are \emph{isomorphic}, 
in symbols $G\simeq H$, if there is a bijective map $f\colon V(G) \to V(H)$
such that $(u,v)\in E(G)$ if and only if $(f(u),f(v))\in E(H)$.

A (directed) path $(v_1,\dots,v_n)$ in a di-graph $G$ consists of pairwise
distinct vertices $v_1,\dots,v_n$ and edges $(v_i,v_{i+1})$, $1\leq i\leq n-1$.
Paths in undirected graphs are defined analogously. We will also refer to paths
$(v_1,\dots,v_n)$ as $v_1v_n$-paths. A path is \emph{Hamiltonian} in $G$ if it
contains all vertices of $G$. An undirected graph is \emph{connected} if, for
every two vertices $u,v\in V$, there is a path connecting $u$ and $v$. A
di-graph is \emph{connected} if its underlying undirected graph is
connected. A connected component of $G$ is a maximal induced subgraph that is
connected. A vertex $v$ is a \emph{cut-vertex} in a graph $G$ if
$G[V(G)\setminus\{v\}]$ consists of more connected components than $G$.
Similarly, a directed or undirected edge $(u,v)$ is a \emph{cut edge} in $G$ if
the graph $G'$ with vertex set $V(G')=V(G)$ and edge set
$E(G')=E(G)\setminus\{(u,v)\}$ consists of more connected components than $G$.

\paragraph{\bf Directed acyclic graphs, rooted phylogenetic networks, and rooted phylogenetic trees}
A di-graph $G=(V,E)$ is \emph{acyclic} (a \emph{DAG} for short) if it does not contain a directed
cycle. Note, however, that the underlying undirected graph of a DAG may contain
(undirected) cycles.
In a DAG $G$, a vertex $u\in V$ is called an
\emph{ancestor} of $v\in V$ and $v$ a \emph{descendant} of $u$, in symbols $v
\preceq_G u$, if there is a directed path (possibly reduced to a single vertex)
in $G$ from $u$ to $v$. We write $v \prec_G u$ if $v \preceq_G u$ and $u\neq v$.
If $u \preceq_G v$ or $v \preceq_G u$, then $u$ and $v$ are
\emph{$\preceq_G$-comparable} and otherwise, they are
\emph{$\preceq_G$-incomparable}. Moreover, if $(u,v)\in E$, we say that $u$ is a
\emph{parent} of $v$, $u\in\parent_{G}(v)$, and $v$ is a \emph{child} of $u$,
$v\in \child_{G}(u)$.  
An edge $(u,w)$ in a DAG $G$ is a \emph{shortcut} if there is a vertex
$v\in\child(u)\setminus\{w\}$ such that $w\prec_G v$ (or, equivalently, if there
is a vertex $v'\in V(G)$ such that $w\prec_G v'\prec_G u$).

We define phylogenetic networks here as a slightly more general class of DAGs than what is customarily considered in most of the literature on the topic.
\begin{definition}\label{def:N}
A \emph{(rooted) network} is a connected directed acyclic graph $N=(V,E)$ such that
    \begin{itemize}
		\item[(N1)] There is a unique vertex $\rho_N$, called the \emph{root} of $N$, with $\indeg(\rho_N)=0$.
	\end{itemize}
\noindent A network is \emph{phylogenetic} if 		
	\begin{itemize}		
	\item[(N2)] There is no vertex $v\in V$ with $\outdeg(v)=1$ and $\indeg(v)\le 1$. 
	\end{itemize}
A vertex $v\in V$ is a \emph{leaf} if $\indeg(v)=1$ and $\outdeg(v)=0$, a \emph{hybrid} or \emph{reticulation vertex} if $\indeg(v)>1$, and a \emph{tree vertex} if $\indeg(v)\le 1$ and $\outdeg(v) > 0$.  
The set of leaves is denoted by $L(N)$ and the set of \emph{inner} vertices is defined as $\mathring{V}(N)\coloneqq V(N)\setminus L(N)$.
\end{definition}

Note that leaves of $N$ are \emph{not} tree vertices.
Moreover, by (N2), the root of a phylogenetic network is either the single leaf or has
$\outdeg(\rho_N)\ge 2$. In contrast to the even more general definition
\cite[Definition~3]{Huson:11}, we use the term \enquote{phylogenetic} here to mean that vertices
with in-degree $1$ and out-degree $1$ do not appear. Rooted phylogenetic networks thus generalize
rooted phylogenetic trees, i.e., phylogenetic networks without any hybrid vertices. We emphasize
that all networks considered here are rooted and thus, we always use the term \enquote{network}
instead of \enquote{rooted network}. We denote with $\mscr{T}_n$ and $\mscr{N}_n$ the set of all
(isomorphism classes of not necessarily phylogenetic) trees and networks with $n$ leaves,
respectively.  

Given a network $N=(V,E)$ with leaf set $L(N)$, we use $N(v)$ to denote the subnetwork rooted in
$v$, that is, the subgraph of $N$ that is induced by the vertices in $\{w\mid w\in V, w\preceq_N
v\}$. Note that $N(v)$ is not necessarily phylogenetic, even in case $N$ is phylogenetic. For better
readability, we often write $L_N(v)$ instead of $L(N(v))$. If $x\in L_N(v)$, then $x$ is called a
\emph{leaf-descendant} of $v$.

\begin{remark}
In all drawings that include networks, we omitted drawing the directions of edges as arcs, 
that is, edges are usually drawn as simple lines. However, in all cases, 
the directions are implicitly given by directing the edges \enquote{away} from the root.
\end{remark}

\paragraph{\bf Blocks and level-$\mathbf{1}$ networks}
An undirected or directed graph is \emph{biconnected} if it contains no vertex whose removal
disconnects the graph. A \emph{block} of an undirected or a directed graph is a maximal biconnected
component. A block $B$ is called \emph{non-trivial} if it contains an (underlying undirected) cycle.
Equivalently, a block is non-trivial if it is not a single vertex or a single edge. An edge that is
at the same time a trivial block is a cut edge. We denote with $size(B)$ the \emph{size} of a block
$B$, i.e., the number of vertices in $B$. A block of size three is called a \emph{triangle}. We need
to distinguish between tree vertices that are part of blocks and those that are not. Hence, we say a
\emph{true tree vertex} is a tree vertex that is not contained in a non-trivial block. 
As discussed in Section~\ref{APPX:prelim}, each block $B$ has a unique $\preceq_N$-maximal vertex
$\rho_B$ (the \emph{root of $B$}) and every $\preceq_{N}$-minimal vertex $v$ in $B$ must be a hybrid vertex. 
The notion of $\preceq_{N}$-minimal vertices also allows us to define yet another special type of networks, namely level-$k$ networks.

\begin{cdefinition}{\ref{def:level-k-N}}
A network $N$ is \emph{level-$k$} if each block $B$ of $N$ contains at most $k$ hybrid vertices distinct from its root $\rho_B$.
\end{cdefinition}

We remark that level-$k$ networks are sometimes alternatively defined as networks that can be converted into trees by deleting at most $k$ edges from each block (see, e.g., \cite[p. 247]{steel_phylogeny_2016}). In case the network is binary, this definition coincides with Definition\ \ref{def:level-k-N}.

In the following, we focus mainly on level-$1$ networks in which there is a unique
$\preceq_N$-minimal vertex $\eta_B$ in each block $B$ (cf.\ Lemma \ref{lem:unique-min-B}). Given a
block $B$, we denote with $V^-(B)$ the set of \emph{internal} vertices of $B$, i.e., vertices that
are distinct from $\rho_B$ and that have at least one child that is located in $B$. Hence, for
level-$1$ networks $N$ we have $V^-(B)\coloneqq V(B)\setminus \{\rho_B,\eta_B\}$ for each
non-trivial block $B$ of $N$. Two paths in $B$ are \emph{internal vertex-disjoint}, if they do not
have any internal vertices (of $B)$ in common. The set $\mathcal{B}(N)$ denotes the collection of
all non-trivial blocks in a network $N$ and $\mathcal{B}_m(N)$ is the subset of blocks in
$\mathcal{B}(N)$ that have precisely $m$ vertices. Note that $m\geq 3$ whenever
$\mathcal{B}_m(N)\neq \emptyset$. Moreover, for all $v\in V(N)$, let $\mathcal{B}^v(N)\subseteq
\mathcal{B}(N)$ be the set of non-trivial blocks $B$ in $N$ for which $\rho_B=v$, i.e., the set of
those non-trivial blocks \emph{rooted in} $v$. Whenever there is no ambiguity concerning $N$, we
refer to these sets as $\mathcal{B}$, $\mathcal{B}_m$, and $\mathcal{B}^v$, respectively.

A particular role in this contribution is played by restricted versions $\partN(v)$ of the subnetwork $N(v)$ of $N$ that is rooted in $v$.
\begin{cdefinition}{\ref{def:relevant-part}}
Let $N$ be a network. If (a) $v$ is not contained in any non-trivial block of $N$ or $v = \rho_B$ for every non-trivial block $B$ in $N$ that contains $v$,
then put $\partN(v) \coloneqq N(v)$ and, otherwise,  if (b) $v$ is part of some non-trivial block $B$ but not the root of $B$, then $\partN(v)$ is the subgraph of $N$ induced by all vertices $u\prec_N v$ for which $u\not\preceq_N w$ for all $w\in B$ with $w\prec_N v$. 
\end{cdefinition}

As outlined in Section~\ref{APPX:prelim}, $\partN(v)$ is well-defined and the operator \enquote{$\partial$} is idempotent.
Definition~\ref{def:relevant-part}(b) can sloppily  be rephrased to: $\partN(v)$ is the subnetwork rooted in $v$ without any descendants of vertices $w\prec_N v$ that are also located in the block $B$ that contains $v$ but for which $v$ is not the root; see Figure~\ref{fig:partN1} and \ref{fig:relevant_neighbors} for illustrative examples.

\paragraph{\bf Subsets of children, relevant vertices, and least common ancestors}

An important subset of children of a vertex $v$ in a network $N$ is provided by the set
$\child_N^*(v)\subseteq \child_N(v)$  of all children of $v$ that are not contained
in a non-trivial block $B\in \mathcal{B}(N)$ with $v\in V(B)\setminus \{\rho_B\}$, i.e., $B$
contains $v$ but $v$ is not the root of $B$. Hence, $\child_N^*(v) = \child_N(v)$ if $v$ is not
contained in any non-trivial block or, otherwise, if $v$ is the root of every non-trivial block it
is contained in. A vertex $v$ is \emph{relevant} in $N$ if $|\child_N^*(v)|>1$. For $N$, we
denote with $\relV_N$ the set of its relevant vertices. The set $V(N)\setminus \relV_N$ comprises
all \emph{irrelevant} vertices. As we shall see, relevant vertices play a crucial role to ensure that we can
express the difference in the subsequently defined balance index of two networks in terms of local
differences in the respective networks.

A \emph{least common ancestor} (LCA) of a subset $Y\subseteq V$ in a DAG $N$ is an
$\preceq_N$-minimal vertex of $V$ that is an ancestor of all vertices in $Y$. In general DAGs $N$, a
LCA does not necessarily exist for a given vertex set. Moreover, the LCA is not unique in general.
We write $\LCA(Y)$ for the (possibly) empty set of $\preceq_{N}$-minimal ancestors of the elements
in $Y$. In a (phylogenetic) network $N$, the root $\rho_N$ is an ancestor of all vertices in $V$,
and thus a least common ancestor exists for all $Y\subseteq V$. If $\LCA(Y) = \{u\}$ consists of a
single element $u$ only, we write $\lca(Y)=u$. In other words, $\lca(Y) = u$ always implies that the
$\preceq_{N}$-minimal ancestor of the elements in $Y$ exists and is uniquely determined. We leave
$\lca(Y)$ undefined for all $Y$ with $\vert\LCA(Y)\vert\ne1$.

\begin{clemma}{\ref{L7.9HSS}}
For all level-$1$ networks $N$ and $Y\subseteq V(N)$,  $\lca_N(Y)$ is well-defined.
\end{clemma}

\section{The weighted total cophenetic index} \label{sec:wtci}

Before we can introduce the weighted total cophenetic index as a measure of network balance, we need
to formalize the term balance. Indeed, 
Bienvenu et al. \cite{Bienvenu2021} and Zhang
\cite{Zhang2022} have considered balance indices for networks, however, a precise definition of the term
\enquote{network balance index} or \enquote{network imbalance index} seems to be missing in the
literature. A \emph{network shape statistic} $\nss\colon \mscr{N}_n \to \mathbb{R}$  is a function
assigning a value to a network $N$ that solely depends on the topology of $N$ and 
is invariant under vertex or leaf labelings, edge lengths or any 
hybridization probabilities associated with hybrid nodes. 
For trees, Fischer et al. \cite{Fischer2021Book} and, independently, Lemant et al.
\cite{Lemant2022} have recently developed criteria a network shape statistic has to
satisfy to be a tree (im)balance index. 
Before recalling these criteria we provide two types of trees that play an important role
in this context.

\begin{itemize}
 \item  A \emph{star tree} is a tree $T$ such that either (i) $T$ consists of
	     precisely one vertex, or (ii) $T$ has at least two leaves and all leaves
	     of $T$ are adjacent to the root $\rho_T$ of $T$. A star tree on $n$
	     leaves is denoted by $\Tstar$.

 \item The \emph{caterpillar tree} $\Tcat$, is a tree on $n$ leaves such that
       each inner vertex has exactly two children and the subgraph induced by
       the inner vertices is a path with the root $\rho_T$ at one end of this
       path. 
\end{itemize}

\begin{definition} \label{def-networkbal}
    A network shape statistic $\nss\colon \mscr{T}_n \to \mathbb{R}$, where the domain is 
    restricted to trees, is called a \emph{tree balance index} if and only if (i) the caterpillar
    tree $\Tcat$ is the unique tree minimizing $\nss$ across all trees in $\mscr{T}_n$ and the
    star tree $\Tstar$ is the unique tree maximizing $\nss$ across all trees in $\mscr{T}_n$.
    A tree \emph{imbalance index} is defined
    analogously with the extremal trees swapped. 
\end{definition}

The possibly most generic generalization of Definition\ \ref{def-networkbal} is as follows.

\begin{definition}\label{def:network-bal-index} \leavevmode
	 A network shape statistic $\nss$ is called a \emph{network balance index} if and only if 
    \begin{enumerate}[(i)]
        \item $\nss$ is a tree balance index when the domain of $\nss$ is restricted to $\mscr{T}_n$, i.e., when only trees are considered.
        
        \item  The tree $\Tstar$ is the network maximizing $\nss$ across all networks in $\mscr{N}_n$.
    \end{enumerate}

    A network \emph{imbalance index} is defined
    analogously assuming that $\Tstar$ is the network in $\mscr{N}_n$ minimizing $\nss$.
\end{definition}

We remark that we only state a condition for the most balanced network, but not for the most imbalanced one. 
The reason is as follows. For trees $T$, the number of vertices in $T$ is bounded by the number of its leaves. This, however, is not true for networks. That is, although a network may have a bounded number $n$ of leaves, 
it may have an unbounded number of vertices, see Figure\ \ref{fig:N_cres} and Section \ref{APPX:subsection:MaxN}
for more details. As a consequence, a network shape statistic and therefore any meaningful (im)balance
index on such networks cannot be bounded by the number of leaves.
The latter observation motivates the question as whether more restrictive networks
may have bounded (im)balance indices and, if so, if one can characterize the structure
of networks achieving it. As we shall see later, when restricting our attention to  
binary networks and all networks with the property that all hybrid nodes have in-degree two, 
then the network $\NC_n$ (defined in detail below) is the unique network maximizing the \emph{im}balance defined below. 
In simple words, $\NC_n$ is a reminiscence of the caterpillar tree $\Tcat$ to which an artificial root
and a shortcut have been added,  see Figure\ \ref{fig:N_cres} for a generic example.

We are now in a position to introduce the central concept of this manuscript, namely the weighted
total cophenetic index. This index was introduced by Mir et al. \cite{Mir2013} and is one of  
the most popular balance indices for phylogenetic trees and serves as the basis for the weighted version 
that we employ for networks later on.
For a rooted phylogenetic tree $T$, the \emph{total cophenetic index} $\Phi(T)$ is expressed as
\[\Phi(T) = \sum\limits_{v\in \mathring{V}(T)\setminus \{\rho_T\}} \binom{|L_T(v)|}{2},\]
where $\Phi(T) \coloneqq 0$ in case $\mathring{V}(T)\setminus  \{\rho_T\} = \emptyset$ (which is precisely the case when $T$ is a so-called star tree). 

Of course, we may directly use  \[\Phi(N)\coloneqq \sum\limits_{v\in  \mathring{V}(N)\setminus \{\rho_N\}}\binom{|L_N(v)|}{2}\] 
and obtain a network imbalance index according to 
Definition~\ref{def:network-bal-index}.
However, as shown in Figure~\ref{fig:phi-N}, several highly distinct networks then become indistinguishable
under the value $\Phi$. To overcome these issues,
we consider weighted versions $\Phi^*$ and $\Phi^{**}$ of $\Phi$. These are based on assigning
weights to the non-trivial blocks of the underlying network which, in turn, are used to assign
weights to all vertices of the network.

\paragraph{\bf Weights of blocks}

We first define the weight $\omega$ of a non-trivial block $B$ based on the size $\kappa_v = |K_v|$
of the set $K_v\coloneqq \{w\in V(B)\mid w\preceq_N v\}$ of vertices that are \enquote{below} $v$ in
$B$ (including $v$), for all vertices $v \in V(B)$.  
We now define the \emph{weight} of a non-trivial block $B$ as
\[\omega(B)\coloneqq \sum\limits_{v\in V(B)\setminus\{\rho_B\}} \binom{\kappa_v}{2} = \sum\limits_{v\in V^-(B)} \binom{\kappa_v}{2}\]
The equality $\sum_{v\in V(B)\setminus\{\rho_B\}} \binom{\kappa_v}{2} = \sum_{v\in V^-(B)} \binom{\kappa_v}{2}$ is due to the fact that all hybrids $\eta$ in $B$ are $\preceq_B$-minimal vertices and thus, satisfy 
	$\kappa_\eta = 1$, i.e.,  $\binom{\kappa_\eta}{2}= 0$.

Note that the choice of the weight $\omega(B)$ for non-trivial blocks $B$ is intentional. In fact,
as demonstrated in Section \ref{APPX:sec:wtci}, it always holds that 
\[\omega(B) = \Phi(B^*)  \coloneqq \sum\limits_{v\in \mathring{V}(B^*)\setminus\{\rho_{B^*}\}} \binom{|L_{B^*}(v)|}{2},\]
where $B^*$ is the network obtained from $B$ by adding, for all $v \in V(B)$, a new vertex $x_v$
along with the edge $(v, x_v)$. Thus, $B^*$ contains the single block $B$ to which a leaf is
attached to each vertex. 
Consequently, $\omega(B)$ is a measure of the balance of $B^*$ in terms of
$\Phi$ and it, thus,  distills the information about $B$ to its essential core structure, quantifying the
balance of $B$ without disturbing it with information about network structures \enquote{surrounding}
$B$ in $N$. In other words, any information that is dispensable to understand the structure of  $B$ is omitted, 
see Figure\ \ref{fig:omega-leaf-ext} for an illustrative example.

\begin{figure}[tbp]
\centering
\includegraphics[width = 0.65\textwidth]{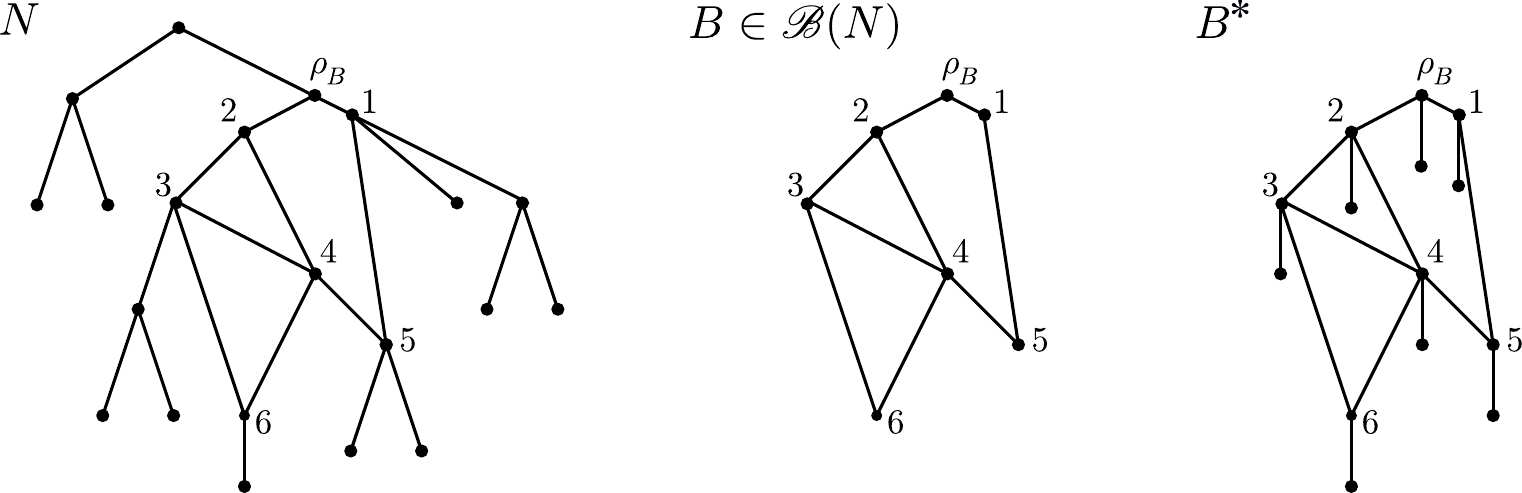}
\caption{Phylogenetic level-3 network $N$ (left) that contains the non-trivial block $B \in \mathcal{B}(N)$ (middle) whose 
				leaf extended version is $B^*$ (right).
				$V^-(B)$ consists of the vertices $1,2,3$ and $4$. Hence, $\omega(B) = 
				\sum_{i=1}^4 \binom{\kappa_i}{2}  = \binom{2}{2} + \binom{5}{2} + \binom{4}{2} + \binom{3}{2}$.
				One easily observes that $\kappa_i = |L_{B^*}(i)|$, $1\leq i \leq 6$ and that 
				$\omega(B) = \Phi(B^*)$.
				} 
				
\label{fig:omega-leaf-ext}
\end{figure}

\paragraph{\bf Weights of vertices}

We now proceed with defining weights $\phi$ for vertices $v\in V(N)$. Let $0\leq \epsilon\leq 2$. Then, we define:
\[
\phi(v) = \phi_N(v) \coloneqq 
\begin{cases}
	1  & \text{, if }  \mathcal{B}^v(N) = \emptyset  \\
	\epsilon + \sum_{B\in \mathcal{B}^v(N)} 	\omega(B)   &\text{, otherwise.}
\end{cases}
\]

\noindent If the context is clear we write $\phi(v)$ rather than $\phi_N(v)$.  Hence, if $v$ is not the root of any non-trivial block, then $\phi(v)=1$ and, otherwise $\phi(v)$ becomes the sum of $\epsilon$ and the sum of the weights $\omega(B)$ of non-trivial blocks $B$ for which $\rho_B=v$. Note, however, that by Observation~\ref{obs:omega}, $\omega(B)=1$ may be possible in case that $B$ is a triangle.
In particular, if $\mathcal{B}^v(N) = \{B\}$ and $B$ is a triangle, then $\sum_{B\in \mathcal{B}^v(N)}\omega(B)  =1$ which is the same value as assigned e.g.,  to true tree vertices $w$ (where trivially $\mathcal{B}^w(N) = \emptyset$ holds).
To distinguish  between the weights of roots of such triangle blocks and 
those vertices $w$ for which $\mathcal{B}^w(N) = \emptyset$, we employ the extra constant $\epsilon$. Note that by the previous argument, setting $\epsilon = 0$, will lead to the same weight for true tree vertices and tree vertices that are roots of triangles, whereas tree vertices that are roots of larger
or more than two non-trivial blocks  will receive a larger weight. As far as the upper bound, $\epsilon=2$, is concerned, we choose this number primarily for technical reasons (in particular, for the proof of Lemma \ref{lem:max_is_binary} and \ref{lem:exactly_one_block}) as well as to not over-penalize tree vertices contained in non-trivial blocks.

\paragraph{\bf The weighted total cophenetic index}

\begin{figure}[tbp]
\centering
\includegraphics[width = 1.\textwidth]{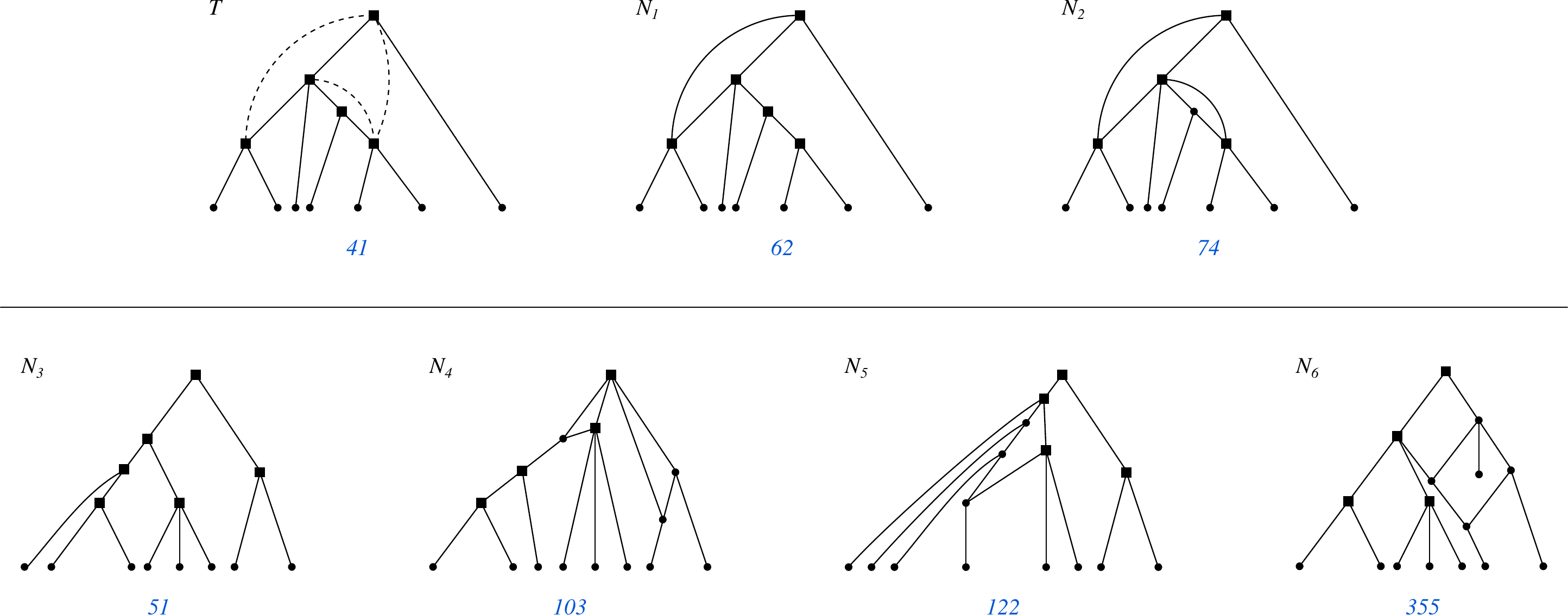}
\caption{Shown are several networks together with their respective values $\Phi^{**}(\cdot)$ written below the networks where we have chosen $\epsilon=1$.
							Relevant vertices are highlighted by $\blacksquare$.\newline
\emph{Upper Panel:}  Independent of which shortcuts (dashed edges) are added to $T$, the resulting networks $N_i$ satisfy $\Phi(N_i)=\Phi(T)=20$, $i\in \{1,2\}$. 
					  Moreover, if only shortcuts that are adjacent to the root are added to obtain $N$, then $\Phi^*(N)=\Phi^*(T)$ is possible. By way of example, for $T$ and $N_1$ we have  
					  	$\Phi^*(T)=20=\Phi^*(N_1)$.  \newline
\emph{Lower Panel:}   Four networks with 8 leaves for which we have $\Phi(N_i) = \binom{6}{2}+2\binom{3}{2}+2\binom{2}{2} = 23$ for all $i\in \{3,4,5,6\}$. 
							Furthermore, $\Phi^*(N_4)= \Phi^*(N_6) = \binom{6}{2}+\binom{3}{2}+\binom{2}{2} = 19$. 
							In the latter case, the structure of 
							non-trivial blocks rooted at the root of $N_4$, resp., $N_6$ is not taken into account. 
							To better distinguish between networks that may contain such blocks, we
							we use $\Phi^{**}(N_i) =  \Phi^{*}(N_i) + \phi_{N_i}(\rho_{N_i})\binom{|L(N_i)|}{2}$.
							}

\label{fig:phi-N}
\end{figure}

We are now in the position to generalize the total cophenetic index $\Phi$ from trees to networks by
defining a weighted version of it. To this end, we consider relevant vertices in $\relV_N$ and their
particular weights. Note that we only consider relevant vertices here as information about the
irrelevant vertices is already accounted for in the weights $\omega(B)$ of non-trivial blocks.

\begin{cdefinition}{\ref{def:wpci}}
The \emph{weighted partial (total) cophenetic index (wPCI) $\Phi^*(N)$} 
of a network $N\in \mscr{N}_n$ is defined as

\[ \Phi^*(N) \coloneqq 
\begin{cases}
 \sum\limits_{v\in \relV_N\setminus\{\rho_N\}} \phi(v) \binom{|L_N(v)|}{2} 
                    & \text{, if  $\relV_N\setminus\{\rho_N\}\neq \emptyset$}\\
 \hfill 0 & \text{, otherwise.}
\end{cases}
\]

\noindent
The \emph{weighted (total) cophenetic index (wCI) $\Phi^{**}(N)$} is then defined as
\[\Phi^{**}(N) := \Phi^*(N) + \phi(\rho_N) \binom{n}{2}.\]
\end{cdefinition}

By convention, we assume that $\binom{1}{2}\coloneqq 0$.
As we shall see later, the distinction between the parts of the sums $\Phi^{**}(N)$ into $\Phi^{*}(N)$ and $\phi(\rho_N) \binom{n}{2}$ will not only be handy in upcoming proofs but is crucial for rigorously stating certain properties of the wCI such as locality (see Theorem~\ref{thm:local-phi}).

\begin{figure}[t!]
\centering
\includegraphics[width = .95\textwidth]{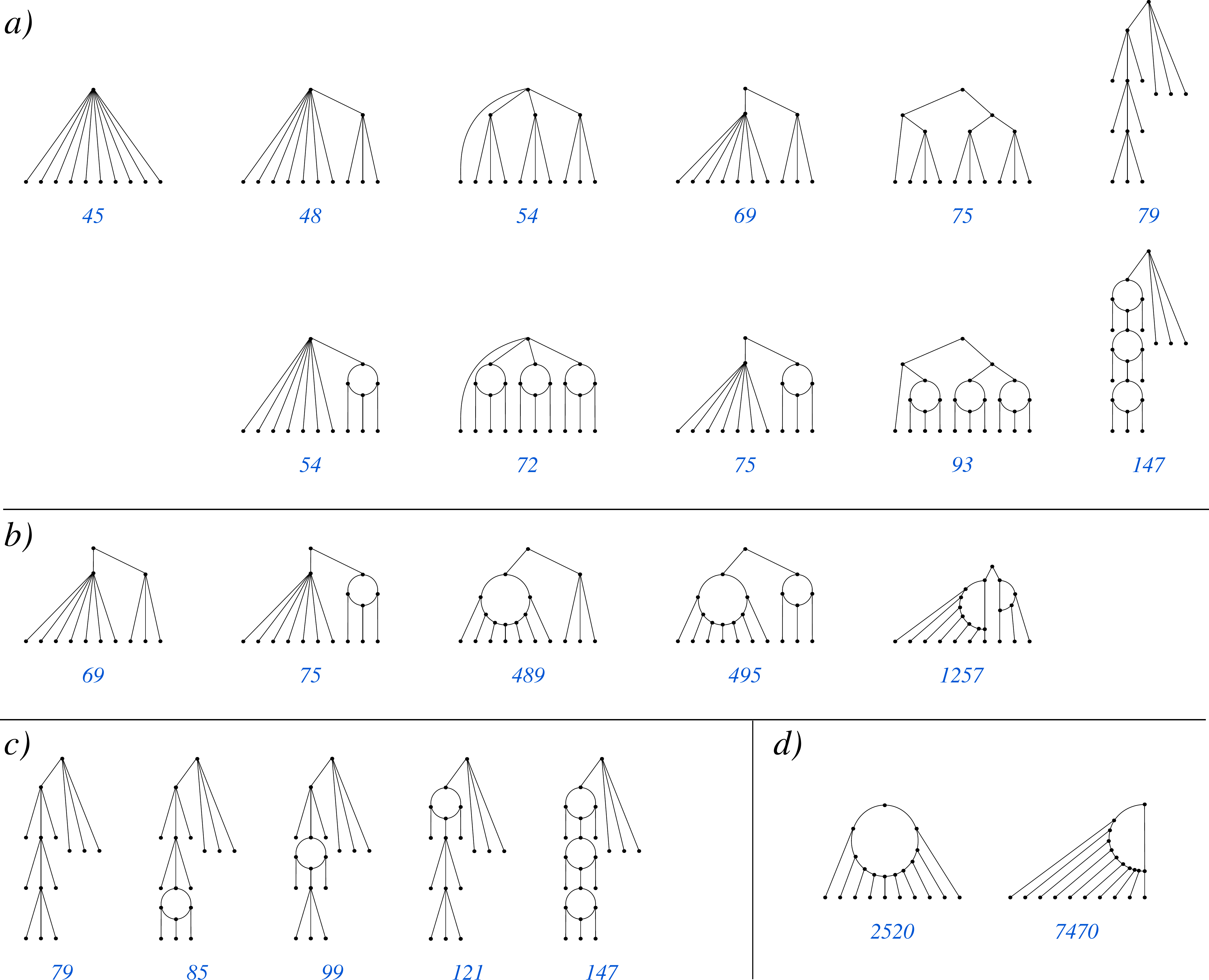}
\caption{Shown are several level-$1$ networks $N$ on $n=10$ leaves together with their respective values $\Phi^{**}(N)$ (written directly below the networks), where we chose $\epsilon=1$. 
\emph{Panel~a).} Several trees (top) together with networks (drawn below the respective trees) obtained from these trees by replacing certain vertices by a fixed block of size $4$.
							Here, the root $\rho_B$ of each block $B$ has the same weight $\phi(\rho_B) = \omega(B) + \epsilon = 2\binom{2}{2} +1 = 3$. 
							One easily observes that the order of $\Phi^{**}$-values of the networks in the 2nd row and the $\Phi$-values $\Phi(T) = \Phi^{**}(T) - \binom{10}{2} = \Phi^{**}(T) - 45$ of the trees $T$ in the 1st row coincide. 
							The tie for trees and networks with  $\Phi^{**}$-value $54$, resp., $75$ could be resolved by using a larger $\epsilon$-value (e.g.\ $\epsilon=1.5$). 
	\emph{Panel~b).} Several networks were obtained from the tree (left) by replacing certain vertices by blocks of size $4$ or $8$. First, one can observe that 
							the $\Phi^{**}$-value of each network increases with the inclusion of more blocks as well as with the inclusion of blocks $B$ with higher weight $\omega(B)$.
							The two blocks in the rightmost network are crescents (defined below) and, as we shall see later, thus have the maximum weight among all blocks of size $4$ and $8$, respectively. 
							Hence, $\Phi^{**}$ reaches its maximum among all networks that were created in this way with this network.
		\emph{Panel~c).} Several networks were obtained from the tree (left) by replacing certain vertices by a fixed block of size $4$. 
							Again,	$\phi(\rho_B) = \omega(B) + \epsilon =3$ for the root $\rho_B$ of every block $B$. Moreover, 
							 we again observe that the $\Phi^{**}$-value of each network increases with the inclusion of more blocks. At the same time one can 
							 observe that $\Phi^{**}$-values increase whenever a block $B$ is closer to the root, in which case more leaves are located
							 below $\rho_B$ which increases the influence of the value $\phi(\rho_B)$. 
	  \emph{Panel~d).} Shown are two networks with a single block $B$ where $B$ is a full-moon (left) and a crescent (right). Both are formally defined below. }
							
\label{fig:partN}
\end{figure}

Note, if $T$ is a phylogenetic tree, then $\child_T^*(v)= \child_T(v)$ for all $v\in V(T)$ and thus, $\relV_T = V(T)\setminus L(N)$, i.e., all inner vertices are relevant. Moreover, $\mc{B}^v = \emptyset$ for all $v\in V(T)$ and thus, all inner vertices $v$ of $T$ have weight $\phi(v)=1$. 
As discussed in more detail in Section \ref{APPX:sec:wtci}, the following conditions hold. 
\begin{itemize}
	\item For all phylogenetic trees $T\in \mscr{T}_n$, it holds that $\Phi^{*}(T) = \Phi(T)$ and $\Phi^{**}(T) = \Phi(T) + \binom{n}{2}$. 
	\item For all phylogenetic networks $N$, it holds that  $\phi(v)\geq 1$ for all non-leaf vertices of $N$ and, in particular, 
		\[\Phi^{**}(N) =\sum\limits_{v\in \relV_N} \phi(v) \binom{|L_N(v)|}{2} \geq   \binom{|L(N)|}{2}.\]

\end{itemize}

To gain some intuition into the index $\Phi^{**}$, we refer to Figure~\ref{fig:partN}, where several
networks $N$, accompanied by their respective $\Phi^{**}(N)$-values (for $\epsilon=1$), are shown; see the
caption for more details. Figure~\ref{fig:partN}(a) presents examples where the order of the
$\Phi^{**}$-values of the networks in the 2nd row matches the order of $\Phi$-values of their
corresponding trees in the 1st row. The networks in the 2nd row are obtained from the trees in the
1st row by replacing all vertices with exactly three children by a fixed block size of four. This
highlights, to some extent, the correspondence between $\Phi$-values in trees and $\Phi^{**}$-values
in networks. Figure~\ref{fig:partN}(b) demonstrates the effect of the number, size, and internal
structure of non-trivial blocks. In a nutshell, the $\Phi^{**}$-values increase with a higher number
of non-trivial blocks and also with increasing size as well as increasing $\omega$-values of
the respective blocks. The latter observation is also reflected by the examples in
Figure~\ref{fig:partN}(c). Here, the additional effect of the distance of blocks to the root of the
networks is highlighted.

An appealing property of the wCI is that it can be computed efficiently. 
\begin{cproposition}{\ref{prop:algo}}
	The weighted total cophenetic index can be computed in linear-time for
   level-$1$ networks and in  $O(|V|^2+|V||E|)$  time for general networks $N=(V,E)$.
\end{cproposition}

\section{Locality and recursiveness of the weighted total cophenetic index} \label{sec:otherproperties}
In order to investigate the structure of phylogenetic level-$1$ networks that maximize or minimize the weighted total cophenetic index $\Phi^{**}$, it will be helpful to have a mechanism that \enquote{controls} the replacement of local networks. In particular, as we shall see, $\Phi^{**}$ is local in the sense that if two phylogenetic level-$1$ networks $N_1$ and $N_2$ differ only in some well-defined local subnetworks $N'_1$ and $N'_2$, then the difference $\Phi^{**}(N_1)-\Phi^{**}(N_2)$ can be expressed in terms of the differences of $\Phi^*(N'_1)$ and $\Phi^*(N'_2)$, and the weights of the roots of $N'_1$ and $N'_2$ only. Moreover, $\Phi^{**}$ can always be expressed as a sum of $\Phi^*(N')$ for networks $N'\subsetneq N$ rooted at well-chosen relevant vertices in $N$. These two properties are often referred to as \enquote{locality} and \enquote{recursiveness} in the literature, and are indeed satisfied by the total cophenetic index on trees \cite[Lemma 4, Lemma 5]{Mir2013}. 
To provide the main results of this section, we first  introduce the notion of relevant neighbors and relevant-vertex-free paths.

\begin{cdefinition}{\ref{def:IRfree}}
A $uv$-path in $N$ whose internal vertices (i.e., vertices that are distinct from $u$ and $v$) are not relevant is \emph{internal relevant-vertex-free (IR-free)}.
\end{cdefinition}

Note that every edge $(u,v)$ is trivially an IR-free path.
Moreover, if $v$ is an inner vertex in $N$ that is not contained in a non-trivial block of $N$,
then its children are either leaves, true tree vertices, or roots of non-trivial blocks. 
Hence, all non-leaf children of $v$ are relevant. The latter can be rephrased to:
any relevant vertex $u\prec_N v$ for which there is no
relevant vertex $u'\prec_N v$ with $u\prec_N u'$ must be a child of $v$.

\begin{cdefinition}{\ref{def:rel-neighbor}}
Let $v,u\in V(N)$ be such that $u\prec_N v$ and let $\mathcal{B}$ be the set of non-trivial blocks of $N$ that contain $v$.
 Then, $u$ is a \emph{relevant neighbor of $v$} if $u$ is relevant and precisely one of the following conditions is satisfied:
\begin{enumerate}
    \item $v$ is not contained in a non-trivial block of $N$ and $u$ is child of $v$
    \item $v$ is contained in a non-trivial block of $N$ and precisely one condition holds 
        \begin{enumerate}
            \item $u$ is contained in some $B\in \mathcal{B}$ as well; or 
            \item $u$ is not contained in any $B \in \mathcal{B}$ and
            			there is an IR-free $wu$-path such that $w\in B$ for some $B\in \mathcal{B}$ and $w$ is not relevant in $N$ and $w\prec_N v$. 
            \item $u$ is not contained in any $B \in \mathcal{B}$ and $u$ is a child of $v$. 
        \end{enumerate}
    \end{enumerate}
The \emph{relevant neighborhood $\relN(v)$ of $v$} is the set of its relevant neighbors (see Figure \ref{fig:relevant_neighbors} for an example).
\end{cdefinition}
Note that $u$ being a relevant neighbor of $v$ does not necessarily imply that $u$ and $v$ are
adjacent. Moreover, observe that Conditions (a), (b) and (c) in Definition\ \ref{def:rel-neighbor}(2) are mutually exclusive, since Condition (a) 
does not cover any of the properties in Condition (b) and (c) while 
Condition (b) covers the case that $u$ and $v$ are not adjacent and (c) the case that $u$ and $v$ are adjacent.

We are now in the position to provide the two main results (Theorem\ \ref{thm:local-phi} and
\ref{thm:sum-up-part}) of this section. While Theorem\ \ref{thm:local-phi} ensures locality,
Theorem\ \ref{thm:sum-up-part} ensures recursiveness of the weighted total cophenetic index.

\begin{ctheorem}{\ref{thm:local-phi}}
Let $N \in \mathscr{N}_n$ be a phylogenetic level-$1$ network and let $v$ be a relevant vertex in $N$. Moreover, let $N'$ be the phylogenetic level-$1$ network obtained from $N$ by replacing $\partN(v)$ by a phylogenetic level-$1$ network $\tilde N$ with $L(\tilde N) = L(\partN(v))$. Then, \[ \Phi^{**}(N)-\Phi^{**}(N') =  \left(\Phi^{*}(\partN(v))-\Phi^{*}(\tilde N)\right) + \left((\phi_N(v)-\phi_{N'}(\rho_{\tilde N})) \binom{|L(v)|}{2}\right), \]
where $L(v) \coloneqq L_N(v)=L_{N'}(v)$.
\end{ctheorem}

For an illustrative example of Theorem\ \ref{thm:local-phi} consider the networks $N$  and $N'$
    as shown in Figure\ \ref{fig:exmpl-local}. 
   One easily observes that  
   $N'$ is obtained from $N$ by replacing $\partN(v) = N(v)$ with the network $\tilde N\coloneqq  N'(v)$.
    In this example, $\Phi^{**}(N) = 75$ and $\Phi^{**}(N') = 69$ for $\epsilon=1$.
    It holds that $\Phi^{*}(\partN(v)) = \Phi^{*}(\tilde N) = 0$, $\phi_N(v) = 2\cdot\binom{2}{2}+\epsilon =3$ and $\phi_{N'}(v) = 1$.
    According to Theorem\ \ref{thm:local-phi}, 
    \[\Phi^{**}(N)-\Phi^{**}(N') =  \left(\Phi^{*}(\partN(v))-\Phi^{*}(\tilde N)\right) + \left((\phi_N(v)-\phi_{N'}(v)) \binom{|L(v)|}{2}\right)
     = 0 + (3-1)\binom{3}{2} = 6  = 75-69.\]

\begin{ctheorem}{\ref{thm:sum-up-part}}
Let $N\in \mathscr{N}_n$ be a phylogenetic level-$1$ network and let $v_1,\dots,v_k$ be the relevant neighbors of the root $\rho_N$ in $N$. Let $\phi(\rho_N)\coloneqq \phi_N(\rho_N)$, $\phi(v_i) \coloneqq \phi_N(v_i)$, $n_i\coloneqq |L_N(v_i)|$, and $\partN_i\coloneqq \partN(v_i)$, $1\leq i\leq k$. Then, \[\Phi^{**}(N) = 
\sum_{i=1}^k \Phi^{**}(\partN_i) + \phi(\rho_N)\binom{n}{2} =\sum_{i=1}^k \left(\Phi^*(\partN_i) + \phi(v_i)\binom{n_i}{2}\right) + \phi(\rho_N)\binom{n}{2}. \]
\end{ctheorem}

To illustrate Theorem\ \ref{thm:sum-up-part}, consider the network $N$ in Figure\ \ref{fig:exmpl-local} with relevant vertices $v,w,\rho_N$ highlighted by $\blacksquare$.
	The relevant neighbors of $\rho_N$ are, thus, $v$ and $w$.  
	 In this example, it holds that  $\partN(w)=N(w)$,  $\partN(v) = N(v)$, $\phi_N(v) = 3$  and $\phi_N(\rho_N) = \phi_N(w)=1$.  
	 According to Theorem\ \ref{thm:sum-up-part}, we have
	 \[ \Phi^{**}(N) = \Phi^{**}(\partN(w)) +  \Phi^{**}(\partN(v)) + \phi(\rho_N)\binom{n}{2} = 1\cdot \binom{7}{2}+3\cdot \binom{3}{2} + 1\cdot \binom{10}{2} = 21 + 9 + 45 = 75.\]

\begin{figure}[t]
\centering
\includegraphics[width = .4\textwidth]{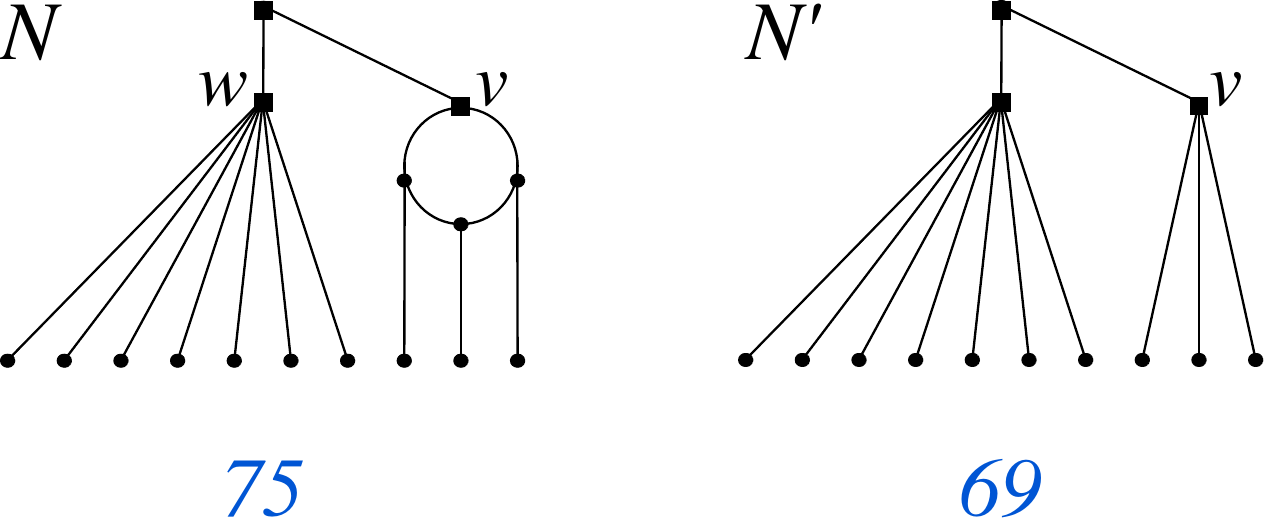}
\caption{Shown are two level-$1$ networks $N$ and $N'$ on $n=10$ leaves together with their respective $\Phi^{**}$-values written directly below the networks, where we chose $\epsilon=1$.
			All relevant vertices are highlighted by $\blacksquare$.}
\label{fig:exmpl-local}
\end{figure}

\section{Structure of blocks with extremal weights}
\label{sec:structure-blocks}

The index $\Phi^{**}$ of general level-$1$ phylogenetic networks involves the value $\omega(B)$ of non-trivial blocks $B$. Hence, to understand $\Phi^{**}$, we investigate here 
the structure of blocks that minimize and maximize $\omega$. These results are, in particular, required to determine the structure of 
extremal networks and their $\Phi^{**}$-values in Section \ref{sec:extrema}.

\begin{itemize}
 \item A block $B \in \mathcal{B}_m(N)$ is a \emph{lantern}, if $B$ consists of $m-2$ internally vertex-disjoint $\rho_B\eta_B$-paths  and every such path contains precisely one vertex of $V^{-}(B)$. The edge $(\rho_B, \eta_B)$ may or may not be part of $B$.

 \item  A block $B \in \mathcal{B}_m(N)$ with $V^-(B) = \{v_1,\dots,v_{m-2}\}$ is a \emph{crescent}, if it contains a Hamiltonian path $(\rho_B, v_1,\dots,v_{m-2},\eta_B)$, a shortcut $(\rho_B,\eta_B)$, and all other edges, if any, are of the form $(v_i,\eta_B)$. 

 \item 
A block $B \in \mathcal{B}_m(N)$ is a \emph{full-moon}, if $B$ consists of precisely two internal vertex disjoint paths one containing $\lceil \frac{m-2}{2} \rceil$ internal vertices and the other $\lfloor \frac{m-2}{2} \rfloor$ internal vertices.
\end{itemize}

An illustrative example of lanterns, crescents, and full-moons is provided in Figure~\ref{fig:lantern-crescent}.

\begin{figure}[t]
	\begin{center}
		\includegraphics[width=0.9\textwidth]{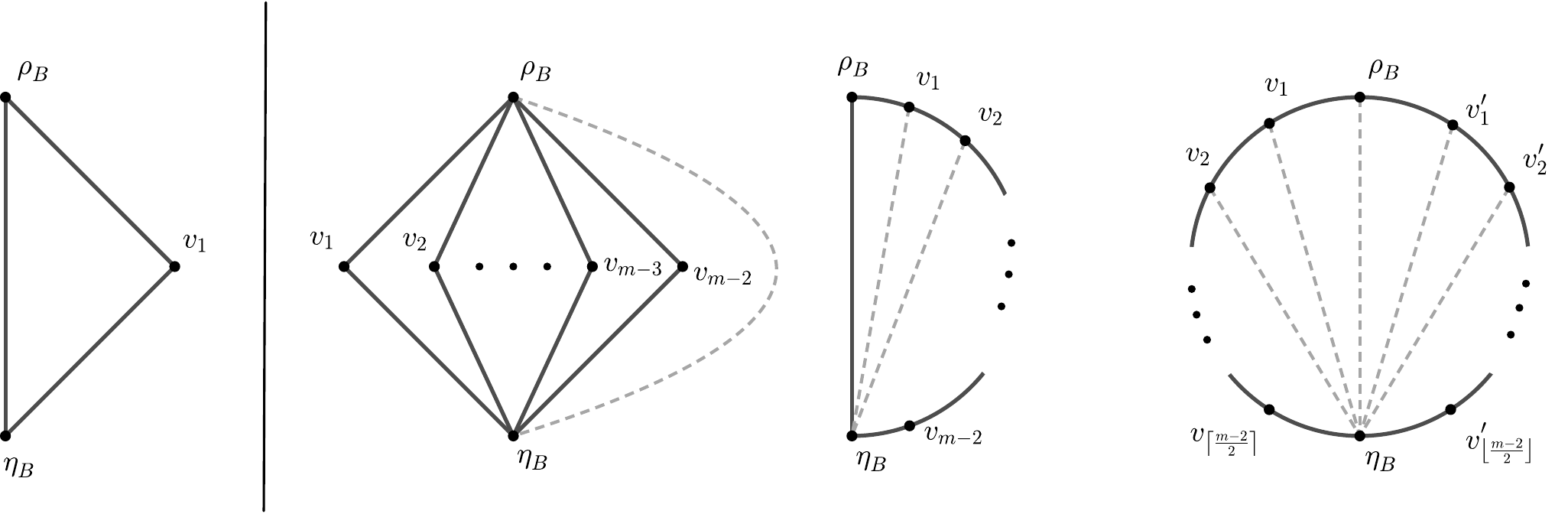}
	\end{center}
	\caption{Generic examples of lanterns (2nd-left), crescents (3rd-left), and full-moons (right).
 Solid-drawn edges must exist while dashed edges (shortcuts) may or may not exist.
  The triangle on the left is a lantern, crescent, and full-moon. 
  }
	\label{fig:lantern-crescent}
\end{figure}

\begin{cproposition}{\ref{prop:B_maxmin_arbitrary}}
Let $N$ be a phylogenetic level-$1$ network and let $B \in \mathcal{B}_m(N)$. Then, 
\begin{enumerate}[(1)]
	\item $\omega(B) \geq m-2$. This bound is tight and is, in particular, achieved if and only if $B$ is a lantern.
	\item $\omega(B) \leq \binom{m}{3}$. This bound is tight and is, in particular, achieved if and only if $B$ is a crescent.
\end{enumerate}
\end{cproposition}

We now consider non-trivial blocks of binary phylogenetic level-$1$ networks. Recall that any such block $B$ is an (undirected) cycle and contains precisely two paths from $\rho_B$ to $\eta_B$. Note that every crescent consisting only of the edges of the 
underlying Hamiltonian path and the shortcut $(\rho_B, \eta_B)$ satisfies the properties of being such an (undirected) cycle. Hence, for the maximum weight $\omega$, we can apply Proposition~\ref{prop:B_maxmin_arbitrary}(2) to obtain the 
following result.
\begin{ccorollary}{\ref{cor:crecent}}
Let $N$ be a binary phylogenetic level-$1$ network and let $B\in \mathcal{B}_m(N)$. Then, $\omega(B) \leq \binom{m}{3}$. Moreover, this bound is tight and is achieved if and only if $B$ is a crescent.
\end{ccorollary}

For the non-trivial blocks $B$ with minimum weight $\omega$ the situation becomes, however, more complicated. Observe that the number of children of $\rho_B$ in a lantern $B\in \mathcal{B}_m(N)$ is $m-2$ or $m-1$. This immediately implies that
in a binary phylogenetic level-$1$ network there cannot be a lantern except for the case $m=3$ and $m=4$. 

\begin{cproposition}{\ref{prop:B_min_binary}}
Let $N$ be a binary phylogenetic level-$1$ network and let $B\in \mathcal{B}_m(N)$. Then, 
$\omega(B) \geq \binom{\left \lceil \frac{m}{2} \right \rceil+1}{3} + \binom{\left \lfloor \frac{m}{2} \right \rfloor+1}{3}$. This bound is tight and is, in particular, achieved if and only if $B$ is a full-moon.
\end{cproposition}

\section{Phylogenetic level-1 networks with extremal indices} \label{sec:extrema}
Having characterized the structure and weights of extremal blocks in any phylogenetic level-$1$
network in the previous section, in this section we study the extremal values of the weighted total
cophenetic index across phylogenetic level-$1$ networks and characterize networks achieving these
values. Additionally, we will focus on the following sub classes of networks.

\begin{cdefinition}{\ref{def:LEVELoneDEGtwo}}
$\BinLevelOneN$ denotes the class of all binary phylogenetic level-$1$ networks on $n$ leaves
and $\HybdidDegTwoN$ the class of all phylogenetic level-$1$ networks on $n$ leaves where each vertex has in-degree at most $2$. 
\end{cdefinition}
\noindent
Note that  $\BinLevelOneN\subseteq \HybdidDegTwoN$. Considering networks within the classes $\BinLevelOneN$ and $\HybdidDegTwoN$
is, in particular, motivated from a biological point of view, where it is often assumed that hybrids have two but not more parental species. 
Networks $N\in \HybdidDegTwoN$ are also called \emph{galled-trees} \cite{HSS:22,Gusfield2003}.
Galled-trees do not only play an important role in phylogenetics but are useful as a framework to store structural information 
of so-called \textsc{GaTeX} graphs that allows to solve several computationally hard problems in linear-time \cite{HS-GatexLinTime:23,HS-GatexForbSubg:23,HS:22}. 
Galled-trees form a subclass of so-called tree-child phylogenetic network \cite{CRV07}.

A particular role is played by the following types of networks which are also illustrated in Figure\ \ref{fig:mod_and_pd_star} and \ref{fig:N_cres}.

\begin{itemize}	      

	\item A \emph{triangle network} $N_{\Delta}$ is the network that contains
	      precisely one block $B$ that is a triangle with $\rho_B = \rho_N$ and
	      with precisely two leaves and each of the two vertices distinct from
	      $\rho_B$ is adjacent to precisely one leaf.

	\item $\Tmod$ denotes the \emph{modified star} on $n$ leaves that is
	      obtained from the star tree $\mathrm{T}^{star}_{n-2}$ and $N_{\Delta}$
	      by identifying their roots.

	\item  $\Tpush$ denotes the \emph{push-down star} on $n$ leaves that is
	       obtained from the star tree $\mathrm{T}^{star}_{n-2}$ and $N_{\Delta}$
	       by adding the edge $(\rho_{\mathrm{T}^{star}_{n-2}}, \rho_N)$. If
	       $n=2$, then both $\Tmod$ and $\Tpush$ are isomorphic to $N_{\Delta}$.
	       
	\item $\NC_n$ is the  binary phylogenetic level-$1$ network with $n$ leaves that consists of one non-trivial block $B\in \mathcal{B}_{n+1}$ 
            with root $\rho_B = \rho_N$
             that is a crescent where $|\child^*_{\NC_n}(v)|=1$ for all $v\in V(B)\setminus\{\rho_N\}$ (see Figure\ \ref{fig:N_cres} for a generic example).       
\end{itemize}

\subsection{Minimum networks} 
\label{subsection:MinN}

We begin by considering the minimum value of the weighted total cophenetic index both across
arbitrary level-$1$ networks as well as binary level-$1$ networks. As we shall see, the minimum
value and the networks achieving it, depend on the choice of $\epsilon$ used when assigning weights
to vertices, reflecting the fact that different values of $\epsilon$ treat true tree vertices and
tree vertices contained in blocks differently. 

We first consider arbitrary level-$1$ networks.

\begin{ctheorem}{\ref{thm:minimum-L1}}
Let $N$ be a  phylogenetic level-$1$ network on $n\geq 2$ leaves. Then $N$ minimizes the weighted total cophenetic index within the class
of phylogenetic level-$1$ networks on $n$ leaves if and only if $\Phi^{**}(N)=\binom{n}{2}$. Moreover, 
if $\epsilon >0$, then $N\simeq \Tstar$ and, otherwise, if $\epsilon=0$, then $N\simeq \Tstar$ or $N\simeq \Tmod$. 
\end{ctheorem}

\begin{remark}
    Theorem~\ref{thm:minimum-L1} together with Observation~\ref{obs:cp} and Proposition~\ref{prop:cp-properties} (Section~\ref{APPX:sec:wtci}) imply that the weighted total cophenetic index is a network imbalance index according to Definition~\ref{def:network-bal-index}.
\end{remark}

Additionally, note that the star tree $\Tstar$ and the modified star $\Tmod$  with $n\geq2$ leaves are both contained in $\HybdidDegTwoN$.
This together with Theorem~\ref{thm:minimum-L1} implies 
\begin{ccorollary}{\ref{cor:min-hybrid-two}}
Let $N\in \HybdidDegTwoN$ be a phylogenetic level-$1$ network with $n\geq 2$  leaves.
Then, $N$ minimizes the weighted total cophenetic index within the class $\HybdidDegTwoN$
if and only if $\Phi^{**}(N)=\binom{n}{2}$. Moreover, 
if $\epsilon >0$, then $N\simeq \Tstar$ and, otherwise, if $\epsilon=0$, then $N\simeq \Tstar$ or $N\simeq \Tmod$.
\end{ccorollary}

To characterize level-$1$ networks with minimum weighted total cophenetic index
among those networks that contain at least one non-trivial block (and thus, level-$1$ networks that are not trees)
we first provide 

\begin{cproposition}{\ref{prop:min-withBlock-basic}}
		Let $\hat{\mathscr{N}}_n$ be the class of phylogenetic level-$1$ networks on $n \geq 2$ leaves  that contain at least one non-trivial block. 
	  If $N\in \hat{\mathscr{N}}_n$ minimizes the weighted total cophenetic index within the class $\hat{\mathscr{N}}_n$, 
		then $N\simeq\Tmod$ or $N\simeq\Tpush$. 
\end{cproposition}

\begin{ctheorem}{\ref{thm:minimum-L1-block}}
Let $\hat{\mathscr{N}}_n$ be the class of phylogenetic level-$1$ networks on  $n \geq 2$ leaves  that contain at least one non-trivial block. 
	  Moreover,  let $N\in \hat{\mathscr{N}}_n$ be a network that minimizes the weighted total cophenetic index within the class $\hat{\mathscr{N}}_n$.	
	  \begin{itemize}
	   \item If $n=2$, then $N\simeq N_{\Delta}$ and $\Phi^{**}(N) =1+\epsilon$ for all $\epsilon\geq 0$. 
	   \item If $n\geq 3$, 
	  then
       	 \begin{itemize}
            \item[$\circ$] $0\leq  \epsilon<\frac{1}{\binom{n}{2}-1}$ implies $N\simeq \Tmod$ and $\Phi^{**}(N)=\binom{n}{2}(1+\epsilon)$.  
            
            \item[$\circ$] $\epsilon>\frac{1}{\binom{n}{2}-1}$ implies $N\simeq \Tpush$ and $\Phi^{**}(N) =  \binom{n}{2}+1+\epsilon$. 
            
            \item[$\circ$]  $\epsilon=\frac{1}{\binom{n}{2}-1}$ implies $N\simeq \Tmod$ or $N\simeq \Tpush$ and $\Phi^{**}(N) = \frac{\binom{n}{2}^2}{\binom{n}{2}-1} = 
            \frac{n^2(n-1)^2}{2 (n-2)(n+1)}$.
        \end{itemize}
        \end{itemize}
\end{ctheorem}

Note that $N\simeq N_{\Delta}$ for all $N\in \HybdidDegTwoNvar{2}$ that contain at least one non-trivial block and that the modified star $\Tmod$ and the push-down star $\Tpush$ are both contained in $\HybdidDegTwoN$ for $n\geq3$. This together with Proposition~\ref{prop:min-withBlock-basic} and Theorem~\ref{thm:minimum-L1-block} implies

\begin{ccorollary}{\ref{cor:bin-withB}}
    Let $\hat{\mscr{N}}_n^{\mathrm{L1,in}\leq 2}$ be the class of phylogenetic level-$1$ networks on $n \geq 2$ leaves that contain at least one non-trivial block and whose vertices have at most in-degree 2. 
	  Moreover,  let $N\in \hat{\mscr{N}}_n^{\mathrm{L1,in}\leq 2}$ be a network that minimizes the weighted total cophenetic index within the class $\hat{\mscr{N}}_n^{\mathrm{L1,in}\leq 2}$.	
	  \begin{itemize}
	   \item If $n=2$, then $N\simeq N_{\Delta}$ and $\Phi^{**}(N) =1+\epsilon$ for all $\epsilon\geq 0$. 
	   \item If $n\geq 3$, 	  then
       	 \begin{itemize}
            \item[$\circ$] $0\leq  \epsilon<\frac{1}{\binom{n}{2}-1}$ implies $N\simeq \Tmod$ and $\Phi^{**}(N)=\binom{n}{2}(1+\epsilon)$.  
            
            \item[$\circ$] $\epsilon>\frac{1}{\binom{n}{2}-1}$ implies $N\simeq \Tpush$ and $\Phi^{**}(N) =  \binom{n}{2}+1+\epsilon$. 
            
            \item[$\circ$]  $\epsilon=\frac{1}{\binom{n}{2}-1}$ implies $N\simeq \Tmod$ or $N\simeq \Tpush$ and $\Phi^{**}(N) = \frac{\binom{n}{2}^2}{\binom{n}{2}-1} = 
            \frac{n^2(n-1)^2}{2 (n-2)(n+1)}$.
        \end{itemize}
      \end{itemize}
\end{ccorollary}

In the following, we will focus on \emph{binary} level-$1$ networks
with minimum weighted total cophenetic index. To this end, we will
consider an operation of associating  trees with certain  level-$1$
networks via expanding vertices to triangles and collapsing triangles into
vertices. More formally, let $T$ be a rooted  tree and let $u$ be one of
its inner vertices with children $u_1$ and $u_2$. 
Then, by \emph{expanding $u$ to a triangle (along $u_1,u_2$)}, we mean the process of subdividing
the edges $(u,u_1)$ and $(u,u_2)$ with vertices $v_1$ and $v_2$, respectively,
and introducing either the edge $(v_1,v_2)$ or the edge $(v_2,v_1)$. On the
other hand, if $N$ is a level-$1$ network containing a triangle, i.e., a
block $B$ of size three with root $\rho_B$, by \emph{collapsing the triangle
$B$}, we mean deleting the (unique) hybrid edge of $B$ that is not incident to
$\rho_B$ and suppressing the two resulting degree-2 vertices. An example of both
processes is depicted in Figure \ref{Fig_ExpandingCollapsing} for the case
that $T$ or $N$ is binary and phylogenetic. 
We emphasize that distinct triangles in $N$ created in the process of expanding
inner vertices in a binary tree $T$ must be vertex-disjoint. Moreover, the degrees of
existing vertices in $T$ remain the same as in the resulting network $N$, 
while all newly created vertices have either
in-degree 1 and out-degree 2 or in-degree 2 and out-degree 1 in $N$. Consequently, the
\enquote{expansion to triangles} process results in a binary level-$1$ network when
applied to binary trees. Moreover, given a binary network, distinct blocks must
be vertex-disjoint. By similar arguments, the process of collapsing all
triangles in a binary network yields a binary tree. Furthermore, observe that
for $\epsilon=0$, obviously $\phi(v)=\phi(x)=1$, where $v$ is a true tree vertex
and $x$ is the root of a triangle in a binary phylogenetic level-$1$ network. Hence, 
we obtain 

\begin{cfact}{\ref{fact:triangle-1}}
Let $N$ be a network obtained from a binary  tree $T$ by expanding an arbitrary number of inner vertices to triangles. 
Then, $N$ is a binary phylogenetic level-$1$ network. In particular, if $N$ does not contain non-trivial blocks or 
in case $\epsilon=0$, we have $\Phi^{**}(N) = \Phi^{**}(T)$. Otherwise, if $N$ contains non-trivial blocks
and $\epsilon>0$, we have $\Phi^{**}(N) > \Phi^{**}(T)$. 
Moreover, if $T'$ is the tree obtained from  $N$ by collapsing all triangles, then $T' \simeq T$.

In addition, if $N$ is a binary phylogenetic level-$1$ network where each non-trivial block (if there are any)
is a triangle and $T$ is the tree obtained from $N$ by collapsing all triangles, then $T$ is binary and
$\Phi^{**}(N) \geq \Phi^{**}(T)$.
\end{cfact}

A tree $T$ is called \emph{maximally balanced} if each inner vertex $v$ has precisely two
children $v_1$ and $v_2$ and these two children satisfy $\lvert L_T(v_1) -
L_T(v_2) \rvert \leq 1$.
\begin{cdefinition}{\ref{def:max_bal}}
We denote with $\Nmb$ the set of networks that
are obtained from a maximally balanced tree with $n$ leaves by 
expanding an arbitrary number of inner vertices to triangles.
\end{cdefinition}

Note that each maximally balanced tree with $n$ leaves is always contained in
$\Nmb$. 
\begin{cfact}{\ref{fact:triangle}}
It holds that 
$\Nmb\subseteq \BinLevelOneN$ and all non-trivial blocks in $N\in \Nmb$ are triangles
and pairwise vertex-disjoint.
 In particular, if $N \in \Nmb$, 
then the tree $T$ from which $N$ was obtained by expanding an arbitrary number of inner vertices to triangles is a maximally balanced tree on $n$ leaves. 
\end{cfact}

\begin{figure}[tbp]
    \centering
    \includegraphics[width=0.8\textwidth]{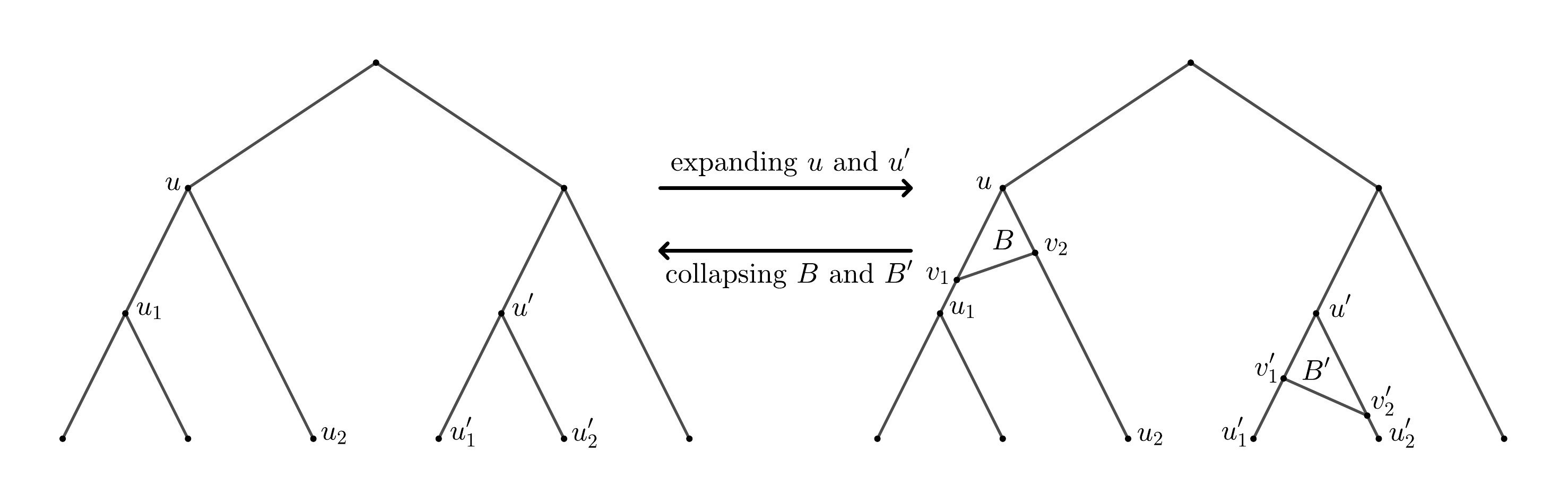}
    \caption{Tree $T$ and level-$1$ network $N$ that can be obtained from each other by expanding vertices to triangles, respectively collapsing triangles into vertices. Note that $T$ is a maximally balanced tree and, thus, the network obtained from $T$ is an element of $\Nmb$ with $n=6$.}
    \label{Fig_ExpandingCollapsing}
\end{figure}

As discussed in Section \ref{APPX:subsection:MinN}, for any $N\in \BinLevelOneN$ that 
minimizes the weighted total cophenetic index within the class $\BinLevelOneN$, 
each non-trivial block is a full-moon of size at most $4$.

Binary phylogenetic level-$1$ networks with minimum weighted total cophenetic index are characterized as follows.
\begin{ctheorem}{\ref{thm:binary_min}}
    Let $N\in \BinLevelOneN$ be a binary phylogenetic level-$1$ network with $n \geq 2$ leaves. Then, 
    \[\Phi^{**}(N)\geq \sum\limits_{k=1}^{n-1} a(k) + \binom{n}{2},\] where $a(k)$ is the highest power of 2 	that divides $k!$. 
    This bound is tight and is, in particular, achieved 
    precisely if $N\in \Nmb\subseteq  \BinLevelOneN$.
If $N\in \Nmb$ and $\epsilon>0$, then $N$ is a maximally balanced tree.
\end{ctheorem}

\subsection{Maximum networks}
\label{subsection:MaxN}
As discussed in 
Section \ref{APPX:subsection:MaxN},  the weighted total cophenetic index is unbounded on general phylogenetic level-$1$ networks
even if the number of leaves is restricted.  
Hence, the question arises whether more restrictive networks
may have a bounded weighted total cophenetic index and, if so, if one can characterize the structure
of networks achieving it. To shed some light into this question, we focus  on the  class $\BinLevelOneN$
of binary phylogenetic level-$1$ networks and the more general class $\HybdidDegTwoN$ of phylogenetic level-$1$ networks
where all hybrids have in-degree two. 
\begin{ctheorem}{\ref{thm:binary_galledtree_maximum}}
    Let $N\in \BinLevelOneN$ (resp.\ $N\in \HybdidDegTwoN$) be a  network with $n \geq 3$ leaves. Then, 
    \[\Phi^{**}(N)\leq \left(\binom{n+1}{3}+\epsilon\right)\binom{n}{2}.\] This bound is tight and is, in particular, achieved 
    for $N\in \BinLevelOneN$  (resp.\ $N\in \HybdidDegTwoN$) if and only if $N\simeq \NC_n$.
\end{ctheorem}

\section{Discussion and outlook} \label{sec:discussion}

In the recent literature, first attempts have been made to make balance indices,
which play an important role in the analysis of phylogenetic trees, available
for phylogenetic networks. On the one hand, this concerns the \emph{Sackin
index} \cite{Sackin1972}, which is one of the oldest and most widely used tree
balance indices and has recently been generalized to phylogenetic networks
\cite{Zhang2022}. The Sackin index also plays a role in theoretical computer
science, where it is known as \emph{total external path length} \cite{Knuth3}.
On the other hand, the \emph{$B_2$ index} \cite{Shao1990}, which is a lesser
known but in certain aspects very useful tree balance index
\cite{Fischer2021Book}, has also recently been analyzed in the context of
phylogenetic networks \cite{Bienvenu2021}. 
In this contribution, we generalized  one of the most common and widely used tree balance indices, 
namely the total cophenetic index \cite{Mir2013}, to networks.

Balance indices differ in many ways. In fact, for trees there are more than 20
such indices available \cite{Fischer2021Book}, all of which have their own
advantages: Some have a wider range, some lead to fewer ties, some are more
suitable to detect an underlying evolutionary process like the so-called Yule
model, and so forth. It therefore makes sense to have various indices at hand to
choose the best one for certain purposes. In case of the weighted total cophenetic index, the introduced 
real-valued parameter $\epsilon$,
that can take any value between $0$ and $2$, allows for such additional adjustments.

A thorough comparison of the three different network balance indices
available so-far is beyond the scope of the present work.
Nevertheless, to give at least some insights, we provide Figure \ref{fig:nwIndexComp1} which
shows a simple comparison of the rankings induced by the two mentioned indices from the literature as well as the
weighted total cophenetic index for $\epsilon\in (0,2)$ on all binary phylogenetic level-$1$ networks with three leaves.
Rankings along the $x$-axis of Figure \ref{fig:nwIndexComp1} are chosen such that 
a network $N$ is left from network $N'$  precisely if the respective index of $N$  is smaller than that of  $N'$. 
In case of ties, the respective networks are shown on top of each other and appear at the same $x$-position.

\begin{figure}[htbp]
	\begin{center}
	\includegraphics[width=0.95\textwidth]{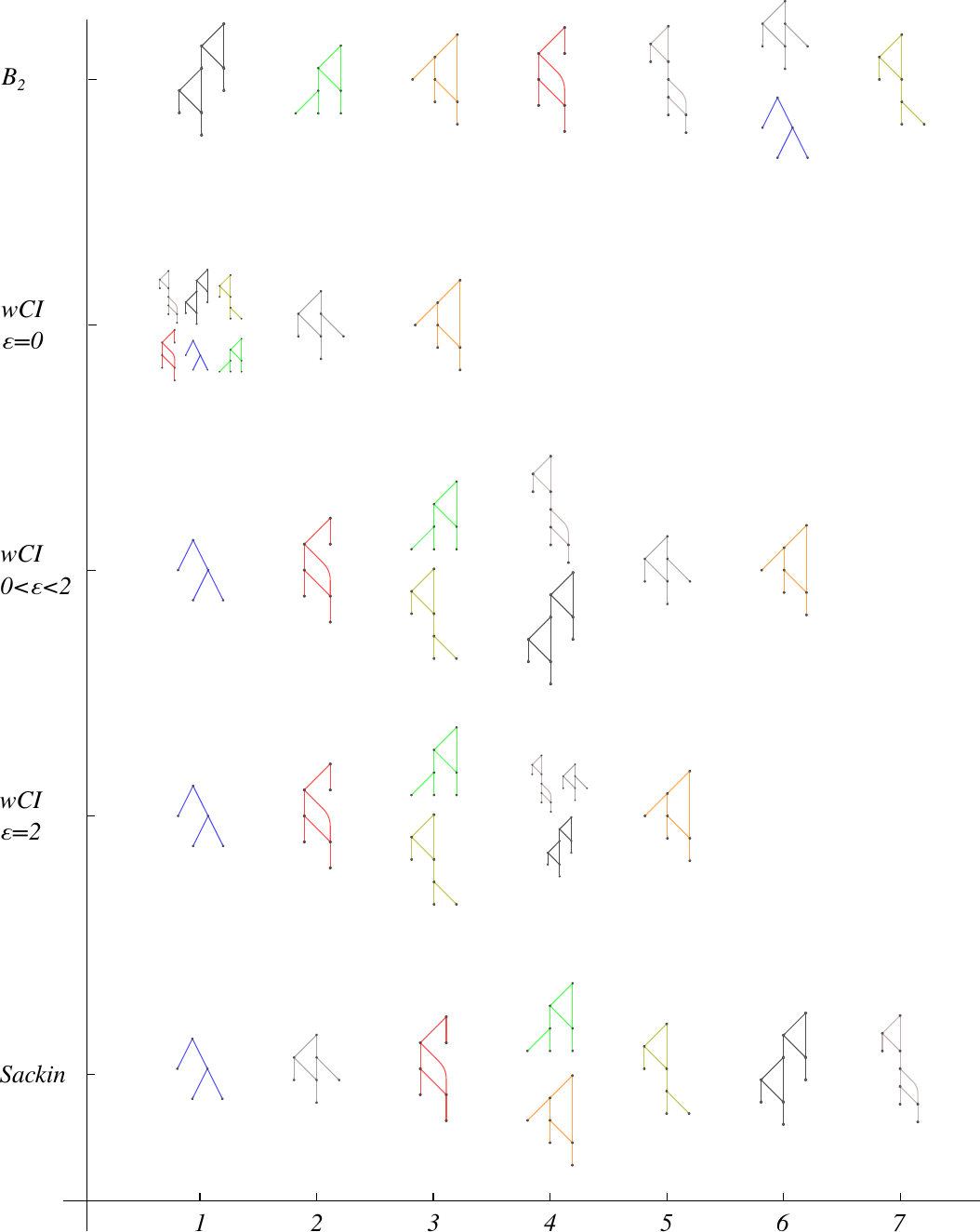}
	\end{center}
	\caption{Comparison of the rankings induced by the weighted total cophenetic index (wCI) and the Sackin index, both of which are \emph{imbalance} indices, as well as the $B_2$ index, a \emph{balance index,} for all binary phylogenetic level-$1$ networks with three leaves. For these small cases,  $\epsilon$ can be chosen arbitrarily in the open interval $(0,2)$ without changing the ranking. To resolve ties in the case  $\epsilon\in (0,2)$, one may consider an additional parameter $\epsilon'$ added for vertices whose parent is a hybrid. Note that for larger networks, the ranking may differ depending on the choice of $\epsilon\in (0,2)$; see e.g.,  caption of Figure \ref{fig:partN}.
     }
	\label{fig:nwIndexComp1}
\end{figure}

We emphasize that there is a clear difference between balance and imbalance
indices \cite{Fischer2021Book}: While the Sackin and total cophenetic indices
for phylogenetic trees increase with decreasing balance of the underlying tree
(i.e., they measure its imbalance and are thus considered as imbalance indices),
the $B_2$ index increases with increasing balance of the underlying tree, i.e.,
it is a balance index. In Figure \ref{fig:nwIndexComp1}, this is reflected by
the rankings: While for instance the only binary tree on three leaves is
considered to be the most balanced network under consideration both by the
Sackin index and our version of the total cophenetic index, the $B_2$ index considers 
this tree to be the second most balanced network, which is why it receives rank 6 
(together with another network considered equally balanced by $B_2$). In any case, it can easily be seen that
all analyzed indices are very different, and that in particular the weighted
total cophenetic index has different properties (e.g., different numbers of
ties). Which one of these indices is best
for certain biological applications is an interesting question for future
research.

Moreover, we want to point out that balance indices for networks may have other
useful properties. For instance, as Figure \ref{fig:nwIndexComp1} suggests, the
weighted total cophenetic index as well as the Sackin index seem to also take
tree-likeness into account: In the given set of networks, they both consider the
only tree as most balanced. Moreover, the Sackin index tends to rank networks
with fewer cycles lower than networks with more cycles, whereas the weighted
total cophenetic index tends to rank networks with smaller cycles lower than
those with larger ones. The $B_2$ index, on the other hand, does not appear to
have such a bias, which may or may not be desirable under certain circumstances.
However, the $B_2$ ranking assigns the network on rank 7
a similar rank as the only tree in the set (rank 6) (and even considers it
to be more balanced if -- as for trees -- a higher $B_2$-value induces a higher
degree of balance), which from an intuitive view on balance might not really be
plausible.

In summary, compared to trees, only very few balance indices are available
to-date. This work adds an index to this short list by making one of the
most famous and widely used tree balance indices available for general networks. By
allowing for different choices of $\epsilon$ ($0\leq \epsilon \leq 2$), it can
even be considered a family of indices. We have presented a thorough
investigation of the extremal properties of this novel index
for phylogenetic level-$1$ networks, both for the binary
and arbitrary cases. Moreover, we have shown that this index has other desirable
properties, such as locality and recursiveness, which also play an important
role in the area of tree balance indices \cite{Fischer2021Book}. 
It would be of interest to characterize the structure
of those level-$1$ networks that allow for ties in the weigted total cophenetic index, i.e.,
that have the same index value. We suppose that those ties can easily be resolved by adding 
an additional parameter $\epsilon'$ to hybrids or vertices whose parent is a hybrid.
Moreover, the understanding
of other types of networks (e.g.,\ level-$k$ or tree-child networks) that
maximize or minimize the weighted total cophenetic index will be of relevance. Hence, many questions concerning network balance indices are
still open and will surely inspire more research in this area.

\section*{TECHNICAL PART}
\setcounter{section}{0}
\renewcommand{\thesection}{\Alph{section}}

\section{Further Preliminaries}
\label{APPX:prelim}

\paragraph{\bf Blocks and level-$\mathbf{1}$ networks}
An undirected or directed graph is \emph{biconnected} if it contains no vertex whose removal disconnects the graph. A \emph{block} of an undirected or a directed graph is a maximal biconnected component. A block $B$ is called \emph{non-trivial} if it contains an (underlying undirected) cycle. Equivalently, a block is non-trivial if it is not a single vertex or a single edge. An edge that is at the same time a trivial block is a cut edge. We denote with $size(B)$ the \emph{size} of a block $B$, i.e., the number of vertices in $B$. A block of size three is called a \emph{triangle}. We need to distinguish between tree vertices that are part of blocks and those that are not. Hence, we say a \emph{true tree vertex} is a tree vertex that is not contained in a non-trivial block. For later reference, we state here the following observations:

\begin{fact} \label{obs:identical-block}
If biconnected components share two vertices, then their union is contained in a common block.
\end{fact}

\begin{fact}\label{obs:biConn-edge-disjoint}
If $B$ and $B'$ are distinct blocks of a di-graph, then $B$ and $B'$ are edge-disjoint.
\end{fact}

We now focus on blocks induced by a network $N$. Here, we first note that when considering the partial order $\preceq_N$ defined above, every block $B$ has a unique $\preceq_N$-maximal vertex:

\begin{lemma}[{\cite[Lemma~8]{HSS:22}}]\label{lem:unique-max-B}
Every block $B$ in a network $N$ has a unique $\preceq_N$-maximal vertex denoted by $\rho_B$ and called the \emph{root of $B$}. In particular,  $v\preceq_N \rho_B$ for all $v\in V(B)$.
\end{lemma}

Moreover, a $\preceq_{N}$-minimal vertex $v$ in a block $B$ of $N$ must be a hybrid vertex \cite{HSS:22}. The notion of $\preceq_{N}$-minimal vertices also allows us to define yet another special type of networks, namely level-$k$ networks.

\begin{definition}
A network $N$ is \emph{level-$k$} if each block $B$ of $N$ contains at most $k$ hybrid vertices distinct from its root $\rho_B$.
\label{def:level-k-N}
\end{definition}

In the following, we focus mainly on level-$1$ networks. 
Blocks in level-$1$ networks satisfy the following condition.
\begin{lemma}[{\cite[Lemma 47]{HSS:22}}]\label{lem:unique-min-B}
Every block $B$ in a level-$1$ network $N$ has a unique $\preceq_N$-minimal vertex  denoted by $\eta_B$.
In particular, if $B$ is non-trivial, then $\eta_B$ is the unique hybrid of $B$.
\end{lemma}

Given a block $B$, we denote with $V^-(B)$ the  set of \emph{internal} vertices of $B$, i.e., 
vertices that are distinct from $\rho_B$ and that have at least one child that is located in $B$. 
Hence, for  level-$1$ networks $N$ we have  $V^-(B)\coloneqq V(B)\setminus \{\rho_B,\eta_B\}$
for each non-trivial block $B$ of $N$.
Two paths in $B$ are \emph{internal vertex-disjoint}, if they do not have any internal vertices (of $B)$ in common.  The set $\mathcal{B}(N)$ denotes the collection of all non-trivial blocks in a network $N$ and $\mathcal{B}_m(N)$ is the subset of blocks in $\mathcal{B}(N)$ that have precisely $m$ vertices. Note that $m\geq 3$ whenever $\mathcal{B}_m(N)\neq \emptyset$.
Moreover, for all $v\in V(N)$, let $\mathcal{B}^v(N)\subseteq \mathcal{B}(N)$ be the set of	non-trivial blocks $B$ in $N$ for which $\rho_B=v$, i.e., the set of those non-trivial blocks \emph{rooted in} $v$. Whenever there is no ambiguity concerning $N$, we refer to these sets  as $\mathcal{B}$, $\mathcal{B}_m$, and $\mathcal{B}^v$, respectively. 

Throughout this manuscript, we will make frequent use of the results that are summarized in the following
\begin{lemma}\label{lem:results-B}
Let $N=(V,E)$ be a network. 
\begin{enumerate}
\item If $v \in V$ is contained in blocks $B$ and $B'$ of $N$ and $v\notin \{\rho_B, \rho_{B'}\}$, then $B = B'$ \cite[L.~9]{HSS:22}. \label{L3.19HSS}
\item  $w \in V(B) \setminus \{\rho_B\}$ if and only if $w$ and all of its parents are contained in $B$        \cite[L.~11]{HSS:22}. \label{L3.23HSS}
\item Suppose that $N$ is a phylogenetic level-$1$ network and let $u,v\in V$. Then, $u$ and $v$ are $\preceq_N$-comparable if and only if $L_N(u)\subseteq L_N(v)$ or $L_N(v)\subseteq L_N(u)$ \cite[Def.~12 and L.~44]{HSS:22}. \label{L7.1HSS} 

Moreover, if $v$ is a hybrid vertex in $N$, then $L_N(v)\subsetneq L_N(u)$ for all $u \in V (N )$ with $v \prec_N u$ \cite[L.~46]{HSS:22}. \label{L7.5HSS} 

In addition, if $u, v \in  V$ are $\prec_N$-incomparable, then
$u$ and $v$ are located in a common non-trivial block $B$ of $N$ if and only if $L_N(u) \cap L_N(v) \neq \emptyset$. 
In particular, if $u, v \in  V$ are $\prec_N$-incomparable and $u,v\in V(B)$, then  $L_N(u) \cap L_N(v) =L_N(\eta_B)$ \cite[L.~48]{HSS:22}. \label{L7.8HSS} 

\item  If $L_N (u)\cap L_N (v)\notin\{L_N(u),L_N(v),\emptyset\}$  for $u, v \in V$, then $u$ and $v$ are $\preceq_N$-incomparable and both contained in some non-trivial block $B\in \mathcal{B}(N)$ \cite[L.~19]{HSS:22}.\label{L3.35HSS}
\item Let $B\in \mathcal{B}(N)$ with unique hybrid vertex $\eta_B$ and $u,v\in V(B)$. Then, $L_N(v) \cap L_N(u) \in \{L_N(u), L_N(v), L_N(\eta_B)\}$ \cite[L.~20]{HSS:22}.\label{L3.37HSS}
\item  If $B\in \mathcal{B}(N)$ and $u,v\in V(B)$, then there is an undirected cycle in $N$ that contains $u$ and $v$ \cite[Thm.~4.2.4]{westGTbook}. \label{T4.2.4WestBook}
\end{enumerate}
\end{lemma}

A particular role in this contribution is played by restricted versions $\partN(v)$ of the subnetwork $N(v)$ of $N$ that is rooted in $v$.
\begin{definition}\label{def:relevant-part}
Let $N$ be a network. If (a) $v$ is not contained in any non-trivial block of $N$ or $v = \rho_B$ for every non-trivial block $B$ in $N$ that contains $v$,
then put $\partN(v) \coloneqq N(v)$ and, otherwise,  if (b) $v$ is part of some non-trivial block $B$ but not the root of $B$, then $\partN(v)$ is the subgraph of $N$ induced by all vertices $u\prec_N v$ for which $u\not\preceq_N w$ for all $w\in B$ with $w\prec_N v$. 
\end{definition}

Note that every vertex $v$ of a network $N$ satisfies either 
Condition (a) or (b) in Definition\ \ref{def:relevant-part}. 
Moreover,  we emphasize   that $\partN(v)$ in Definition~\ref{def:relevant-part}(b) is well-defined since, by Lemma \ref{lem:results-B}\eqref{L3.19HSS}, such a non-trivial block $B$ is always uniquely determined for $v$. 
Definition~\ref{def:relevant-part}(b) can sloppily  be rephrased to: $\partN(v)$ is the subnetwork rooted in $v$ without any descendants of vertices $w\prec_N v$ that are also located in the block $B$ that contains $v$ but for which $v$ is not the root; see Figure~\ref{fig:partN1} and \ref{fig:relevant_neighbors} for illustrative examples. 
\begin{figure}[htbp]
\centering
\includegraphics[width = 0.9\textwidth]{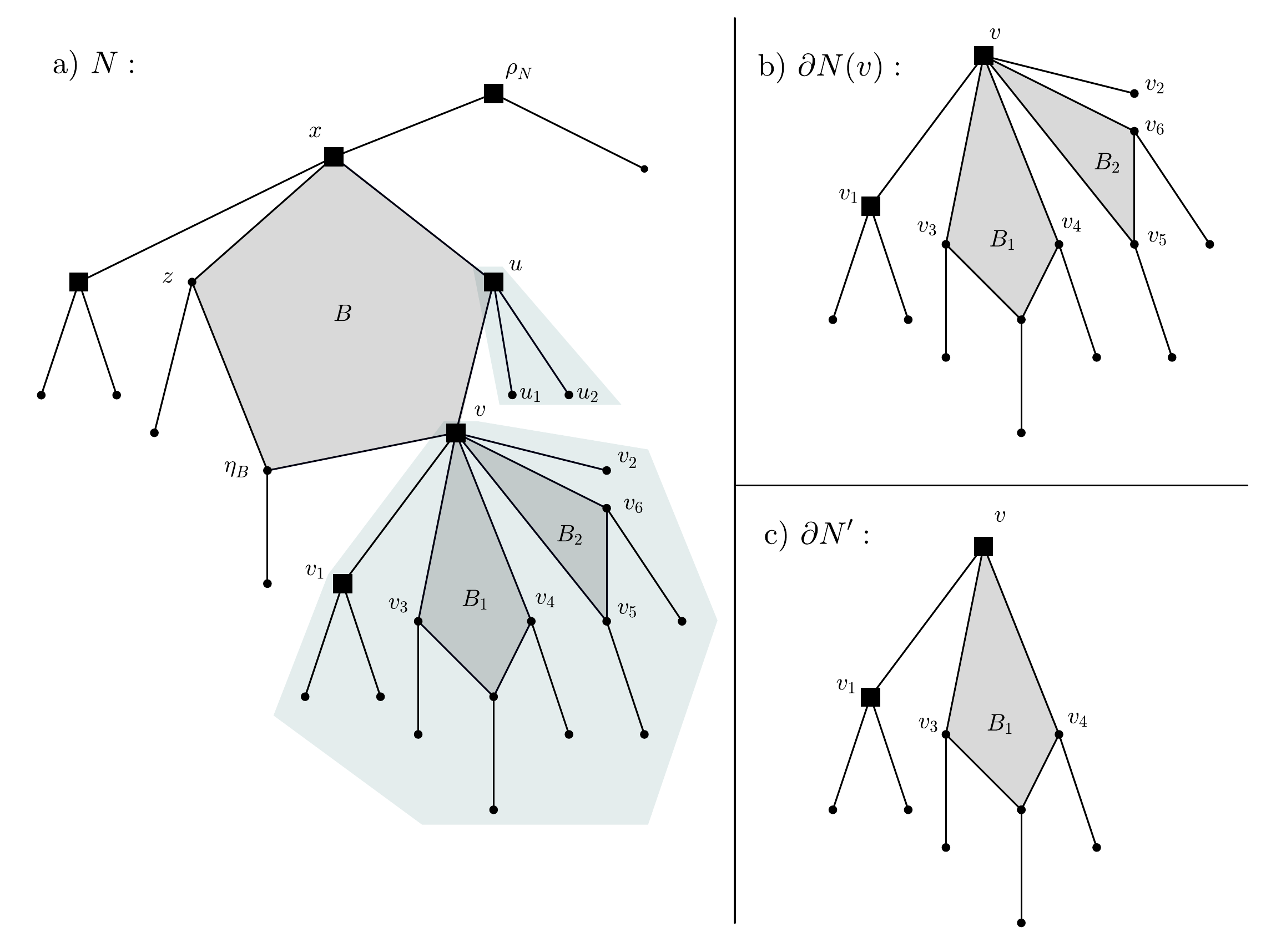}
\caption{
Examples to illustrate the notations established in Section~\ref{sec:prelim} using the vertices $\rho_N$, $u$, $v$, and $w$. 
Shown are three distinct networks in which all relevant vertices are highlighted by $\blacksquare$. Moreover, non-trivial blocks are highlighted by gray-shaded areas.
\newline
Panel a) shows a phylogenetic level-$1$ network $N$.  
Although $N$ is phylogenetic, neither $N(u)$ nor $N(v)$ are phylogenetic, since they both contain the vertex $\eta_B$ that has in- and out-degree $1$ in $N(u)$ as well as in $N(v)$. Furthermore, we have $\child_N^*(\rho_N) = \child_N(\rho_N)$, since $\rho_N$ is not contained in any non-trivial block. Moreover,  $\child_N^*(x) = \child_N(x)$ since, for all non-trivial blocks $B$ that contain $x$, it holds that $x=\rho_B$. In contrast, $\child_N^*(u) = \{u_1,u_2\}\neq \{u_1,u_2,v\} = \child_N(u)$, since $v$ is contained in a non-trivial block $B$ that also contains $u$ but for which $u$ is not the root. Similarly, $\child_N^*(v) = \{v_1,\ldots,v_6\}\neq \{v_1,\ldots,v_6,\eta_B\}= \child_N(v)$. Furthermore, vertex $\rho_N$ is relevant by Observation\ \ref{obs:relevant}(i) and $x$ and $v$ are relevant by Observation\ \ref{obs:relevant}(ii). Vertex $u$ is relevant since $|\child_N^*(u)|>1$, although $u$ does neither satisfy Observation~\ref{obs:relevant}(i) nor (ii). In contrast, vertex $z$ is not relevant since $|\child_N^*(z)|=1$. 
\newline
Vertices $\rho_N$ and $x$ satisfy Definition~\ref{def:relevant-part}(a) and, therefore, $\partN(\rho_N) = N(\rho_N)=N$ and $\partN(x)=N(x)$. Vertices $u$ and $v$ satisfy Definition~\ref{def:relevant-part}(b), since they are part of non-trivial block $B$ with $\rho_B = x\neq u,v$.  
The subnetworks $\partN(u)$ and $\partN(v)$ (highlighted by turquoise-shaded areas; see also Panel b) for $\partN(v)$) are vertex-disjoint and
distinct from $N(u)$ and $N(v)$, respectively.%
\newline 
In this example, we have $\mathcal B^v(N) = \{B_1, B_2\}$. In particular, $v$
is the root of both blocks $B_1$ and $B_2$ but not of block $B$ that also contains $v$. Hence, $\child^*_N(v)=\{v_1,\dots,v_6\}$, $\child^*_{B_1}(v)=\{v_3,v_4\}$, $\child^*_{B_2}(v)=\{v_5,v_6\}$,
and $\child^*_{\mathcal{B}^v(N)} =\{v_3,v_4,v_5,v_6\}$. Thus, 
$\child^*_{\overline{\mathcal{B}^v(N)}}$ consists exactly of $v_1$ and $v_2$. 
Panel c) shows the network 
$\partN':=\partN(v,\mathcal{B}',child')$ for $\mathcal{B}'=\{B_1\}$ and $\child'=\{v_1\}$.}
\label{fig:partN1}
\end{figure}

For later reference we provide the following result that shows that the operator \enquote{$\partial$} is idempotent. 
\begin{lemma}\label{lem:idempotent}
	Let $N$ be a network and $v\in V(N)$. Then, for $\widetilde N =\partN(v)$ we have 
	$\partial \widetilde N (v) = \partN(v)$.
\end{lemma}
\begin{proof}
	Assume that $v$ satisfies  Definition\ \ref{def:relevant-part}(a) in $N$ in which case we have $\widetilde N =\partN(v) = N(v)$. By construction, 
	if  $v$ is not contained in any non-trivial block of $N$ or $v = \rho_B$ for every non-trivial block $B$ in $N$ that contains $v$,
	then $v$ is not contained in any non-trivial block of $N(v)$ or $v = \rho_B$ for every non-trivial block $B$ in $N(v)$ that contains $v$.
	Hence, $v$ satisfies Definition\ \ref{def:relevant-part}(a) in $N(v)$. Therefore and since $\rho_{\widetilde N}=v$, we have
	$\partial \widetilde N(v) = \widetilde N = \partN(v)$.
	
	Assume that $v$ satisfies  Definition\ \ref{def:relevant-part}(b) in $N$. In this case, 
	there is a non-trivial block $B$ such that $v\in V(B)\setminus \{\rho_B\}$. 
	By Lemma \ref{lem:results-B}\eqref{L3.19HSS}, this block $B$ must be 
	unique. 
	By definition, $\widetilde N =\partN(v)$ is the subgraph of $N$ induced by all vertices 
	$u\prec_N v$ for which $u\not\preceq_N w$ for all $w\in B$ with $w\prec_N v$. 
	This together with Lemma \ref{lem:results-B}\eqref{L3.23HSS} implies that 
	any non-trivial block $B'$ in $\widetilde N$ satisfies $\rho_{B'}\preceq_{\widetilde N} v$. 
	One easily verifies that, therefore, $v$ satisfies  Definition\ \ref{def:relevant-part}(a) in $\widetilde N$. 	
	Therefore and since $\rho_{\widetilde N}=v$, we have
	$\partial \widetilde N(v) = \widetilde N(v) = \widetilde N = \partN(v)$. 
\end{proof}

\paragraph{\bf Further special types of networks,  trees, and blocks}
We introduce here further types of networks and blocks that will play an important role in this contribution.
We start with the structure of particular blocks.

\begin{itemize}
 \item A block $B \in \mathcal{B}_m(N)$ is a \emph{lantern}, if $B$ consists of $m-2$ internally vertex-disjoint $\rho_B\eta_B$-paths  and every such path contains precisely one vertex of $V^{-}(B)$. The edge $(\rho_B, \eta_B)$ may or may not be part of $B$.

 \item  A block $B \in \mathcal{B}_m(N)$ with $V^-(B) = \{v_1,\dots,v_{m-2}\}$ is a \emph{crescent}, if it contains a Hamiltonian path $(\rho_B, v_1,\dots,v_{m-2},\eta_B)$, a shortcut $(\rho_B,\eta_B)$, and all other edges, if any, are of the form $(v_i,\eta_B)$. 

 \item 
A block $B \in \mathcal{B}_m(N)$ is a \emph{full-moon}, if $B$ consists of precisely two internal vertex disjoint paths one containing $\lceil \frac{m-2}{2} \rceil$ internal vertices and the other $\lfloor \frac{m-2}{2} \rfloor$ internal vertices.
\end{itemize}

An illustrative example of lanterns, crescents, and full-moons is provided in Figure~\ref{fig:lantern-crescent}.

\begin{itemize}
 \item A network $N$ is \emph{separated} if every hybrid vertex $v$ has  $\outdeg(v)=1$. 
  A \emph{binary} network is a separated network in which every tree vertex $v$ has $\outdeg(v)=2$ and and every hybrid vertex $v$ has $\indeg(v)=2$.
\end{itemize}

\begin{itemize}
 \item The \emph{caterpillar tree} $\Tcat$, is a tree on $n$ leaves such that
       each inner vertex has exactly two children and the subgraph induced by
       the inner vertices is a path with the root $\rho_T$ at one end of this
       path. 

	\item A tree $T$ is \emph{maximally balanced} if either (i) $T$ consists of precisely one vertex, 
        or (ii) $T$ has at least two leaves and each inner vertex $v$ of $T$ has
	      precisely two children $v_1$ and $v_2$ and these two children satisfy
	      $\lvert L_T(v_1) - L_T(v_2) \rvert \leq 1$. 
	      
	      Note that the quantity $\lvert L_T(v_1) - L_T(v_2) \rvert$ is also
	      referred to as the \emph{balance value} $\mathrm{bal}_T(v)$ of $v$ and
	      $v$ is called \emph{balanced} if $\mathrm{bal}_T(v) \leq 1$. Thus, all
	      vertices of a maximally balanced tree are balanced \cite{Mir2013}.
	      As argued in \cite{Mir2013}, for a fixed leaf number $n$, the
	      maximally balanced tree on $n$ leaves is uniquely determined (up to
	      isomorphism).
	      
	\item A \emph{star tree} is a tree $T$ such that either (i) $T$ consists of
	     precisely one vertex, or (ii) $T$ has at least two leaves and all leaves
	     of $T$ are adjacent to the root $\rho_T$ of $T$. A star tree on $n$
	     leaves is denoted by $\Tstar$.

	\item A \emph{triangle network} $N_{\Delta}$ is the network that contains
	      precisely one block $B$ that is a triangle with $\rho_B = \rho_N$ and
	      with precisely two leaves and each of the two vertices distinct from
	      $\rho_B$ is adjacent to precisely one leaf.   

\end{itemize}

Note that triangle networks, caterpillars, and maximally balanced trees are binary. 

\begin{itemize}

	\item $\Tmod$ denotes the \emph{modified star} on $n$ leaves that is
	      obtained from the star tree $\mathrm{T}^{star}_{n-2}$ and $N_{\Delta}$
	      by identifying their roots.

	\item  $\Tpush$ denotes the \emph{push-down star} on $n$ leaves that is
	       obtained from the star tree $\mathrm{T}^{star}_{n-2}$ and $N_{\Delta}$
	       by adding the edge $(\rho_{\mathrm{T}^{star}_{n-2}}, \rho_N)$. If
	       $n=2$, then both $\Tmod$ and $\Tpush$ are isomorphic to $N_{\Delta}$.

\item $\NC_n$ is the  binary phylogenetic level-$1$ network with $n$ leaves that consists of one non-trivial block $B\in \mathcal{B}_{n+1}$ 
            with root $\rho_B = \rho_N$
             that is a crescent where $|\child^*_{\NC_n}(v)|=1$ for all $v\in V(B)\setminus\{\rho_N\}$ (see Figure\ \ref{fig:N_cres} for a generic example
             and Definition~\ref{def:child*} for the definition of $\child^*$).
\end{itemize}

\begin{figure}[t]
    \centering
    \includegraphics[width=1.\textwidth]{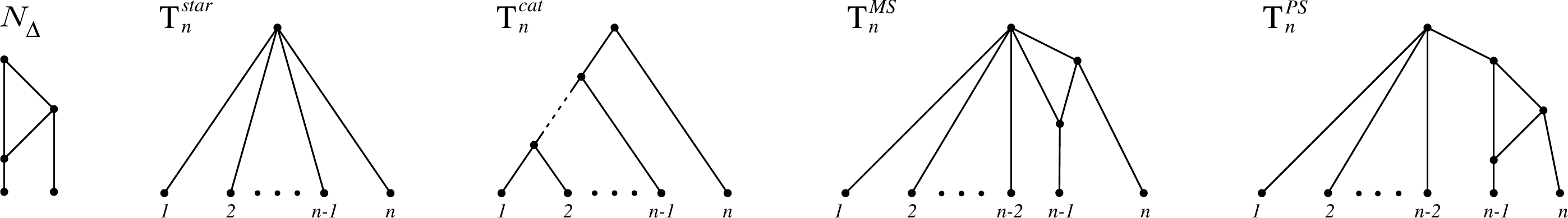}
    \caption{Shown are the triangle network $N_{\Delta}$ and generic examples of the networks 
    $\Tstar$, $\Tcat$, $\Tmod$, and $\Tpush$.
    }
    \label{fig:mod_and_pd_star}
\end{figure}

\begin{figure}[t]
    \centering
    \includegraphics[width=0.6\textwidth]{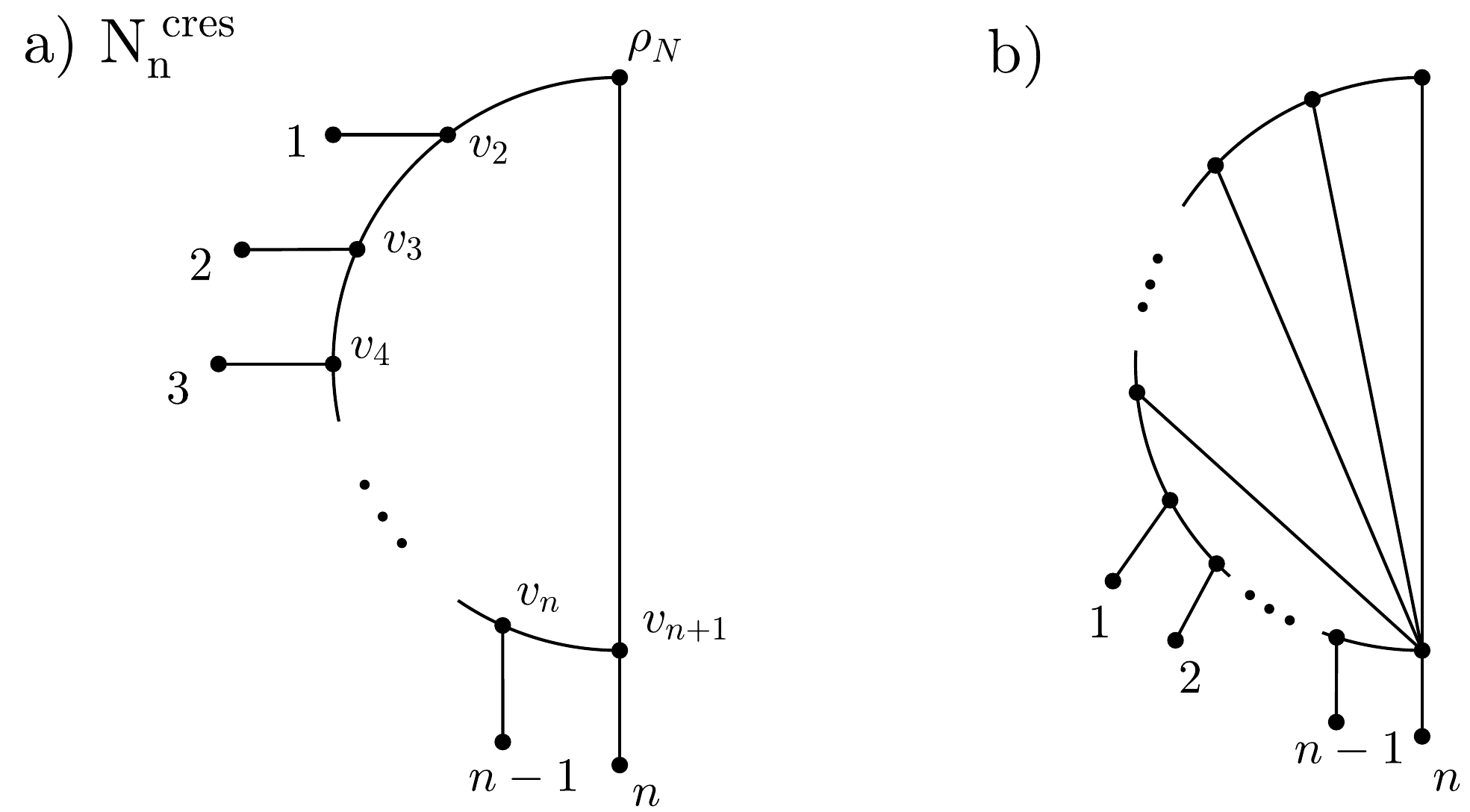}
    \caption{Panel a) shows the network $\NC_n$
     and panel b) shows a generic example for a class of infinitely many phylogenetic level-$1$ networks with $n$ leaves.}
    \label{fig:N_cres}
\end{figure}

\paragraph{\bf Subsets of children, relevant vertices, and least common ancestors}
We finally provide notation and terminology for certain subsets of vertices of a network and start with particular subsets of the children of a vertex.

\begin{definition}\label{def:child*}
Let $N$ be a phylogenetic network. For $v \in V(N)$, let $\child_N^*(v)\subseteq \child_N(v)$
be the subset of children of $v$ that are not contained in any non-trivial block $B\in
\mathcal{B}(N)$ with $v\in V(B)\setminus \{\rho_B\}$, i.e., $B$ contains $v$ but $v$ is not the root
of $B$.

Moreover, $\child^*_B(v)\subseteq child_N(v)$ denotes the set of children of $v$ that are located in
$B\in \mathcal B^v(N) $. For a subset $\mathcal{B}'\subseteq \mathcal{B}^v(N)$, we define
$\child^*_{\mathcal{B}'}\coloneqq \cupdot_{B\in \mathcal{B}'} \child^*_B(v)$. Furthermore,
$\child^*_{\overline{\mathcal{B}^v}} \coloneqq \child_N^*(v)\setminus \child^*_{\mathcal{B}^v}$
denotes the set of children of $v$ that are not located in any non-trivial block rooted in $v$.
\end{definition}

Hence, $\child_N^*(v) = \child_N(v)$ if $v$ is not contained in any non-trivial block or, otherwise, if $v$ is the root of every non-trivial block it is contained in. Note that Observation \ref{obs:identical-block} implies that  two different non-trivial blocks $B, B'\in \mathcal B^v(N)$ can share at  most one vertex and since both are rooted in $v$ it follows that $\child^*_{B}(v)\cap \child^*_{B'}(v) = \emptyset$.

\begin{fact}\label{obs:dictintBlocks-disjointChild}
Let $N$ be a phylogenetic network and $v\in V(N)$. For two different blocks $B, B'\in \mathcal B^v(N)$ it holds that $\child^*_{B}(v)\cap \child^*_{B'}(v) = \emptyset$.
\end{fact}

We now introduce the notion of relevant vertices.
\begin{definition}\label{def:relevant}
A vertex $v$ is \emph{relevant} in $N$ if $|\child_N^*(v)|>1$. For $N$, we denote with $\relV_N$ the set of its relevant vertices. The set 
$V(N)\setminus \relV_N$ comprises all \emph{irrelevant} vertices.
\end{definition}

\begin{fact}\label{obs:relevant}
Let $N\in \mscr{N}_n$ be phylogenetic and $n\geq 2$. 
If (i) $v=\rho_N$ or (ii) $v$ is the root of a non-trivial block in $N$ or (iii) $v$ is a true tree vertex of $\outdeg_N(v)>1$, then $v$ is relevant. 

In particular, if $v$ is relevant, then $v$ is the root of $\partN(v)$ with children $\child_{\partN}(v) = \child_N^*(v)$. Moreover, $v$ is irrelevant in $N$
if and only if $v$ is a leaf of $N$, $\outdeg_N(v)=1$ or all but possibly one child of $v$ is contained in a unique non-trivial block $B$ for which $\rho_B\neq v$. 
\end{fact}

Let $N$ be a binary phylogenetic level-$1$ network. Since $N$ is separated, 
only its tree vertices can be relevant. Note that all tree vertices $v$
of $N$ satisfy $|\child_N(v)| = 2$. Moreover, every hybrid vertex $v$
has $\indeg(v)=2$ and, by \cite[Obs.~16]{HSS:22}, every 
non-trivial block forms an \enquote{undirected cycle}. 
Hence,  if $v$ is a tree vertex in $N$ that is contained in some non-trivial block $B$ but different from its root, then only one child of $v$ is part of $B$ while the other child is not.
In this case, $|\child^*_N(v)|=1$ must hold, which, in turn, implies that $v$ is irrelevant. In particular, a tree vertex $v$ is relevant in $N$ if and only if $v$ is either the root of a non-trivial block or a true tree vertex in $N$ (in which case
 $\child^*_N(v) = \child_N(v)$ and thus, $|\child^*_N(v)|=2>1 $ must hold).
We summarize the latter discussion into the following 
\begin{fact}\label{fact:binary_rel_vert}
In binary phylogenetic level-$1$ networks the set of relevant vertices consists exactly of roots of non-trivial blocks and inner vertices that are not part of non-trivial blocks.
\end{fact}

Moreover, we will consider subsets of $\partN$.  

\begin{definition}
Let $N$ be a level-$1$ network, $v\in V(N)$, $\mathcal{B}'\subseteq \mathcal{B}^v(N)$, and $\child'\subseteq \child^*_{\overline{\mathcal{B}^v(N)}}$.  Then, $\partN(v,\mathcal{B}',\child')$ denotes the subgraph of $N$
that is induced by $v$ and all vertices $w\preceq u$ with $u\in \child^*_{\mathcal{B}'} \cup \child'$.
\end{definition}

\begin{fact}\label{obs:part-partN}
Let $N$ be a level-$1$ network and $v\in V(N)$. Then, $\partN(v,\mathcal{B}',\child')\subseteq \partN(v)$ and  $\partN(v) = \partN(v,\mathcal{B}^v(N),\child^*_{\overline{\mathcal{B}^v(N)}})$, cf. Figure \ref{fig:partN1}.
\end{fact}

Finally, a \emph{least common ancestor} (LCA) of a subset $Y\subseteq V$ in a DAG $N$ is an $\preceq_N$-minimal vertex of $V$ that is an ancestor of all vertices in $Y$.  In general DAGs $N$, a LCA does not necessarily exist for a given vertex set.  Moreover, the LCA is not unique in general.  We write $\LCA(Y)$ for the (possibly) empty set of $\preceq_{N}$-minimal ancestors of the elements in $Y$. In a (phylogenetic) network $N$, the root $\rho_N$ is an ancestor of all vertices in $V$, and thus a least common ancestor exists for all $Y\subseteq V$. If $\LCA(Y) = \{u\}$ consists of a single element $u$ only, we write $\lca(Y)=u$. In other words, $\lca(Y) = u$ always implies that the $\preceq_{N}$-minimal ancestor of the elements in $Y$ exists and is uniquely determined. We leave $\lca(Y)$ undefined for all $Y$ with $\vert\LCA(Y)\vert\ne1$.

By Lemma~49 in \cite{HSS:22}, level-$1$ networks $N$ are \enquote{$\lca$-networks}, i.e., $\lca_N(Y)$ is well-defined for all $Y\subseteq L(N)$. 
	In addition, the network $N^*$ obtained from a level-1 network $N$ by adding for every non-leaf vertex $v \in V(N)$ a new vertex $x_v$
	and the edge $(v, x_v)$ to $N^*$ remains a level-1 network. The latter allows us to apply Proposition~3.11 in \cite{SCHS:24} to conclude that 
	$N$ has the global $\lca$-property, i.e., $\lca_N(Y)$ is well-defined for all $Y\subseteq V(N)$. 
	We summarize the latter discussion in 
\begin{lemma}\label{L7.9HSS}
For all level-$1$ networks $N$ and $Y\subseteq V(N)$,  $\lca_N(Y)$ is well-defined.

\end{lemma}

\begin{table}[t]
    \centering
    \begin{tabular}{l|l}
        $\mathcal{B}(N)$  & collection of all non-trivial blocks in a network $N$  \\
        $\mathcal{B}_m(N)$ & subset of blocks in $\mathcal{B}(N)$ with precisely $m$ vertices \\
        $\mathcal{B}^v(N)$ & subset of blocks in $\mathcal{B}(N)$ rooted in $v$ \\
        $\partN(v)$ & partial network with root $v$ as specified in Definition\ \ref{def:relevant-part}\\
        $\child_N(v)$ & set of children of $v$ in $N$\\
        $\child_N^*(v)$ & set of children of $v$ in $N$ that are not contained in any \\                         &  non-trivial block $B\in \mathcal{B}(N)$ with  $v\in V(B)\setminus \{\rho_B\}$ \\
        $\child^*_B(v)$ & set of children of $v$ in $N$ that are contained in $B\in \mathcal B^v(N)$\\
        $\child^*_{\mathcal{B}'}$ & is the set $\cupdot_{B\in \mathcal{B}'} \child^*_B(v)$ for  $\mathcal{B}'\subseteq \mathcal{B}^v(N)$. \\
        $\child^*_{\overline{\mathcal{B}^v}}$ & is the set $\child_N^*(v)\setminus \child^*_{\mathcal{B}'}$
        with $\mathcal{B}' = \mathcal{B}^v(N)$.\\
        $v$ is relevant &  $|\child_N^*(v)|>1$\\
        $\relV_N$ & set of relevant vertices of $N$ \\
    \end{tabular}
    \caption{For the readers convenience, we here give an overview on our main notation, see also Figure\ \ref{fig:partN1}
    for an illustrative example.}
    \label{tab:sum-defs}
\end{table}

\section{The weighted total cophenetic index} \label{APPX:sec:wtci}
We are now in a position to introduce the central concept of this manuscript, namely the weighted
total cophenetic index. We begin with introducing the \enquote{classical} total cophenetic index
defined by Mir et al. \cite{Mir2013}. 

\begin{definition}\label{def:total_cophenetic_index}
Let $T$ be a rooted phylogenetic tree with leaf set $L(T)$. For any two distinct leaves $i,j\in
L(T)$, let $\varphi_T(i,j)$ be the depth of the least common ancestor of $i$ and $j$, i.e., the
number of vertices (excluding $v$) on a shortest $\rho_Nv$-path in $T$ where $v$ is the least common
ancestor of $i$ and $j$. Then, the total cophenetic index $\Phi(T)$ is defined as 
\[ \Phi(T) \coloneqq \sum\limits_{\substack{i,j \in L(T) \\ i \neq j}} \varphi_T(i,j). \]
\end{definition}

Lemma 2 in \cite{Mir2013} implies that the total cophenetic index can equivalently be expressed as
\[\Phi(T) = \sum\limits_{v\in \mathring{V}(T)\setminus  \{\rho_T\}} \binom{|L_T(v)|}{2},\] 
where $\Phi(T) \coloneqq 0$ in case $\mathring{V}(T)\setminus  \{\rho_T\} = \emptyset$ (which is precisely the case if $T$ is a star tree). 
We will primarily use this second representation in what follows. 

Of course, we may directly use \[\Phi(N) \coloneqq \sum\limits_{v\in  \mathring{V}(N)\setminus\{\rho_N\}} \binom{|L_N(v)|}{2}\] also for phylogenetic networks $N$.
However, this type of measure does not address the difference between trees $T$
and networks obtained from $T$ by adding shortcuts (i.e., by adding additional edges $(u,w)$ between existing vertices $u,w\in V(T)$ for which there is a vertex $v'\in V(G)$ such that $w\prec_G v'\prec_G u$). In the latter case, we obtain 
$\Phi(T) = \Phi(N)$ which is due to the fact that adding shortcuts does not change the set of descendant leaves of a given vertex, that is, $L_T(v)=L_N(v)$ for all $v\in V(T)=V(N)$; see Figure~\ref{fig:phi-N}(a) for an example. 
Moreover, as a side issue of the fact that  the total cophenetic index $\Phi$ was primarily constructed as an 
index for trees, the total cophenetic index $\Phi$ does not make direct use of additional information 
about non-trivial blocks.
Thus,
one can easily construct examples of even \enquote{highly diverse} networks with the same total cophenetic index $\Phi$; see Figure~\ref{fig:phi-N} for some examples. 
To overcome this issue, 
we consider weighted versions $\Phi^*$ and $\Phi^{**}$ of $\Phi$. These are based on assigning weights to the non-trivial blocks of the underlying network which, in turn, are used to assign weights to all vertices of the network.

\paragraph{\bf Weights of blocks}

We first define the weight $\omega$ of a non-trivial block $B$ based on the size $\kappa_v = |K_v|$
of the set $K_v\coloneqq \{w\in V(B)\mid w\preceq_N v\}$ of vertices that are \enquote{below} $v$ in
$B$ (including $v$), for all vertices $v \in V(B)$ of $B$.

We define the \emph{weight} of a non-trivial block $B$ as
\[\omega(B)\coloneqq \sum\limits_{v\in V(B)\setminus\{\rho_B\}} \binom{\kappa_v}{2} = \sum\limits_{v\in V^-(B)} \binom{\kappa_v}{2}\]
The equality $\sum_{v\in V(B)\setminus\{\rho_B\}} \binom{\kappa_v}{2} = \sum_{v\in V^-(B)} \binom{\kappa_v}{2}$ is due to the fact that all hybrids $\eta$ in $B$ are $\preceq_B$-minimal vertices and thus, satisfy 
	$\kappa_\eta = 1$, i.e.,  $\binom{\kappa_\eta}{2}= 0$.

Note that the choice of the weight $\omega(B)$ for non-trivial blocks $B$ is intentional. In fact,
by taking the block $B$ and adding, for all $v \in V(B)$, a new vertex $x_v$  along with the
edge $(v, x_v)$, we obtain the \emph{leaf-extended network} $B^*$ of $B$. As demonstrated below,
$\omega(B) = \Phi(B^*)$ holds, and therefore, $\omega(B)$ is a measure of the balance of $B^*$ in terms of
$\Phi$. 
Consequently, it distills the information about $B$ to its essential core structure, quantifying the
balance of $B$ without disturbing it with information about network structures \enquote{surrounding}
$B$ in $N$. In other words, any information that is dispensable to understand the structure of  $B$ is omitted.

\begin{lemma}\label{lem:omega=phi}
	Let $N$ be a network, $B\in  \mathcal{B}(N)$ be a non-trivial block and 
	$B^*$ the leaf-extended network of $B$. Then, 
	\[\omega(B) = \Phi(B^*)  \coloneqq \sum\limits_{v\in \mathring{V}(B^*)\setminus\{\rho_{B^*}\}} \binom{|L_{B^*}(v)|}{2}.\]
\end{lemma}
\begin{proof}
	Let $B\in  \mathcal{B}(N)$. Since $B$ is a non-trivial block, $B$ does not contain leaves. 
	By construction, every vertex $v\in V(B)$ is adjacent to exactly one leaf $x_v$ in $B^*$.
	Hence, we readily obtain $\kappa_v = |L_{B^*}(v)|$ for all $v\in V(B)$. 
	Moreover, it holds that $\mathring{V}(B^*) = V(B)$. 		
	Taken the latter arguments together, we obtain 
	\[\sum\limits_{v\in \mathring{V}(B^*)\setminus\{\rho_{B^*}\}} \binom{|L_{B^*}(v)|}{2}
	  = 
	 \sum\limits_{v\in V(B)\setminus\{\rho_{B}\}} \binom{\kappa_v}{2}. 
	\]		
	Now consider a vertex $v\in V(B)\setminus V^-(B)$ such that $v\neq \rho_B$. 
	By definition, 
	there is no child of $v$ that is located in $N$ and thus, $\kappa_v=1$ must hold. 
	Hence, $\binom{\kappa_v}{2} = 0$ for all  $v\in V(B) \setminus (V^-(B)\cup \{\rho_B\})$
	and, therefore, 
	\[\sum\limits_{v\in V(B)\setminus\{\rho_{B}\}} \binom{\kappa_v}{2} 
	  = 
	 \sum\limits_{v\in V^-(B)} \binom{\kappa_v}{2}, 
	\]		
	which completes the proof. 
\end{proof}
We provide for later reference
\begin{lemma} \label{lem:upperbound-kappa}
Let $N$ be a network and $B \in \mathcal{B}_m(N)$ be a non-trivial block with precisely $m$ vertices
and root $\rho_B$. Then, $\kappa_v \leq m-1$ for all $v \in V^-(B)$. 
\end{lemma}
\begin{proof}
Let $v\in V^-(B)$. In particular, $v\neq \rho_B$. Together with Lemma~\ref{lem:unique-max-B} this
implies $v\prec\rho_B$. Together with the fact that $N$ is acyclic this in turn implies that
$\rho_B\preceq v$ is not possible. Hence, $|K_v|\leq m-1$.
\end{proof}

\begin{fact}\label{obs:omega} 
For all non-trivial blocks $B$ we have $\omega(B)\geq 1$. In particular, $\omega(B) = 1$ if and only if $B$ is a triangle. 
\end{fact}

\paragraph{\bf Weights of vertices}

We now proceed with defining weights $\phi$ for vertices $v\in V(N)$. Let $0\leq \epsilon\leq 2$. Then, we define:
\[
\phi(v) = \phi_N(v) \coloneqq 
\begin{cases}
	1  & \text{, if }  \mathcal{B}^v(N) = \emptyset  \\
	\epsilon + \sum_{B\in \mathcal{B}^v(N)} 	\omega(B)   &\text{, otherwise.}
\end{cases}
\]

\noindent If the context is clear we write $\phi(v)$ rather than $\phi_N(v)$.  Hence, if $v$ is not the root of any non-trivial block, then $\phi(v)=1$ and, otherwise $\phi(v)$ becomes the sum of $\epsilon$ and the sum of the weights $\omega(B)$ of non-trivial blocks $B$ for which $\rho_B=v$. Note, however, that by Observation~\ref{obs:omega}, $\omega(B)=1$ may be possible in case that $B$ is a triangle.
In particular, if $\mathcal{B}^v(N) = \{B\}$ and $B$ is a triangle, then $\sum_{B\in \mathcal{B}^v(N)}\omega(B)  =1$ which is the same value as assigned e.g.,  to true tree vertices $w$ (where trivially $\mathcal{B}^w(N) = \emptyset$ holds).
To distinguish  between the weights of roots of such triangle blocks and 
those vertices $w$ for which $\mathcal{B}^w(N) = \emptyset$, we employ the extra constant $\epsilon$. Note that by the previous argument, setting $\epsilon = 0$, will lead to the same weight for true tree vertices and tree vertices that are roots of triangles, whereas tree vertices that are roots of larger
or more than two non-trivial blocks  will receive a larger weight. As far as the upper bound, $\epsilon=2$, is concerned, we choose this number primarily for technical reasons (in particular, for the proof of Lemma \ref{lem:max_is_binary} and \ref{lem:exactly_one_block}) as well as to not over-penalize tree vertices contained in non-trivial blocks.

\paragraph{\bf The weighted total cophenetic index}

We are now in the position to generalize the total cophenetic index $\Phi$ from trees to networks by defining a weighted version of it. To this end, we consider relevant vertices in $\relV_N$ and their particular weights. Note that we only consider relevant vertices here as information about the irrelevant vertices is already accounted for in the weights $\omega(B)$ of non-trivial blocks.

\begin{definition}\label{def:wpci}
The \emph{weighted partial (total) cophenetic index (wPCI) $\Phi^*(N)$} 
of a network $N\in \mscr{N}_n$ is defined as

\[ \Phi^*(N) \coloneqq 
\begin{cases}
 \sum\limits_{v\in \relV_N\setminus\{\rho_N\}} \phi(v) \binom{|L_N(v)|}{2} 
                    & \text{, if  $\relV_N\setminus\{\rho_N\}\neq \emptyset$}\\
 \hfill 0 & \text{, otherwise.}
\end{cases}
\]

\noindent
The \emph{weighted (total) cophenetic index (wCI) $\Phi^{**}(N)$} is then defined as
\[\Phi^{**}(N) := \Phi^*(N) + \phi(\rho_N) \binom{n}{2}.\]
\end{definition}
By convention, we assume that $\binom{1}{2}\coloneqq 0$.
As we shall see later, the distinction between the parts of the sums $\Phi^{**}(N)$ into $\Phi^{*}(N)$ and $\phi(\rho_N) \binom{n}{2}$ will become quite handy in upcoming proofs and is also necessary to better distinguish between distinct networks as shown in Figure~\ref{fig:phi-N}(c). 
Further examples of networks and their wCI are provided in Figure\ \ref{fig:partN}.

Note, if $T$ is a phylogenetic tree, then $\child_T^*(v)= \child_T(v)$ for all $v\in V(T)$ and thus, $\relV_T = V(T)\setminus L(N)$, i.e., all inner vertices are relevant. Moreover, $\mc{B}^v = \emptyset$ for all $v\in V(T)$ and thus, all inner vertices $v$ of $T$ have weight $\phi(v)=1$. 
Hence, we obtain

\begin{fact}\label{obs:cp}
For all phylogenetic trees $T\in \mscr{T}_n$, it holds that $\Phi^{*}(T) = \Phi(T)$ and $\Phi^{**}(T) = \Phi(T) + \binom{n}{2}$.
\end{fact}

By Observation~\ref{obs:relevant}, the root $\rho_N$ of a phylogenetic network is always relevant. Hence, we obtain

\begin{fact}\label{obs:root-rel}
For all phylogenetic networks $N$ it holds that $\Phi^{**}(N) =\sum\limits_{v\in \relV_N} \phi(v) \binom{|L_N(v)|}{2}$.
\end{fact}

 \begin{fact}\label{obs:phi**-lower-bound}
 For every phylogenetic network $N$ it holds that  $\phi(v)\geq 1$ for all non-leaf vertices of $N$. 
 Moreover,  since by Observation~\ref{obs:relevant} the root $\rho_N$ is relevant, Observation\ \ref{obs:root-rel}
 implies that  $\Phi^{**}(N)\geq 1 \cdot \binom{|L(N)|}{2} = \binom{|L(N)|}{2}$.  
 \end{fact}

An appealing property of the wCI is that it can be computed efficiently. 
\begin{proposition}\label{prop:algo}
	The weighted total cophenetic index can be computed in linear-time for
   level-$1$ networks and in  $O(|V|^2+|V||E|)$  time for general networks $N=(V,E)$. 
\end{proposition}
\begin{proof}
	Let $N=(V,E)$ be network and thus, in particular, a DAG with a unique root $\rho_N$. Put
	$n\coloneqq |V|$ and $m\coloneqq |E|$. We first determine the non-trivial biconnected components
	of $N$ via depth-ﬁrst search in $O(n+m)$ time \cite{cormen2022introduction} and collect them in
	the set $\mathcal{B}$. Note that, by Observation\ \ref{obs:identical-block}, two distinct blocks
	in $\mathcal{B}$ have at most one vertex in common. Thus $\sum_{B\in \mathcal B} |V(B)| \in O(n)$
	and $\sum_{B\in \mathcal B} |E(B)| \in O(m)$. Each $B\in \mathcal B$ has a unique root $\rho_B$,
	i.e., a vertex of in-degree $0$ in $B$. By Observation\ \ref{obs:relevant}, the vertices $\rho_N$
	and $\rho_B$ for all $B\in \mathcal B$ and every vertex $v$ not contained in any $B\in \mathcal
	B$ with $\outdeg_N(v)>1$ are relevant. Pre-computation of the set $\mathcal B$ allows to
	determine the latter type of relevant vertices (i.e., $\rho_N$, $\rho_B$, and the true tree
	vertices $v$ with $\outdeg_N(v)>1$) in $O(n)$ time and we collect them in the set $\relV$. In
	particular, we can store for each $v\in V$ the set $\mathcal B^v$ of blocks in $B \in \mathcal B$
	for which $v=\rho_B$. 
 
   Put $W \coloneqq V\setminus \relV$. Each edge either belongs to a unique $B\in \mathcal B$ or to
   no non-trivial block of $N$ at all. Hence, we can mark all edges $e$ with a \enquote{$1$}
   precisely if it belongs to some $B\in \mathcal B$ by traversing over each edge in each $B$ in
   $O(n+m)$ time. For the remaining vertices $v\in W$ we must check if $|\child^*_N(v)|>1$. Note
   that none of these vertices are roots of a block $B\in \mathcal B$ since we already collected
   these vertices. One easily verifies that $|\child^*_N(v)|\leq 1$ if \ all out-edges of $v$,
   except for possibly one edge $(v,w)$, are marked with a \enquote{$1$}. Hence, as soon as we
   observe that there are two edges $(v,w)$ and $(v,w')$ not marked with a \enquote{$1$}, then we
   add $v$ to $\relV$. The latter task can be done for all vertices in $W$ in $\sum_{v\in W}
   \outdeg(v)\in O(m)$ time. In this way, we computed $\relV = \relV(N)$ in $O(n+m)$ time. 	

	We finally have to compute $\phi(v)$ for all $v\in \relV$. To this end, first compute, for each
	$B\in \mathcal B$, the weight $\omega(B)$. This requires, for each $w\in V(B)$, the computation
	of $\kappa_w = |K_w|$. To this end, we employ, for each $w\in V(B)$, a depth-first search in $B$
	with start vertex $w$ to obtain its descendants in $B$ and, therefore, $\kappa_w$ in
	$O(|V(B)|+|E(B)|)$ time. Since the latter approach must be repeated for each vertex in $B$, we
	end in an overall time-complexity of $O(|V(B)|(|V(B)|+|E(B)|))$ time to compute $\omega(B)$.
	Since different non-trivial blocks are edge-disjoint and share at most one vertex, computing
	$\kappa_w$ for each vertex $w$ in each block $B\in \mathcal B$ requires $O(n^2+nm)$ time. Based
	on the determined values $\kappa_w$, $\omega(B)$ can be computed in $O(|V(B)|)$ time for each
	$B\in \mathcal B$. Again, since distinct non-trivial blocks share at most one vertex, $\omega(B)$
	can be computed for all $B\in \mathcal B$ in $O(n)$ time. 
 
   We then put, for each $v\in \relV$, $\phi(v) = 1$ in case $\mathcal B^v = \emptyset$ and,
   otherwise, $\phi(v) = \epsilon + \sum_{B\in \mathcal B^v} \omega(B)$. The latter task can be done
   in $|\mathcal B^v|$ time for each $v\in \relV$. Note that $\mathcal{B}^v \cap \mathcal{B}^w =
   \emptyset$ for all distinct $v,w\in V$. Hence, $\mathcal B = \cupdot_{v\in \relV} \mathcal{B}^v$
   and thus, $ \sum_{v\in \relV} |\mathcal B^v| =\sum_{B\in \mathcal B} |B| = O(n)$. In summary,
   $\phi(v)$ can be computed for all $v\in \relV$ in $O(n^2+nm)$ time. By Observation\
   \ref{obs:root-rel}, the weighted total cophenetic index of $N$ can be computed via $\sum_{v\in
   \relV} \phi(v) \binom{|L_N(v)|}{2}$. Thus, we need to determine the values $|L_N(v)|$. Again,
   this can be done via depth-first search for each $v$ in $O(n^2+nm)$ time.

   In summary, the overall complexity to compute the weighted total cophenetic index in general
   networks is dominated by depth-first search for each of its vertices and is, thus, in
   $O(n^2+nm)$.  
 	
	Suppose now that $N$ is level-$1$ network. Computation of $\relV = \relV(N)$
    can be done in $O(n+m)$ as for general networks. However, we can improve the
    runtime to compute the respective values $\omega(B)$ and $|L_N(v)|$. 
    First observe that each $B\in \mathcal B$ has a unique hybrid $\eta_B$ and
    removing $\eta_B$ from $B$ results in a rooted tree. Thus,
    we can set, for each $B\in \mathcal B$, first  $\kappa_{\eta_B}=1$ and
	afterwards traverse $B$ in postorder which ensures that whenever we reach 
	a vertex $w\in V(B)$ that all its children in $B$ have been processed and we
	put $\kappa_{w} = 1+ (\sum_{u\in \child_B(w)} \kappa_u) - (\outdeg_B(w)-1)$. 
    Removing the term $(\outdeg_B(u)-1)$ is necessary to avoid 
    multiple counts of $\eta_B$ as descendant which is counted already once for 
    each of the $\outdeg_B(u)$ children of $u$ (cf. Lemma~\ref{lem:results-B}\eqref{L7.5HSS}).
    The postorder traversal takes $O(|V(B)| + |E(B)|)$ time. The computation of $\kappa_{w}$
	takes $O(\outdeg_B(w))$ time for each vertex $w$. Hence, the time needed to
	compute $\kappa_{w}$ for all $w\in V(B)$ is in $O(\sum_{w\in V(B)} \outdeg_B(w))
	= O(|E(B)|)\subseteq O(|V(B)| + |E(B)|)$ for each block $B\in \mathcal B$. 
    By similar arguments as in the previous case, computing
	$\kappa_w$ for each vertex $w$ in each block $B\in \mathcal B$ requires
	$O(n+m)$ time and, afterwards,  $\omega(B)$ can be
	computed for all $B\in \mathcal B$ in $O(n)$ time.

    To compute, for each $v\in \relV$, the value $\phi(v)$, we give now a more efficient
    way to compute $|L_N(v)|$ for $v\in \relV$.  To this end, we put $l(v)=1$ for each
	$v\in L(N)$. We then traverse $N$ in post-order and put, for each visited vertex 
    $v$, $l(v)=\sum_{u\in \child_N(v)} l(u)$ if $v$ is a true tree vertex or
    a hybrid but not the root of some non-trivial block,  and, otherwise, 
    if $v$ is the root
    or an internal vertex of a non-trivial block $B$, we 
    we collect all these blocks in the set $\mathcal{B}^*(v)$ and
    put 
    $l(v)=\sum_{u\in\child_N(v)} l(u) - \sum_{B\in \mathcal{B}^*(v)}(\outdeg_B(v)-1)\cdot l(\eta_B)$ for each visited vertex $v$. 
    Removing the term $(\outdeg_B(v)-1)\cdot l(\eta_B)$ for each $B$ in $\mathcal{B}^*(v)$ is necessary to avoid 
    multiple counts of $l(\eta_B)$ which is counted already once for 
    each of the $\outdeg_B(v)$ children of $v$ in $B\in \mathcal{B}^*(v)$ (cf. Lemma~\ref{lem:results-B}\eqref{L7.8HSS}).
    It is easy  to see that
	$l(v)=|L_N(v)|$ is correctly computed in this manner. While post-order
	traversal of $N$ can be done in $O(n+m)$ time, the computation of $|L_N(v)|$
	for all $v\in V$ can be done in $O(\sum_{v\in V} \deg(v)) = O(m)$ time. 
    This in turn, allows to compute $\phi(v)$ for all $v\in \relV$
    in $O(n+m)$ time. 
    
    In summary, the overall complexity to compute the weighted total cophenetic
	index of a level-$1$ network $N$ is dominated by the complexity of the post-order traversal of $N$ and
	is, thus, in $O(n+m)$.  
\end{proof}

As a simple consequence of Observation~\ref{obs:cp} and the results provided in \cite{Mir2013}, we obtain 

\begin{proposition}\label{prop:cp-properties}
The following properties are satisfied for  phylogenetic trees.
\begin{enumerate}  \setlength{\itemsep}{2pt}
\item The phylogenetic trees in $ \mscr{T}_n$ with maximum weighted (partial) cophenetic index are exactly the caterpillars, and the maximum is $\binom{n}{3}$ for $\Phi^*$ and $\binom{n+1}{3}$ for $\Phi^{**}$.
\item The phylogenetic trees in $ \mscr{T}_n$ with minimum weighted (partial) cophenetic index are exactly the star trees, and the minimum is $0$ for $\Phi^*$ and $\binom{n}{2}$ for $\Phi^{**}$.
\item If $T\in \mscr{T}_n$ is binary, then $\Phi^{*}(T)$ and $\Phi^{**}(T)$ are minimum among all binary phylogenetic trees in $\mscr{T}_n$ if and only if $T$ is maximally balanced. The minimum is $\sum\limits_{k=1}^{n-1} a(k)$ for $\Phi^{*}$ and $\sum\limits_{k=1}^{n-1} a(k) + \binom{n}{2}$ for $\Phi^{**}$, where $a(k)$ is the highest power of 2 that divides $k!$. 
\end{enumerate}
\end{proposition}

\begin{proof}
To see Conditions (1)-(3), first note that $\Phi(T)$ is maximized (respectively, minimized) among phylogenetic trees in $\mscr{T}_n$ if $T$ is a caterpillar (respectively, star tree), in which case $\Phi(T) = \binom{n}{3}$ (respectively, $\Phi(T) = 0$) and that $\Phi(T)$ is minimum among all binary phylogenetic trees in $\mscr{T}_n$ if and only if $T$ is maximally balanced (cf. \cite{Mir2013}). 
Since $\Phi(T)=\Phi^*(T)$ differs from	$\Phi^{**}(T)$ only by a constant summand $\binom{n}{2}$, Conditions (1)-(3) readily follow. 
\end{proof}

For later reference, we provide 
\begin{lemma}\label{lem:non_separated} 
For every phylogenetic level-$1$ network $N$, there is a separated phylogenetic level-$1$ network $\tilde N$ with 
$L(N) = L(\tilde N)$ and $\Phi^{**}(N)=\Phi^{**}(\tilde N)$. 
\end{lemma}
\begin{proof}
Suppose that $N$ is a phylogenetic level-$1$  network. If $N$ is already separated, we are done. Suppose that $N$ is not separated. Hence, $N$ contains a hybrid vertex $\eta$ with 
$\child_N(\eta) = \{x_1,\dots, x_m\}$, $m\geq 2$ and thus, $\outdeg_N(\eta)=m\geq 2$. 
Construct now a network $N'$ by removing all edges $(\eta,x_1),\ldots,(\eta,x_m)$ from $N$ and, afterwards, 
adding a new vertex $\eta'$ and edges $(\eta,\eta')$ as well as the edges $(\eta',x_1),\ldots,(\eta',x_m)$.
Clearly, $L(N) = L(N')$ holds. Moreover, 
by \cite[Lemma~5 \& Theorem~8]{HSS:22}, $N'$ remains a  phylogenetic level-$1$ network.

Since $N$ is a level-$1$ network, the vertex $\eta$ cannot be an internal vertex  of any other non-trivial block of $N$ (cf.\ Lemma \ref{lem:results-B}\eqref{L3.23HSS}). 
By Definition\ \ref{def:child*}, $\child^*_N(\eta) = \child_N(\eta)$ and therefore, $\child^*_N(\eta) = m>1$. Hence, $\eta$ is relevant in $N$. 
Moreover, by construction, $|\child_{N'}(\eta)|=1$ and thus, $\eta$ is not relevant in $N'$. Nevertheless, we have 
$\child_{N'}(\eta') = \{x_1,\dots, x_m\}$ and, by similar arguments, $\eta'$ is relevant in $N'$. 
Moreover, \cite[Lemma~5]{HSS:22} implies that $L_N(\eta) = L_{N'}(\eta')$
and $L_N(v) = L_{N'}(v)$ for all $v\in V(N)$.  
One easily verifies also that $\phi_{N}(\eta)=\phi_{N'}(\eta')$ and 
$\phi_{N}(v)=\phi_{N'}(v)$ for all $v\in V(N)\setminus\{\eta\}$. 
In particular, since all $v\in V(N)\setminus\{\eta\}$ remain unaffected in the construction
of $N'$, it is a straightforward task to verify that $v\in V(N)\setminus\{\eta\}$ is
relevant in $N$ if and only if $v$ is relevant in $N'$. 
Taking the last arguments together, we obtain $\Phi^{**}(N)=\Phi^{**}(N')$. 

Note that in a level-$1$ network distinct hybrid vertices must be contained in different
non-trivial blocks. Hence, we can independently apply the latter arguments  to 
all hybrid vertices of $N'$ that have out-degree distinct from $1$ until
we end in a separated phylogenetic level-$1$ network $\tilde N$ with 
$L(N) = L(\tilde N)$ and $\Phi^{**}(N)=\Phi^{**}(\tilde N)$. 
\end{proof}

 It is clear that the latter arguments can be used to also show that for each 
 separated phylogenetic level-$1$ network $\tilde N$ there is a  phylogenetic level-$1$ network $N$
 for which each hybrid vertex has outdegree greater than one (except in case it is adjacent to a leaf)
 and such that $L(N) = L(\tilde N)$ and $\Phi^{**}(N)=\Phi^{**}(\tilde N)$. 
In other words, the weighted total cophenetic index cannot distinguish between
 phylogenetic level-$1$ networks that are isomorphic up to \enquote{contraction} of  single edges $(\eta_B,v)$ 
 as well as \enquote{expansion} of $\eta_B$ to a single edge $(\eta_B,\eta'_B)$ 
  where $\outdeg_N(\eta_B)>1$ and $\eta_B$ is a hybrid. 
  One may address this issue by considering an additional parameter $\epsilon'$ added for hybrids or for 
  vertices whose parent is a hybrid. 
  
Note that Lemma \ref{lem:non_separated} does not hold for arbitrary networks in general. To see this,  
suppose that $B$ is a non-trivial block in a network $N$ that contains 
two hybrid vertices $u,v$ such that $u\prec_B v$ and 
$\outdeg_B(v)>1$. Using the approach as proposed in 
the proof of Lemma \ref{lem:non_separated} on vertex $v$ would increase the number of internal vertices in $B$ and, eventually
change $\omega(B)$ yielding a network $N'$ for which $\Phi^{**}(N')\neq \Phi^{**}(N)$ may be possible.

\section{Locality and recursiveness of the weighted total cophenetic index} \label{APPX:sec:otherproperties}
In order to investigate the structure of phylogenetic level-$1$ networks that maximize or minimize the weighted total cophenetic index $\Phi^{**}$, it will be helpful to have a mechanism that \enquote{controls} the replacement of local networks. In particular, as we shall see, $\Phi^{**}$ is local in the sense that if two phylogenetic level-$1$ networks $N_1$ and $N_2$ differ only in some well-defined local subnetworks $N'_1$ and $N'_2$, then the difference $\Phi^{**}(N_1)-\Phi^{**}(N_2)$ can be expressed in terms of the differences of $\Phi^*(N'_1)$ and $\Phi^*(N'_2)$, and the weights of the roots of $N'_1$ and $N'_2$ only. Moreover, $\Phi^{**}$ can always be expressed as a sum of $\Phi^*(N')$ for networks $N'\subsetneq N$ rooted at well-chosen relevant vertices in $N$. These two properties are often referred to as \enquote{locality} and \enquote{recursiveness} in the literature, and are indeed satisfied by the total cophenetic index on trees \cite[Lemma 4, Lemma 5]{Mir2013}. 
To provide the main results of this section, we first  introduce the notion of relevant neighbors and relevant-vertex-free paths.

\begin{definition}\label{def:IRfree}
A $uv$-path in $N$ whose internal vertices (i.e., vertices that are distinct from $u$ and $v$) are not relevant is \emph{internal relevant-vertex-free (IR-free)}.
\end{definition}

Note that every edge $(u,v)$ is trivially an IR-free path.
Moreover, if $v$ is an inner vertex in $N$ that is not contained in a non-trivial block of $N$,
then its children are either leaves, true tree vertices, or roots of non-trivial blocks. 
Hence, all non-leaf children of $v$ are relevant. The latter can be rephrased to:
any relevant vertex $u\prec_N v$ for which there is no
relevant vertex $u'\prec_N v$ with $u\prec_N u'$ must be a child of $v$.

\begin{definition}\label{def:rel-neighbor}
Let $v,u\in V(N)$ be such that $u\prec_N v$ and let $\mathcal{B}$ be the set of non-trivial blocks of $N$ that contain $v$.
 Then, $u$ is a \emph{relevant neighbor of $v$} if $u$ is relevant and precisely one of the following conditions is satisfied:
\begin{enumerate}
    \item $v$ is not contained in a non-trivial block of $N$ and $u$ is child of $v$
    \item $v$ is contained in a non-trivial block of $N$ and precisely one condition holds 
        \begin{enumerate}
            \item $u$ is contained in some $B\in \mathcal{B}$ as well; or 
            \item $u$ is not contained in any $B \in \mathcal{B}$ and
            			there is an IR-free $wu$-path such that $w\in B$ for some $B\in \mathcal{B}$ and $w$ is not relevant in $N$ and $w\prec_N v$. 
            \item $u$ is not contained in any $B \in \mathcal{B}$ and $u$ is a child of $v$. 
        \end{enumerate}
    \end{enumerate}
The \emph{relevant neighborhood $\relN(v)$ of $v$} is the set of its relevant neighbors (see Figure \ref{fig:relevant_neighbors} for an example).
\end{definition}
Note that $u$ being a relevant neighbor of $v$ does not necessarily imply that $u$ and $v$ are
adjacent. Moreover, observe that Conditions (a), (b) and (c) in Definition\ \ref{def:rel-neighbor}(2) are mutually exclusive, since Condition (a) 
does not cover any of the properties in Condition (b) and (c) while 
Condition (b) covers the case that $u$ and $v$ are not adjacent and (c) the case that $u$ and $v$ are adjacent.

We are now in the position to provide the two main results (Theorem\ \ref{thm:local-phi} and \ref{thm:sum-up-part}) of this section. While Theorem\ \ref{thm:local-phi} ensures locality, Theorem\ \ref{thm:sum-up-part} ensures recursiveness of the  weighted total cophenetic index. The proofs of these results are provided after we have established some additional technical results.

\begin{theorem}\label{thm:local-phi}
Let $N \in \mathscr{N}_n$ be a phylogenetic level-$1$ network and let $v$ be a relevant vertex in $N$. Moreover, let $N'$ be the phylogenetic level-$1$ network obtained from $N$ by replacing $\partN(v)$ by a phylogenetic level-$1$ network $\tilde N$ with $L(\tilde N) = L(\partN(v))$. Then, \[ \Phi^{**}(N)-\Phi^{**}(N') =  \left(\Phi^{*}(\partN(v))-\Phi^{*}(\tilde N)\right) + \left((\phi_N(v)-\phi_{N'}(\rho_{\tilde N})) \binom{|L(v)|}{2}\right), \]
where $L(v) \coloneqq L_N(v)=L_{N'}(v)$.
\end{theorem}

\begin{theorem}\label{thm:sum-up-part}
Let $N\in \mathscr{N}_n$ be a phylogenetic level-$1$ network and let $v_1,\dots,v_k$ be the relevant neighbors of the root $\rho_N$ in $N$. Let $\phi(\rho_N)\coloneqq \phi_N(\rho_N)$, $\phi(v_i) \coloneqq \phi_N(v_i)$, $n_i\coloneqq |L_N(v_i)|$, and $\partN_i\coloneqq \partN(v_i)$, $1\leq i\leq k$. Then, \[\Phi^{**}(N) = 
\sum_{i=1}^k \Phi^{**}(\partN_i) + \phi(\rho_N)\binom{n}{2} =\sum_{i=1}^k \left(\Phi^*(\partN_i) + \phi(v_i)\binom{n_i}{2}\right) + \phi(\rho_N)\binom{n}{2}. \]
\end{theorem}

\begin{figure}[ht!]
\centering
\includegraphics[width=.8\textwidth]{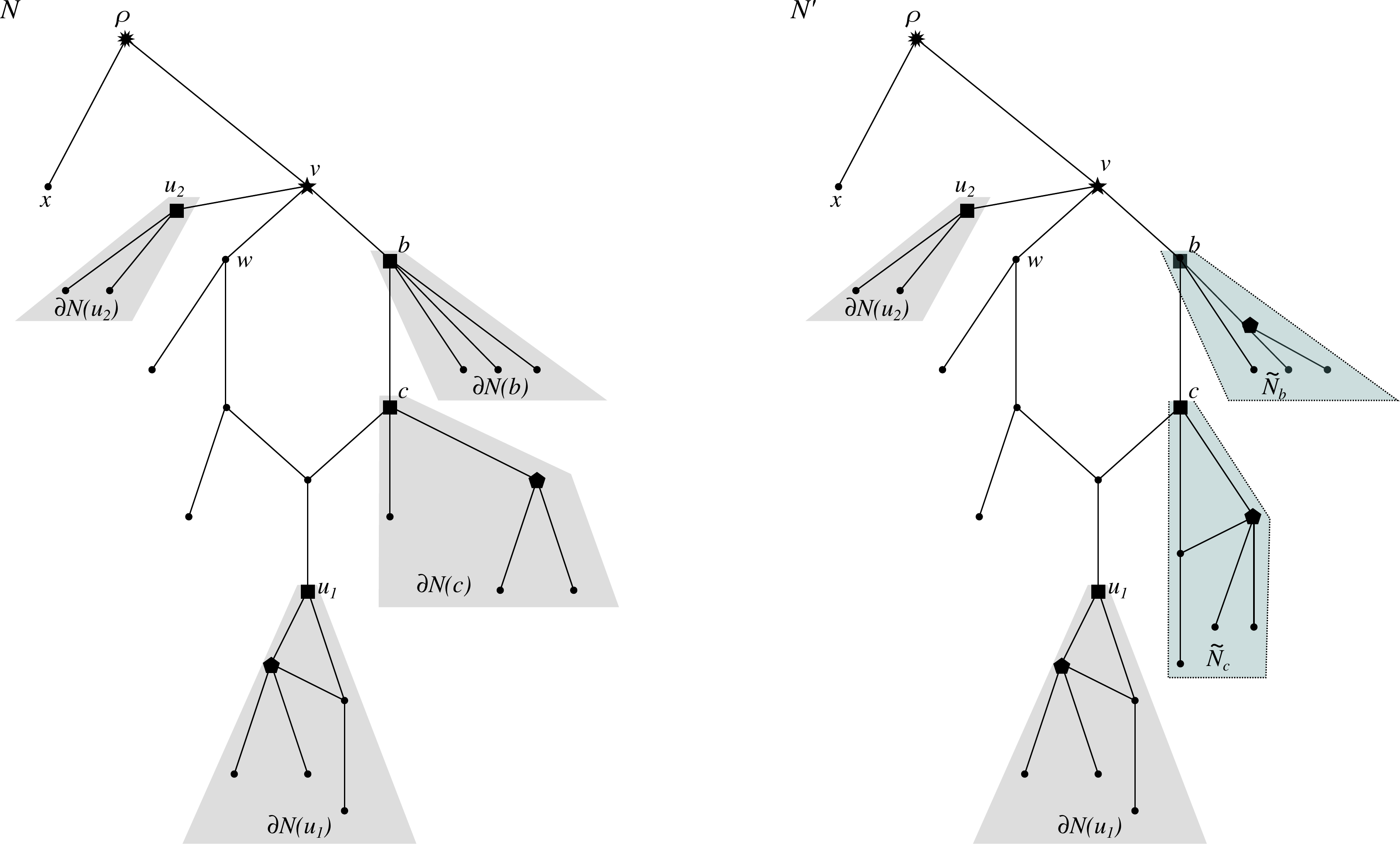}
\caption{
An example to illustrate Definition~\ref{def:rel-neighbor} of relevant neighbors
as well as Theorems~\ref{thm:local-phi} and \ref{thm:sum-up-part}. All
relevant vertices in $N$ and $N'$ are highlighted with a symbol different from
$\bullet$. 
\newline 
Consider first the network $N$ (left).  Since $\rho$ is not contained in a non-trivial
block and $v$ is relevant and a child of $\rho$, the vertices $\rho$ and $v$
satisfy Definition~\ref{def:rel-neighbor}(1). Moreover, $x$ is irrelevant in $N$ since it is a leaf. 
Hence, $\relN_N(\rho)=\{v\}$. Since $v$
is contained in a non-trivial block $B$, we need to employ Definition~\ref{def:rel-neighbor}(2) to detect the relevant neighbors
 of $v$.  To be more precise, the vertices $b$
and $c$ are relevant and contained in the same block $B$ as $v$ and thus satisfy
Definition~\ref{def:rel-neighbor}(2.a). The vertex $u_1$ satisfies  Definition~\ref{def:rel-neighbor}(2.b), since
there is an IR-free $wu_1$-path such that $w\in B$ is not
relevant in $N$ and $w\preceq_N v$. Moreover, $u_2$ is a relevant neighbor of
$\rho$, since it satisfies Definition~\ref{def:rel-neighbor}(2.c).
Since no further relevant vertices $z\prec_N v$ exist, we have $\relN_N(v)=\{b,c,u_1,u_2\}$.
The respective restricted subnetworks $\partN(z)$ of all relevant neighbors $z$ of $v$ are
highlighted in gray-shaded areas. Note that $N(v) = \partN(v)$. According to Corollary~\ref{cor:relevant-partN}, 
the relevant vertices in $N$ that are located in $\partN(v)$ are precisely
the relevant vertices in $\partN(v)$. 
Since $v$ is the only relevant neighbor of $\rho$ and by Theorem~\ref{thm:sum-up-part},
we obtain $\Phi^{**}(N) = (\Phi^{*}(\partN(v)) + \phi_N(v)\binom{13}{2}) + \phi_N(\rho)\binom{14}{2}$.
Moreover, by Theorem~\ref{thm:sum-up-part} (and, in particular, Corollary~\ref{fact:sum-up-part}), 
$\Phi^{*}(\partN(v)) = (\Phi^*(\partN(b)) + \phi_N(b)\binom{|L_N(b)|}{2})
		+(\Phi^*(\partN(c)) + \phi_N(c)\binom{|L_N(c)|}{2})
		+(\Phi^*(\partN(u_1)) + \phi_N(u_1)\binom{|L_N(u_1)|}{2})
		+(\Phi^*(\partN(u_2)) + \phi_N(u_2)\binom{|L_N(u_2)|}{2})
= (0+1\cdot 3) + (1+1\cdot 3) + (1 + (\epsilon + 1)\cdot 3)) + (0+1\cdot 1)		
= 12 + 3\epsilon$.
In addition, since $\phi_N(v)=\epsilon+\omega(B) = \epsilon + 8$
and  $\phi_N(\rho)=1$, we obtain 
$\Phi^{**}(N) = ((12+3\epsilon) + (\epsilon+8)\binom{13}{2}) + \binom{14}{2} = 727 + 81\epsilon$.
\newline
The network $N'$ (right) is obtained from $N$ by replacing  $\partN(b)$ and $\partN(c)$ by the network $\tilde N_b$ and $\tilde N_c$, 
respectively (highlighted by turquoise-shaded areas). According to Theorem~\ref{thm:local-phi}, 
$\Phi^{**}(N)-\Phi^{**}(N') =  (\Phi^{*}(\partN(b))-\Phi^{*}(\tilde N_b)) + ((\phi_N(b)-\phi_{N'}(b)) \binom{3}{2})
				+ (\Phi^{*}(\partN(c))-\Phi^{*}(\tilde N_c)) + ((\phi_N(c)-\phi_{N'}(c)) \binom{3}{2})
				=  (0-1) + ((1-1)\cdot 3) + (1-1)+((1-(\epsilon+1))\cdot 3) = -(1+3\epsilon)$.
}				
\label{fig:relevant_neighbors}
\end{figure}

To formally establish the corresponding statements, we first need some technical lemmas concerning relevant vertices, non-trivial blocks, partial neighborhoods, and their  mutual dependencies.

\begin{lemma}\label{lem:blocks-in-partN}
Let $N\in \mscr{N}_n$ be a level-$1$ network, $v\in V(N)$ be a relevant vertex, $\mathcal{B}'\subseteq \mathcal{B}^v(N)$ and $\child'\subseteq \child^*_{\overline{\mathcal{B}^v}(N)}$. Let $\partN' \coloneqq \partN(v,\mathcal{B}',\child')$ and $\partN \coloneqq \partN(v)$.
Then, the following statements hold:
\begin{enumerate}
    \item $\partN'$ and $\partN$ are level-$1$ networks with root $v$.
    \item All non-trivial blocks in $N$ are either entirely contained in $\partN'$ (respectively, $\partN$) or intersect $\partN'$ (respectively, $\partN$) at most in vertex $v$. 
    \item Every non-trivial block $B$ in $\partN'$ (respectively, $\partN$) satisfies $B\in \mathcal{B}(N)$.
    \item $\phi_N(w) = \phi_{\partN}(w)$ for all $w$ in $\partN$ and $\phi_N(w) = \phi_{\partN'}(w)$ for all $w\neq v$ in $\partN'$.
    \item Suppose that $N$ is phylogenetic. Then, $\partN$ is phylogenetic. Moreover, $\partN'$ is phylogenetic if $|\child'\cup \child^*_{\mathcal{B}'}|\geq 2$.
    \end{enumerate}
\end{lemma}

\begin{proof}
We show first that $\partN'$ is a level-$1$ network. By definition, $\partN'$ does only contain descendants of $v$ and since $N$ is acyclic, $\partN'$ is a DAG with unique root $v$ and, thus, a network. Moreover, $\partN'$ is a subgraph of $N$ and thus, cannot have a higher level than $N$, that is, $\partN'$ is level-$1$. By Observation~\ref{obs:part-partN} and since $\partN'$ is chosen arbitrarily, $\partN$ is a level-$1$ network. Hence, Statement (1) is satisfied.

We continue with showing that all non-trivial blocks of $N$ are either entirely contained in $\partN'$ or intersect $\partN'$ at most in vertex $v$. Let $B$ be a non-trivial block of $N$. Assume, for the sake of a contradiction, that $B$ is not entirely contained in $\partN'$ but intersects $\partN'$ in a vertex $w$ different from $v$. Since $B$ is not entirely contained in $\partN'$, the root $\rho_{B}$ of $B$ cannot be contained in $\partN'$ since, otherwise, it would hold that $\rho_{B}\preceq_N v$ which, for $\rho_B\in\partN'$, implies that all descendants of $\rho_{B}$ and, therefore, the entire block $B$ would be contained in $\partN'$. Without loss of generality let $w\neq v$ be a $\preceq_N$-maximal vertex in $\partN'$ that is also contained in $B$, i.e., there is no vertex $w'\neq v$ in $\partN'$ with  $w\prec_N w'$ that is contained in $\partN'$ and $B$.

\begin{figure}[ht]
    \centering
    \includegraphics[width=0.9\textwidth]{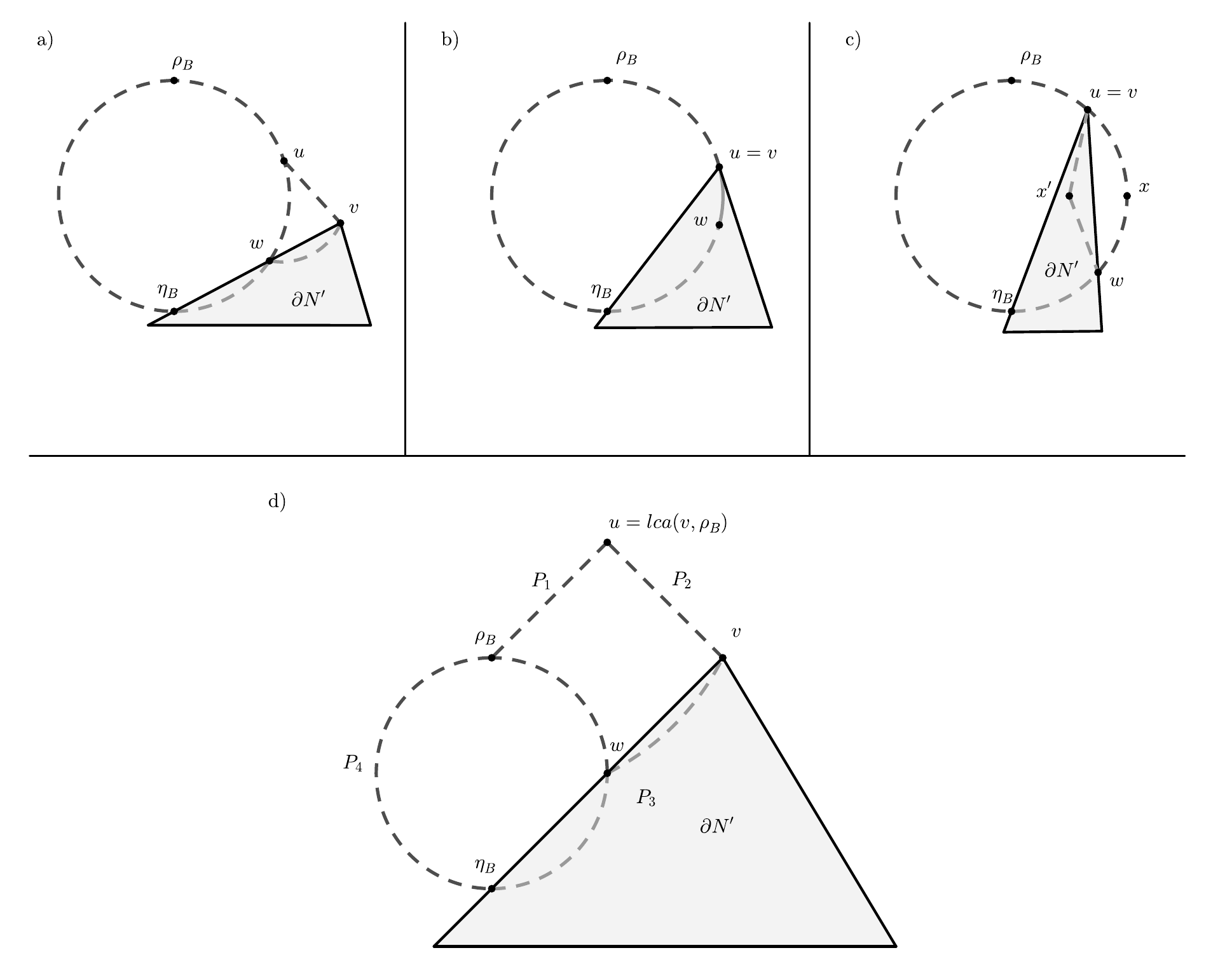}
    \caption{Depiction of the different cases used to create contradictions in the proof of Lemma \ref{lem:blocks-in-partN}(2). The dashed lines represent paths in the networks. In Panels a)-c), $v$ and $\rho_B$ are $\preceq_N$-comparable. Panel a) represents the subcase $v\neq V(B)$, Panel b) the case $w\in \child(v)$, and Panel c) the remaining subcase $v\in V(B)$ but $w\neq \child(v)$. In Panel d), $v$ and $\rho_B$ are $\preceq_N$-incomparable. }
    \label{fig:statement2}
\end{figure}

We now distinguish between the cases that $v$ and $\rho_B$ are $\preceq_N$-comparable or not (see Figure \ref{fig:statement2}). Assume first that $v$ and $\rho_B$ are $\preceq_N$-comparable. Hence, by the previous arguments, $v\prec_N \rho_B$. Since $w,v \prec_{N} \rho_B$ we can choose a $\preceq_N$-minimal vertex $u$ that satisfies $v,w \preceq_N u \preceq_N \rho_B$ and $u\in V(B)$. Either $v\in V(B)$ or not.  If $v$ is not contained in $B$, then $\vert \{u,v,w\}\vert =3$ and thus, there are pairwise internal vertex-disjoint $uv$-, $uw$-, and $vw$-paths that together form an undirected cycle in $N$. This means that $u,v,w$ are contained in a non-trivial block $B'$ that shares at least two vertices with $B$. This together with Observation~\ref{obs:identical-block} implies that $B'=B$; a contradiction to $v\notin V(B)$. If $v$ is contained in $B$, and $w$ is a child of $v$, then $w$ can not be contained in $\partN'$ as, since $v\neq \rho_B$, none of the children of $v$ in $B$ can be contained in $\child^*_N(v)$; a contradiction. If $v\in V(B)$ and $w$ is not a child of $v$, then there has to be a vertex $x$ with $w\prec_N x$ and $x\in V(B)$. Let $P$ the $vw$-path whose vertices are all in $B$. None of the inner vertices of $P$ can be contained in $\partN'$ as $w$ is maximal. Since $v,w\in \partN'$ and $w\prec_N v$ there is another $vw$-path $P'$ in $N$ containing some inner vertex $x'$, whose internal vertices are not in $B$. Hence, $P$ and $P'$ form an undirected cycle and $v,x'$ ,and $w$ are contained in a non-trivial block $B'$ that shares two vertices with $B$. Therefore, $B'=B$; a contradiction since $w$ is maximal but $w,x'\in \partN'$, $w,x\in V(B)$ and $w\prec_N x'$.
    
Assume that $v$ and $\rho_B$ are $\preceq_N$-incomparable. 
Let $u = \lca_N(v,\rho_{B})$, which is uniquely determined by Lemma \ref{L7.9HSS}. By definition of the lca, there are internal vertex disjoint paths $P_1$ and $P_2$ from $u$ to $\rho_{B}$ and $v$, respectively. Moreover, since $\eta_B\preceq_N w\prec_N v$, there is a path $P_3$ from $v$ to $\eta_B$ that contains $w$. In addition, there is a path $P_4$ from $\rho_B$ to $\eta_B$. Hence, the subgraph induced by the vertices in $P_1,\dots,P_4$ contains an undirected cycle that contains $v$ and $\rho_{B}$. Hence, $v$ and $\rho_{B}$ must lie in a common non-trivial block $B'$. Together with Observation~\ref{obs:identical-block} this implies that $B'=B$. But then, $v\preceq_N\rho_B$ as all vertices in $B$ are descendants of $\rho_B$; a contradiction.
    
In summary, all non-trivial blocks of $N$ are either entirely contained in $\partN'$ or intersect at most in vertex $v$ and, thus, Statement (2) is satisfied for $\partN'$. Moreover, since $\partN'$ is always an induced subgraph, no further non-trivial blocks can exist in $\partN'$. Hence, every non-trivial block $B$ in $\partN'$ satisfies $B\in \mathcal{B}(N)$ and, thus, Statement (3) is satisfied for $\partN'$. By Observation~\ref{obs:part-partN} and since $\partN'$ is chosen arbitrarily, all non-trivial blocks of $N$ are either entirely contained in $\partN$ or intersect $\partN$ at most in vertex $v$ and Statements (2) and (3) hold for $\partN$ as well.
    
Note that both $\phi_N(w)$ and $\phi_{\partN}(w)$ are defined in terms of the non-trivial blocks for which $w$ is the root. By the previous arguments, $\phi_N(w) = \phi_{\partN}(w)$ for all $w$ in $\partN$. The same is true for all $w\neq v$ in $\partN'$, since the only weight in $\partN'$ that may change is that of vertex $v$ as $\partN'$ does not necessarily contain all non-trivial blocks rooted in $v$. Therefore, Statement (4) is satisfied.

We proceed with showing that $\partN'$ is phylogenetic provided that $N$ is phylogenetic and $|\child'\cup \child^*_{\mathcal{B}'}|\geq 2$. Since $|\child'\cup \child^*_{\mathcal{B}'}|\geq 2$, there are at least two children of $v$ contained in $\partN'$ and so $\outdeg_{\partN'}(v) \geq 2$ which implies that (N2) is satisfied for $\partN'$ and $v$. Now, let $w$ be a vertex in $\partN'$ with $w\prec_N v$. Since $N$ is phylogenetic, we have $\indeg_N(w)\geq 2$ or  $\outdeg_N(w) \geq 2$. Assume first that $\indeg_{N}(w)\geq 2$. Hence, $w = \eta_B$ for some  non-trivial block $B$ in $N$. By  Lemma \ref{lem:results-B}\eqref{L3.23HSS}, all parents must be contained in $B$. As shown above, $B$ must be entirely contained in $\partN'$ and so the parents of $w$ are as well. Hence, $\indeg_{\partN'}(w)\geq 2$. Assume now that $\outdeg_{N}(w)\geq 2$. By definition of $\partN'$ and since $w\in \partN'$, the children of $w$ must also be contained in $\partN'$ and thus,  $\outdeg_{\partN'}(w)\geq 2$, which completes the proof for $\partN'$. To see that $\partN = \partN(v)$ is phylogenetic provided that $N$ is phylogenetic, observe first that $v$ is relevant. By Observation~\ref{obs:relevant}, there are at least two children of $v$ contained in $\partN$ and, thus, $\outdeg_{\partN}(v) \geq 2$ which implies that (N2) is satisfied for $\partN$ and $v$. Now, we can reuse exactly the same arguments as for $\partN'$ to conclude that $\partN$ is phylogenetic. Hence, Statement (5) is satisfied. 
\end{proof}

The next lemma establishes a correspondence of relevant vertices in $N$ and in $\partN'$.
\begin{lemma}\label{lem:relevant-partN}
Let $N\in \mathscr{N}_n$ be a level-$1$ network and $v$ be a relevant vertex of $N$. Moreover, let $\mathcal{B}'\subseteq \mathcal{B}^v(N)$, $\child'\subseteq \child^*_{\overline{\mathcal{B}^v}(N)}$, $\partN'\coloneqq \partN(v,\mathcal{B}',\child')$, and $w\in V(\partN')\setminus\{v\}$.  Then, $\child_N^*(w) = \child_{\partN'}^*(w)$ and $\child_{\partN'}(v) = \child_{\partN'}^*(v) = \child'\cup \child^*_{\mathcal{B}'}(v)$. In particular, $w$ is relevant in $N$  if and only if $w$ is relevant in $\partN'$. Moreover, if $|\child'\cup \child^*_{\mathcal{B}'}(v)|\geq 2$, then $v$ is relevant in $\partN'$.
\end{lemma}

\begin{proof}
Let $N$ and $\partN'\coloneqq \partN(v,\mathcal{B}',\child')$ be chosen as in the statement of the lemma and let $\mathcal{B}^v\coloneqq \mathcal{B}^v(N)$. Let $v\in V(N)$ be relevant in $N$ and $w\in V(\partN')\setminus\{v\}$.  Hence, $w\prec_N v$.
    
To recall, $\child_N^*(w)\subseteq \child_N(v)$ is the set of children of $w$ that are not contained in any non-trivial block that contains $w$ and for which $w$ is not the root. Lemma \ref{lem:blocks-in-partN} implies that all non-trivial blocks are entirely contained in $\partN'$ or intersect $\partN'$ at most in vertex $v\neq w$. Hence, every non-trivial block in $N$ that contains $w$ is entirely contained in $\partN'$. Therefore, $\child_N^*(w) = \child_{\partN'}^*(w)$. Thus, $|\child_N^*(w)|>1$ if and only if $|\child_{\partN'}^*(w)|>1$ which implies that $w$ is relevant in $N$ if and only if $w$ is relevant in $\partN'$.
	
Consider now vertex $v$. Since $v$ is the root of $\partN'$, there cannot be any non-trivial block in $\partN'$ that contains $v$ but for which $v$ is not the root. Hence, $\child_{\partN'}^*(v)  = \child_{\partN'}(v)$. By definition and construction, $\child'\cup \child^*_{\mathcal{B}'}(v) \subseteq \child_{\partN'}(v)$. We show now that $\child_{\partN'}(v)$ is a subset of $\child'\cup \child^*_{\mathcal{B}'}(v)$. Note that $\child'$ and $\child^*_{\mathcal{B}'}$ are defined in terms of $N$ and not of $\partN'$. Let $u\in \child_{\partN'}^*(v)  = \child_{\partN'}(v)$. Hence, by the previous arguments, either $u$ is not contained in any non-trivial block of $\partN'$ or $u$ is contained in a non-trivial block $B$ of $\partN'$ with $\rho_B=v$.
    
Assume first that $u$ is contained in a non-trivial block $B$ of $\partN'$ with $\rho_B=v$. By Lemma \ref{lem:blocks-in-partN}, $B\in \mathcal{B}(N)$ and still $\rho_B=v$. Hence, $u\in \child^*_B(v)$. Note that Observation~\ref{obs:dictintBlocks-disjointChild} implies that $\child^*_B(v)\cap \child^*_{B'}(v)=\emptyset$ for two distinct non-trivial blocks $B$ and $B'$ that are rooted in $v$. Hence, there is precisely one set $\child^*_B(v)\subseteq \child^*_{\mathcal{B}^v}$ that contains $u$ and since $B$ is a non-trivial block in $\partN'$ it must hold that $B\in \mathcal{B}'\subseteq \mathcal{B}^v$. Consequently, $u\in \child^*_{\mathcal{B}'}$.
    
Assume now that $u$ is not contained in any non-trivial block $B$ of $\partN'$ with root $\rho_B=v$. Note that if there is a non-trivial block of $N$ that contains $u$, then $B$ must be entirely contained in $\partN'$ since, otherwise, $B$ would intersect $\partN'$ in vertex $u\neq v$; a contradiction to Lemma \ref{lem:blocks-in-partN}. Taken these arguments together, there is no non-trivial block in $N$ that contains $u$ and thus, $u\in \child^*_{\overline{\mathcal{B}^v}}$. In this case, $v$ must be the unique parent of $u$ and thus, since $u\in child_{\partN'(v)}$, $u\in \child'$. 
    
Hence, $\child_{\partN'}(v) = \child'\cup \child^*_{\mathcal{B}'}$. Therefore, if $| \child'\cup \child^*_{\mathcal{B}'}|\geq 2$, then $|\child_{\partN'}(v)|\geq 2$ and, by definition, $v$ is relevant in $\partN'$. 
\end{proof}

Observation~\ref{obs:part-partN} together with the fact that $\partN'$ is chosen arbitrarily in Lemma \ref{lem:relevant-partN} implies
\begin{corollary}\label{cor:relevant-partN}
Let $N\in \mathscr{N}_n$ be a level-$1$ network and $v,w\in V(N)$ such that $v$ is relevant and $w\in V(\partN(v))$. Then, $w$ is relevant in $N$  if and only if $w$ is relevant in $\partN(v)$. 
\end{corollary}

The next lemma establishes some properties of phylogenetic level-$1$ networks and its relevant vertices when replacing certain subnetworks.

\begin{lemma}\label{lem:replacement-partN-1}
Let $N$ be a phylogenetic level-$1$ network and let $v$ be a relevant vertex in $N$. Moreover, let $N'$ be the directed graph obtained from $N$ by replacing $\partN(v)$ by a phylogenetic level-$1$ network $\tilde N$ with $L(\tilde N) = L_N(v)$. Then, the following statements are satisfied.
\begin{enumerate}
	\item $N'$ is a phylogenetic level-$1$ network with leaf set $L(N)$.
	\item All relevant vertices of $N$ that are not contained in $V(\partN(v))\setminus \{v\}$ remain relevant vertices in $N'$, and the relevant vertices in $\tilde N$ become relevant vertices in $N'$. 
	\item $\phi_N(w) = \phi_{N'}(w)$ for all relevant vertices $w$ in $(V(N)\setminus V(\partN(v))$.
\end{enumerate}
\end{lemma}

\begin{proof}
Let $N$, $\tilde N$, and $N'$ be as specified in the statement. Since $\partN(v)$ contains all descendants of $v$ with the properties as in Definition~\ref{def:relevant-part}, one easily sees that the digraph $N''$ obtained from removing $\partN(v)$, i.e., all descendants of $v$ (except for $v$ and those descendants located in non-trivial blocks for which $v$ is not the root)  and its incident edges, remains a connected directed graph and, in particular, a network. Since $N''$ is a subgraph of $N$, it remains level-$1$. Now, adding the network $\tilde N$ to $N''$ to obtain $N'$ by identifying its root with $v$ clearly results in a level-$1$ network since both $N''$ and $\tilde N$ are level-$1$ networks. Since $\tilde N$ is phylogenetic, the root $\rho_{\tilde N}$ must have out-degree greater than two, and thus $v$ has out-degree greater than two in $N'$. For all other vertices in $V(N)\setminus V(\partN(v))$ and $V(\tilde N)$ the degrees remain as in $N$ and $\tilde N$, respectively. Hence, $N'$ is phylogenetic.
    
Let $w$ be a relevant vertex of $N$ that is not contained in $V(\partN(v))\setminus \{v\}$. Assume first that $w=v$. By construction, and since $\tilde N$ is phylogenetic, $v$ has at least two children in $\tilde N$ that are also children of $v$ in $N'$. Hence, one easily verifies that Definition~\ref{def:relevant} is satisfied for $w=v$ in $N'$. Assume $w\neq v$ and thus, $w\notin V(\partN(v))$. Every non-trivial block in $N$ that may contain $w$ can, by Lemma \ref{lem:blocks-in-partN}, intersect $\partN(v)$ at most in one vertex and thus, remains a non-trivial block in $N'$. Hence, the property of being relevant in $N$ does not depend on the structure of $\partN(v)$ and, therefore, $w$ remains relevant in $N'$. This also implies that $\phi_N(w) = \phi_{N'}(w)$ for all relevant vertices $w$ in $(V(N)\setminus V(\partN(v))$. Now, let $w$ be a relevant vertex in $\tilde N$. By construction, $N''$ and $\tilde N$ intersect only in vertex $v$. Note that $\partN'(v)=\tilde N$. By Corollary~\ref{cor:relevant-partN}, $w$ is a relevant vertex in $N'$. This completes the proof.   
\end{proof}

We are now in a position to show the locality of the weighted total cophenetic index.

\begin{proof}[Proof of Theorem~\ref{thm:local-phi}]
Let $N, N',\tilde N$, and $v$ be as specified in the statement. By construction $n\coloneqq |L_N(\rho_N)| = |L_{N'}(\rho_N')|$. Moreover, let $\relV$ and $\relV'$ denote the set of relevant vertices in $N$ and $N'$, respectively. Let us first denote with $R$, respectively, $R'$ the set of relevant vertices in $N$, respectively $N'$ that are not located in $V(\partN(v))\setminus \{v\}$, respectively, $V(\tilde N)\setminus \{\rho_{\tilde N}\}$. By Lemma \ref{lem:replacement-partN-1}, $R=R'$. Let $R^v$ and $\tilde R$ be the set of relevant vertices in $\partN(v)$ and $\tilde N$, respectively. By Corollary~\ref{cor:relevant-partN} and Lemma~\ref{lem:replacement-partN-1}, $R^v$ and $\tilde R$ are relevant vertices in $N$, respectively, $\tilde N$. These arguments imply that $\relV = R\cup R^v$ and $\relV' = R\cup \tilde R$. By construction, $L(\tilde N) = L(\partN(v))$ and  $L_N(v)= L_{N'}(v)$. Hence, \[\Phi^{**}(N) = \Phi^*(N) + \phi_N(\rho_N) \binom{n}{2} = \sum_{w\in \relV(N)\setminus\{\rho_N\}} \phi(w) \binom{|L_N(w)|}{2} +  \phi_N(\rho_N) \binom{n}{2}\] can be rewritten as \[\Phi^{**}(N)= \sum_{w\in R} \phi_N(w) \binom{|L_N(w)|}{2} + \sum_{w\in R^v}\phi_N(w) \binom{|L_N(w)|}{2} =\sum_{w\in R} \phi_N(w) \binom{|L_N(w)|}{2}  + \Phi^*(\partN(v)) + \phi_N(v)\binom{|L_N(v)|}{2}.\] 
By Lemma \ref{lem:replacement-partN-1}, $\phi_N(w) = \phi_{N'}(w)$ for all $w\in R=R'$, and $ \phi_{\tilde N}(w) = \phi_{N'}(w)$ for all $w\in \tilde R = \relV'\setminus R$. Moreover, by construction, $L_N(w) = L_{N'}(w)$ for all $w\in R$. Hence, we can rewrite $\Phi^{**}(N')$ as \[\Phi^{**}(N') = \sum_{w\in R} \phi_N(w) \binom{|L_N(w)|}{2} + \sum_{w\in \tilde R}  \phi_{N'}(w) \binom{|L_{N'}(w)|}{2} =  \sum_{w\in R} \phi_N(w) \binom{|L_N(w)|}{2} + \Phi^*(\tilde N) + \phi_{N'}({\rho_{\tilde N}})\binom{|L_{N'}(\rho_{\tilde N})|}{2}.\] 
Observe that $\sum_{w\in R^v}\phi_N(w) \binom{|L_N(w)|}{2} = \Phi^*(\partN(v)) + \phi_N(v)\binom{|L_N(v)|}{2}$
and $\sum_{w\in \tilde R} \phi_{N'}(w) \binom{|L_N(w)|}{2}= \Phi^*(\tilde N) + \phi_{N'}({\rho_{\tilde N}})\binom{|L_{N'}(\rho_{\tilde N})|}{2}$. Moreover, since by assumption $L(\tilde N) = L(\partN(v))$ and $\rho_{\tilde N}$ is identified with $v$, we can observe that $L_N(v)=L_{N'}(\rho_{\tilde N})$. Thus,
\begin{align*}
\Phi^{**}(N)-\Phi^{**}(N') 
    =& \  \sum_{w\in R} \phi_N(w) \binom{|L_N(w)|}{2} + \Phi^*(\partN(v)) +                                          \phi_N(v)\binom{|L_N(v)|}{2} -  \sum_{w\in R} \phi_N(w) \binom{|L_N(w)|}{2} \\
    &- \Phi^*(\tilde N) - \phi_{N'}({\rho_{\tilde N}})\binom{|L_N(v)|}{2} \\
    =&  \ \Phi^*(\partN(v)) + \phi_N(v)\binom{|L_N(v)|}{2} - \Phi^*(\tilde N) - \phi_{N'}({\rho_{\tilde N}})\binom{|L_N(v)|}{2}\\
    =& \left(\Phi^{*}(\partN(v)-\Phi^{*}(\tilde N)\right) + \left((\phi_N(v)-\phi_{N'}({\rho_{\tilde N}})) \binom{|L(v)|}{2}\right).
\end{align*} 
This completes the proof.
\end{proof}

We can easily conclude a similar statement for phylogenetic trees, which coincides with the local property of the total cophenetic index $\Phi$ for phylogenetic trees as stated in \cite{Mir2013}.

\begin{corollary}[{\cite[Lemma 5]{Mir2013}}]
The weighted total cophenetic index $\Phi^{**}$ is local for trees, i.e., if $T'$ is a phylogenetic tree that is obtained from a phylogenetic tree $T$ by replacing the subtree $T(v)$ rooted in $v\in V(T)$ by a phylogenetic tree $\tilde T$ with $L(\tilde T) = L(T(v))$, then \[\Phi^{**}(T)-\Phi^{**}(T') = \Phi^{**}(T(v))-\Phi^{**}(\widetilde T)= \Phi^{*}(T(v))-\Phi^{*}(\widetilde T) = \Phi(T(v))-\Phi(\widetilde T). \]
\end{corollary}

\begin{proof}
By construction, $n\coloneqq |L(\tilde T)| = |L(T(v))|$. Since $T$ and $\tilde T$ are phylogenetic trees, it holds that $\phi_T(w)=1$ and $\phi_{\tilde T}(w)=1$ for all vertices $w$ in $T$ and $\tilde T$, respectively. In particular, $\phi_{\tilde T}({\rho_{\tilde N}}) = \phi_T(v)$. Note that by Definition~\ref{def:relevant-part}(a), we have $\partT(v) = T(v)$ since $T$ does not contain non-trivial blocks. This together with Theorem~\ref{thm:local-phi} implies
\begin{align*}
\Phi^{**}(T)-\Phi^{**}(T') 
 =&\  \left(\Phi^{*}(T(v))-\Phi^{*}(\tilde T)\right) + \left((1-1) \binom{|L(v)|}{2}\right) = \Phi^{*}(T(v))-\Phi^{*}(\tilde T)\\
\overset{\text{Obs.\ \ref{obs:cp}}}{=}&\  \Phi(T(v))  -\Phi(\tilde N)  = \Phi(T(v)) + \binom{n}{2}  -\Phi(\tilde N) -\binom{n}{2}   \\
\overset{\text{Obs.\ \ref{obs:cp}}}{=}&\ \Phi^{**}(T(v))-\Phi^{**}(\tilde T),
\end{align*} 
which completes the proof.
\end{proof}

We now provide a sufficient condition that ensures that the root of a network is the only
relevant vertex.
\begin{lemma}\label{lem:rel_neigh_empty}
If the relevant neighborhood $\relN(\rho_N)$ of the root $\rho_N$ of a phylogenetic network $N$ is empty, then $\rho_N$ is the only relevant vertex in $N$.
\end{lemma}
\begin{proof}
Suppose that $\relN(\rho_N)=\emptyset$ for $\rho_N$ the root of a phylogenetic network $N$. 
By Observation\ \ref{obs:relevant}, $\rho_N$ is a relevant vertex in $N$.
Assume, for contradiction, that there is an additional relevant vertex in $N$. 
In this case, there is, in particular, a relevant vertex $u$ such that
no possibly further relevant vertex $u'$ satisfies $u\prec_N u'\prec_N \rho_N$. 
Hence, there is an IR-free $\rho_Nu$-path $P$. Since $\relN(\rho_N)=\emptyset$, 
none of the conditions in Definition~\ref{def:rel-neighbor} can be satisfied
for $\rho_N$ and $u$. Assume first that $\rho_N$ is not the root of a non-trivial
block $B$ in $N$. Since Definition~\ref{def:rel-neighbor}(1) is not satisfied for $\rho_N$ and $u$, 
$u$ cannot be a child of $\rho_N$. By the arguments preceding Definition~\ref{def:rel-neighbor}, 
$u$ must be a leaf and is, therefore, irrelevant; a contradiction. Hence, 
$\rho_N$ must be the root of a non-trivial block.
Since Definition~\ref{def:rel-neighbor}(2.a) is not satisfied for $\rho_N$ and $u$, 
vertex $u$ cannot be contained in any $B \in  \mathcal{B}_N^{\rho_N}$. This together with $\relN(\rho_N)=\emptyset$
implies that $u$ is not a child of $\rho_N$ as, otherwise,  Definition~\ref{def:rel-neighbor}(2.c)
would be satisfied. Note that along the IR-free $\rho_Nu$-path there cannot be
internal vertices that are true tree vertices $u'$ as they are relevant 
and would satisfy $u\prec_N u'\prec_N \rho_N$; a contradiction to 
the choice of $u$. The latter arguments and the choice of $u$ implies that
an IR-free $\rho_Nu$-path $P$ contains at least one internal vertex $w$ such that
$w$ is irrelevant and $w$ is contained in some non-trivial block  $B \in  \mathcal{B}_N^{\rho_N}$. But then the sub-path of $P$ connecting $w$ and $u$ is an IR-free $wu$-path  
and thus, Definition~\ref{def:rel-neighbor}(2.b) is satisfied; a contradiction. 
Consequently, $N$ does not contain any relevant vertex except the root $\rho_N$.
\end{proof}

We now provide sufficient conditions that ensure that distinct relevant neighborhoods are vertex-disjoint.
\begin{lemma}\label{lem:disjoint-relN} 
Let $N\in \mscr{N}_n$ be a phylogenetic level-$1$ network and $u,v\in V(N)$. Then, $\partN(u)$ and $\partN(v)$ are vertex disjoint, if $u$ and $v$ satisfy (at least) one of the following two conditions:
\begin{enumerate}
    \item $u$ and $v$ are $\preceq_N$-incomparable 
    \item $v$ is contained in a non-trivial block $B$ of $N$ but $v\neq \rho_B$ and either 
        \begin{enumerate}
            \item $u$ is contained in $B$ as well and $u\neq \rho_B$; or 
            \item there is an IR-free $wu$-path such that $w\in B$ is not relevant in $N$ and $w\prec v$.
        \end{enumerate}
\end{enumerate}
\end{lemma}

\begin{proof}
Let $N\in \mscr{N}_n$ be a phylogenetic level-$1$ network with $u,v \in V(N)$. Consider first Condition (1) and assume that $u$ and $v$ are $\preceq_N$-incomparable. By Lemma \ref{lem:results-B}\eqref{L7.1HSS}, neither $L_N(u)\subseteq L_N(v)$ nor $L_N(v)\subseteq L_N(u)$ is satisfied. Hence, $L_N(v)$ and  $L_N(u)$ are either disjoint or have a non-empty overlap. If  $L_N(v) \cap L_N(u) = \emptyset$, then $L_{\partN}(v)\subseteq L_N(v)$ and $L_{\partN}(u)\subseteq L_N(u)$ imply that $L_{\partN}(v) \cap L_{\partN}(u) = \emptyset$. Hence, $\partN(v)$and $\partN(u)$ must be vertex disjoint, since $\partN(v)$ (respectively, $\partN(u)$) contains precisely all ancestors $w\preceq_N v$ (respectively, $w\preceq_N u$) of the leaves in $L_{\partN}(v)$ (respectively, $L_{\partN}(u)$). Assume that $L_N(v) \cap L_N(u) \neq \emptyset$. In this case, Lemma \ref{lem:results-B}\eqref{L3.35HSS} implies that $u,v$ are contained in some non-trivial block $B$ of $N$ and, by Lemma \ref{lem:results-B}\eqref{L3.37HSS}, $L_N(v) \cap L_N(u) = L_N(\eta_B)$. By construction, neither $L_{\partN}(v)$ nor $L_{\partN}(u)$ contains $L_N(\eta_B)$ and, since $L_{\partN}(v)\subseteq L_N(v)$ and $L_{\partN}(u)\subseteq L_N(u)$ it holds that $L_{\partN}(v) \cap L_{\partN}(u) = \emptyset$. By the same arguments as before, $\partN(v)$ and $\partN(u)$ must be vertex disjoint.
    
Consider now Condition (2.a) and assume $u$ and $v$ are contained in the same non-trivial block of $N$ and $u,v\neq \rho_B$. The case that $u$ and $v$ are $\preceq_N$-incomparable was already treated in Case (1). Thus, assume that $u$ and $v$ are $\preceq_N$-comparable. Without loss of generality let $u\prec_N v$. By construction,  $\partN(v)$ is the subnetwork rooted in $v$ without any proper descendants of vertices $w\prec_N v$ with $w\in B$. Moreover, $u \not\preceq_N w$ for any vertex $w$ of $\partN(v)$ since $N$ is acyclic and level-$1$. Hence, $\partN(v)$and $\partN(u)$ are vertex disjoint.
    
Consider now Condition (2.b) and assume $v$ is contained in a non-trivial block $B$ of $N$ such that $v\neq \rho_B$ and there is an IR-free $wu$-path connecting the non-relevant vertex $w\in B$ and $u$ with $w\prec v$. Let $w'$ be the $\preceq_N$-minimal vertex in $B$ that is also contained in this $wu$-path. Since $v$ and $w'$ are both contained in the same block $B$, we can apply Condition (2.a) to conclude that $\partN(v)$ and $\partN(w')$ are vertex disjoint. By definition of $\partN(w')$,  $\partN(u) $ is a subgraph of $\partN(w')$ and, therefore, $\partN(v)$ and $\partN(u)$ are vertex disjoint.
\end{proof}

We are now in the position to show that $\Phi^{**}(N)$ is recursive, i.e.,  it can be expressed as the sum of the weighted cophenetic indices of the  networks rooted in relevant neighbors of the root $\rho_N$.

\begin{proof}[Proof of Theorem~\ref{thm:sum-up-part}]
The proposition trivially holds if $\rho_N$ has no relevant neighbors. In this case, by Lemma \ref{lem:rel_neigh_empty}, there are no additional relevant vertices and $\Phi^{**}(N)=\phi_N \left(\rho_N \right)\binom{n}{2}$. 

Hence, we assume in the following that $\rho_N$ has  relevant neighbors
and thus $\relV'\coloneqq \relV_N\setminus \{\rho_N\}$ is not empty.
Let $v_1,\dots,v_k$ be all relevant neighbors of $\rho_N$. By definition, it suffices to show that $\sum_{i=1}^k \left(\Phi^*(\partN_i) +\phi(v_i)\binom{n_i}{2}\right) = \Phi^{*}(N) = \sum_{v\in \relV'} \phi(v) \binom{|L_N(v)|}{2}$. 
 Let $R_i$ be the sets of vertices in $\partN(v_i)$ that are relevant in $\partN(v_i)$. We show first that $R_1,\dots,R_k$ partition the set $\relV'$. Note that $n_i \coloneqq |L_N(v_i)| \neq |L(\partN(v_i))|$ may hold.

Relevant neighbors of $\rho_N$ are based first on the properties of $\rho_N$, i.e., either (1) $\rho_N$ is not contained in a non-trivial block or (2) $\rho_N$ is contained in a non-trivial block $B$.

\begin{figure}[ht!]
    \centering
    \includegraphics[width=0.9\textwidth]{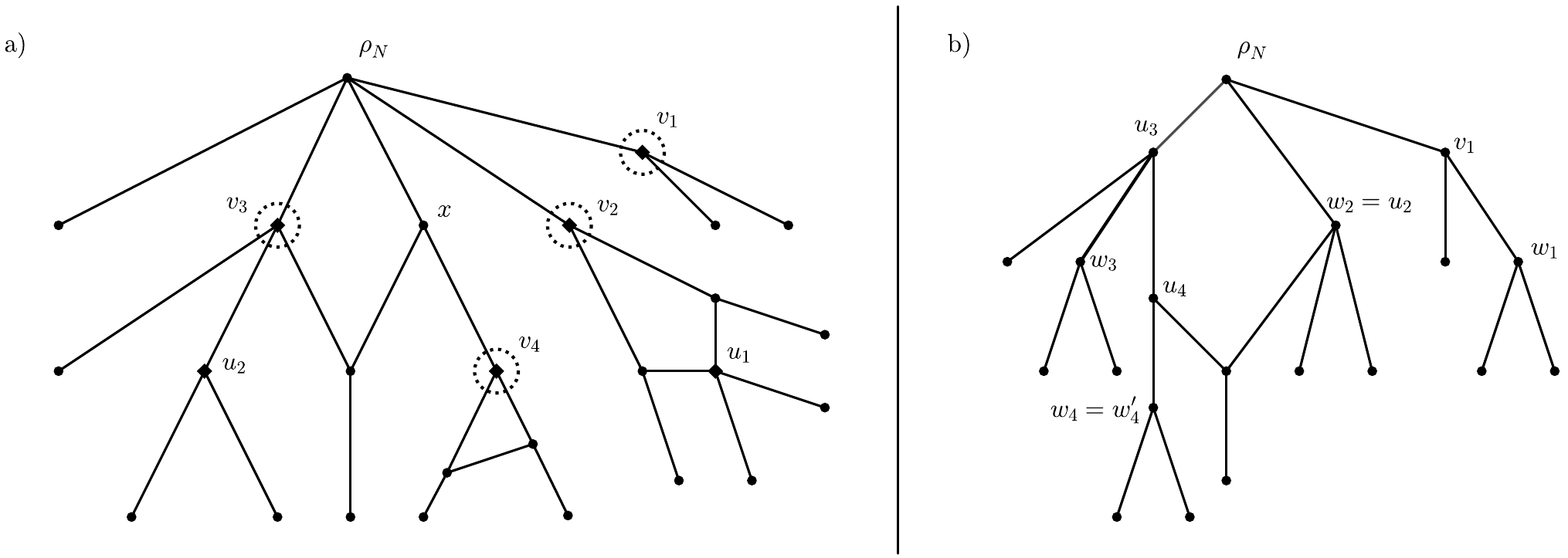}
    \caption{
    Panel a) shows an example network $N$ illustrating the notations for sets $C, C''$, and $R_i$ used in the proof of Theorem~\ref{thm:sum-up-part}. $C$ is comprised of vertices $v_3$ and $x$, whereas $C''$ consists of $v_1$ and $v_2$. The relevant neighbors (dotted circles) of $\rho_N$ are $v_1,v_2,v_3$, and $v_4$. The sets of relevant vertices in $\partN(v_1),\ldots,\partN(v_4)$ (diamond shaped vertices) are as follows: $R_1=\{v_1\}$, $R_2=\{v_2,u_1\}$, $R_3=\{v_3,u_2\}$, and $R_4=\{v_4\}$. In this example it is obvious that $relV\setminus{\{\rho_N\}}=relV'=\cupdot_{i=1}^4 R_i$.
    Panel b) shows an example network illustrating the different possible locations of relevant vertices used in the proof of Theorem~\ref{thm:sum-up-part} to show that $\relV' \subseteq \cupdot_{i=1}^k R_i$.
    Vertex $w_1$ is a descendant of $v_1$, that is, a descendant of a relevant neighbor of $\rho_N$ that is contained in $C'$. In this example, $w_1$ fulfills Condition~(2.c) of Definition~\ref{def:rel-neighbor}, but if $\rho_N$ was not the root of some non-trivial block, $w_1$ would fulfill Condition~(1) of that definition.
    Vertex $w_2$ is in the same block $B$ as $\rho_N$ and therefore fulfills Condition~(2.a) of Definition~\ref{def:rel-neighbor}. In this case, the $\preceq_N$-minimal vertex $u_2$ that is also in $B$ and on the path from $\rho_N$ to $w_2$ is $w_2$ itself. 
    Both $w_3$ and $w_4$ fulfill Condition~(2.b) of Definition~\ref{def:rel-neighbor}, and $u_3$ and $u_4$ are the $\preceq_N$-minimal vertices in $B$ on the path from $\rho_N$ to $w_3$ and $w_4$ respectively. Vertex $u_3$ is itself a relevant vertex, whereas $u_4$ is not. Here we define vertex $w_4'$ with $w_4'\prec_N u_4$ such that $w_4'$ is on the $\rho_Nw_4$ path and relevant.}
    \label{fig:proof_recursive}
\end{figure}

Consider Case (1). In this case, any child $w$ of $\rho_N$ must be either the root of some non-trivial block or not, where in both cases, $\child_N^*(w) = \child_N(w)$. Since $N$ is phylogenetic,  $|\child_N^*(w)| = |\child_N(w)|>1$ for all non-leaf vertices $w$. Together with Definition~\ref{def:rel-neighbor}, this implies that every child of $\rho_N$ that is not a leaf is relevant. In other words, all relevant neighbors of $\rho_N$ are children of $\rho_N$ as for any additional relevant vertices $u$ with $u\prec v$ for some child $v$ of $\rho_N$, there would be no IR-free paths. Let $R\subseteq \child_N(\rho_N)$ be the set of relevant neighbors of $\rho_N$. Note that $R=\emptyset$ is possible if all children of $\rho_N$ are leaves. 

Distinct children $v_i,v_j\in R$ of $\rho_N$ must be $\preceq_N$-incomparable, since otherwise, we would have, say $v_i\preceq_N v_j$, and so, $\rho_N$ would be contained in a non-trivial block; a contradiction. Lemma \ref{lem:disjoint-relN} implies that $\partN(v_i)$ and $\partN(v_j)$ are vertex-disjoint. In particular, we have $V(N) = \{\rho_N\}\cup \bigcupdot_{v\in R} V(\partN(v))$. By Corollary\ \ref{cor:relevant-partN}, all vertices in $\partN(v_i)$ are relevant in $\partN(v_i)$ if and only if they are relevant in $N$.

Consider Case (2). Let $B_1,\dots,B_r$, $r\geq 1$, be the non-trivial blocks that are rooted in $\rho_N$. We can bipartition $\child_N(\rho_N)$ into
 the set $C$ of children that are contained in some $B_\ell$ with $\ell\in \{1,\dots,r\}$ and the set $C'\coloneqq \child_N(\rho_N)\setminus C$. Moreover, let $C''\coloneqq C'\setminus L(N)$. 
See Figure~\ref{fig:proof_recursive}(a) for an illustrative example. Definition \ref{def:rel-neighbor}(2.c) applies for children in $C''$ (cf.\ e.g.,  $v_1$ and $v_2$ in Figure \ref{fig:proof_recursive}(a)) and these children satisfy the same conditions as the children in Case (1) and, thus, all vertices in $C''$ are relevant. Let $u$ be a relevant neighbor of $\rho_N$ that is not contained in $C''$ (as e.g.,  $v_3$ and $v_4$ in Figure~\ref{fig:proof_recursive}(a)). Hence, Definition\ \ref{def:rel-neighbor}(2.a) or (2.b) applies for $u$ and, thus, either (i) $u$ is contained in some non-trivial block $B_\ell$ or (b) there is an IR-free $wu$-path for some non-relevant vertex $w\in B_\ell$ and for which $u\prec_N w$. For each relevant neighbor $v_i$ of $\rho_N$ we have
$v_i\prec_N \rho_N$ and, since $N$ is acyclic, none of these relevant neighbors $v_i$ can be the root of any of the $B_\ell$, $1\leq \ell\leq r$. One easily verifies that all vertices in $B_\ell$, $1\leq \ell\leq r$, are $\preceq_N$-incomparable with the vertices in $C''$. Taken these arguments together, we can apply Lemma \ref{lem:disjoint-relN} to all relevant neighbors $v_1,\dots,v_k$ of $\rho_N$ and conclude that $\partN(v_i)$ and $\partN(v_j)$ are vertex-disjoint for $i\neq j$. Consequently, $R_i \cap R_j = \emptyset$, $1\leq i<j\leq k$. Therefore it remains to show that $\relV' = \cupdot_{i=1}^k R_i$. By Corollary~\ref{cor:relevant-partN}, $\cupdot_{i=1}^k R_i \subseteq \relV'$ holds.

We show now that $\relV' \subseteq \cupdot_{i=1}^k R_i$. To this end, let $w\in \relV'$. Thus, $w\prec_N \rho_N$. We now distinguish several cases (see Figure~\ref{fig:proof_recursive}(b) for an illustrative example). If $w\preceq v_i$ for some $v_i\in C'$, then by Definition, $w\in V(\partN(v_i))$ and, by Corollary~\ref{cor:relevant-partN}, $w\in R_i$ 
(in Figure~\ref{fig:proof_recursive}(b), this case is represented by the vertices $w_1\preceq v_1$).  Assume now that $w\not\preceq v_i$ for some $v_i\in C'$. Hence, $w\preceq_N u$ for some vertex $u \prec_N \rho_N$ that is located in some non-trivial block $B$ rooted in $\rho_N$. Let $P$ be a $\rho_Nw$-path and assume without loss of generality that $u$ is a $\preceq_N$-minimal vertex in $B$ that is located on $P$. If $u$ is relevant, then by Definition~\ref{def:rel-neighbor}(2i), $u$ is a relevant neighbor of $\rho_N$ and thus, $u = v_i$ for some $i\in \{1,\dots,k\}$. By Definition~\ref{def:relevant-part}, $w\in \partN(u)$, and Corollary~\ref{cor:relevant-partN} implies that $w\in R_i$
(in Figure~\ref{fig:proof_recursive}(b) this case applies to $w_2$ and $w_3$ with $u_2$ and $u_3$ respectively). Assume that $u$ is not relevant in $N$. Then $w\prec_N u$ and there is a vertex $w'\prec_N u$ located on $P$ that is relevant and such that the $uw'$-path $P'$ that is a sub-path of $P$ is an IR-free path. In other words, there is a relevant neighbor $v_j = w'$ of $\rho_N$ with $v_j\prec_N u$. Note that there cannot be a non-trivial block $B'\neq B$ for which $v_j$ is not the root (cf.~Lemma \ref{lem:results-B}\eqref{L3.19HSS}). Hence, only Definition~\ref{def:relevant-part}(a) can apply to $v_j$, and so $\partN(w') = N(w')$ and thus, $w$ is a vertex in $\partN(w')$. Corollary~\ref{cor:relevant-partN} implies that $w\in R_j$
(in Figure~\ref{fig:proof_recursive}(b), this last case is represented by $w_4=w_4'$).
In summary, $\relV' \subseteq \cupdot_{i=1}^k R_i$ and, therefore, 
$R_1,\dots,R_k$ partition the set $\relV'$.

Lemma~\ref{lem:blocks-in-partN} now implies $\phi_N(w) = \phi_{\partN(v)}(w)$ for all $w$ in $\partN(v)$. Taken the previous arguments together, we can write $\sum_{v\in \relV'} \phi_N(v)\binom{|L_N(v)|}{2} =\sum_{i=1}^k \sum_{w\in R_i}\phi_{\partN(w)}(w)\binom{|L_N(w)|}{2} =\sum_{i=1}^k (\Phi^*(\partN_i) +\phi_N(v_i)\binom{n_i}{2})$. This together with
$\Phi^{**}(\partN_i) = \Phi^*(\partN_i) + \phi(v_i)\binom{n_i}{2}$ completes the proof.
\end{proof}

When studying maximum and minimum phylogenetic level-$1$ networks in Section~\ref{sec:extrema}, we will often make use of a combination of Theorem~\ref{thm:local-phi} and the following corollary of Theorem~\ref{thm:sum-up-part}
\begin{corollary}\label{fact:sum-up-part}
   Let $N\in \mathscr{N}_n$ be a phylogenetic level-$1$ network and let $v_1,\dots,v_k$ be the relevant neighbors of the root $\rho_N$. Let $\phi(v_i) \coloneqq \phi_N(v_i)$, $n_i\coloneqq |L_N(v_i)|$, and $\partN_i\coloneqq \partN(v_i)$, $1\leq i\leq k$. Then, \[\Phi^{*}(N) = \sum_{i=1}^k \left(\Phi^*(\partN_i) + \phi(v_i)\binom{n_i}{2}\right). \]
\end{corollary}
When applied to phylogenetic trees, Theorem~\ref{thm:sum-up-part} yields the following well-known result for the total cophenetic index $\Phi$ on trees.

\begin{corollary}[{\cite[Lemma 4]{Mir2013}}]
Let $T\in \mscr{T}_n$ be a phylogenetic tree and $T_1,\dots,T_k$ be the subtrees rooted in the children of $\rho_T$. Then, \[\Phi^{**}(T) = \sum_{i=1}^k \Phi^{**}(T_i) + \binom{n}{2}.\]
\end{corollary}
\begin{proof}
In a phylogenetic tree $T$, the relevant neighbors of $\rho_T$ are precisely the children of $\rho_T$. 
By Lemma~\ref{lem:blocks-in-partN} and Theorem~\ref{thm:sum-up-part}, we have  
$\Phi^{**}(T) = \sum_{i=1}^k \left(\Phi^*(T_i) + \phi(v_i)\binom{n_i}{2}\right) + \phi(\rho_N)\binom{n}{2} =\sum_{i=1}^k \Phi^{**}(T_i) + \binom{n}{2}.$
\end{proof}

Note that, if $v$ is a leaf in $N$, then $V(\partN(v)) = L_N(v) = \{v\}$ and, therefore, 
    $\Phi^*(\partN(v))+\phi_N(v)\binom{|L_N(v)|}{2}=0+1\cdot \binom{1}{2} = 0$.

\section{Structure of blocks with extremal weights}
\label{APPX:sec:structure-blocks}

The index $\Phi^{**}$ of general level-$1$ phylogenetic networks involves the value $\omega(B)$ of non-trivial blocks $B$. Hence, to understand $\Phi^{**}$, we first investigate the structure of blocks that minimize and maximize $\omega$. As we shall see, lanterns, crescents, and full-moons defined earlier will play crucial role in this regard.
We start with 
\begin{fact} \label{obs:Bm3-isom}
For all phylogenetic level-$1$ networks $N$, every block in $\mathcal{B}_3(N)$, i.e., every block of size three, is simultaneously a lantern, a crescent, and a full-moon as there is only one such block (up to isomorphism), cf. Figure~\ref{fig:lantern-crescent}.
\end{fact}

To characterize blocks that maximize and minimize $\omega$ we first need
\begin{lemma}\label{lem:not-crescent}
Let $N$ be a phylogenetic level-$1$ network and
let $B \in \mathcal{B}_m(N)$ such that $B$ is not a crescent. Then,  $m>3$  and there is a vertex $u\in V^-(B)\cup \{\rho_B\}$ that has at least two children that are contained in $B$ and that are both distinct from $\eta_B$. Moreover, there are two internal vertex-disjoint $u\eta_B$-paths in $B$ that both contain at least one internal vertex of $B$ that is distinct from $u$.
\end{lemma}
\begin{proof} 
Assume that $B \in \mathcal{B}_m(N)$ is not a crescent. By Observation~\ref{obs:Bm3-isom}, we have $m>3$. We now show that there are internal vertices in $B$ that are $\prec_N$-incomparable. Assume, for contradiction, that all internal vertices of $B$ are $\prec_N$-comparable. Since $B$ is acyclic, $\prec_N$ imposes a total order among the $m-2$ vertices in $V^-(B)$, say $v_{m-2} \prec_N v_{m-1}\prec_N\cdots \prec_N v_1$. This implies that $B$ contains a directed path $(v_1,\dots,v_{m-2})$. At least one of these vertices is incident to $\eta_B$ and thus, all internal vertices and $\eta_B$ can be reached from $v_1$ by some directed path. Since, however, $B$ contains the unique root $\rho_B$, it must hold that $(\rho_B,v_1)\in E(B)$. This implies that $B$ contains a directed path $(\rho_B,v_1,\dots,v_{m-2})$. Since $\eta_B\neq v_{m-2}$ and $N$ is level-$1$, $v_{m-2}$ must have a unique parent $p$ in $B$. Since $B$ is acyclic, there cannot be an edge $(v_{m-2},w)$ for all $w\in V^-(B)\cup \{\rho_B\}$. But since $B$ is biconnected, $v_{m-2}$ must have some child, since otherwise $p$ would be a cut-vertex. There is, however, only one remaining choice for children of $v_{m-2}$, namely $\eta_B$. Hence, $B$ contains a  directed Hamiltonian path $(\rho_B,v_1,\dots,v_{m-2},\eta_B)$. This together with the fact that $N$ is level-$1$ and $B$ is biconnected implies that there is a shortcut $(\rho_B,\eta_B)$. It is now easy to verify that $B$ is a crescent; a contradiction.
	
Hence, $B$ contains two internal $\prec_N$-incomparable vertices $v$ and $w$. Now consider the two distinct paths from $\rho_B$ to $v$, respectively $w$. These paths must contain a $\prec_N$-minimal vertex $u$, i.e., there is no vertex $u'$ contained in these paths such that $u'\prec_N u$. Note that $u=\rho_B$ is possible. Hence, $u$ has children $v$ and $w$. Note, $v,w\neq \eta_B$ since otherwise they would be $\prec_N$-comparable. Note that all internal vertices must have unique parents, since $N$ is level-$1$ and $\eta_B$ is the unique hybrid. Since $u$ is already a parent of $v$ and $w$, neither $(v,w)$ nor $(w,v)$ can be an edge in $N$. Consequently, we can choose a $u\eta_B$-path $P_1$ that contains $v$ but not $w$ and a $u\eta_B$-path $P_2$ that contains $w$ but not $v$. Hence, $P_1$ and $P_2$ are distinct and contain inner vertices. Assume, for contradiction, that $P_1$ and $P_2$ are not internal vertex disjoint. Let $z$ be the first internal vertex that is contained in both $P_1$ and $P_2$ when following these paths from $u$ to $\eta_B$. By construction, $z\neq u,v$. Since $z$ is the first common internal vertex of $P_1$ and $P_2$, it follows that $z$ must have a parent $p_1\preceq_N v$ that is located in $P_1$ but not in $P_2$ and a parent $p_2\preceq_N w$ that is located in $P_2$ but not in $P_1$. Hence, $w$ is a hybrid. Since $w$ is internal, $z\neq \eta_B$. Taken the latter arguments together, $B$ and, thus, $N$ is not level-$1$; a contradiction. Hence, there are two internal vertex-disjoint $u\eta_B$-paths in $B$ that both contain at least one internal vertex of $B$ distinct from $u$.
\end{proof}

\begin{proposition} \label{prop:B_maxmin_arbitrary}
Let $N$ be a phylogenetic level-$1$ network and let $B \in \mathcal{B}_m(N)$. Then, 
\begin{enumerate}[(1)]
	\item $\omega(B) \geq m-2$. This bound is tight and is, in particular, achieved if and only if $B$ is a lantern.
	\item $\omega(B) \leq \binom{m}{3}$. This bound is tight and is, in particular, achieved if and only if $B$ is a crescent.
\end{enumerate}
\end{proposition}
\begin{proof}
Let $N$ be a phylogenetic level-$1$ network and let $B \in \mathcal{B}_m(N)$ be a non-trivial block. Hence, $B$ contains $m \geq 3$ vertices.  
We start with (1). We will make frequent use of the fact that, for every vertex $v \in  V^-(B)$, it holds that $\kappa_v \geq 2$ since  $v \preceq_N v$ and $\eta_B\prec_N v$. Thus,
	\begin{align}
	\omega(B) &= \sum\limits_{v \in  V^-(B)} \binom{\kappa_v}{2} \geq \sum\limits_{v \in  V^-(B)} \binom{2}{2} = |V^{-}(B)| = |V(B)|-2 = m-2. \label{LowerBound_nonbinary_block}
	\end{align}
We continue with showing that $\omega(B) = m-2$ precisely if $B$ is a lantern. Assume first that $B\in \mathcal{B}_m(N)$ is a lantern. Since $B$ consists of $m-2$ internal vertex disjoint paths from $\rho_B$ to $\eta_B$ containing precisely one vertex of $V^{-}(B)$ each, then, irrespective of whether $B$ also contains the edge $(\rho_B, \eta_B)$, we must have $\kappa_v = 2$ for all $v \in V^{-}(B)$. Hence, the inequality in Equation~\eqref{LowerBound_nonbinary_block} becomes an equality and we obtain $\omega(B) = m-2$. 
		
For the converse, assume that $\omega(B) = m-2$. Assume, for contradiction, that $B$ is not a lantern. By Observation~\ref{obs:Bm3-isom}, it must hold that $m>3$. As $B$ is not a lantern, $B$ contains at least one path from $\rho_B$ to $\eta_B$ with at least two internal vertices. Thus, there exists, in particular, a vertex $u \in V^{-}(B)$ with $\kappa_u \geq 3$. Recall that	$\kappa_v \geq 2$ for all $v \in V^-(B)$. Taken the latter arguments together, we obtain
	\begin{align*}
		\omega(B) &= \sum\limits_{v \in V^{-}(B)}  \binom{\kappa_v}{2}  
        =  \binom{\kappa_u}{2}  + \sum\limits_{v \in V^{-}(B)\setminus  \{u\}}  \binom{\kappa_v}{2}  \geq \binom{3}{2} + \sum\limits_{v \in V^{-}(B)\setminus \{u\}} \binom{2}{2} \\
		&= 3 + |V^{-}(B)|-1
		= 3 + (m-2) -1 
		= m 
		> m-2,
	\end{align*}
contradicting $\omega(B) = m-2$. Thus, $B$ must be a lantern. This completes the proof of Part (1).
		
We now proceed with Part (2). We first show that $\omega(B) \leq \binom{m}{3}$ by induction on $m$ for all non-trivial blocks $B$ with $m$ vertices. As base case, let $m=3$. Hence, $B$ is a crescent with one internal vertex $u$ (cf.~Observation~\ref{obs:Bm3-isom}) and we obtain $\omega(B) = \binom{\kappa_u}{2} = \binom{2}{2}=1 = \binom{3}{3} = \binom{m}{3}$. Now, let $m \geq 4$ and suppose that the statement is true for all non-trivial blocks with up to $m-1$ vertices in any phylogenetic level-$1$ network $N$. Let $B \in \mathcal{B}_m$ be a non-trivial block with $m$ vertices. Let $u \in V^{-}(B)$ be a vertex adjacent to $\rho_B$. Note that $u$ must have in-degree 1 since $\eta_B$ is the unique hybrid in $B$. Let $u_1, \ldots, u_l$ denote the children of $u$ in $B$. We now delete $u$ and its incident edges, add the edges $(\rho_B,u_1), \ldots, (\rho_B,u_l)$, except possibly $(\rho_B,\eta_B)$ in case it already exists. Note that if $l=1$, the process described here is simply the suppression of $u$, and also note that $u$ and $\rho_B$ cannot have any child in common except possibly $\eta_B$ (as otherwise, $B$ would contain two hybrids, a contradiction).  It can easily be seen that this yields a block $B'$ with $m-1$ vertices, unique root $\rho_{B'}$, and unique hybrid $\eta_{B'}$. This is due to the fact that all $\geq 2$ paths from $\rho_B$ to $\eta_B$ are also present in $B'$, except that the ones going through $u$ have been shortened in length by 1. In particular, $B'$ is biconnected. We can now apply the inductive hypothesis and conclude that $\omega(B') \leq \binom{m-1}{3}$. Moreover, by the choice of $u$ as being adjacent to $\rho_B$ in $B$, the \enquote{removal} of $u$ to obtain $B'$ does not affect $\kappa_v$ for any vertex $v \in V^{-}(B) \setminus \{u\} = V^{-}(B')$.  Moreover, we have $\kappa_u \leq m-1$ for all $u\in V(B)\setminus\{\rho_B\}$, since in a block $B$ with $m$ vertices, there can be at most $m-1$ vertices $w\preceq_ B u$ as $u \neq \rho_B$ (Lemma \ref{lem:upperbound-kappa}). Taking these facts together, we obtain
	\begin{align*}
	\omega(B) &= \sum\limits_{v \in V^{-}(B)} \binom{\kappa_v}{2} 
		= \binom{\kappa_u}{2} + \sum\limits_{v \in V^{-}(B) \setminus{u}} \binom{\kappa_v}{2} 
		= \binom{\kappa_u}{2} + \underbrace{\sum\limits_{v \in V^{-}(B')} \binom{\kappa_v}{2}}_{= \omega(B')} \leq \binom{m-1}{2} + \binom{m-1}{3} = \binom{m}{3}.
	\end{align*}
It remains to show that this bound is tight and is achieved if and only if $B$ is a crescent.
		
To see this, first suppose that $B$ is a crescent. Let  $P=(\rho_B, v_{m-2}, v_{m-3}, \ldots,  v_2, v_1, \eta_B)$ be the Hamiltonian path in $B$ containing all vertices in $V^{-}(B)=\{v_1,\dots, v_{m-2}\} $. Since $B$ contains $(\rho_B,\eta_B)$ and all other possibly existing  edges that are not located on $P$ are of the form $(v_i,\eta_B)$ one easily verifies that $\kappa_{v_1}=2, \kappa_{v_2}=3, \ldots, \kappa_{v_{m-2}}=m-1$. More precisely, $\kappa_{v_i}=i+1$ for $i=1, \ldots, m-2$. Therefore,
	\begin{align*}
		\omega(B) &= \sum\limits_{v \in V^{-}(B)}  \binom{\kappa_v}{2}  
		= \sum\limits_{v \in \{v_1, \ldots, v_{m-2}\}}  \binom{\kappa_v}{2}  
		= \sum\limits_{i=1}^{m-2} \binom{\kappa_{v_i}}{2}  
		= \sum\limits_{i=1}^{m-2} \binom{i+1}{2} = \binom{m}{3}.
	\end{align*}
Hence, if $B$ is a crescent, we have $\omega(B) =\binom{m}{3}$. Thus, the bound is tight.
		
For the converse, suppose that $B$ is a non-trivial block with $m$ vertices of any phylogenetic level-$1$ network and such that that $\omega(B) =\binom{m}{3}$. Assume, for contradiction, that $B$ is not a crescent. Lemma \ref{lem:not-crescent} implies that $m>3$ and, moreover, $B$ contains at least one vertex $u$ such that there are internally vertex-disjoint $u\eta_B$-paths $P_1$ and $P_2$ in $B$ that both contain at least one internal vertex of $B$ (note that $u=\rho_B$ is possible). We may assume without loss of generality that $|P_1|\leq |P_2|$. Let $a$ denote the  child of $u$ that is located in $P_1$ and let $b$ denote the vertex of $P_2$ that is incident to $\eta_B$. We now modify $B$ to get a block $B'$ as follows: In case $B$ does \emph{not} contain the edge $(\rho_B,\eta_B)$, we add this edge. We then delete the edges $(u,a)$ and $(b,\eta_B)$ and introduce a  new edge $(b,a)$. By construction, this yields a block as the edge $(\rho_B,\eta_B)$ guarantees biconnectedness. For all $v$ on $P_2$, this strictly increases $\kappa_v$, namely at least by $|P_1|-2\geq 1$ (where $-2$ refers to $u$ and $\eta_B$), but possibly even by more (if there is more than one path from $a$ to $\eta_B$). However, by construction all other $\kappa$-values remain unchanged, which shows that overall, $\omega(B')>\omega(B) = \binom{m}{3}$, which is a contradiction as, by the above considerations, $\binom{m}{3}$ is an upper bound for $\omega$. Hence, $B$ must be a crescent, which completes the proof.
\end{proof}

We now consider non-trivial blocks of binary phylogenetic level-$1$ networks. Recall that any such
block $B$ is an (undirected) cycle and contains precisely two paths from $\rho_B$ to $\eta_B$. Note
that every crescent consisting only of the edges of the underlying Hamiltonian path and the shortcut
$(\rho_B, \eta_B)$ satisfies the properties of being such an (undirected) cycle. Hence, for the
maximum weight $\omega$, we can apply Proposition~\ref{prop:B_maxmin_arbitrary}(2) to obtain
the following result.
\begin{corollary}\label{cor:crecent}
Let $N$ be a binary phylogenetic level-$1$ network and let $B\in \mathcal{B}_m(N)$. Then, $\omega(B)
\leq \binom{m}{3}$. Moreover, this bound is tight and is achieved if and only if $B$ is a crescent.
\end{corollary}

For the non-trivial blocks $B$ with minimum weight $\omega$ the situation becomes, however, more
complicated. Observe that the number of children of $\rho_B$ in a lantern $B\in \mathcal{B}_m(N)$ is
$m-2$ or $m-1$. This immediately implies that in a binary phylogenetic level-$1$ network there
cannot be a lantern except for the case $m=3$ or $m=4$.

\begin{proposition}\label{prop:B_min_binary}
Let $N$ be a binary phylogenetic level-$1$ network and let $B\in \mathcal{B}_m(N)$. Then, 
$\omega(B) \geq \binom{\left \lceil \frac{m}{2} \right \rceil+1}{3} + \binom{\left \lfloor \frac{m}{2} \right \rfloor+1}{3}$. This bound is tight and is, in particular, achieved if and only if $B$ is a full-moon.
\end{proposition}
\begin{proof}
Let $N$ be a binary phylogenetic level-$1$ network and let $B\in \mathcal{B}_m(N)$ for some fixed $m$. Note that $m\geq 3$. Since $N$ is binary, $B$  consists of two internal vertex-disjoint $\rho_B\eta_B$-paths $P_1$ and $P_2$. Let $U_1,U_2\subseteq V^-(B)$ be the internal vertices that are contained in $P_1$ and $P_2$, respectively. Moreover let $m_1\coloneqq \vert U_1 \vert$ and $m_2\coloneqq \vert U_2 \vert$. Without loss of generality assume that $m_1 \geq m_2$. Note that $m_1,m_2\geq 0$ and $m_1+m_2=m-2$. The latter two arguments imply that $m_1 \in M_1\coloneqq \{\lceil \frac{m-2}{2} \rceil, \lceil \frac{m-2}{2} \rceil+1, \ldots, m-2\}$ and $m_2 \in \{0, 1, \ldots, \lfloor \frac{m-2}{2}\rfloor\} $.
Let $P_1 = (\rho_B, u_{m_1}, \ldots, u_2, u_1, \eta_B)$ and $P_2 = (\rho_B, v_{m_2}, \ldots, v_2, v_1, \eta_B)$.
Then, $\kappa_{u_i} = i+1$ for all $i = 1, \ldots, m_1$, and $\kappa_{v_i} = i+1$ for all $i = 1, \ldots, m_2$. Thus, the weight of $B$ is given by 
	\begin{align*}
		\omega(B) = \sum\limits_{v \in V^{-}(B)} \binom{\kappa_v}{2} 
		&=   \sum\limits_{v \in \{u_1, \ldots, u_{m_1}\}} \binom{\kappa_v}{2} +  
		    \sum\limits_{v \in \{v_1, \ldots, v_{m_2}\}} \binom{\kappa_v}{2} 
		\ = \  \sum\limits_{i=1}^{m_1} \binom{i+1}{2} + 
		    \sum\limits_{i=1}^{m_2} \binom{i+1}{2} \\
	    &=  \sum\limits_{i=0}^{m_1-1} \binom{i+2}{2} + 
		    \sum\limits_{i=0}^{m_2-1} \binom{i+2}{2} 
		\ = \  \binom{2+(m_1-1)+1}{3} +  \binom{2+(m_2-1)+1}{3}  \\
		 &=   \binom{2+m_1}{3} +  \binom{2+m_2}{3} 
		\label{eq:wB}
	\end{align*} 
Now observe that $m_1 \geq \lceil \frac{m-2}{2} \rceil$ implies that $2+m_1\geq 2+\lceil \frac{m-2}{2} \rceil = 2+ \lceil \frac{m}{2} -1 \rceil = 2+\lceil \frac{m}{2} \rceil  -1 =\lceil \frac{m}{2} \rceil +1$. 
In the latter, equality is achieved if and only if $m_1 = \lceil \frac{m-2}{2} \rceil$. 
Moreover, $m\geq 3$ implies that $m_1\geq 1$. This together with $m_2=m-m_1 - 2$ implies that $2+m_2 = 2+m-m_1 - 2 \geq m-1 \geq \lfloor \frac{m}{2} \rfloor +1$. In the latter, equality is achieved if and only if $m_2 = \lfloor \frac{m}{2} \rfloor +1-2 =\lfloor \frac{m}{2} \rfloor -1 = \lfloor \frac{m-2}{2} \rfloor $. Taken the latter arguments together, we obtain
	\[
	\binom{2+m_1}{3} \geq \binom{\lceil \frac{m}{2} \rceil +1}{3}
	\text{\quad and \quad }
	\binom{2+m_2}{3} \geq \binom{\lceil \frac{m}{2} \rceil +1}{3}
	\label{eq:2}
	\]
and, therefore,
	\[\omega(B) = \binom{2+m_1}{3} +  \binom{2+m_2}{3}  \geq 
		\binom{\left \lceil \frac{m}{2} \right \rceil+1}{3} + \binom{\left \lfloor 
			\frac{m}{2} \right \rfloor+1}{3}.\]
Using the aforementioned statements concerning \enquote{equality} for $m_1$ and $m_2$, it is now straightforward to verify that $\omega(B) =\binom{\left \lceil \frac{m}{2} \right \rceil+1}{3} + \binom{\left \lfloor \frac{m}{2} \right \rfloor+1}{3} $ if and only if $m_1 = \lceil \frac{m-2}{2} \rceil$ and thus, $m_2=m-m_1-2 = \lfloor \frac{m-2}{2}\rfloor$ which is, by definition, precisely if $B$ is a full-moon.
\end{proof}

\section{Phylogenetic level-1 networks with extremal indices} \label{APPX:sec:extrema}
Having characterized the structure and weights of extremal blocks in any phylogenetic level-$1$ network in the previous section, in this section we study the extremal values of the weighted total cophenetic index across phylogenetic level-$1$ networks and characterize networks achieving these values.
Additionally, we will focus on the following sub classes of networks.

\begin{definition}\label{def:LEVELoneDEGtwo}
$\BinLevelOneN$ denotes the class of all binary phylogenetic level-$1$ networks on $n$ leaves
and $\HybdidDegTwoN$ the class of all phylogenetic level-$1$ networks on $n$ leaves where each vertex has in-degree at most $2$. 
\end{definition}
\noindent
Note that  $\BinLevelOneN\subseteq \HybdidDegTwoN$. Considering networks within the classes $\BinLevelOneN$ and $\HybdidDegTwoN$
is, in particular, motivated from a biological point of view, where it is often assumed that hybrids have two but not more parental species. 
Networks $N\in \HybdidDegTwoN$ are also called \emph{galled-trees} \cite{HSS:22,Gusfield2003}.
Galled-trees do not only play an important role in phylogenetics but are useful as a framework to store structural information 
of so-called \textsc{GaTeX} graphs that allows to solve several computationally hard problems in linear-time \cite{HS-GatexLinTime:23,HS-GatexForbSubg:23,HS:22}. 
Galled-trees form a subclass of so-called tree-child phylogenetic network \cite{CRV07}. As shown in \cite[Prop.~1]{CRV07}, the number of vertices in galled-trees
on $n$ leaves is bounded by  $4(n-1)+1\in O(n)$.  For later reference, we summarize this in
\begin{lemma}\label{lem:vertices-bounded}
For each network in $\HybdidDegTwoN$ and thus, for each network in $\BinLevelOneN$, the number of vertices is bounded and, in particular,  in $O(n)$.
\end{lemma}

\subsection{Minimum networks} 
\label{APPX:subsection:MinN}

We begin by considering the minimum value of the weighted total cophenetic index both across
arbitrary level-$1$ networks as well as binary level-$1$ networks. As we shall see, the minimum
value and the networks achieving it, depend on the choice of $\epsilon$ used when assigning weights
to vertices, reflecting the fact that different values of $\epsilon$ treat true tree vertices and
tree vertices contained in blocks differently. We are now characterize the level-$1$ networks with
minimum weighted total cophenetic index. We first consider arbitrary level-$1$ networks.

\begin{theorem}\label{thm:minimum-L1}

Let $N$ be a  phylogenetic level-$1$ network on $n\geq 2$ leaves. Then $N$ minimizes the weighted total cophenetic index within the class
of phylogenetic level-$1$ networks on $n$ leaves if and only if $\Phi^{**}(N)=\binom{n}{2}$. Moreover, 
if $\epsilon >0$, then $N\simeq \Tstar$ and, otherwise, if $\epsilon=0$, then $N\simeq \Tstar$ or $N\simeq \Tmod$. 
\end{theorem}
\begin{proof} 
Let $N$ be a phylogenetic level-$1$ network that minimizes the weighted total cophenetic index within the class
of phylogenetics level-$1$ networks on $n\geq 2$ leaves. The weighted total cophenetic index of the star tree is
$\Phi^{**}(\Tstar)=\binom{n}{2}$ since $\Phi^*(\Tstar)=0$ and
$\phi({\rho_{\Tstar}})=1$.
Since $\Tstar$ is a phylogenetic level-$1$ network on $n$ leaves and 
by choice of $N$,   $\Phi^{**}(N)\leq \Phi^{**}(\Tstar) = \binom{n}{2}$ must hold. 
By Observation\ \ref{obs:phi**-lower-bound}, $\Phi^{**}(N)\geq \binom{n}{2}$
and therefore, $\Phi^{**}(N) =  \binom{n}{2}$. Using the latter arguments, one easily
derives the \emph{if} direction.
        
Assume, for contradiction, that $N$  contains a non-trivial block $B$ that is not rooted in the root $\rho_N$. 
Note that $\omega(B)\geq 1$. By Observation\ \ref{obs:relevant}, $\rho_B\neq \rho_N$ is another relevant vertex. In particular,  $\phi(\rho_B)\geq 1+\epsilon \geq 1$
and $|L_N(\rho_B))|>1$. Taken the latter arguments together, we obtain $\Phi^{**}(N)> \binom{n}{2}$; a contradiction. 
Hence, if $N$ contains a non-trivial block $B$, then $\rho_B = \rho_N$ must hold.

Suppose now that $\epsilon>0$. In this case, $N$ cannot contain any non-trivial block, since otherwise
such a block $B$ satisfies $\rho_B = \rho_N$ and thus, $\phi(\rho_N)>1$ which implies 
 $\Phi^{**}(N)> \binom{n}{2}$; a contradiction. Hence, $N$ must be a tree. 
 By Proposition \ref{prop:cp-properties}, $N \simeq   \Tstar$ must hold.
 
Suppose now that $\epsilon = 0$. 
If $N$ contains more than one non-trivial block, then, as argued above, all of them must be rooted in $\rho_N$ and we
obtain,  $\phi(\rho_N)\geq \epsilon + 2 = 2 > 1$ in which case $\Phi^{**}(N)>\binom{n}{2}$; a contradiction. 
Hence, $N$ contains at most one non-trivial block. If $N$ does not contain
such a block, then $N$ is a tree and, by similar arguments as above,  we have $N \simeq   \Tstar$. 
Suppose that  $N$ contains one non-trivial block $B$. Again $\rho_B=\rho_N$ must hold. 
By Observation~\ref{obs:omega}, $\omega(B)>1$, if $B$ consists of more than three vertices
and, therefore, $\phi(\rho_N)>1$ and, thus, $\Phi^{**}(N)>\binom{n}{2}$; a contradiction. 
Hence, $B$ contains precisely three vertices. In fact, we have in this case, 
$\omega(B)=1$ and thus,  $\phi(\rho_N)=1$ and $\Phi^{**}(N) =  \binom{n}{2}$. 
Thus, $N\simeq \Tmod$. 
\end{proof}

Note that the star tree $\Tstar$ and the modified star $\Tmod$  with $n\geq2$ leaves are both contained in $\HybdidDegTwoN$.
This together with Theorem~\ref{thm:minimum-L1} implies 
\begin{corollary}\label{cor:min-hybrid-two}
Let $N\in \HybdidDegTwoN$ be a phylogenetic level-$1$ network with $n\geq 2$  leaves.
Then, $N$ minimizes the weighted total cophenetic index within the class $\HybdidDegTwoN$
if and only if $\Phi^{**}(N)=\binom{n}{2}$. Moreover, 
if $\epsilon >0$, then $N\simeq \Tstar$ and, otherwise, if $\epsilon=0$, then $N\simeq \Tstar$ or $N\simeq \Tmod$.
\end{corollary}

To characterize level-$1$ networks with minimum weighted total cophenetic index
among those networks that contain at least one non-trivial block (and thus, level-$1$ networks that are not trees)
we first provide 

\begin{proposition}\label{prop:min-withBlock-basic}
		Let $\hat{\mathscr{N}}_n$ be the class of phylogenetic level-$1$ networks on $n \geq 2$ leaves  that contain at least one non-trivial block. 
	  If $N\in \hat{\mathscr{N}}_n$ minimizes the weighted total cophenetic index within the class $\hat{\mathscr{N}}_n$, 
		then $N\simeq\Tmod$ or $N\simeq\Tpush$. 
\end{proposition}
\begin{proof}
Let $N\in \hat{\mathscr{N}}_n$ be a network that minimizes the weighted total
cophenetic index within the class $\hat{\mathscr{N}}_n$. Suppose first that
$n=2$. One easily verifies that (up to isomorphism) there is only one phylogenetic level-$1$
network that contains at least one non-trivial block, namely the triangle
network $N_{\Delta}$. Hence, $N\simeq N_{\Delta}$ must hold. By definition,
$N_\Delta \simeq \TmodTWO \simeq \TpushTWO$ which verifies the statement for the
case $n=2$. 
 
Suppose now that $n\geq 3$. Observe that $\Phi^{**}( \Tmod) =
(1+\epsilon)\binom{n}{2}$ and $\Phi^{**}(\Tpush) = \binom{n}{2} +
(1+\epsilon)\binom{2}{2} = \binom{n}{2} +1+\epsilon$. Since $N$ minimizes the weighted total cophenetic index
within the class $\hat{\mathscr{N}}_n$ and $\Tmod,  \Tpush\in
\hat{\mathscr{N}}_n$, we have $\Phi^{**}(N) \leq
(1+\epsilon)\binom{n}{2}$ and $\Phi^{**}(N)\leq \binom{n}{2} +1+\epsilon$. Assume, for contradiction, that $N$ contains at
least two non-trivial blocks.
Suppose first, that all these non-trivial blocks are rooted in the root $\rho_N$
of $N$. By Observation\ \ref{obs:omega}, each non-trivial block $B$ has weight
$\omega(B)\geq 1$. Hence, $\phi(\rho_N) \geq 2+\epsilon$ and thus,
\[\Phi^{**}(N)\geq \phi(\rho_N)\binom{n}{2} \geq (2+\epsilon)\binom{n}{2} >(1+\epsilon)\binom{n}{2};\] a contradiction. 
Hence, there is a vertex $w\neq \rho_N$ with $w = \rho_B$ for some non-trivial block $B$ of $N$. By Observation\
\ref{obs:relevant}, $w$ is relevant. Moreover, $\phi(w)\geq 1+\epsilon$ and
$|L_N(w)|\geq 2$ must hold. If $\rho_N$ is the root of another non-trivial block
of $N$ we have $\phi(\rho_N) \geq 1+\epsilon$. Therefore, \[\Phi^{**}(N)\geq
\phi(\rho_N)\binom{n}{2} + \phi(w)\binom{|L_N(w)|}{2} \geq
(1+\epsilon)\binom{n}{2}+(1+\epsilon)\binom{|L_N(w)|}{2}>(1+\epsilon)\binom{n}{2};\]
a contradiction. Hence, none of the non-trivial blocks of $N$ that are not rooted in $w$ can be rooted in
$\rho_N$ and thus, $\phi(\rho_N) =1$. In this case, there is either (i) a vertex
$v\neq w$ such that $v$ is the root of some non-trivial block $B'\neq B$ or (ii) all
non-trivial blocks are rooted in $w$. 
In Case (i), $v$ is relevant (cf.\ Observation\ \ref{obs:relevant}) and it holds that $\phi(v)\geq 1+\epsilon$ and $|L_N(v)|\geq 2$. 
Therefore, 
\begin{align*}
\Phi^{**}(N) &\geq \phi(\rho_N)\binom{n}{2} + \phi(w)\binom{|L_N(w)|}{2} +\phi(v)\binom{|L_N(v)|}{2} \\
						&\geq \binom{n}{2}+(1+\epsilon)\binom{|L_N(w)|}{2} + (1+\epsilon)\binom{|L_N(v)|}{2} >\binom{n}{2} +1+\epsilon.
\end{align*}
In Case (ii), $\phi(w)\geq 2+\epsilon$ must hold and we obtain
\[\Phi^{**}(N)\geq \phi(\rho_N)\binom{n}{2} + \phi(w)\binom{|L_N(w)|}{2} \geq \binom{n}{2}+(2+\epsilon)\binom{|L_N(w)|}{2}>\binom{n}{2} +1+\epsilon.\]
Thus, both Cases (i) and (ii) yield a contradiction. 
In summary, $N$ contains precisely one non-trivial block. 

Assume now, for contradiction, that neither $N\simeq \Tmod$ nor $N\simeq \Tpush$ holds. 
Let $B$ be the unique non-trivial block in $N$. 

Assume first that $\rho_N=\rho_B$. Since $N\not\simeq \Tmod$, there is a
relevant vertex $w\neq \rho_N$ in $N$ or $B$ is not a triangle. 
If there is a relevant vertex $w\neq \rho_N$, then $|L_N(w)|\geq 2$ and thus, 
\[\Phi^{**}(N) \geq \phi(\rho_N)\binom{n}{2} + \phi(w)\binom{|L_N(w)|}{2}  
						\geq (1+\epsilon) \binom{n}{2}+\binom{|L_N(w)|}{2} >(1+\epsilon)\binom{n}{2};
\]  a contradiction. Hence, there is no further relevant vertex than the root $\rho_N$
in $N$. Therefore, $B$ is not a triangle, which implies that $\omega(B)>1$. 
Hence, 
\[\Phi^{**}(N) = \phi(\rho_N)\binom{n}{2} = ( \omega(B)+\epsilon )\binom{n}{2}	> (1+\epsilon) \binom{n}{2};\]
a contradiction. Thus, in case $\rho_N=\rho_B$, the root $\rho_N$ is the
only relevant vertex of $N$ and the single non-trivial block $B$ rooted
in $\rho_N$ must be a triangle. This, in particular, implies that 
$N\simeq \Tmod$ in case $\rho_N=\rho_B$.

Assume now that $\rho_N\neq \rho_B$. Since $B$ is the unique non-trivial block in $N$, 
we have $\phi(\rho_N) = 1$. Let $w = \rho_B$. By Observation\
\ref{obs:relevant}, $\rho_N$ and $w$ are relevant. Since $N\not\simeq \Tpush$, there is 
third relevant vertex $v\neq w,\rho_N$, or $B$ is not a triangle. 
If there is a third relevant vertex $v$, then $|L_N(v)|\geq 2$ and thus, 
\begin{align*}
\Phi^{**}(N) &\geq \phi(\rho_N)\binom{n}{2} + \phi(w)\binom{|L_N(w)|}{2} +\phi(v)\binom{|L_N(v)|}{2} \\
						&\geq \binom{n}{2}+(1+\epsilon)\binom{|L_N(w)|}{2} + \binom{|L_N(v)|}{2} >\binom{n}{2} +1+\epsilon;
\end{align*}
a contradiction. Hence, $\rho_N$ and $w$ are the only relevant vertices in $N$. 
Thus, $B$ is not a triangle, which implies that $\omega(B)>1$. 
Therefore, 
\begin{align*}
\Phi^{**}(N) = \phi(\rho_N)\binom{n}{2} + \phi(w)\binom{|L_N(w)|}{2} = 
								\binom{n}{2} + (\omega(B)+\epsilon)\binom{|L_N(w)|}{2} > \binom{n}{2} +1+\epsilon;
\end{align*}
a contradiction. Hence, in case  $\rho_N\neq \rho_B$, there are only two relevant vertices,
$\rho_N$ and $w$, and the unique block in $N$ must be a triangle that is rooted in $w$. 
Note that $w$ and $\rho_N$ must be adjacent, since otherwise there is a vertex $w'$
that is adjacent to $\rho_N$ and satisfies $w\prec_N w'\prec_N	\rho_N$. 
In this case, $w'$ is not a leaf and either a true tree vertex or the root of a non-trivial
block. In either case, Observation\ \ref{obs:relevant} implies that $w'$ is relevant; a contradiction. 
Taken the latter arguments together, $N\simeq \Tpush$ in case $\rho_N\neq \rho_B$.
\end{proof}

\begin{theorem}\label{thm:minimum-L1-block}
Let $\hat{\mathscr{N}}_n$ be the class of phylogenetic level-$1$ networks on  $n \geq 2$ leaves  that contain at least one non-trivial block. 
	  Moreover,  let $N\in \hat{\mathscr{N}}_n$ be a network that minimizes the weighted total cophenetic index within the class $\hat{\mathscr{N}}_n$.	
	  If $n=2$, then $N\simeq N_{\Delta}$ and $\Phi^{**}(N) =1+\epsilon$ for all $\epsilon\geq 0$. If $n\geq 3$, 
	  then
       	 \begin{itemize}
            \item $0\leq  \epsilon<\frac{1}{\binom{n}{2}-1}$ implies $N\simeq \Tmod$ and $\Phi^{**}(N)=\binom{n}{2}(1+\epsilon)$.  \vspace{-0.1in}
            
            \item $\epsilon>\frac{1}{\binom{n}{2}-1}$ implies $N\simeq \Tpush$ and $\Phi^{**}(N) =  \binom{n}{2}+1+\epsilon$. \vspace{-0.1in}
            
            \item  $\epsilon=\frac{1}{\binom{n}{2}-1}$ implies $N\simeq \Tmod$ or $N\simeq \Tpush$ and $\Phi^{**}(N) = \frac{\binom{n}{2}^2}{\binom{n}{2}-1} = 
            \frac{n^2(n-1)^2}{2 (n-2)(n+1)}$.
        \end{itemize}
\end{theorem}

\begin{proof} 
Let $N\in \hat{\mathscr{N}}_n$ be a network that minimizes the weighted total cophenetic index within the class $\hat{\mathscr{N}}_n$.
By Proposition\ \ref{prop:min-withBlock-basic}, $N\simeq \Tmod$ or $N\simeq \Tpush$ holds. 
Observe that  $\Phi^{**}( \Tmod) =
(1+\epsilon)\binom{n}{2}$ and $\Phi^{**}(\Tpush) = \binom{n}{2} +
(1+\epsilon)\binom{2}{2} = \binom{n}{2} +1+\epsilon$. 
Suppose first that $\epsilon = 1 /(\binom{n}{2}-1)$. One easily verifies that 
\[\Phi^{**}( \Tmod) = \left(1+ \frac{1}{\binom{n}{2}-1}\right)\binom{n}{2} =  \frac{\binom{n}{2}^2}{\binom{n}{2}-1} =
 \binom{n}{2} +1+ \frac{1}{\binom{n}{2}-1} = \Phi^{**}(\Tpush).\] 
In this case, in particular,  $N\simeq \Tmod$ or $N\simeq \Tpush$ holds. 
Suppose now that $\epsilon\geq 0$. It holds that 
\begin{equation*}
        \begin{aligned}
            \Phi^{**}(\Tmod)\leq\Phi^{**}(\Tpush)
            &\iff (1+\epsilon)\binom{n}{2}\leq \binom{n}{2} +1+\epsilon\\
            &\iff \binom{n}{2}\epsilon\leq 1+\epsilon\\
            &\iff  \epsilon\leq \frac{1}{\binom{n}{2}-1}.
        \end{aligned}
    \end{equation*}
Hence, in case $\epsilon<\frac{1}{\binom{n}{2}-1}$ we have  $\Phi^{**}(\Tmod)<\Phi^{**}(\Tpush)$ and, 
therefore, $N\simeq \Tmod$ and  $\Phi^{**}(N) =(1+\epsilon)\binom{n}{2}$. 
In case, $\epsilon>\frac{1}{\binom{n}{2}-1}$ we have  $\Phi^{**}(\Tmod)>\Phi^{**}(\Tpush)$ and, 
therefore, $N\simeq \Tpush$ and $\Phi^{**}(N) =\binom{n}{2} +1+\epsilon$. 
\end{proof} 

Note that $N\simeq N_{\Delta}$ for all $N\in \HybdidDegTwoNvar{2}$ that contain at least one non-trivial block and that the modified star $\Tmod$ and the push-down star $\Tpush$ are both contained in $\HybdidDegTwoN$ for $n\geq3$. This together with Proposition~\ref{prop:min-withBlock-basic} and Theorem~\ref{thm:minimum-L1-block} implies

\begin{corollary}\label{cor:bin-withB}
    Let $\hat{\mscr{N}}_n^{\mathrm{L1,in}\leq 2}$ be the class of phylogenetic level-$1$ networks on $n \geq 2$ leaves that contain at least one non-trivial block and whose vertices have at most in-degree 2. 
	  Moreover,  let $N\in \hat{\mscr{N}}_n^{\mathrm{L1,in}\leq 2}$ be a network that minimizes the weighted total cophenetic index within the class $\hat{\mscr{N}}_n^{\mathrm{L1,in}\leq 2}$.	
	  If $n=2$, then $N\simeq N_{\Delta}$ and $\Phi^{**}(N) =1+\epsilon$ for all $\epsilon\geq 0$. If $n\geq 3$, 
	  then
       	 \begin{itemize}
            \item $0\leq  \epsilon<\frac{1}{\binom{n}{2}-1}$ implies $N\simeq \Tmod$ and $\Phi^{**}(N)=\binom{n}{2}(1+\epsilon)$.  \vspace{-0.1in}
            
            \item $\epsilon>\frac{1}{\binom{n}{2}-1}$ implies $N\simeq \Tpush$ and $\Phi^{**}(N) =  \binom{n}{2}+1+\epsilon$. \vspace{-0.1in}
            
            \item  $\epsilon=\frac{1}{\binom{n}{2}-1}$ implies $N\simeq \Tmod$ or $N\simeq \Tpush$ and $\Phi^{**}(N) = \frac{\binom{n}{2}^2}{\binom{n}{2}-1} = 
            \frac{n^2(n-1)^2}{2 (n-2)(n+1)}$.
        \end{itemize}
\end{corollary}

In the following, we will focus on \emph{binary} level-$1$ networks
with minimum weighted total cophenetic index. To this end, we will
consider an operation of associating  trees with certain  level-$1$
networks via expanding vertices to triangles and collapsing triangles into
vertices. More formally, let $T$ be a rooted  tree and let $u$ be one of
its inner vertices with children $u_1$ and $u_2$. 
Then, by \emph{expanding $u$ to a triangle (along $u_1,u_2$)}, we mean the process of subdividing
the edges $(u,u_1)$ and $(u,u_2)$ with vertices $v_1$ and $v_2$, respectively,
and introducing either the edge $(v_1,v_2)$ or the edge $(v_2,v_1)$. On the
other hand, if $N$ is a level-$1$ network containing a triangle, i.e., a
block $B$ of size three with root $\rho_B$, by \emph{collapsing the triangle
$B$}, we mean deleting the (unique) hybrid edge of $B$ that is not incident to
$\rho_B$ and suppressing the two resulting degree-2 vertices. An example of both
processes is depicted in Figure \ref{Fig_ExpandingCollapsing} for the case
that $T$ or $N$ is binary and phylogenetic. 
We emphasize that distinct triangles in $N$ created in the process of expanding
inner vertices in a binary tree $T$ must be vertex-disjoint. Moreover, the degrees of
existing vertices in $T$ remain the same as in the resulting network $N$, 
while all newly created vertices have either
in-degree 1 and out-degree 2 or in-degree 2 and out-degree 1 in $N$. Consequently, the
\enquote{expansion to triangles} process results in a binary level-$1$ network when
applied to binary trees. Moreover, given a binary network, distinct blocks must
be vertex-disjoint. By similar arguments, the process of collapsing all
triangles in a binary network yields a binary tree. Furthermore, observe that
for $\epsilon=0$, obviously $\phi(v)=\phi(x)=1$, where $v$ is a true tree vertex
and $x$ is the root of a triangle in a binary phylogenetic level-$1$ network. Hence, 
we obtain 

\begin{fact}\label{fact:triangle-1}
Let $N$ be a network obtained from a binary  tree $T$ by expanding an arbitrary number of inner vertices to triangles. 
Then, $N$ is a binary phylogenetic level-$1$ network. In particular, if $N$ does not contain non-trivial blocks or 
in case $\epsilon=0$, we have $\Phi^{**}(N) = \Phi^{**}(T)$. Otherwise, if $N$ contains non-trivial blocks
and $\epsilon>0$, we have $\Phi^{**}(N) > \Phi^{**}(T)$. 
Moreover, if $T'$ is the tree obtained from  $N$ by collapsing all triangles, then $T' \simeq T$.

In addition, if $N$ is a binary phylogenetic level-$1$ network where each non-trivial block (if there are any)
is a triangle and $T$ is the tree obtained from $N$ by collapsing all triangles, then $T$ is binary and
$\Phi^{**}(N) \geq \Phi^{**}(T)$.
\end{fact}

To recall, a tree
$T$ is called maximally balanced if each inner vertex $v$ has precisely two
children $v_1$ and $v_2$ and these two children satisfy $\lvert L_T(v_1) -
L_T(v_2) \rvert \leq 1$.
\begin{definition}\label{def:max_bal}
We denote with $\Nmb$ the set of networks that
are obtained from a maximally balanced tree with $n$ leaves by 
expanding an arbitrary number of inner vertices to triangles.
\end{definition}

Note that each maximally balanced tree with $n$ leaves is always contained in
$\Nmb$. 
\begin{fact}\label{fact:triangle}
It holds that 
$\Nmb\subseteq \BinLevelOneN$ and all non-trivial blocks in $N\in \Nmb$ are triangles
and pairwise vertex-disjoint.
 In particular, if $N \in \Nmb$, 
then the tree $T$ from which $N$ was obtained by expanding an arbitrary number of inner vertices to triangles is a maximally balanced tree on $n$ leaves. 
\end{fact}

For the characterization of binary phylogenetic level-$1$ networks  $N$ with minimum weighted total cophenetic we first start with showing 
that any non-trivial blocks in any such $N$ (if there are any) must be full-moons.
\begin{lemma}\label{lem:bin_min_fullmoon}
    Let $N\in \BinLevelOneN$ be a binary level-$1$ network that minimizes the weighted total cophenetic index within the class $\BinLevelOneN$. 
    If $N$ contains a non-trivial block $B$, then $B$ is a full-moon.
\end{lemma}

\begin{proof}
Let $N$ be a binary phylogenetic level-$1$ network that
minimizes the weighted total cophenetic index within the class $\BinLevelOneN$.
Assume that $N$ contains a non-trivial block $B$. 
Hence, $n\geq 2$ and the size of $B$ is $m
\geq 3$. If $m=3$, $B$ is a full-moon by Observation~\ref{obs:Bm3-isom} and
there is nothing to show. Thus, let $m \geq 4$ and assume, for contradiction,
that $B$ is not a full-moon. Since $N$ is binary and level-$1$, $B$ consists of
two internally vertex-disjoint $\rho_B \eta_B$-paths, say $P$ and $P'$.
Moreover, since $B$ is not a full-moon, one of these paths, say $P$, contains
strictly more than $\lceil \frac{m-2}{2} \rceil$ internal vertices, whereas $P'$
contains strictly less than $\lfloor \frac{m-2}{2} \rfloor$ internal vertices. 
Let $\rho_B, v_k, \ldots, v_1, \eta_B$  with $k \geq \lceil \frac{m-2}{2} \rceil
+ 1$ be the vertices in $P$, and let $\rho_B, v'_{m-k-2}, \ldots , v'_1, \eta_B$
be the vertices in $P'$.

We now construct a network $N'$ from $N$ by moving vertex $v_{1}$ (together with the
subnetwork $\partN(v_1)$) from $P$ to $P'$. More precisely, we delete $v_1$ and its
incident edges as well as the edge $(v'_{1},\eta_B)$ and insert the edges
$(v_{2},\eta_B)$, $(v'_1,v_1)$, and $(v_1,\eta_B)$ (see
Figure~\ref{fig:bin_min_only_fullmoons} for an illustration). Let $B'$ denote the resulting block in
$N'$.

\begin{SCfigure}[.25][ht]
\caption{Subnetworks $N$ and $N'$ with $\Phi^{**}(N)>\Phi^{**}(N')$
			used in the proof of Lemma \ref{lem:bin_min_fullmoon}. The
			subnetwork $\partN(v_1)$ that is relocated when constructing $N'$ is
			marked with bold edges and vertices. }
\includegraphics[width=.65\textwidth]{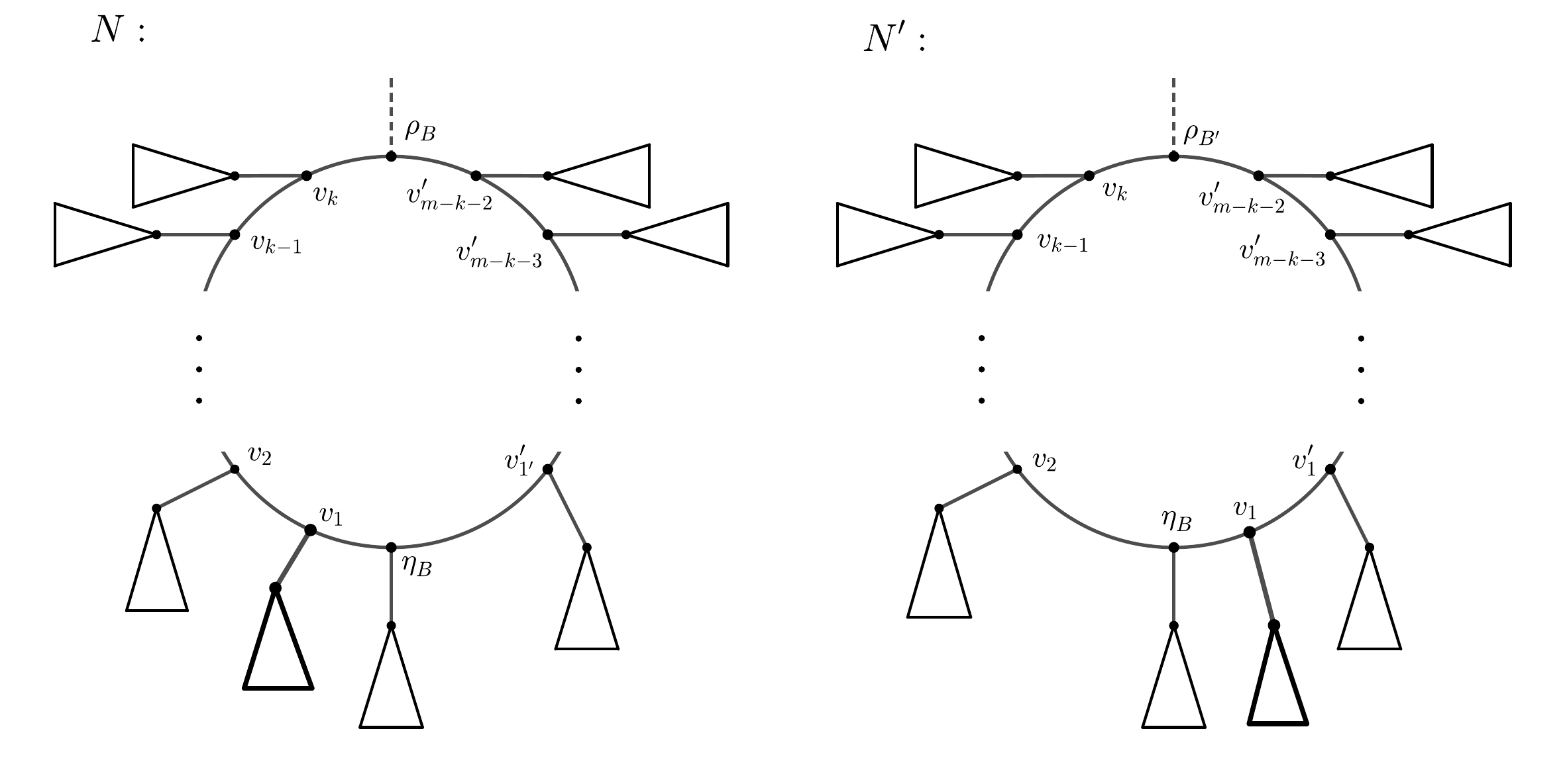}
\label{fig:bin_min_only_fullmoons}
\end{SCfigure}

By  Lemma \ref{lem:blocks-in-partN} and \ref{lem:relevant-partN}, $\partN(v_1)$ is a phylogenetic level-$1$ network and all relevant vertices in $N$ remain relevant in $N'$ and vice versa, i.e.,  $\relV_N=\relV_{N'}$. In particular, 
the values $\phi(w)$ of relevant vertices $w$ in $N$ and $N'$ do not change. 
One easily verifies that $N'$ is a binary level-$1$ network. Moreover, since $\rho_B$ has precisely two children, 
we have, in essence, just replaced $\partN(\rho_B)$ in $N$ by the binary level-$1$ network $\partN'(\rho_{B'})$ to obtain $N'$
and both subnetworks have the same number $\ell \coloneqq |L_N(\rho_B)|$ of leaves. It is easy to see that, by the latter arguments, 
$\Phi^*(\partN(\rho_B)) = \Phi^{*}(\partN'(\rho_{B'}))$. This and Theorem~\ref{thm:local-phi} implies
that $\Phi^{**}(N)-\Phi^{**}(N') = (\phi_N(\rho_B) - \phi_{N'}(\rho_{B'})) \binom{\ell}{2} = 
(\epsilon + \omega(B)) - (\epsilon + \omega(B')) \binom{\ell}{2}  =  ( \omega(B) - \omega(B')) \binom{\ell}{2} $.

We can directly compute $\omega(B)$ and $\omega(B')$ using the same approach as in the proof of Proposition~\ref{prop:B_min_binary}. More precisely, we obtain
\begin{align*}
    \omega(B) &= \binom{2+k}{3} + \binom{m-k}{3} \quad \text{and} \\
    \omega(B') &= \binom{1+k}{3} + \binom{m-k+1}{3}. 
\end{align*}
Now, taking the difference between $\omega(B)$ and $\omega(B')$ and simplifying, we have
\begin{align*}
    \omega(B) - \omega(B') &= \frac{1}{2} (1+2k-m) m > 0,
\end{align*}
where the strict inequality follows from the fact that $k \geq \lceil \frac{m-2}{2} \rceil + 1 = \lceil \frac{m}{2} \rceil$ and thus $1+2k-m > 0$. This directly implies $\Phi^{**}(N) > \Phi^{**}(N')$, thereby contradicting the minimality of $N$. In particular, all non-trivial blocks in a binary phylogenetic level-$1$ network minimizing the weighted total cophenetic index have to be full-moons, which completes the proof. 
\end{proof}

We proceed by showing that binary phylogenetic level-$1$ networks minimizing $\Phi^{**}$ cannot contain large blocks.

\begin{lemma}\label{lem:bin_min_at_most_triangles}
Let $N\in \BinLevelOneN$ be a binary level-$1$ network with $n\geq 1$ leaves that minimizes the weighted total cophenetic index within the class $\BinLevelOneN$.
Then, $N$ does not contain blocks of size $m\geq 4$.
\end{lemma}
\begin{proof}
  Let $N\in \BinLevelOneN$, $n\geq1$, be a binary level-$1$ network that minimizes the weighted
  total cophenetic index within the class $\BinLevelOneN$. If $n=1$, the statement is vacuously
  true. Hence, we may assume in the following that $n\geq 2$. Assume, for contradiction, that $N$
  contains a non-trivial block $B$ of size $m\geq 4$. By Lemma~\ref{lem:bin_min_fullmoon}, $B$ has
  to be a full-moon. Hence, $B$ consist of two internal vertex disjoint paths
  $\rho_B,l_{\left\lceil\frac{m-2}{2}\right\rceil},
  l_{\left\lceil\frac{m-2}{2}\right\rceil-1},\ldots,l_2,l_1,\eta_B$ and $\rho_B,
  r_{\left\lfloor\frac{m-2}{2}\right\rfloor},
  r_{\left\lfloor\frac{m-2}{2}\right\rfloor-1},\ldots,r_2,r_1,\eta_B$. Since $m\geq 4$, both
  $\left\lceil\frac{m-2}{2}\right\rceil\geq 1$ and $\left\lfloor\frac{m-2}{2}\right\rfloor\geq 1$.
  We now obtain a network $N'$ from $N$ as follows: We delete edges $(l_2,l_1)$,
  $(l_1,\eta_B)$, and, if it exists, also $(\parent(\rho_B),\rho_B)$, and insert edges
  $(l_1,\rho_B)$, $(l_{2},\eta_B)$, and also $(\parent(\rho_B),l_1)$ if $\rho_B\neq
  \rho_{N}$ (see Figure \ref{fig:bin_min_only_triangles} for an illustration). In the following, we
  denote the resulting block of size $m-1$ in $N'$ obtained from $B$ by $B'$.

\begin{SCfigure}[.3][ht]
	\caption{Schematic view of the subnetworks $N_1$ and $N_2$ in  $N$ and $N'$, respectively,  
					  used in the proof of Lemma \ref{lem:bin_min_at_most_triangles}. The
	         subnetwork $\partN(l_{1})$ that is
	         relocated when constructing $N'$ is marked with bold edges and
	         vertices.}
   \includegraphics[width=0.65\textwidth]{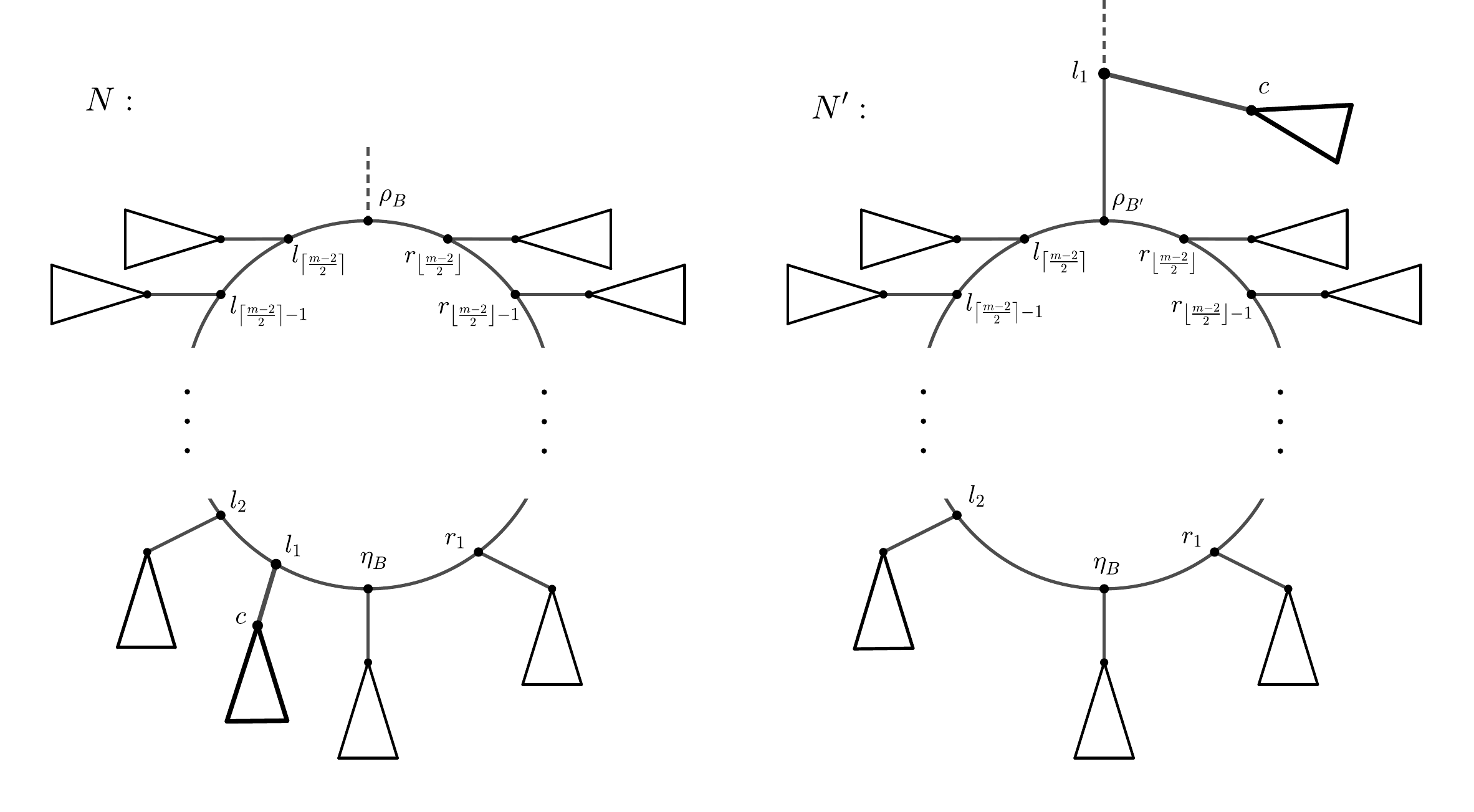}
   \label{fig:bin_min_only_triangles}
\end{SCfigure}
    
  We now show that $\Phi^{**}(N)-\Phi^{**}(N')>0$, contradicting the minimality of $N$. Since $N$ is
 binary, there is exactly one vertex in $\child_N^*(l_1)=:\{c\}$ and one easily observes that
 $N'$ is obtained from $N$ by replacing $N_1\coloneqq\partN(\rho_B)=N(\rho_B)$ by
 $N_2\coloneqq\partN'(l_{1})=N'(l_{1})$. In particular $|L(N_1)| = |L(N_2)|\eqqcolon \ell$. By
 Theorem~\ref{thm:local-phi},
    \[\Phi^{**}(N)-\Phi^{**}(N')=(\Phi^*(N_1)-\Phi^*(N_2))+(\phi_{N}(\rho_B)-\phi_{N'}(l_1)\binom{\ell}{2}.\]
    
  For a vertex $v$ in $B$ in $N$, let us denote with $v'$ the unique child of $v$ that is not
  located in $B$. Observe that none of the vertices $v\in V(B)\setminus \{\rho_B\}$ are relevant.
  Hence, all the relevant neighbors of $\rho_B$ in $N$ (and thus, by
  Corollary~\ref{cor:relevant-partN}, in $N_1$), must satisfy
  Definition~\ref{def:rel-neighbor}(2.b). 
     
 Hence, if $u$ is a relevant neighbor of $\rho_B$, then $u=v'$ must hold for some vertex $v$ in $B$.
 Let $u_1,\ldots,u_k$ be the relevant neighbors of $\rho_B$ in $N$. Note that $k=0$ is possible if
 the child $v'$ of each vertex $v$ in $B$ is a leaf. In $N'$, vertex $l_1$ is a true tree
 vertex and, therefore, it has at most two relevant neighbors, namely $\rho_B$ and $c$. While
 $\rho_B$ is always a relevant neighbor of $l_1$ in $N'$, we can observe that $c$ is a leaf in
 $N$ if and only if $c$ is a leaf in $N'$ and thus, $c$ is a relevant neighbor of $l_1$ in
 $N'$ if and only of $c$ is a relevant neighbor of $\rho_N$ in $N$. 
     
 This and the construction of $N'$, in particular, implies that each $u_i$, $1\leq i\leq k$ with
 $u_i\neq c$ is a relevant neighbor of $\rho_B$ in $N$ if and only if $u_i$ is a relevant neighbor
 of $\rho_B$ in $N'$. It is, therefore, easy to see that (I) $\Phi^*(N_1)-\Phi^*(N'(\rho_{B'}))=0$
 and (II) $\Phi^{*}(N_2)=\Phi^{**}(N'(\rho_{B'}))$ if $c$ is a leaf in $N$ as well as that (I')
 $\Phi^*(N_1)-\Phi^*(N'(\rho_{B'}))=\Phi^{**}(N(c)) = \Phi^{**}(N'(c))$ and (II')
 $\Phi^{*}(N_2)=\Phi^{**}(N'(\rho_{B'})) + \Phi^{**}(N'(c))$ if $c$ is not a leaf in $N$ and, thus,
 relevant. Hence, if $c$ is a leaf, we combine the two equations (I) and (II) and use the fact that
 $\Phi^{**}(N'(\rho_{B'})) = \Phi^{*}(N'(\rho_{B'})) + \phi_{N'}(\rho_{B'})
 \binom{|L_{N'}(\rho_{B'})|}{2}$ to obtain $\Phi^*(N_1) - \Phi^{*}(N_2) = -\phi_{N'}(\rho_{B'})
 \binom{|L_{N'}(\rho_{B'})|}{2}$. Similarly, if $c$ is relevant, the two equations (I') and (II')
 together with $\Phi^{**}(N'(\rho_{B'})) = \Phi^{*}(N'(\rho_{B'})) + \phi_{N'}(\rho_{B'})
 \binom{|L_{N'}(\rho_{B'})|}{2}$ yield $\Phi^*(N_1) - \Phi^{*}(N_2) = -\phi_{N'}(\rho_{B'})
 \binom{|L_{N'}(\rho_{B'})|}{2}$. Put $\ell'\coloneqq |L_{N'}(\rho_{B'})|$. The latter arguments
 together with Corollary~\ref{fact:sum-up-part} imply
    \begin{equation*}
        \begin{aligned}
             \Phi^{**}(N)-\Phi^{**}(N') &= (\Phi^*(N_1)-\Phi^*(N_2))+(\phi_{N}(\rho_B)-\phi_{N'}(l_1))\binom{\ell}{2}\\
             &=-\phi_{N'}(\rho_{B'})\binom{\ell'}{2}
            +(\phi_{N}(\rho_B)-\underbrace{\phi_{N'}(l_1)}_{=1})\binom{\ell}{2}\\
            &=\phi_{N}(\rho_B)\binom{\ell}{2}- \phi_{N'}(\rho_{B'})\binom{\ell'}{2}-\binom{\ell}{2}.
        \end{aligned}
    \end{equation*}
    By the arguments as used in the proof of Proposition\ \ref{prop:B_min_binary},
       \[\omega(B) = \binom{\left\lceil\frac{m-2}{2}\right\rceil+2}{3}+\binom{\left\lfloor\frac{m-2}{2}\right\rfloor+2}{3}
       \text{ and } \omega(B') = \binom{\left\lceil\frac{m-2}{2}\right\rceil+1}{3}+\binom{\left\lfloor\frac{m-2}{2}\right\rfloor+2}{3}.\]
       Now put $\alpha\coloneqq \binom{\left\lceil\frac{m-2}{2}\right\rceil+2}{3}$, $\beta\coloneqq \binom{\left\lfloor\frac{m-2}{2}\right\rfloor+2}{3}$, and $\gamma\coloneqq \binom{\left\lceil\frac{m-2}{2}\right\rceil+1}{3}$
       and observe that $\phi_{N}(\rho_B) = \alpha +\beta +\epsilon$ and $\phi_{N'}(\rho_{B'})=\gamma+\beta+\epsilon$. Moreover, we have $\ell>\ell'\geq 3$, $m\geq 4$ and $\beta>1$ which implies that 
       $(\gamma+\beta+\epsilon)\binom{\ell'}{2} < (\gamma+\beta+\epsilon)\binom{\ell}{2}$.
       Hence, 
    \begin{equation*}
        \begin{aligned}
            \Phi^{**}(N)-\Phi^{**}(N') &= \phi_{N}(\rho_B)\binom{\ell}{2}- \phi_{N'}(\rho_{B'})\binom{\ell'}{2}-\binom{\ell}{2}\\
                                       & = (\alpha +\beta +\epsilon)\binom{\ell}{2} - (\gamma+\beta+\epsilon)\binom{\ell'}{2} -\binom{\ell}{2}\\
                                       & > (\alpha +\beta +\epsilon)\binom{\ell}{2} - (\gamma+\beta+\epsilon)\binom{\ell}{2} -\binom{\ell}{2}\\
                                       & =  (\alpha-\gamma-1) \binom{\ell}{2}\\ 
                                       & =  \underbrace{\left(\binom{\left\lceil\frac{m-2}{2}\right\rceil+2}{3} - \binom{\left\lceil\frac{m-2}{2}\right\rceil+1}{3} -1\right)}_{\geq 0\text{ since } m\geq 4}\binom{\ell}{2} \geq 0.
        \end{aligned}
    \end{equation*}   
    Therefore, $\Phi^{**}(N)-\Phi^{**}(N')>0$ and thus, $\Phi^{**}(N)>\Phi^{**}(N')$; a contradiction to the assumption that 
    $N$ minimizes the weighted total cophenetic index within the class $\BinLevelOneN$. In summary, $N$ cannot contain blocks of size $m\geq 4$. 
\end{proof}

Binary phylogenetic level-$1$ networks with minimum weighted total cophenetic index are characterized as follows.
\begin{theorem}\label{thm:binary_min}
    Let $N\in \BinLevelOneN$ be a binary phylogenetic level-$1$ network with $n \geq 2$ leaves. Then, 
    \[\Phi^{**}(N)\geq \sum\limits_{k=1}^{n-1} a(k) + \binom{n}{2},\] where $a(k)$ is the highest power of 2 	that divides $k!$. 
    This bound is tight and is, in particular, achieved 
    precisely if $N\in \Nmb\subseteq  \BinLevelOneN$. If $N\in \Nmb$ and $\epsilon>0$, then $N$ is a maximally balanced tree.
\end{theorem}
\begin{proof}
 Let $N$ be a binary level-$1$ network with $n \geq 2$ leaves. 
 Assume first that $N$ is a tree. Proposition \ref{prop:cp-properties} implies that 
 $\Phi^{**}(N) \geq \sum_{k=1}^{n-1} a(k) + \binom{n}{2}$, where $a(k)$ is the highest power of 2 that divides $k!$
 (equality holds precisely if $N$ is maximally balanced).
 
 Suppose now that $N$ is not a tree and, therefore, that $N$ contains a non-trivial block $B$. 
 Assume first that $N\in \Nmb$. Hence, $B$ is a triangle. Moreover, since $N$ is binary and level-$1$, 
 $B$ is vertex-disjoint from possibly other non-trivial blocks in $N$. Hence, $\phi_{N}(\rho_B) = 1+\epsilon$. 
 If we now collapse $B$, we obtain a network $\tilde N$ in which $\rho_{B}$ is not
 the root of any non-trivial block. Hence, $\phi_{\tilde N}(\rho_B)=1$. One easily verifies that, 
 therefore, $\Phi^{**}(N)\geq \Phi^{**}(\tilde N)$ (equality holds precisely if $\epsilon=0$). 
 Repeating the latter arguments implies that 
 $\Phi^{**}(N) =  \Phi^{**}(T)$ if $\epsilon=0$ and  
 $\Phi^{**}(N) >  \Phi^{**}(T)$ if $\epsilon>0$
 where $T$ is the tree obtained from  $N$ by collapsing all triangles. 
 By Observation\ \ref{fact:triangle}, $T$ is maximally balanced and thus, $T\in \Nmb$. By Proposition \ref{prop:cp-properties},  
 $\Phi^{**}(T) = \sum_{k=1}^{n-1} a(k) + \binom{n}{2}$. 

 Assume now that $N\notin \Nmb$. 
 Since $N$ is not a tree and binary level-$1$, 
 it contains a non-trivial block $B$ that is not a triangle (and thus, $B$ is of size of size $m>3$) or $N$ is not obtained from a maximally
 balanced tree by expanding vertices to triangles. 
 Assume first that $N$ contains a non-trivial block $B$ of size $m>3$. 
 As shown in the proof of Lemma \ref{lem:bin_min_at_most_triangles}, one can construct a binary level-$1$ network $N'$
 with a block of size $m-1\geq 3$ such that $\Phi^{**}(N)>\Phi^{**}(N')$. Repeating this process, eventually
 will lead to a binary level-$1$ network $N''$ such that $\Phi^{**}(N)>\Phi^{**}(N'')$
 and where all non-trivial blocks in $N''$ are triangles.  
 By Observation~\ref{fact:triangle-1}, one can collapse all of these triangles to obtain a binary tree $T$ such that $\Phi^{**}(N'')\geq\Phi^{**}(T)$. Proposition \ref{prop:cp-properties} implies that 
 $\Phi^{**}(T) \geq \sum_{k=1}^{n-1} a(k) + \binom{n}{2}$.
 By the latter arguments, $\Phi^{**}(N)>\Phi^{**}(N'') \geq  \Phi^{**}(T) \geq \sum_{k=1}^{n-1} a(k) + \binom{n}{2}$. 
  Assume now that all non-trivial blocks in $N$ (if there are any) are triangles.
 This and the fact that $N$ is binary allows us to apply Observation\ \ref{fact:triangle-1} and
 to conclude that one can obtain a binary tree $T$ by collapsing all triangles in $N$ such that 
 $\Phi^{**}(N) \geq \Phi^{**}(T)$. In particular, since $N\notin \Nmb$ and by Observation\ \ref{fact:triangle-1} again,  
 $T$ is not maximally balanced. Proposition \ref{prop:cp-properties} implies that 
 $\Phi^{**}(N) > \sum_{k=1}^{n-1} a(k) + \binom{n}{2}$. 
 
 Let $T$ be a maximally balanced tree and therefore, $\Phi^{**}(T) = \sum_{k=1}^{n-1} a(k) + \binom{n}{2}$
and $T\in \Nmb$. Based on the latter arguments, we obtain the following results.
If $N\in \Nmb$ and $N$ is a tree, we have $\Phi^{**}(N)=\Phi^{**}(T)$. 
If $N\in \Nmb$ and $N$ is not a tree, then $\Phi^{**}(N)\geq \Phi^{**}(T)=\sum_{k=1}^{n-1} a(k) + \binom{n}{2}$ (equality holds precisely if $\epsilon=0$). 
If $N\notin \Nmb$, then $\Phi^{**}(N)> \Phi^{**}(T)$.
In other words, the bound is tight and is, in particular, achieved if and only if $N\in \Nmb$. 
In the latter case, $\epsilon>0$ implies that $N$ is a maximally balanced tree.
\end{proof}

\subsection{Maximum networks}
\label{APPX:subsection:MaxN}
In the following, we will investigate the structure of phylogenetic level-$1$
networks with maximum weighted total cophenetic index. Observe first that, for
every $n\geq 2$, there are infinitely many phylogenetic level-$1$ networks with
$n$ leaves and different balance values (cf.\ Figure~\ref{fig:N_cres}(b)). In
particular, Figure~\ref{fig:N_cres}(b) provides a generic example of a class of
networks $N$ with $n$ leaves whose single non-trivial block $B$ can be made
arbitrarily large, showing that $\omega(B)$ and thus also the weighted total
cophenetic index is unbounded.
\begin{fact}\label{obs:max-unbounded}
    The weighted total cophenetic index is unbounded on phylogenetic level-$1$ networks with $n\geq 2$ leaves.
\end{fact}

Observation\ \ref{obs:max-unbounded} motivates the question as whether more restrictive networks
may have a bounded weighted total cophenetic index and, if so, if one can characterize the structure
of networks achieving it. To shed some light into this question, we focus now on the  class $\BinLevelOneN$
of binary phylogenetic level-$1$ networks and the more general class $\HybdidDegTwoN$ of phylogenetic level-$1$ networks
where all hybrids have in-degree two. In particular, Lemma \ref{lem:vertices-bounded} implies
\begin{fact}\label{obs:bounded-indegree2}
	    The weighted total cophenetic index is bounded for all networks within the classes 
	     $\HybdidDegTwoN$ and $\BinLevelOneN$. 
\end{fact}

Moreover, we obtain

\begin{fact}\label{fact:n2_max_bin}
$\BinLevelOneNvar{2}$ and $\HybdidDegTwoNvar{2}$ each consist of precisely two networks (up to isomorphism), namely the star tree $\TstarTwo$ and the triangle network $N_\Delta$.
For $\epsilon=0$, both  $\TstarTwo$ and  $N_\Delta$ have a weighted total cophenetic index of 1. For $\epsilon>0$, 
the triangle network $N_\Delta$ is the unique network within the classes $\BinLevelOneNvar{2}$ and $\HybdidDegTwoNvar{2}$ with maximum weighted total cophenetic index, namely $\Phi^{**}(N_\Delta) = 1+\epsilon$.
\end{fact}

In the remainder of this section, we will prove the following characterization of networks that 
maximize the weighted total cophenetic index within the classes $\BinLevelOneN$ and $\HybdidDegTwoN$. 
\begin{theorem}\label{thm:binary_galledtree_maximum}
    Let $N\in \BinLevelOneN$ (resp.\ $N\in \HybdidDegTwoN$) be a  network with $n \geq 3$ leaves. Then, 
    \[\Phi^{**}(N)\leq \left(\binom{n+1}{3}+\epsilon\right)\binom{n}{2}.\] This bound is tight and is, in particular, achieved 
    for $N\in \BinLevelOneN$  (resp.\ $N\in \HybdidDegTwoN$) if and only if $N\simeq \NC_n$.
\end{theorem}

In order to prove Theorem \ref{thm:binary_galledtree_maximum},
we first provide a necessary condition for the non-trivial blocks in networks that achieve a maximum weighted total cophenetic index
within the classes $\BinLevelOneN$ and $\HybdidDegTwoN$. 

\begin{lemma} \label{lem:one_block}
   Let $N\in \HybdidDegTwoN$ (resp.\ $N\in \BinLevelOneN$) be a phylogenetic level-$1$ network with $n \geq 3$ leaves  that 
   has maximum weighted total cophenetic index among all networks within the class $\HybdidDegTwoN$ (resp.\ $\BinLevelOneN$).
   Then, $N$ contains at least one non-trivial block and all non-trivial blocks of $N$ are crescents.
\end{lemma}
\begin{proof}
		By Observation\ \ref{obs:bounded-indegree2}, networks in $\HybdidDegTwoN$ and $\BinLevelOneN$ have bounded maximum weighted total cophenetic index. 
		Hence, we can choose $N\in \HybdidDegTwoN$  (resp.\ $N\in \BinLevelOneN$) as 
		 one of the networks that maximizes the weighted total cophenetic index within the class $\HybdidDegTwoN$ (resp.\ $\BinLevelOneN$).
    By Proposition \ref{prop:cp-properties}, the weighted total cophenetic index on trees is maximized by the caterpillar tree with a value of $\binom{n+1}{3}$
    and the caterpillar tree is contained in $\BinLevelOneN\subset \HybdidDegTwoN$.
    However, there are phylogenetic level-$1$ networks containing non-trivial blocks that have a higher value, e.g., the network $\NC_n$ is contained in  $\BinLevelOneN\subset \HybdidDegTwoN$ which satisfies
    	$\Phi^{**}(\NC_n) = (\binom{n+1}{3}+\epsilon)\binom{n}{2}$.
     Thus, $N$ must contain at least one non-trivial block $B$. Suppose that $B$ has size $m$.
     Assume, for contradiction, that $B$ is not a crescent. By Observation\ \ref{obs:Bm3-isom}, $m>3$.
     Since $\eta_B$ has in-degree two, $B$ is composed of two internally vertex-disjoint $\rho_B \eta_B$-paths
     $P=(\rho_B,v_1,\dots, v_k, \eta_B)$ and $P'=(\rho_B,w_1,\dots, w_{k'}, \eta_B)$ (cf.\ \cite[Cor.\ 39]{HSS:22}).
     Note that  $k,k'\geq 1$ since $B$ is not a crescent.
     We now replace $B$ in $N$ by a crescent $B'$. To be more precise, 
     we remove  all edges within $B$ from $N$ and add 
     the two internally vertex-disjoint $\rho_B \eta_B$-paths $(\rho_B,\eta_B)$
     and $(\rho_B,v_1,\dots, v_k,w_1,\dots, w_{k'},\eta_B)$ to obtain the network $N'$ (cf.\ Figure\ \ref{fig:hybridtwo_one_block}).

\begin{SCfigure}[.4][ht]
	\caption{Construction of $B'$ from $B$ in the proof of Lemma~\ref{lem:one_block}.}
   \includegraphics[width=0.6\textwidth]{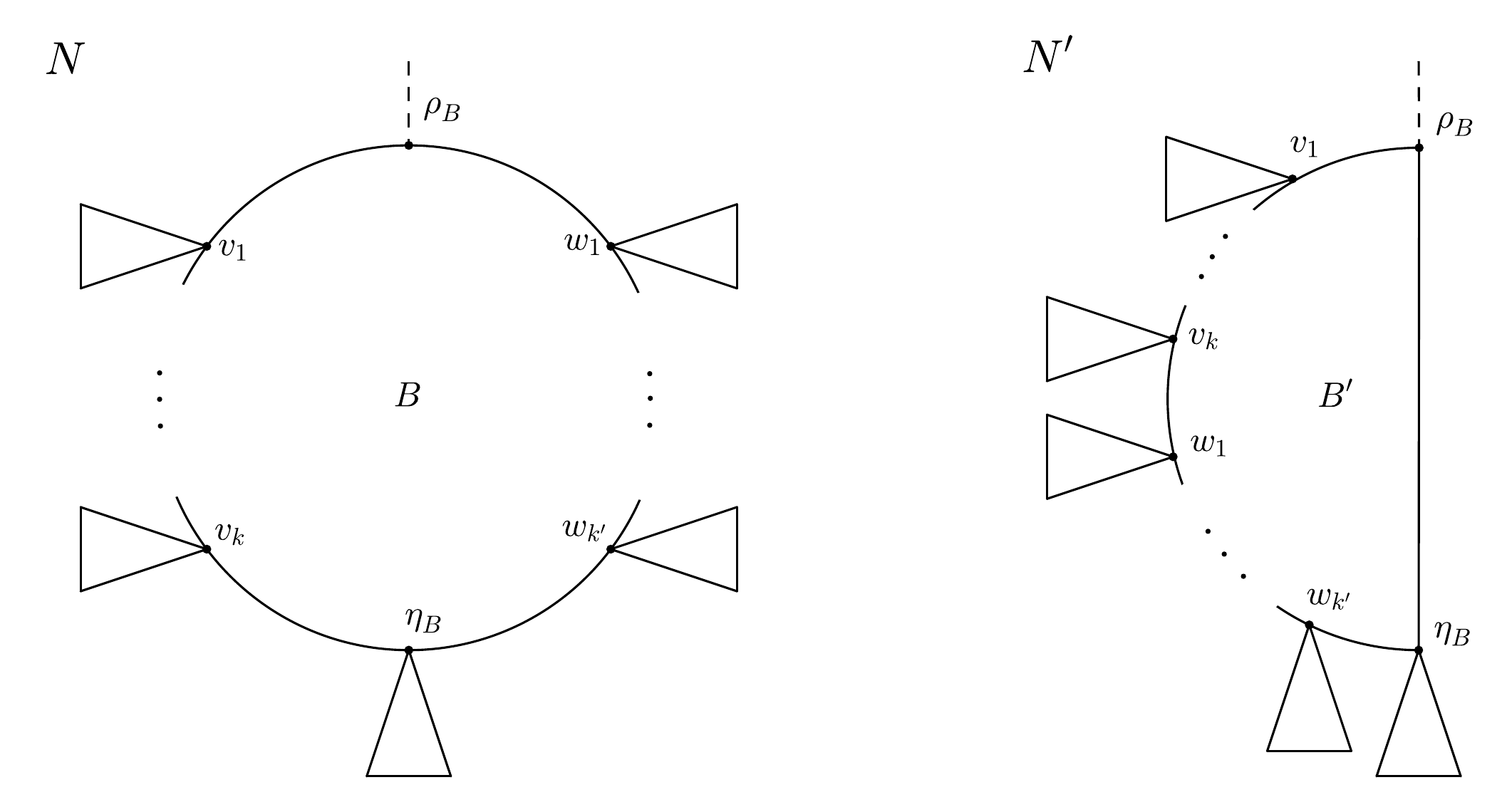}
    \label{fig:hybridtwo_one_block}
\end{SCfigure}

     By construction, $V\coloneqq V(N) = V(N')$ and $W \coloneqq V(B)=V(B')$ .
     One easily verifies that $N'\in \HybdidDegTwoN$ (resp.\ $N\in \BinLevelOneN$) and that
     $\rho_{B'}=\rho_B$, $\eta_{B'} = \eta_{B}$, and $B'$ has size $|V(B')|= |W|=m$. 
		 It is now a straightforward task for verify that 
		 $v\in V$ is relevant in $N$ if and only if $v$ is relevant in $N'$
		 and that $\phi_N(v)=\phi_{N'}(v)$ for all $v\neq \rho_B$. Put $R\coloneqq \relV(N)$. By the latter arguments,  $R = \relV(N')$. Moreover, we have $L_N(v) = L_{N'}(v)$ for all $v\in V\setminus W$
		 and $L_N(w) \subseteq  L_{N'}(w)$  and, thus, $|L_N(w)|\leq  |L_{N'}(w)|$
				 for all $w\in W$.   On easily observes that,  in particular, $\ell\coloneqq |L_{N}(\rho_{B})| = |L_{N'}(\rho_{B'})|$.
		By Proposition\ \ref{prop:B_maxmin_arbitrary}, $\omega(B)<\omega(B') = \binom{m}{3}$. 
		Hence, $\phi_N(\rho_B) = (\omega(B)+\epsilon)\binom{\ell}{2} <  (\omega(B')+\epsilon)\binom{\ell}{2} = \phi_{N'}(\rho_B)$.

		 Taken the latter arguments together, we obtain 
		 \[\alpha\coloneqq \sum_{v\in R\cap (V\setminus W)} \phi_N(v) \binom{|L_N(v)|}{2} = \sum_{v\in  R\cap (V\setminus W)} \phi_{N'}(v) \binom{|L_{N'}(v)|}{2}\]
		 and \[\beta\coloneqq \sum_{v\in R\cap W} \phi_N(v) \binom{|L_N(v)|}{2} < \sum_{v\in R\cap  W} \phi_{N'}(v) \binom{|L_{N'}(v)|}{2}\eqqcolon \gamma.\] 
		 Since $\Phi^{**}(N)=\alpha+\beta$ and $\Phi^{**}(N')=\alpha+\gamma$, we can conclude that $\Phi^{**}(N)<\Phi^{**}(N')$; a contradiction to the assumption that 
		$N$ is a network that maximizes the weighted total cophenetic index within the class $\HybdidDegTwoN$.
		In summary, $N$ contains at least one non-trivial block and all non-trivial blocks of $N$ are crescents. 
\end{proof}

For the proof of Theorem \ref{thm:binary_galledtree_maximum}, we require further additional results. 
We first characterize the structure of binary networks $N\in \BinLevelOneN$ that achieve maximum 
weighted total cophenetic index. Afterwards, we prove that  each network $N\in \HybdidDegTwoN$ that has 
maximum weighted total cophenetic index within the class $\HybdidDegTwoN$ must be binary. 
To establish the results for binary networks, we need a particular
\enquote{swapping operation} to exchange the positions of certain subnetworks. As we shall see below, 
the weighted total cophenetic index for binary phylogenetic level-$1$ networks is invariant under this operation.

\begin{definition}\label{def:swapping}
    Let $B$ be a non-trivial block in a binary phylogenetic level-$1$ network $N$ and suppose that there are two vertices $p_1,p_2\in V(B)$ 
    with children  $v_1$ and $v_2$, respectively, that are not contained in $V(B)$. 
    A \emph{swap} of the subnetworks $N(v_i)$ and $N(v_j)$ is achieved by
    replacing the edges $(p_1,v_1)$ and $(p_2,v_2)$ by 
    the edges $(p_1,v_2)$ and $(p_2,v_1)$. We denote the resulting network by $\mathrm{swap}(N,v_1,v_2)$.

\end{definition}

Definition\ \ref{def:swapping} is illustrated in Figure \ref{fig:swapping}.

\begin{figure}[htbp]
    \centering
    \includegraphics[width=0.8\textwidth]{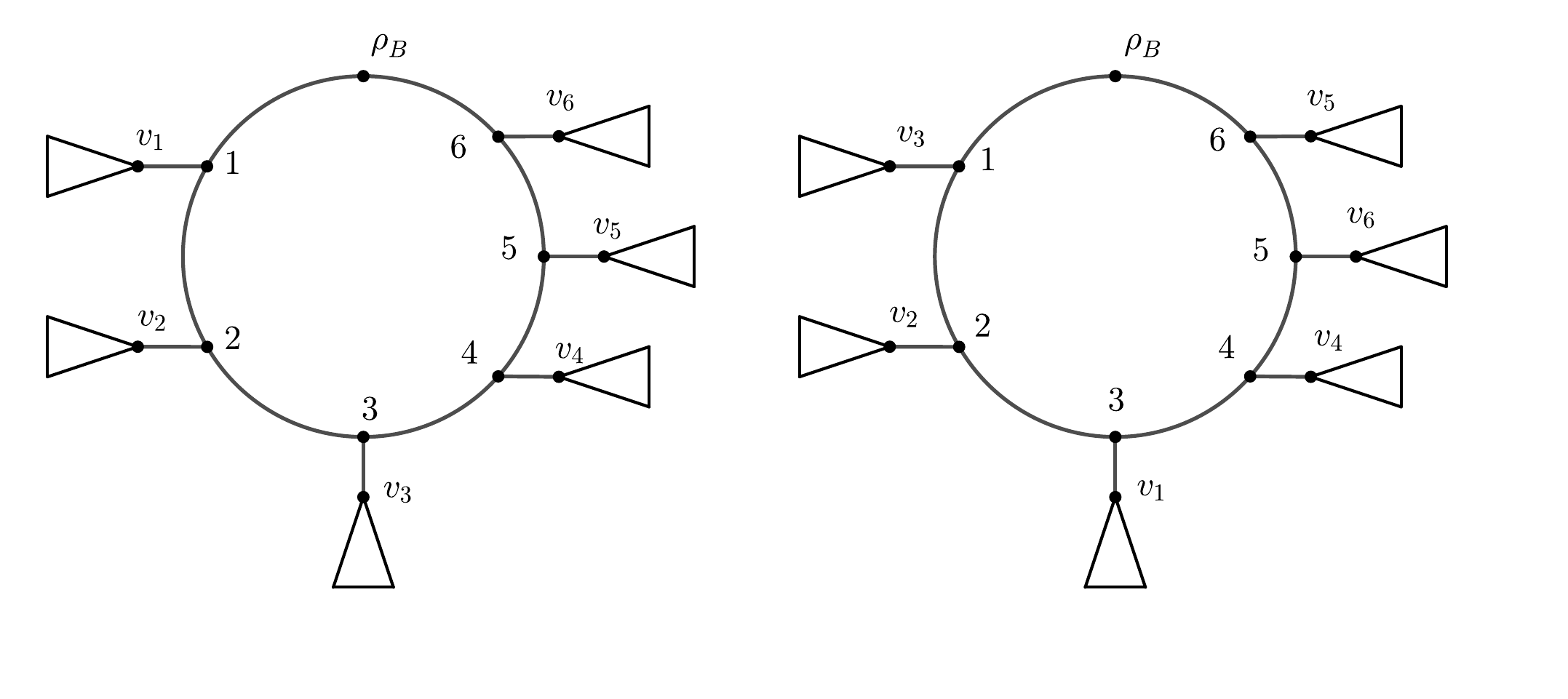}
    \caption{Example for Definition \ref{def:swapping}. Left: a binary phylogenetic level-$1$ network $N$ with a single non-trivial block $B$
    with \enquote{attached} subnetworks $N(v_1),\ldots,N(v_6)$. Right: the network $N'\coloneqq \mathrm{swap}(\mathrm{swap}(N,v_1,v_3),v_5,v_6)$
    that is obtained from $N$ by first swapping  $N(v_1)$ and $N(v_3)$ and then swapping  $N(v_5)$ and $N(v_6)$.
    }
    \label{fig:swapping}
\end{figure}

\begin{lemma}\label{lem:binary_swap}
    Let $N$ be a binary phylogenetic level-$1$ network containing a non-trivial
    block $B$ and suppose that there are vertices $v_1,v_2\in V(N)\setminus V(B)$
    that are children of vertices in $B$. Then, $N'=\mathrm{swap}(N,v_1,v_2)$
    is a phylogenetic level-$1$ network that satisfies 
    $\Phi^{**}(N)=\Phi^{**}(N')$.    
\end{lemma}
\begin{proof}
Let $N$ be a binary phylogenetic level-$1$ network containing a non-trivial
    block $B$ and let $W$ be the set of vertices $w\in V(N)\setminus V(B)$ that
    are children of  vertices in $B$. Let $v_1,v_2\in W$ be chosen arbitrarily. Put $N'=\mathrm{swap}(N,v_1,v_2)$. 
    One easily verifies that swapping networks does not change the in- and out-degree
    of any vertex and thus, $N'$ remains a phylogenetic network.
    In particular, since $N$ is binary,  $v_1$ and $v_2$ must have distinct parents in $B$
    and hence, $N(v_1)$ and $N(v_2)$ are vertex-disjoint. 
    It is now an easy task to verify that $N'$ remains a level-$1$ network.
    By Observation \ref{fact:binary_rel_vert}, the only relevant vertices in $N$ are roots of non-trivial blocks and inner vertices that are not contained in a non-trivial block. 
    Hence, the only relevant vertex of $N$ that is contained in $B$ is the root $\rho_B$ of $B$. 
    One easily verifies that $|L_N(\rho_B)| = |L_{N'}(\rho_B)|$ and, in particular, $\phi_N({\rho_B})=\phi_{N'}(\rho_B)$.
    By Corollary\ \ref{cor:relevant-partN}, a vertex $v\in N(w)$ is relevant in $N$ if and only if $v$ is relevant in $N(w) $
    for all  $w\in W$. By construction, $N(w)=N'(w)$ for all  $w\in W$ and thus, 
    $R\coloneqq \relV_N \bigcap ( \cup_{w\in W} V(N(w)) ) = \relV_{N'} \cap ( \cup_{w\in W} V(N'(w)) )$.
    In particular,  $|L_N(v)|=|L_{N'}(v)|$ and $\phi_N(v)=\phi_{N'}(v)$ for all $v\in R$. 
    Note that $N'$ can also be considered as the network obtained from $N$ by
    replacing $N(\rho_B) = \partN(\rho_B)$ by the network $N'(\rho_B)$. Hence, we can apply
    Lemma \ref{lem:replacement-partN-1} and conclude that $v\in V(N)\setminus V(N(\rho_B))$ 
    is relevant in $N$ if and only if $v\in V(N')\setminus V(N'(\rho_B))$ is relevant in $N'$. 
    In particular, $|L_N(v)=|L_{N'}(v)|$ and $\phi_N(v)=\phi_{N'}(v)$ for all relevant vertices $v\in V(N)\setminus V(N(\rho_B))$. 
    In summary, $\relV_N = \relV_{N'}$ and $|L_N(v)=|L_{N'}(v)|$ as well as $\phi_N(v)=\phi_{N'}(v)$ for all
    $v\in  \relV_N$. This, in particular, implies that $\Phi^{**}(N) = \Phi^{**}(N')$. 
\end{proof}

In the following, we show that each binary phylogenetic level-$1$ network maximizing the weighted total cophenetic index cannot 
contain true tree vertices and contains precisely one non-trivial block. In addition, we show that each network $N\in \HybdidDegTwoN$ 
maximizing the  weighted total cophenetic index must be binary. These results are provided in Lemmas \ref{lem:max_is_binary}, \ref{lem:bin_no_true_vert},
and \ref{lem:exactly_one_block}. The proofs of these lemmas are rather lengthy and, although straightforward, 
involve tedious case-by-case analysis. Their proofs are, therefore, placed in the appendix. 

\begin{lemma}\label{lem:max_is_binary}
If $N\in \HybdidDegTwoN$, $n\geq 3$, is a network that maximizes the weighted total cophenetic index within the class $\HybdidDegTwoN$, then
 $N$ is a binary.
\end{lemma}  
\begin{proof}	
By Observation\ \ref{obs:bounded-indegree2}, networks in $\HybdidDegTwoN$ have bounded maximum weighted total cophenetic index. 
			Thus, we can choose $N_1\in \HybdidDegTwoN$ as a network that has finite and maximum weighted total cophenetic index within the class $\HybdidDegTwoN$, $n\geq 3$.
			By Lemma \ref{lem:non_separated}, we can assume w.l.o.g.\ that each hybrid vertex has precisely one child. 
			Assume, for contradiction, that $N_1$ is not binary. Hence, there is a inner vertex $x$ in $N_1$ such that $\outdeg(x)=k\geq 3$.
			By assumption, $x$ is not a hybrid vertex. We consider now the following case: either $x$ is a true tree vertex
			or $x$ is contained in a non-trivial block. 
		
			Assume, first that $x$ is a true tree vertex in $N_1$ and let $x_1,\dots,x_k$ be the children of $x$ in $N_1$. 
			Consider now the network $N_2$ obtained from $N_1$ by removing edges $(x,x_i)$, $2\leq i\leq k$, and
			adding a new vertex $x'$ and the edges $(x,x')$ and $(x',x_i)$, $2\leq i\leq k$. 
			One easily verifies that $N_2\in \HybdidDegTwoN$. 
			Since $x$ is a true tree vertex in $N_1$, $x$ remains a true tree vertex in $N_2$. 
			Moreover, since $x$ is not contained in any non-trivial block of $N_1$,
			$x'$ is not contained in any non-trivial block of $N_2$. 
			Hence, $x$ and $x'$ are true tree vertices in $N_2$. 
			By Observation\ \ref{obs:relevant}, $x$ is relevant in $N_1$ and $x, x'$ are relevant in $N_2$. 
			In particular, $\phi_{N_1}(x) = \phi_{N_2}(x) = \phi_{N_2}(x')=1$.
			It is a straightforward task to verify that 
	    all relevant vertices  $v\in V(N_1)\setminus\{x\}$ remain relevant in $N_2$
	    and satisfy $L_{N_1}(v) = L_{N_2}(v)$ and $\phi_{N_1}(v) = \phi_{N_2}(v)$.
	    Note that $|L_{N_2}(x')|\geq 2$ and thus, $\binom{|L_{N_{2}}(x')|}{2}\geq 1$.	
	    Taken the latter arguments together, we obtain
	    \[\Phi^{**}(N_1) - \Phi^{**}(N_2) =  - \phi_{N_2}(x') \binom{|L_{N_{2}}(x')|}{2} < 0 \]
	    and thus, $\Phi^{**}(N_1) < \Phi^{**}(N_2)$; contradicting the fact
	    that $N_1$ has maximum weighted total cophenetic index within the class $\HybdidDegTwoN$. 
			As the latter arguments can be applied to each of the remaining true tree vertices
			in $N_2$, it can be \enquote{binarized} to obtain a further network with larger 
			maximum weighted total cophenetic index and we can conclude that no true tree vertex $x$ in $N_1$
			can have $\outdeg(x)\neq 2$.

In the following, we can therefore assume that $N_1$ does not contain any
    true vertices with out-degree distinct from $2$. 
    We consider now the case
    that $x$ is not a true tree vertex and, therefore, is contained in a
    non-trivial block in $N_1$. 
    In particular, we can choose w.l.o.g. $x$ as
    a $\preceq_{N_1}$-minimal vertex among all vertices with out-degree larger than 2 in $N_1$.  
    Note that, by Lemma~\ref{lem:one_block}, 
    all non-trivial blocks in $N_1$ are crescents. 
    
    There are now three cases that can apply for $x$ (cf.\ Figure~\ref{fig:max_is_binary}):
    \begin{itemize}
        \item[(1)] $\mathcal{B}^x\coloneqq \mathcal{B}^x(N_1)=\emptyset$.\vspace{-0.1in}
        \item[(2)] $\mathcal{B}^x\neq \emptyset$ and $x$ has out-degree at least three in $\partN_1(x)$. \vspace{-0.1in}
        \item[(3)] $\mathcal{B}^x\neq \emptyset$ and $x$ has out-degree precisely two in $\partN_1(x)$.
    \end{itemize}

\begin{figure}[ht!]
    \centering
    \includegraphics[width=0.9\textwidth]{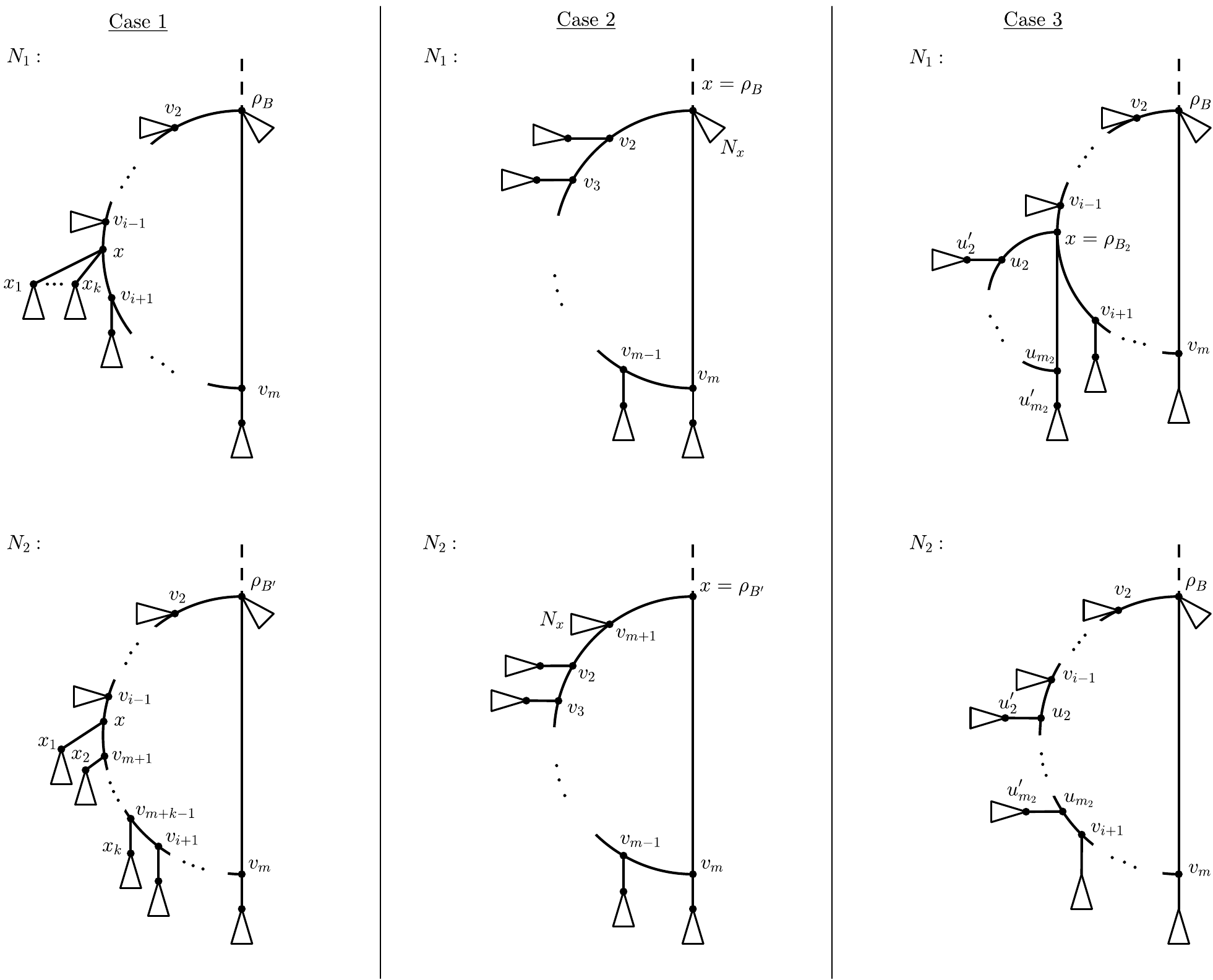}
    \caption{Shown are sketches of the subnetworks $\widetilde N_1$ and  $\widetilde N_2$ of the networks $N_1$ and $N_2$ showing that 
    						 $\Phi^{**}(N_1)<\Phi^{**}(N_2)$ used in the proof of Lemma~\ref{lem:max_is_binary}. Each column corresponds to one of the three cases distinguished in the proof.
  } 
    \label{fig:max_is_binary}
\end{figure}

Before we proceed,  we note that we will make use of the subnetwork $N^B\coloneqq
    \partN(\rho_B,\{B\},\emptyset)$  induced by vertices $\{\rho_B\}\cup\{u\mid
    u\in V(N) \text{ and } u\preceq_{N} v \text{ for all }
    v\in\child^*_B(\rho_B)\}$ where  $N\in \HybdidDegTwoN$ and $B\in \mathcal{B}(N)$. 

\smallskip \noindent
\emph{Case 1:} Suppose that $\mathcal{B}^x=\emptyset$. By Lemma~\ref{lem:results-B}\eqref{L3.19HSS}, $x$ is contained in exactly one non-trivial block $B$.
Since $B$ is a crescent, we can assume that $B$ consists of the two internally
vertex-disjoint paths $(v_1=\rho_B,v_2,\ldots,v_i=x,\ldots,v_m=\eta_B)$ and
$(\rho_B,\eta_B)$. Note that, by assumption, $x$ is not a hybrid vertex of any
block in $N_1$ and since $\mathcal{B}^x=\emptyset$, it holds that $x = v_i \in
\{v_2,\dots,v_{m-1}\}$. Let $x_1,\ldots,x_k$ be the children of $x$ that are not
contained in $B$ and note that $k+1=\outdeg(x)\geq 3$. We now construct a
phylogenetic level-$1$ network $N_2$ by deleting the edges $(x,v_{i+1})$ and
$(x,x_j)$ for $2\leq j\leq k$ and, afterwards, adding vertices
$v_{m+1},\ldots,v_{m+k-1}$ and edges $(x,v_{m+1})$, $(v_{m+k-1},v_{i+1})$, as
well as $(v_{m+j},x_{j+1})$ for $1\leq j\leq k-1$. We denote with $B'$ the
non-trivial block in $N_2$ induced by $v_1,\dots, v_m, v_{m+1}, \dots,
v_{m+k-1}$. One easily verifies that $B'$ is a crescent of size $m+k-1$. By
construction, $x$ and all newly created vertices $v_{m+1},\ldots,v_{m+k-1}$ have
in-degree 1 and out-degree 2 in $N_2$, while $\outdeg_{N_1}(v)=\outdeg_{N_2}(v)$
and $\indeg_{N_1}(v)=\indeg_{N_2}(v)$ for all $v\in
V(N_1)\setminus\{x\}=V(N_2)\setminus\{x,v_{m+1},\ldots,v_{m+k-1}\}$. The latter
implies, that $B'$ only contains one hybrid vertex and that there are no
vertices $v\in V(N_2)$ with $\indeg(v)\leq 1$ and $\outdeg(v)=1$ in $N_2$.
Moreover, since by assumption, $x$ is not a hybrid vertex, the latter
arguments also imply that all vertices in $N_2$ have at most in-degree 2. Hence,
$N_2 \in \HybdidDegTwoN$.

As we will show in the following, $\Phi^{**}(N_1)< \Phi^{**}(N_2)$ holds. To verify the latter, we  employ Theorem~\ref{thm:local-phi} 
and Corollary~\ref{fact:sum-up-part}. To this end, we must first set up the relevant neighbors of $\rho_B$ in $\partN_1(\rho_B)$
and of $\rho_{B'}$ in $\partN_2(\rho_{B'})$. Note that $B\subseteq \partN_1(\rho_B)$ and $B'\subseteq \partN_2(\rho_{B'})$. 
Hence, the respective relevant neighbors of $\rho_B$ in $\partN_1(\rho_B)$ and of $\rho_{B'}$ in $\partN_2(\rho_{B'})$
must satisfy one of the Conditions (a), (b), (c) in Definition~\ref{def:rel-neighbor}(2).  
Let $\mathfrak{R}_B$, resp.,  $\mathfrak{R}_{B'}$ be the relevant neighbors of $\rho_B$ in $\partN_1(\rho_{B})$ resp., $\rho_{B'}$ in $\partN_2(\rho_{B'})$.
We continue to partition $\mathfrak{R}_B$, resp.,  $\mathfrak{R}_{B'}$ into certain subsets that allow us eventually
to employ Theorem~\ref{thm:local-phi} and Corollary~\ref{fact:sum-up-part}.

We first define the subset $V_R$, resp., $V'_R$ of $\mathfrak{R}_B$, resp, $\mathfrak{R}_{B'}$ 
that consists of all those relevant neighbors that are not descendants of vertices in $V(B)\setminus {\rho_B}$, 
resp., $V(B')\setminus {\rho_{B'}}$. To be more precise, we put
           \[ V_R  \coloneqq   \{  v\in \mathfrak{R}_B  \mid v\in V(\partN_1(\rho_B) - N_1^B)\} \text{ and }
            V'_R  \coloneqq  \{ v\in \mathfrak{R}_{B'} \mid v\in V(\partN_2(\rho_{B'}) - N_2^{B'})\}.\]
The modifications necessary to obtain $N_2$ from $N_1$ were only applied to $N_1^B$. 
This and $\rho_B = \rho_{B'}\neq x$ implies that 
all vertices in $N_1 - N_1^B = N_2 - N_2^{B'}$
remain unaffected. Together with Lemma \ref{lem:relevant-partN} it is a straight-forward task
to verify that $V_R = V'_R$.
Note that all relevant  neighbors of $\rho_B$ in $\mathfrak{R}_B\setminus V_R$
must be descendants of vertices in $V(B)\setminus\{\rho_B\}$
and are, therefore, contained in $\partN_1(v_j)$ for some $j\in \{2,\dots,m\}$. 
Hence, none of the relevant  neighbors of $\rho_B$ in $\mathfrak{R}_B\setminus V_R$
can satisfy Definition~\ref{def:rel-neighbor}(2c). 
We collect now all relevant  neighbors of $\rho_B$ in $\mathfrak{R}_B\setminus V_R$
that satisfy Definition~\ref{def:rel-neighbor}(2a) and (2b) and are distinct from $x=v_i$
in  the set 
\[  V_B  \coloneqq \{ v\in \mathfrak{R}_B \mid v\in V(\partN_1(v_j)), 2\leq j\leq m, j\neq i\}.\]
Since $x$ is relevant in $N_1$ and contained in $B$ it is a relevant neighbor of $\rho_B$. 
By construction, $V_B$ and $V_R$ are disjoint and do not contain $x$. 
Hence, 
\[\mathfrak{R}_B = V_B\cup V_R \cup \{x\}.\]

Note that each of  $V_B= \emptyset$ or $V_R=\emptyset$ may be possible.
By putting $\widetilde	N_1\coloneqq \partN_1(\rho_B)$ and by application of 
Corollary~\ref{fact:sum-up-part}, we can write $\Phi^{*}(\partN_1(\rho_B))$
as
\begin{equation}
\Phi^{*}(\widetilde	N_1) =  
\sum\limits_{v\in V_R\cup V_B\cup\{x\}}\left(\Phi^*(\partial \widetilde N_1(v))+\phi_{\widetilde	N_1}(v)\binom{|L_{\widetilde	N_1}(v)|}{2}\right).
\label{eq:Phi*1}
\end{equation}

We continue with further partitioning $\mathfrak{R}_{B'}$. Again, observe that 
all relevant  neighbors of $\rho_{B'}$ in $\mathfrak{R}_{B'}\setminus V'_R$
must be descendants of vertices in $V({B'})\setminus\{\rho_{B'}\}$
and are  therefore, contained in $\partN_2(v_j)$ for some $j\in \{2,\dots,m+k-1\}$.
Similar to the previous part, 
none of those relevant  neighbors of $\rho_{B'}$ in $\mathfrak{R}_{B'}\setminus V'_R$
can satisfy Definition~\ref{def:rel-neighbor}(2c) and we collect 
 all those relevant  neighbors of $\rho_{B'}$
that satisfy Definition~\ref{def:rel-neighbor}(2a) and (2b) and are distinct from $v_i$ and $v_{m+1}, \dots, v_{m+k-1}$ 
in  the set 
\[  V'_B  \coloneqq \{ v\in \mathfrak{R}_{B'} \mid v\in V(\partN_2(v_j)), 2\leq j\leq m, j\neq i\}.\]
By construction, all vertices in $\{v_2,\dots,v_m\}\setminus \{x\}$ remain irrelevant in $N_2$. 
Moreover, we have $\partN_1(v_j)  = \partN_2(v_j)$ for $j\in \{2,\dots,m\}\setminus \{i\}$. 
Hence, IR-free paths from vertices $v_j$ to possible relevant vertices in $\partN_1(v_j)$
remain in $N_2$ for all $j\in \{2,\dots,m\}\setminus \{i\}$. One now easily verifies
that $V_B = V'_B$.
All other relevant vertices in $\mathfrak{R}_{B'}\setminus (V'_R\cup V'_B)$ are 
contained in $\partN_2(v_j)$ for some $j\in \{m+1,\dots, m+k-1\}$. 
By construction, $v_j$ with $j\in \{m+1,\dots, m+k-1\}$ is irrelevant in $\partN_2(\rho_{B'})$. 
Hence, relevant vertices in $\mathfrak{R}_{B'}\setminus V'_R\cup V'_B$ are 
contained in  $\partN_2(x_l)$ for some $l\in \{1,\dots,k\}$
and satisfy, in particular,  Definition~\ref{def:rel-neighbor}(2b). 
These relevant neighbors are collected in the set
\[V'_X\coloneqq \{v\in \mathfrak{R}_{B'}\mid v\in V(\partN_2(x_j)), 1\leq j\leq k\}.\]
By construction, $V'_B, V'_R$ and $V'_X$ are pairwise disjoint and we have
\[\mathfrak{R}_{B'} = V'_B\cup V'_R \cup V'_X.\]
Note that each of  $V'_B,V'_R, V'_X$ could be the empty set. 
By putting $\widetilde	N_2\coloneqq \partN_2(\rho_{B'})$ and by application of 
Corollary~\ref{fact:sum-up-part}, we can write $\Phi^{*}(\partN_2(\rho_{B'}))$
as
\begin{equation}
\Phi^{*}(\widetilde	N_2) =  
\sum\limits_{v\in V'_R\cup V'_B \cup V'_X}\left(\Phi^*(\partial \widetilde N_2(v))+\phi_{\widetilde	N_2}(v)\binom{|L_{\widetilde	N_2}(v)|}{2}\right).
\label{eq:Phi*2}
\end{equation}

In the following we let $\ell \coloneqq |L_{N_1}(\rho_B)| = |L_{N_2}(\rho_{B'})|$. 
Note that $N_2$ can be considered as the network obtained from $N_1$ by replacing
$\widetilde N_1$ by $\widetilde N_2$. In addition, recall that 
Lemma \ref{lem:blocks-in-partN} implies that $\phi_{\widetilde	N_i}(v) = \phi_{N_i}(v)$
for all relevant vertices of $N_i$ that are contained in $\widetilde	N_i$, $i\in \{1,2\}$. 
Taken the latter arguments  together with Theorem~\ref{thm:local-phi} 
and by combining Equation\ \eqref{eq:Phi*1} and \eqref{eq:Phi*2}, we obtain 
\begingroup
\allowdisplaybreaks
   \begin{align}
   &\Phi^{**}(N_1)- \Phi^{**}(N_2) \notag\\ 
    =& \Phi^{*}(\widetilde N_1)-\Phi^*(\widetilde N_2)+
							(\phi_{N_1}(\rho_B)-\phi_{N_2}(\rho_{B'}))\binom{\ell}{2}\notag\\
	  =& \sum\limits_{v\in V_R\cup V_B\cup\{x\}}\left(\Phi^*(\partial \widetilde N_1(v))+\phi_{\widetilde N_1}(v)\binom{|L_{\widetilde	N_1}(v)|}{2}\right) \notag \\
		&-				\sum\limits_{v\in V'_R\cup V'_B \cup V'_X}\left(\Phi^*(\partial \widetilde N_2(v))+\phi_{\widetilde	N_2}(v)\binom{|L_{\widetilde	N_2}(v)|}{2}\right) 
		+(\phi_{ N_1}(\rho_B)-\phi_{N_2}(\rho_{B'}))\binom{\ell}{2}. \label{eq:Case1-step1Q}
   \end{align}
\endgroup 

Moreover, since by construction $\partN_1(x_j) = \partN_2(x_j)$, $1\leq j\leq k$,  
all relevant neighbors of $x$ in $\partial \widetilde N_1(x)$
are precisely the vertices in $V'_X$. This and Corollary~\ref{fact:sum-up-part} implies that we can, therefore, write
\begin{equation}
\alpha\coloneqq \Phi^*(\partial \widetilde N_1(x)) = \sum_{v\in V'_X}\left(\Phi^*(\partial \widetilde N_1(v))+\phi_{\widetilde N_1}(v)\binom{|L_{\widetilde N_1}(v)|}{2}\right). 
\label{eq:alpha1}
\end{equation}
Observe that, by construction, $N_1(v)=N_2(v)$ for all $v\in V'_X$. Hence, we readily obtain
\begin{equation}
\alpha = \sum_{v\in V'_X}\left(\Phi^*(\partial \widetilde N_2(v))+\phi_{\widetilde N_2}(v)\binom{|L_{\widetilde N_2}(v)|}{2}\right).
\label{eq:alpha2}
\end{equation}

Note that by Lemma \ref{lem:idempotent}, we have $\partN'(v)=\partial \widetilde N_1(v)$ for $N' = \partial \widetilde N_1(v)$.
Putting  $N' \coloneqq \partial \widetilde N_1(v)$ and application 
of Lemma \ref{lem:blocks-in-partN} and
 Corollary~\ref{fact:sum-up-part}, therefore, implies that we can write
 $\Phi^*(\partial \widetilde N_1(x))$ as 
\begin{equation}
\Phi^*(\partial \widetilde N_1(x)) = \sum_{v\in V'_X}\left( \Phi^*(\partial N'(v)) + \phi_{\partial \widetilde N_1}(v) \binom{L_{\partial \widetilde N_1}(v)}{2} \right)=  
\sum_{v\in V'_X} \left( \Phi^*(\partial \widetilde N_1(v)) + \phi_{\widetilde N_1}(v) \binom{L_{\partial \widetilde N_1}(v)}{2} \right).
\label{eq:Phi*partialN1xQ}
\end{equation}

Combining Equation\ \eqref{eq:Case1-step1Q} and \eqref{eq:Phi*partialN1xQ}
yields
\begingroup
\allowdisplaybreaks
\ \\

$\Phi^{**}(N_1)- \Phi^{**}(N_2) = $\hfill
   \begin{align}
		&=& 	\sum\limits_{v\in V_R\cup V_B}\left(\Phi^*(\partial \widetilde N_1(v))+\phi_{\widetilde N_1}(v)\binom{|L_{\widetilde	N_1}(v)|}{2}\right) &+
				\Phi^*(\partial \widetilde N_1(x)) +
				\phi_{\widetilde	N_1}(x)\binom{|L_{\widetilde	N_1}(x)|}{2} \notag\\
		&&-		\sum\limits_{v\in V'_R\cup V'_B \cup V'_X}\left(\Phi^*(\partial \widetilde N_2(v))+\phi_{\widetilde	N_2}(v)\binom{|L_{\widetilde	N_2}(v)|}{2}\right) &+ 
				(\phi_{ N_1}(\rho_B) -
				\phi_{N_2}(\rho_{B'}))\binom{\ell}{2} \notag\\
	  & \overset{
			   \begin{subarray}{l}
             		Eq.\ \eqref{eq:alpha1}\\
             		Eq.\ \eqref{eq:alpha2}
        	    \end{subarray}
	  				}{=\hfill}& 
				\sum\limits_{v\in V_R\cup V_B}\left(\Phi^*(\partial \widetilde N_1(v)) +\phi_{\widetilde N_1}(v)\binom{|L_{\widetilde	N_1}(v)|}{2}\right) &+
	  			\alpha
	  			+ \phi_{\widetilde	N_1}(x)\binom{|L_{\widetilde	N_1}(x)|}{2} \notag\\
		&&-		\sum\limits_{v\in V'_R\cup V'_B}\left(\Phi^*(\partial \widetilde N_2(v)) +\phi_{\widetilde	N_2}(v)\binom{|L_{\widetilde	N_2}(v)|}{2}\right) &-
				\alpha +
				(\phi_{ N_1}(\rho_B) -\phi_{N_2}(\rho_{B'}))\binom{\ell}{2} \notag\\
	  &=& 		\sum\limits_{v\in V_R\cup V_B}\left(\Phi^*(\partial \widetilde N_1(v))+\phi_{\widetilde N_1}(v)\binom{|L_{\widetilde	N_1}(v)|}{2}\right) &+ 
	  			\phi_{\widetilde	N_1}(x)\binom{|L_{\widetilde	N_1}(x)|}{2} \notag\\
		&&-		\sum\limits_{v\in V'_R\cup V'_B}\left(\Phi^*(\partial \widetilde N_2(v))+\phi_{\widetilde	N_2}(v)\binom{|L_{\widetilde	N_2}(v)|}{2}\right) &+ 
				(\phi_{ N_1}(\rho_B)-\phi_{N_2}(\rho_{B'}))\binom{\ell}{2} \label{eq:Case1-step2Q}
   \end{align}
\endgroup 

By the latter arguments, we have $V_R\cup V_B =  V'_R\cup V'_B$. By construction, for all $v\in V_R\cup V_B$ we have $\partN_1(v)=\partN_2(v)$ and $L_{\widetilde	N_1}(v) = L_{\widetilde	N_2}(v)$.
Therefore, 
\[\sum\limits_{v\in V_R\cup V_B}\left(\Phi^*(\partial \widetilde N_1(v))+\phi_{\widetilde N_1}(v)\binom{|L_{\widetilde	N_1}(v)|}{2}\right)=
\sum\limits_{v\in V'_R\cup V'_B}\left(\Phi^*(\partial \widetilde N_2(v))+\phi_{\widetilde	N_2}(v)\binom{|L_{\widetilde	N_2}(v)|}{2}\right).\]
The latter arguments together with the fact that $\phi_{\widetilde N_1}(x)=1$  imply that Eq.\ \eqref{eq:Case1-step2Q} simplifies to
\begin{equation}	
		\Phi^{**}(N_1)- \Phi^{**}(N_2) = \binom{|L_{\widetilde	N_1}(x)|}{2} 
		 +(\phi_{ N_1}(\rho_B)-\phi_{N_2}(\rho_{B'}))\binom{\ell}{2}. \label{eq:Case1-step3}
\end{equation}

    Recall that $B$ and $B'$ are crescents of size $m$ and $m+k-1$, respectively.
 By Proposition\ \ref{prop:B_maxmin_arbitrary}, $\omega(B) = \binom{m}{3}$ and $\omega(B') = \binom{m+k-1}{3}$.
     Also note that $|L_{\widetilde N_1}(x)|< \ell$ 
    and that $\mathcal{B}^*\coloneqq \mathcal{B}^{\rho_B}(\widetilde N_1) = \mathcal{B}^{\rho_{B'}}(\widetilde N_2)$. 
    Therefore, we can conclude that
\begingroup  
\allowdisplaybreaks 
        \begin{align*}
        \Phi^{**}(N_1)-\Phi^{**}(N_2) &=  \binom{|L_{\widetilde N_1}(x)|}{2}+(\phi_{N_1}(\rho_B)-\phi_{N_2}(\rho_{B'}))\binom{\ell}{2}\\
           & = \binom{|L_{\widetilde N_1}|}{2}+\left(\omega(B)+\sum_{C\in\mathcal{B}^*}\omega(C)+\epsilon-\left(\omega(B')+\sum_{C\in\mathcal{B}^*}\omega(C)+\epsilon\right)\right)\binom{\ell}{2}\\
            & = \binom{|L_{\widetilde N_1}|}{2}+\left(\binom{m}{3}-\binom{m+k-1}{3}\right)\binom{\ell}{2}\\
            & < \binom{\ell}{2}+\left(\binom{m}{3}-\binom{m+k-1}{3}\right)\binom{\ell}{2}\\
            & \overset{2\leq k}{\leq} \underbrace{\left(1+\binom{m}{3}-\binom{m+1}{3}\right)}_{<0\text{ for $m\geq3$}}\binom{\ell}{2}< 0.
 \end{align*}
\endgroup     
Consequently, $\Phi^{**}(N_1)< \Phi^{**}(N_2)$; a contradiction to the assumption that 
    $N_1$ maximizes the weighted total cophenetic index within the class $\HybdidDegTwoN$. 
    Hence, Case (1) cannot occur. 
    
\smallskip \noindent    
\emph{Case 2:} 
Suppose that $\mathcal{B}^x\neq \emptyset$ and and $x$ has out-degree at least three in $\partN_1(x)$. 

    Let $B\in \mathcal{B}^x$. 
    Again, $B$ must be a crescent and we can assume that $B$ consists of two internally vertex-disjoint paths 
	$(x=\rho_B, v_2, \ldots, v_m=\rho_B)$ and $(\rho_B,\eta_B)$.
    Since $x$ is a $\prec_{N_1}$-minimal vertex of out-degree at least three and is the root of $B$, the subnetwork $N_1^B$
    is binary. Since $\outdeg_{\partN_1(x)}(x)\geq 3$, the subnetwork $N_x \coloneqq \partN_1(x) \setminus (N_1^B-x)$ contains at least one edge. We now construct a phylogenetic level-$1$ network $N_2$ by increasing $size(B)$ by one and attaching the subnetwork $N_x$ to the new vertex in $B$. More precisely, we delete the edges $(x,x_i)$ for all children $x_i$ of $x$ with $x_i\neq v_m$ and, afterwards,  add the new vertex $v_{m+1}$ and edges $(x,v_{m+1})$ and $(v_{m+1},x_i)$ for all previously deleted edges. By similar arguments as in Case (1), $N_2\in \HybdidDegTwoN$.
We denote with  $B'$ the extended block obtained from $B$ in $N_2$ with this construction. 
    As in Case (1), 
    we put $\widetilde N_1 = \partN_1(\rho_B)$ and $\widetilde N_2 = \partN_2(\rho_{B'})$
    and let $\mathfrak{R}_B$, resp.,  $\mathfrak{R}_{B'}$ be the relevant neighbors of $\rho_B = x$ in 
    $\widetilde N_1$, resp., of $\rho_{B'} = x$ in  $\widetilde N_2$. Note that $x\notin \mathfrak{R}_{B}\cup \mathfrak{R}_{B'}$.   
    Moreover, we put
    \begingroup
\allowdisplaybreaks
             \begin{align*}
             V_R  \coloneqq &   \{v\in  \mathfrak{R}_B   \mid v\in V(N_x)\setminus\{x\} \} \text{ and }\\
			 V_B  \coloneqq & \{v\in \mathfrak{R}_B \mid v\in V(\partN_1(v_j)), 2\leq j\leq m\} \text{ and }\\
			 V'_B  \coloneqq & \{v\in \mathfrak{R}_{B'} \mid v\in V(\partN_2(v_j)), 2\leq j\leq m\},
           \end{align*}
\endgroup
where we used in $V_R$ the fact that $N_x-x =  \partN_1(x) \setminus N_1^B$. 
The set $V'_R$ is, however, defined slightly differently from the previous case and covers all the relevant 
neighbors of $\rho_{B'}$ that are contained in $\partN_2(v_{m+1})$. Since $N_x\simeq \partN_2(v_{m+1})$, 
we define 
\[   V'_R  \coloneqq  \{v\in  \mathfrak{R}_{B'}   \mid v\in V(N_x) \}. \]
By similar arguments as in the previous case
$V_R$ and $V_B$ (resp. $V'_R$ and $V'_B$) are vertex-disjoint, do not contain $x$ and
\[\mathfrak{R}_B = V_R\cupdot V_B \text{ and } \mathfrak{R}_{B'} = V'_R\cupdot V'_B.\]
Again, observe that each of $V_R, V_B,V'_R$, and $V'_B$ might be the empty set.

  Note that $N_2$ can be considered as the network obtained from $N_1$ by replacing
$\widetilde N_1$ by $\widetilde N_2$. In addition, recall that 
Lemma \ref{lem:blocks-in-partN} implies that $\phi_{\widetilde	N_i}(v) = \phi_{N_i}(v)$
for all relevant vertices of $N_i$ that are contained in $\widetilde	N_i$, $i\in \{1,2\}$.
By construction, for all $v\in V_B$ it holds that 
 $|L_{\widetilde N_1}(v)|=|L_{\widetilde N_2}(v)|$, $\phi_{\widetilde N_1}(v)=\phi_{\widetilde N_2}(v)$, and $\partial \widetilde N_1(v)=\partial \widetilde N_2(v)$.
Therefore, 
\begin{equation}
    \sum\limits_{v\in V_B}\left(\Phi^*(\partial \widetilde N_1(v))+\phi_{\widetilde N_1}(v)\binom{|L_{\widetilde N_1}(v)|}{2}\right) =
\sum\limits_{v\in V_B}\left(\Phi^*(\partial \widetilde N_2(v))+\phi_{\widetilde N_2}(v)\binom{|L_{\widetilde N_2}(v)|}{2}\right).
\label{IneedABeer}
 \end{equation}

By similar arguments as in Case (1), we obtain $V_B= V'_B$. However, $V_R\neq V'_R$ might be possible.
 Taken the latter arguments  together with Theorem~\ref{thm:local-phi} and
 Corollary~\ref{fact:sum-up-part} and by putting  $\ell \coloneqq |L_{N_1}(\rho_B)| = |L_{N_2}(\rho_{B'})|$, we obtain
\begingroup
\allowdisplaybreaks
       \begin{align}
        \Phi^{**}(N_1)-\Phi^{**}(N_2)
        =&\Phi^*(\widetilde N_1)-\Phi^*(\widetilde N_2)+(\phi_{N_1}(\rho_B)-\phi_{N_2}(\rho_{B'}))\binom{\ell}{2}\notag\\
        =&\sum\limits_{v\in V_B\cup V_R}\left(\Phi^*(\partial \widetilde N_1(v))+\phi_{\widetilde N_1}(v)\binom{|L_{\widetilde N_1}(v)|}{2}\right)
         -\sum\limits_{v\in V_B\cup V'_R}\left(\Phi^*(\partial \widetilde N_2(v))+\phi_{\widetilde N_2}(v)\binom{|L_{\widetilde N_2}(v)|}{2}\right) \notag\\
        &+(\phi_{N_1}(\rho_B)-\phi_{N_2}(\rho_{B'}))\binom{\ell}{2} \notag\\
        \overset{Eq.\ \eqref{IneedABeer}}{=}&\underbrace{\sum\limits_{v\in V_R}\left(\Phi^*(\partial \widetilde N_1(v))+\phi_{\widetilde N_1}(v)\binom{|L_{\widetilde N_1}(v)|}{2}\right)}_{\eqqcolon \alpha}
        -\underbrace{\sum\limits_{v\in V'_R}\left(\Phi^*(\partial \widetilde N_2(v))+\phi_{\widetilde N_2}(v)\binom{|L_{\widetilde N_2}(v)|}{2}\right)}_{\eqqcolon \beta}\notag\\
        &+(\phi_{N_1}(\rho_B)-\phi_{N_2}(\rho_{B'}))\binom{\ell}{2}. \label{eq:case2-1}
 \end{align}
\endgroup 

It holds that $\phi_{N_1}(\rho_B)=\omega(B)+\sum_{C\in\mathcal{B}^{\rho_B}(N_1)\setminus\{ B\}} \omega(C)+\epsilon=\binom{m}{3}+\sum_{C\in\mathcal{B}^{\rho_B}(N_1)\setminus\{B\}} \omega(C)+\epsilon$ 
and $\phi_{N_2}(\rho_B')=\omega(B')+\epsilon=\binom{m+1}{3}+\epsilon$.
 
To recall, the sets $V_R$ and $V'_R$ depend on the structure of the subnetwork $N_x$. 
We distinguish between the two Cases (I) $\outdeg_{\widetilde N_1}(x)=3$
and (II) $\outdeg_{\widetilde N_1}(x)>3$ (in which case  $N_x$ contains at least two edges).

For Case (I), we assume that $\outdeg_{\widetilde N_1}(x)=3$. 
Note that $x=\rho_B$ in $N_1$. Moreover,  $x$ has a unique child $y$ in $N_x$
that becomes the unique child of $v_{m+1}$ that is not in $B'$.
Moreover, $N_x(y) = N_2(y)$. It is now an easy task to verify that 
$V_R=V'_R$. 
 By construction, for all $v\in V_R$ it holds that 
 $\ell_v\coloneqq|L_{\widetilde N_1}(v)|=|L_{\widetilde N_2}(v)|$, $\phi_{\widetilde N_1}(v)=\phi_{\widetilde N_2}(v)$, and $\partial \widetilde N_1(v)=\partial \widetilde N_2(v)$
and, therefore, $\alpha=\beta$. One easily verifies that  $\mathcal{B}^{\rho_B}(N_1)\setminus\{ B\}=\emptyset$.
Hence, we have,  $\phi_{N_1}(\rho_B) = \binom{m}{3}+\epsilon$. 
Taking the later arguments together with Equation\ \eqref{eq:case2-1} we obtain 
\begingroup
\allowdisplaybreaks
   \begin{align*}
           \Phi^{**}(N_1)-\Phi^{**}(N_2)
           &= (\phi_{N_1}(\rho_B)-\phi_{N_2}(\rho_{B'}))\binom{\ell}{2} \\
           &= \left(\binom{m}{3}+\epsilon - \binom{m+1}{3}-\epsilon\right)\binom{\ell}{2}\\
           &= \underbrace{\left(\binom{m}{3}-\binom{m+1}{3}\right)}_{<0\text{ for $m\geq 2$}}\binom{\ell}{2}< 0.
 \end{align*}
\endgroup 
Since $m\geq 3$, we have $\Phi^{**}(N_1)< \Phi^{**}(N_2)$; a contradiction to the assumption that 
    $N_1$ maximizes the weighted total cophenetic index within the class $\HybdidDegTwoN$. 
    Hence, Case (I) cannot occur.

Consider now the Case (II). 
	Since $\outdeg_{\widetilde N_1}(\rho_B)> 3$, it holds that $x=\rho_B$ has at least two children in $N_x$
that are not located in 
	$B$ and these children become the children of $v_{m+1}$ in $N_2$ and are not located in $B'$. 
	One easily verifies that, therefore, $v_{m+1}$ is
	relevant in $N_2$ and, by Lemma~\ref{lem:replacement-partN-1}, $v_{m+1}$ is relevant in $\widetilde N_2$. 
	Consequently, $V'_R = \{v_{m+1}\}$.
	It is an easy task to verify that, by construction, the relevant neighbors
	of $v_{m+1}$ in $N_2$ are precisely the vertices in $V_R$, that is, the relevant
	neighbors of $x$ in $N_x$. 
	By similar arguments as used to derive Equation\ \eqref{eq:Phi*partialN1xQ} and since $V'_R = \{v_{m+1}\}$
 we can write
	 $\Phi^*(\partial \widetilde N_2(v_{m+1}))$ as 
	\begingroup
	\allowdisplaybreaks
        \begin{align*}
		\beta - \phi_{\widetilde N_2}(v_{m+1})\binom{|L_{\widetilde N_2}(v_{m+1})|}{2} = \Phi^*(\partial \widetilde N_2(v_{m+1})) 
		&= \sum_{v\in V_R} \left( \Phi^*(\partial \widetilde N_2(v)) + \phi_{\widetilde N_2}(v) \binom{|L_{\partial \widetilde N_2}(v)|}{2} \right)   \notag\\
	 \end{align*}
	\endgroup 
Moreover, we have by construction and by Lemma \ref{lem:idempotent}, $\partial \widetilde N_1(v)=\partN_1(v) = 
   \partN_2(v) = \partial \widetilde N_2(v)$ for all $v\in V_R$. Hence, one easily verifies that 
   	\begingroup
	\allowdisplaybreaks
        \begin{align}
		\beta - \phi_{\widetilde N_2}(v_{m+1})\binom{|L_{\widetilde N_2}(v_{m+1})|}{2} = \alpha.    \label{eq:Phi*partialN2vm+1}
    \end{align}
	\endgroup 

Combining Equation\ \ref{eq:case2-1} and \ref{eq:Phi*partialN2vm+1} together with $V'_R = \{v_{m+1}\}$ yields
	\begingroup
	\allowdisplaybreaks
        \begin{align}
		   \Phi^{**}(N_1)-\Phi^{**}(N_2) =& \alpha - \beta   + (\phi_{N_1}(\rho_B)-\phi_{N_2}(\rho_{B'}))\binom{\ell}{2}\notag\\
              	=& -\phi_{\widetilde N_2}(v_{m+1})\binom{|L_{\widetilde N_2}(v_{m+1})|}{2}
        +(\phi_{N_1}(\rho_B)-\phi_{N_2}(\rho_{B'}))\binom{\ell}{2}. \label{eq:vm+1}
		\end{align}
	\endgroup
	
	Since $B\in \mathcal{B}^{\rho_B}(N_1)$, the following two cases may occur: $\mathcal{B}^{\rho_B}(N_1)=\{B\}$
	and $\{B\} \subsetneq \mathcal{B}^{\rho_B}(N_1)$.
	If  $\mathcal{B}^{\rho_B}(N_1)=\{B\}$, then 
	$\phi_{N_1}(\rho_B)=\binom{m}{3}+\epsilon$ and $\phi_{N_2}(\rho_{B'})=\binom{m+1}{3}+\epsilon$ and we can conclude
	that
\begingroup
\allowdisplaybreaks
             \begin{align*}
             	\Phi^{**}(N_1)-\Phi^{**}(N_2) = \underbrace{-\phi_{\widetilde N_2}(v_{m+1})\binom{|L_{\widetilde N_2}(v_{m+1})|}{2}}_{<0} 
		        +(\phi_{N_1}(\rho_B)-\phi_{N_2}(\rho_{B'}))\binom{\ell}{2} 
				&<\underbrace{\left(\binom{m}{3}-\binom{m+1}{3}\right)}_{<0\text{ for $m\geq 2$}}\binom{\ell}{2}< 0.
 \end{align*}
\endgroup 	
Again, since $m\geq 3$, we have $\Phi^{**}(N_1)< \Phi^{**}(N_2)$; a contradiction.

Assume now that $\{B\} \subsetneq \mathcal{B}^{\rho_B}(N_1)$ and thus,
$N_x$ contains non-trivial blocks that are rooted in $\rho_B$.
By construction, $\mathcal{B}^{\rho_B}(N_1)\setminus\{ B\} = \mathcal{B}^{v_{m+1}}(N_2)$. 
Therefore, 
$\phi_{\widetilde N_2}(v_{m+1}) = \sum_{C\in\mathcal{B}^{\rho_B}(N_1)\setminus\{ B\}} \omega(C)+\epsilon$. 
Moreover, we have,  by construction $\ell'\coloneqq |L_{\widetilde N_2}(v_{m+1})| = |L_{N_1}(\rho_B)|$. 
Taken the latter arguments together, we obtain

\begingroup
\allowdisplaybreaks
        \begin{align*}
        \Phi^{**}(N_1)-\Phi^{**}(N_2) =	&-\phi_{\widetilde N_2}(v_{m+1})\binom{\ell'}{2}+(\phi_{N_1}(\rho_B)-\phi_{N_2}(\rho_{B'}))\binom{\ell'}{2} \\
             = & -\left(\sum_{C\in\mathcal{B}^{\rho_B}(N_1)\setminus\{ B\}} \omega(C)+\epsilon\right)\binom{\ell'}{2}
            +\left(\omega(B)+\sum_{C\in\mathcal{B}^{\rho_B}(N_1)\setminus\{ B\}}\omega(C)+\epsilon-(\omega(B')+\epsilon)\right)\binom{\ell'}{2} \\
            = & \left(\omega(B)-\omega(B')-\epsilon\right)\binom{\ell'}{2}\\
            = &\underbrace{\left(\binom{m}{3}-\binom{m+1}{3}-\epsilon\right)}_{<0\text{ for $m\geq 2$ and $\epsilon\geq0$}}\binom{\ell'}{2}< 0.
 \end{align*}
\endgroup 
Again, since $m\geq 3$, we have $\Phi^{**}(N_1)< \Phi^{**}(N_2)$; a contradiction.
Hence, neither of the cases (I) and (II) can occur, which implies that Case (2) cannot occur.

\smallskip \noindent    
\emph{Case 3:} Suppose that $\mathcal{B}^x\neq \emptyset$ and $x$ has out-degree two in $\partN_1(x)$. 
Since $x$  has out-degree at least three in $N_1$,
it follows that $x$ has children in $N_1$ that are not contained in $\partN_1(x)$. 
The latter case can only occur if $x$ is contained in some non-trivial block $B$ that is not in $\mathcal{B}^x$. 
Hence, $x$ is not the root of $B$ and, by assumption, $x$ is not a hybrid vertex. 
Moreover, since $x$ has out-degree two in $\partN_1(x)$ and $\mathcal{B}^x\neq \emptyset$, it follows that 
the two children of $x$ that are in $\partN_1(x)$ must be contained in unique block $B_2\in  \mathcal{B}^x$. 
Note $\mathcal{B}^x = \{B_2\}$.
 One easily verifies that, 
therefore, $\partN_1(x) = N_1^{B_2}(\rho_{B_2})$.  
Note that $\eta_B$ has in-degree two in $N_1$ and $B$ contains  precisely one hybrid. 
This together with $x\neq \rho_B$  implies that 
$x$ can only have one child in $B$. 
Taken the latter arguments together, we obtain $\outdeg_{N_1}(x)=3$. 
Note that both $B$ and $B_2$ are crescents of  $size(B)=m$ and $size(B_2)=m_2$. 
Hence,  we can assume that $B$ consists of two internally vertex-disjoint paths 
$(v_1=\rho_B,v_2,\ldots,v_i=x,\ldots,v_m=\eta_B)$ and $(\rho_B,\eta_B)$ and that $B_2$
consists of two internally vertex-disjoint paths 
$(u_1,u_2,\ldots,u_{m_2}=\eta_{B_2})$ and $(u_1,\eta_{B_2})$. Note that $x=\rho_{B_2}=u_1$
and that $\phi_{N_1}(x)=\omega(B_2)+\epsilon$. Since $N_1$ is a level-$1$ network and by Observation\ \ref{obs:identical-block}, 
$B$ and $B_2$ have only vertex $x$ in common. 
Note that each descendant of $x$ that is not a leaf or a hybrid must have out-degree two. 
In particular all non-hybrid vertices in $B_2$ that are distinct from $u_1=x$ 
must have precisely one child that is not contained in $B_2$. 
Moreover, by assumption $u_{m_2}=\eta_{B_2}$ has precisely one child (which is, in particular
not located in $B_2$). 
Let us denote with $u'_i$ the unique child of $u_i$ that is not located in $B_2$ where $2\leq i\leq m_2$.
We construct a phylogenetic level-$1$ network $N_2$ by deleting $x$ and its incident edges from $N_1$
and adding the edges $(v_{i-1},u_2)$, $(u_{m_2},v_{i+1})$.  In this way the block $B_2$ was \enquote{removed}
 and the path $(u_2,\ldots,u_{m_2})$ was \enquote{inserted} into $B$, resulting in a crescent $B'$
 of size $m+m_2-1$. One easily observes that $N_2 \in \HybdidDegTwoN$.
 We put  $\widetilde N_1\coloneqq\partN_1(\rho_{B})$ and $\widetilde N_2 \coloneqq\partN_2(\rho_{B'})$. 
Note that $N_2$ can be considered as the network obtained from $N_1$ by replacing
$\widetilde N_1$ by $\widetilde N_2$. This together with Theorem~\ref{thm:local-phi} and the  fact 
that $v_1=\rho_{B} = \rho_{B'}$ and $\ell\coloneqq |L_{N_1}(\rho_{B})| = |L_{N_2}(\rho_{B'})|$
 implies 
\begin{equation*}
\Phi^{**}(N_1)-\Phi^{**}(N_2)
          = \Phi^*(\widetilde N_1)-\Phi^*(\widetilde N_2)+(\phi_{N_1}(v_1)-\phi_{N_2}(v_1))\binom{\ell}{2}.\\
\label{eq:Case3-1}
\end{equation*} 

Let $\mathfrak{R}_B$, resp.,  $\mathfrak{R}_{B'}$ be the relevant neighbors of $\rho_B$ in $\widetilde N_1$ resp., of $\rho_{B'}$ in $\widetilde N_2$.
Again, as in the previous cases, we will partition $\mathfrak{R}_B$ and  $\mathfrak{R}_{B'}$. 
We first define the subset $V_R$, resp., $V'_R$ of $\mathfrak{R}_B$, resp, $\mathfrak{R}_{B'}$ 
that consists all those relevant neighbors that are not descendants of vertices in $V(B)\setminus {\rho_B}$, 
resp., $V(B')\setminus {\rho_{B'}}$. To be more precise, we put
           \[ V_R  \coloneqq   \{  v\in \mathfrak{R}_B  \mid v\in V(\widetilde N_1 - N_1^B)\} \text{ and }
            V'_R  \coloneqq  \{ v\in \mathfrak{R}_{B'} \mid v\in V(\widetilde N_2 - N_2^{B'})\}.\]
By similar arguments as in Case (1), we have $V_R = V'_R$.
Moreover, we put 
\begingroup
\allowdisplaybreaks
        \begin{align*}
        V_B  \coloneqq& \{ v\in \mathfrak{R}_B \mid v\in V(\partN_1(v_j)), 2\leq j\leq m, j\neq i\}\\
        V'_B  \coloneqq& \{ v\in \mathfrak{R}_{B'} \mid v\in V(\partN_2(v_j)), 2\leq j\leq m, j\neq i\}.
 \end{align*}
\endgroup 
By similar arguments as in Case 1, we observe that $V_B =V'_B$ and that
\[\mathfrak{R}_B = V_B\cup V_R \cup \{x\}.\]

Note that $V_B= \emptyset$ or $V_R=\emptyset$ may be possible.
By application of 
Corollary~\ref{fact:sum-up-part}, we can write $\Phi^{*}(\widetilde N_1(\rho_B))$
as
\begin{equation*}
\Phi^{*}(\widetilde	N_1) =  
\sum\limits_{v\in V_R\cup V_B\cup\{x\}}\left(\Phi^*(\partial \widetilde N_1(v))+\phi_{\widetilde	N_1}(v)\binom{|L_{\widetilde	N_1}(v)|}{2}\right).
\label{eq:Phi*1X}
\end{equation*}

All other relevant vertices in $\mathfrak{R}_{B'}\setminus (V'_R\cup V'_B)$ are 
contained in $\partN_2(u_j)$ for some $j\in \{2,\dots, m_2-1\}$. 
By construction, $u_j$ with $j\in \{2,\dots, m_2\}$ is irrelevant in $\widetilde N_2$. 
Hence, relevant vertices in $\mathfrak{R}_{B'}\setminus V'_R\cup V'_B$ are 
contained in  $\widetilde N_2(u_j')$ for some $j\in \{2,\dots,m_2\}$. 
\[V'_X\coloneqq \{v\in \mathfrak{R}_{B'}\mid v\in V(\widetilde N_2(u_j')), 2\leq j\leq m_2\}.\]
By construction, $V'_B, V'_R$, and $V'_X$ are pairwise disjoint and we have
\[\mathfrak{R}_{B'} = V'_B\cup V'_R \cup V'_X.\]
Note that each of  $V'_B,V'_R, V'_X$ could be the empty set. 
By application of 
Corollary~\ref{fact:sum-up-part}, we can write $\Phi^{*}(\widetilde N_2)$ as
\begin{equation*}
\Phi^{*}(\widetilde	N_2) =  
\sum\limits_{v\in V'_R\cup V'_B \cup V'_X}\left(\Phi^*(\partial \widetilde N_2(v))+\phi_{\widetilde	N_2}(v)\binom{|L_{\widetilde	N_2}(v)|}{2}\right).
\label{eq:Phi*2X}
\end{equation*}

By similar arguments as in Case (1), we obtain 
\begin{equation}
\Phi^{**}(N_1)-\Phi^{**}(N_2)
          = \phi_{\widetilde N_1}(\rho_{B_2})\binom{|L_{\widetilde N_1}(\rho_{B_2)}|}{2}+(\phi_{N_1}(\rho_B)-\phi_{N_2}(\rho_{B'}))\binom{\ell}{2}.
\label{eq:Case3-x}
\end{equation}

Note that $v_1=\rho_{B} = \rho_{B'}$. Moreover, by construction, 
$\mathcal{B}^*\coloneqq \mathcal{B}^{v_1}(N_1)\setminus\{ B\} = \mathcal{B}^{v_1}(N_2)\setminus\{ B'\}$.
Put $\gamma\coloneqq \sum_{C\in\mathcal{B}^*}\omega(C)$. 
By construction, Lemma \ref{prop:B_maxmin_arbitrary}, and the previous arguments we obtain
\begingroup
\allowdisplaybreaks
 \begin{align*}
 &\phi_{\widetilde N_1}(\rho_{B_2})=\omega(B_2)+\epsilon=\binom{m_2}{3}+\epsilon, \\
 &\phi_{N_1}(v_1)=\omega(B)+ \sum_{C\in\mathcal{B}^*}\omega(C) + \epsilon=\binom{m}{3}+\gamma +\epsilon \text{ and } \\
 &\phi_{N_2}(v_1)=\omega(B')+\sum_{C\in\mathcal{B}^*}\omega(C) +\epsilon=\binom{m+m_2-2}{3}+\gamma+\epsilon.
 \end{align*}
\endgroup

This together with Equation~\eqref{eq:Case3-x} and $|L_{\widetilde N_1}(\rho_{B_2})|<|L_{N_1}(\rho_{B})|=\ell$ 
and $\Phi^{**}(N_1)-\Phi^{**}(N_2) \geq 0$ implies 
\begingroup
\allowdisplaybreaks
      \begin{align}
      \Phi^{**}(N_1)-\Phi^{**}(N_2)
          = &\ \phi_{\widetilde N_1}(\rho_{B_2})\binom{|L_{\widetilde N_1}(\rho_{B_2)}|}{2}+(\phi_{N_1}(v_1)-\phi_{N_2}(v_1))\binom{\ell}{2} &\geq &\ 0 \notag\\
     \newline 
         \implies &  \phi_{\widetilde N_1}(\rho_{B_2})\binom{|L_{N}(\rho_{B_2)}|}{2}+\phi_{N_1}(v_1)\binom{\ell}{2}&\geq &\ \phi_{N_2}(v_1)\binom{\ell}{2}\notag\\
        \implies & \phi_{\widetilde N_1}(\rho_{B_2})\binom{\ell}{2}+\phi_{N_1}(v_1)\binom{\ell}{2}& > &\ \phi_{N_2}(v_1)\binom{\ell}{2}\notag\\
         \implies & \phi_{\widetilde N_1}(\rho_{B_2})+\phi_{N_1}(v_1) &> & \ \phi_{N_2}(v_1)\notag\\
        \implies & \binom{m_2}{3}+\epsilon+\binom{m}{3}+\gamma +\epsilon &> &\ \binom{m+m_2-2}{3}+\gamma +\epsilon\notag\\
           \implies & \binom{m_2}{3}+\binom{m}{3}+\epsilon &>& \ \binom{m+m_2-2}{3}\label{eq:fromHereOn}\\
           \implies &\sum_{j=2}^{m_2-1}\binom{j}{2}+\sum_{j=2}^{m-1}\binom{j}{2}+\epsilon&>&\sum_{j=2}^{m+m_2-3}\binom{j}{2}.\notag
 \end{align}
\endgroup 
The latter equation allows us to assume that  $m\leq m_2$ and we obtain
      \begin{align*}
      &\Phi^{**}(N_1)-\Phi^{**}(N_2)
            &\geq &\ 0\\
           \implies &\sum_{j=2}^{m_2-1}\binom{j}{2}+\sum_{j=2}^{m-1}\binom{j}{2}+\epsilon&>&\sum_{j=2}^{m+m_2-3}\binom{j}{2}\\
           \iff &\sum_{j=2}^{m_2-1}\binom{j}{2}+\epsilon&>&\sum_{j=m}^{m+m_2-3}\binom{j}{2}.
 \end{align*}
 Note that $\sum_{j=2}^{m_2-1}\binom{j}{2}+\epsilon $ and $\sum_{j=m}^{m+m_2-3}\binom{j}{2}$ have the same number of 
 of summands, namely $m_2-2$. Thus, we obtain

\begingroup
\allowdisplaybreaks
           \begin{align*}
            &\binom{2}{2}+...+\binom{m_2-1}{2}+\epsilon&>&\binom{m}{2}+...+\binom{m+m_2-3}{2}\\
           \iff &\hspace{3.5cm}\epsilon &>& \binom{m}{2}-\binom{2}{2}+\binom{m+1}{2}-\binom{3}{2}+...+\binom{m+m_2-3}{2}-\binom{m_2-1}{2}\\
           \iff &\hspace{3.5cm}\epsilon &>&\sum_{j=0}^{m_2-3}\left(\binom{m+j}{2}-\binom{2+j}{2}\right)\\
           \overset{m\geq 3}{\implies} &\hspace{3.5cm}\epsilon &>&\sum_{j=0}^{m_2-3}\left(\binom{3+j}{2}-\binom{2+j}{2}\right)\\
              \overset{m_2\geq3}{\implies}   &\hspace{3.5cm}\epsilon&>&\sum_{j=0}^{m_2-3}\left(j+2\right)\geq 2\\
 \end{align*}
\endgroup 
 Hence, $\epsilon >2$; a contradiction to the choice of $\epsilon$ which satisfies by definition $0 \leq  \epsilon \leq 2$. 
 
In summary, we have shown that $N_1$ cannot have maximum weighted total cophenetic index if it is not binary. This completes the proof.
\end{proof}

\begin{lemma}\label{lem:bin_no_true_vert}
A binary phylogenetic level-$1$ network on $n$ leaves that achieves a maximum weighted total cophenetic index within the class $\BinLevelOneN$ cannot  contain true tree vertices.
\end{lemma}
\begin{proof}
	By Observation\ \ref{obs:bounded-indegree2}, networks in $\BinLevelOneN$ have bounded maximum weighted total cophenetic index. 
		Hence, we can choose  $N_1\in \BinLevelOneN$ as one of 
		the networks that maximizes the weighted total cophenetic index within the class $\BinLevelOneN$.

Suppose first that $\epsilon>0$. Assume, for contradiction, that $N_1$ contains a true tree vertex $x$. 
Since $N_1$ is phylogenetic, $x$ must have two children $u_1$ and $u_2$. We construct now a network 
$N_2$ by expanding $x$ to a triangle along $u_1,u_2$. One easily verifies that $N_2$ remains in the class
$\BinLevelOneN$. By construction, we obtain a triangle $B$
induced by the vertices $x,v_1,v_2$ and where $v_1, v_2$ are the two unique children of $x$ and where
$v_1$ has as unique child $u_1$, while $v_2$ has as unique children $v_1$ and $u_2$. 
Consequently, neither $v_1$ nor $v_2$ is relevant in $N_2$. 
 By Observation\ \ref{fact:binary_rel_vert}, $x$ is relevant in $N_1$ and $N_2$. 
In addition, one easily verifies that $u_i$ is relevant in $N_1$ if and only if $u_i$ is relevant in $N_2$, $i\in \{1,2\}$.
Moreover, all remaining vertices $v\in W\coloneqq N_1\setminus \{x,u_1,u_2\} = N_2\setminus \{x,u_1,u_2,v_1,v_2\} $ have not been touched and thus
are relevant in $N_1$ precisely if they are relevant in $N_2$. In particular, all relevant vertices 
$v\in W$ satisfy $L_{N_1}(v)=L_{N_2}(v)$ and $\phi_{N_1}(v) = \phi_{N_2}(v)$. 
This implies $\Phi^*(N_1(x)) = \Phi^*(N_2(x))$. Moreover, $N_2$ can be considered as a network 
obtained from $N_1$ by replacing $N_1(x)$ by $N_2(x)$. The latter two arguments together with
Theorem~\ref{thm:local-phi} and the fact that $\ell\coloneqq |L_{N_1}(x)|= |L_{N_2}(x)|$ imply
    \begin{align*}
     \Phi^{**}(N_1)-\Phi^{**}(N_2) &= \left(\Phi^{*}(N_1(x))-\Phi^{*}(N_2(x))\right) + \left((\phi_{N_1}(x)-\phi_{N_2}(x)) \binom{\ell}{2}\right)\\
        & = (\phi_{N_1}(x)-\phi_{N_2}(x)) \binom{\ell}{2}. 
   \end{align*} 														  
Since $x$ is a true tree vertex, we have $\phi_{N_1}(x) = 1$. Moreover, $x$ is the root of precisely one non-trivial block in $N_2$, namely $B$. 
Since $\omega(B)=1$, we obtain $\phi_{N_2}(x) = 1+\epsilon$. This together with $\ell\geq 2$ and, thus, $\binom{\ell}{2} \geq 1$ implies 
$ \Phi^{**}(N_1)-\Phi^{**}(N_2)  = -\epsilon$ and, therefore,  
$ \Phi^{**}(N_1)<\Phi^{**}(N_2)$; a contradiction to the maximality of $N_1$. In summary, if $\epsilon>0$, then 
$N_1$ cannot contain a true tree vertex. 

Assume now that $\epsilon = 0$. Let $N_1$ be a binary level-$1$ network with
$n\geq 3$ leaves that maximizes the weighted total cophenetic index and assume,
for contradiction, that it contains at least one true tree vertex. Let $x$ with
children $x_1$ and $x_2$ be a $\preceq_{N_1}$-maximal true tree vertex in $N_1$.
Note that $x$ can either be the root $\rho_{N_1}$ of $N_1$ or a different inner
vertex. Moreover, note that by Lemma~\ref{lem:one_block}, $N_1$ contains at
least one non-trivial block.
    
We start with considering the case that $x\neq \rho_{N_1}$. In this case,
the unique parent $par_{N_1}(x)$ of $x$ has to be an inner vertex of some
non-trivial block $B$ of size $m\geq 3$ since $x$ is a $\preceq_{N_1}$-maximal
true tree vertex. By Lemma
\ref{lem:one_block}, $B$ has to be a crescent, consisting of the two internally
vertex-disjoint paths $(v_1=\rho_B,\ldots,v_m=\eta_B)$ and $(\rho_B,\eta_B)$.
Since by Lemma \ref{lem:binary_swap}, rearranging subnetworks attached to a
non-trivial block in a binary phylogenetic level-$1$ network does not change its
weighted total cophenetic index, we can without loss of generality assume that
$par_{N_1}(x)=v_2$. 
We now construct a phylogenetic level-$1$ network $N_2$ by
deleting $x$ (and all incident edges) as well as $(v_{m-1},v_m)$. Next, we add
vertex $v_{m+1}$ as well as the following edges: $(v_2,x_1)$, $(v_{m+1},x_2)$,
$(v_{m-1},v_{m+1})$, and $(v_{m+1},v_m)$ (see Figure
\ref{fig:no_true_tree_vertex}a) for an illustration). 
One easily verifies that $N_2\in \BinLevelOneN$.
Let $B'$ be the extended block
obtained from $B$ in $N_2$ in this manner. Furthermore, let 
$\widetilde N_1\coloneqq
\partN_1(\rho_B)=N_1(\rho_B)$ and
$\widetilde N_2\coloneqq\partN_2(\rho_{B'})=N_2(\rho_{B'})$. 

\begin{figure}[ht]
    \centering
    \includegraphics[width=0.95\textwidth]{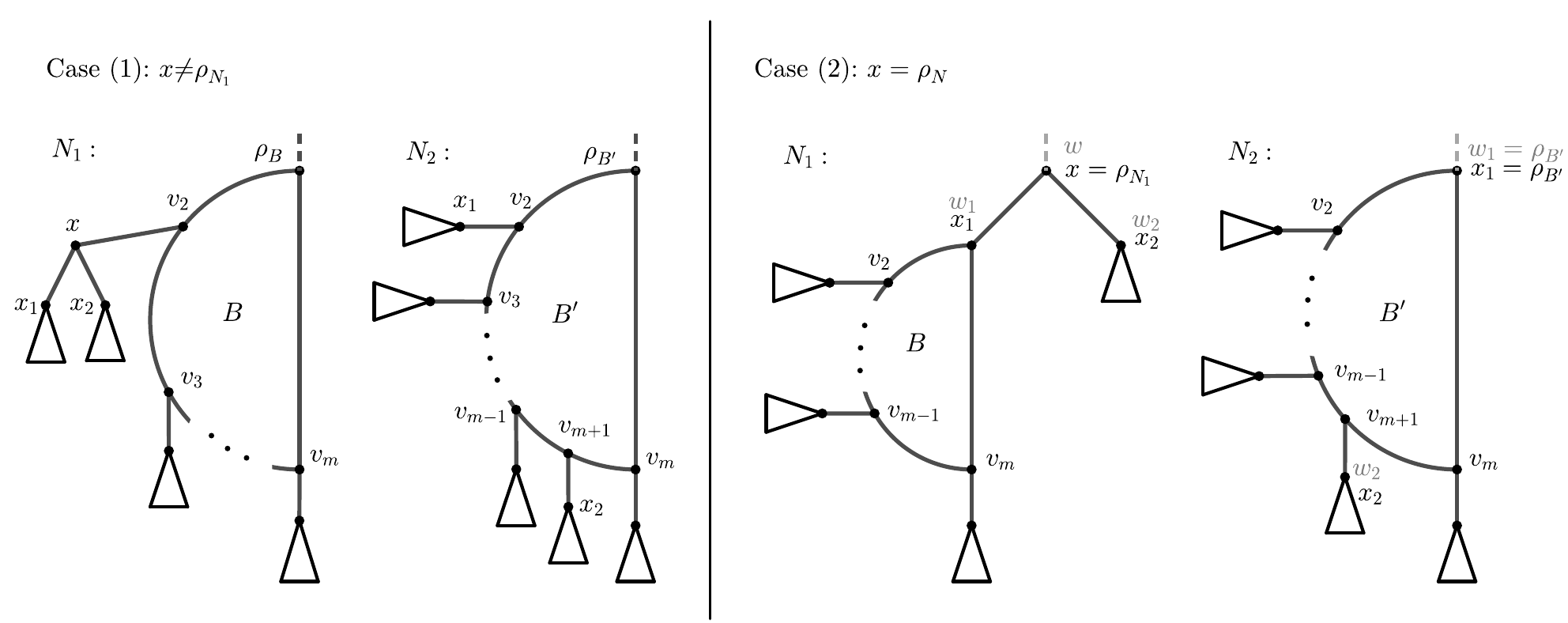}
    \caption{Visualization of the networks $N_1$ and $N_2$ in the proof of Lemma~\ref{lem:bin_no_true_vert} for the choice $\epsilon = 0$. The left panel represents the first case where the $\preceq_{N_1}$-maximal true tree vertex $x$ is different from the root $\rho_{N_1}$. The right panel represents the second case where $x=\rho_{N_1}$. To illustrate the likeness of subcases (a) and (b) of the second case, they are depicted in the same panel. The black labels correspond to (a), while gray labels represent the labels in (b) that differ from (a).}
    \label{fig:no_true_tree_vertex}
\end{figure}

Consider the relevant neighborhoods of $\rho_B$ in $N_1$ and of $\rho_{B'}$
in $N_2$. By Lemma \ref{lem:relevant-partN}, all these relevant
neighbors are contained in $\widetilde N_1$ and $\widetilde N_2$, respectively.
Since $N_1$ and $N_2$ are both binary, all vertices in $B$, resp., $B'$ that are
distinct from $\rho_B$, resp., $\rho_{B'}$ are irrelevant. In the following, let
$\mathfrak{R}_B$, resp., $\mathfrak{R}_{B'}$ be the relevant neighbors of
$\rho_B$ in $\widetilde N_1$ resp., $\rho_{B'}$ in $\widetilde N_2$. Hence, all
relevant neighbors of $\rho_B$, resp., $\rho_{B'}$ (if there are any) must
satisfy Definition~\ref{def:rel-neighbor}(2b). Note, that $x$ is always a
relevant neighbor of $\rho_B$. Moreover, let $V_B$ be the relevant neighbors of 
$\rho_B$ in $N_1$ that are different from $x$ and note that $V_B$ might be empty. By the latter arguments, we
can write \[\mathfrak{R}_B=\{x\} \cupdot V_B.\]  Since $N$ is binary, each $u\in
V_B$ has a unique parent $v\in V(B)\setminus \{\rho_B\}$ and is, therefore,
located in $\partN_1(v)$. Moreover, note that constructing $N_2$ from $N_1$ does
not destroy the IR-paths between $\rho_B$ and any $u\in V_B$. Thus, all $u\in
V_B$ remain relevant neighbors of $\rho_{B'}$. Moreover, while $x\in
\mathfrak{R}_B$, vertex $x$ does not exist in $N_2$. In particular, $v_2$ has
now as the unique child, which is not contained in $B'$, the vertex $x_1$ and we
have a new vertex $v_{m+1}$ that has now as the unique child, which is not
contained in $B'$, the vertex $x_2$. Note that, since $N_2$ is binary and none of
$x_1$ and $x_2$ are hybrids, $x_1$, resp., $x_2$ is a leaf or is relevant. In
particular, $x_i$ is relevant in $N_1$ if and only if $x_i$ is relevant in
$N_2$, $i\in \{1,2\}$. Hence, we can define $V_X$ as the set that contains $x_i$
whenever $x_i$ is relevant in $N_1$ (or equivalently in $N_2$), $i\in \{1,2\}$.
It is straightforward to verify that \[\mathfrak{R}_{B'}= V_X \cupdot V_B.\]
Note that $N_2$ can be considered as the network obtained from $N_1$ by
replacing $\widetilde N_1$ by $\widetilde N_2$. In addition, recall that Lemma
\ref{lem:blocks-in-partN} implies that $\phi_{\widetilde N_i}(v) =
\phi_{N_i}(v)$ for all relevant vertices of $N_i$ that are contained in
$\widetilde N_i$, $i\in \{1,2\}$. One easily observes that, by construction, 
$\partial \widetilde N_1 (v) = \widetilde N_1(v) = \widetilde N_2(v) = \partial \widetilde N_2(v)$ for all $v\in V_B$. 
Hence, we have $\Phi^*(\partial \widetilde N_1 (v) = \Phi^*(\partial \widetilde N_2 (v)$,
$\phi_{\widetilde N_1}(v)= \phi_{\widetilde N_2}(v)$ and $L_{\widetilde
N_1}(v)=L_{\widetilde N_2}(v)$ for all $v\in V_B$. The latter implies
\begin{equation}
            \sum_{v\in V_B} \left(\Phi^*(\partial \widetilde N_1 (v))+\phi_{\widetilde N_1}(v)\binom{|L_{\widetilde N_1}(v)|}{2}\right)
             	=\sum_{v\in V_B} \left(\Phi^*(\partial \widetilde N_2 (v))+\phi_{\widetilde N_2}(v)\binom{|L_{\widetilde N_2}(v)|}{2}\right).
	\label{eq:VBEQUAL}
\end{equation}
The latter arguments together with Theorem~\ref{thm:local-phi},
Corollary~\ref{fact:sum-up-part} and the fact that $\ell \coloneqq |L_{N_1}(\rho_B)| =  |L_{N_2}(\rho_{B'})| $ imply that

 \begingroup
\allowdisplaybreaks
\ \\

          $\Phi^{**}(N_1)-\Phi^{**}(N_2)$ = \hfill
             \begin{align}
           = & (\Phi^*(\widetilde N_1)-\Phi^*(\widetilde N_2))+(\phi_{N_1}(\rho_B)-\phi_{N_2}(\rho_{B'}))\binom{\ell}{2}   \notag\\
           = & \sum_{v\in \{x\} \cupdot V_B} \left(\Phi^*(\partial \widetilde N_1 (v))+\phi_{\widetilde N_1}(v)\binom{|L_{\widetilde N_1}(v)|}{2}\right)
             	-\sum_{v\in V_X   \cupdot V_B} \left(\Phi^*(\partial \widetilde N_2 (v))+\phi_{\widetilde N_2}(v)\binom{|L_{\widetilde N_2}(v)|}{2}\right) \notag\\
             &+	(\phi_{N_1}(\rho_B)-\phi_{N_2}(\rho_{B'}))\binom{\ell}{2} \notag\\
            \overset{Eq.\ \eqref{eq:VBEQUAL}}{=}&  \left(\Phi^*(\partial \widetilde N_1 (x))+\phi_{\widetilde N_1}(x)\binom{|L_{\widetilde N_1}(x)|}{2}\right)
             	-\sum_{v\in V_x} \left(\Phi^*(\partial \widetilde N_2 (v))+\phi_{\widetilde N_2}(v)\binom{|L_{\widetilde N_2}(v)|}{2}\right)\notag\\
            &	+	(\phi_{N_1}(\rho_B)-\phi_{N_2}(\rho_{B'}))\binom{\ell}{2}. \label{eq:RRR}
        \end{align}
\endgroup

Moreover, since by construction $\partN_1(x_j) = \partN_2(x_j)$, $i\in \{1,2\}$,  
all  the relevant neighbors of $x$ in $\partial \widetilde N_1(x)$
are precisely the vertices in $V_X$. This and Corollary~\ref{fact:sum-up-part} implies that we can, therefore, write
\begin{equation}
\alpha\coloneqq \Phi^*(\partial \widetilde N_1(x)) = \sum_{v\in V_X}\left(\Phi^*(\partial \widetilde N_1(v))+\phi_{\widetilde N_1}(v)\binom{|L_{\widetilde N_1}(v)|}{2}\right). 
\label{eq:alpha1XA}
\end{equation}
Furthermore, since by construction $N_1(v)=N_2(v)$ for all $v\in V_X$, we readily obtain
\begin{equation}
\alpha = \sum_{v\in V_X}\left(\Phi^*(\partial \widetilde N_2(v))+\phi_{\widetilde N_2}(v)\binom{|L_{\widetilde N_2}(v)|}{2}\right).
\label{eq:alpha2XA}
\end{equation}

Note that by Lemma \ref{lem:idempotent}, we have $\partN'(v)=\partial \widetilde N_1(v)$ for $N' = \partial \widetilde N_1(v)$.
Putting  $N' \coloneqq \partial \widetilde N_1(v)$ and application 
of Lemma \ref{lem:blocks-in-partN} and
 Corollary~\ref{fact:sum-up-part}, therefore, implies that we can write
 $\Phi^*(\partial \widetilde N_1(x))$ as 
\begin{equation}
\Phi^*(\partial \widetilde N_1(x)) = \sum_{v\in V_X}\left( \Phi^*(\partial N'(v))) + \phi_{\partial \widetilde N_1}(v) \binom{L_{\partial \widetilde N_1}(v)}{2} \right)=  
\sum_{v\in V_X} \left( \Phi^*(\partial \widetilde N_1(v)) + \phi_{\widetilde N_1}(v) \binom{L_{\partial \widetilde N_1}(v)}{2} \right).
\label{eq:Phi*partialN1xXA}
\end{equation}

This together with the fact that $\phi_{\widetilde N_1}(x)=1$ (since $x$ is true tree vertex) implies that 
 Equation\ \eqref{eq:RRR} reduces to
 \begin{equation}
          \Phi^{**}(N_1)-\Phi^{**}(N_2) =
            \binom{|L_{\widetilde N_1}(x)|}{2}	+	(\phi_{N_1}(\rho_B)-\phi_{N_2}(\rho_{B'}))\binom{\ell}{2}. 
            \label{EQ:XXX-1}
	\end{equation}

Both blocks $B$ and $B'$ are crescents of size $m$ and $m+1$ respectively and since $\epsilon=0$ and $\outdeg(\rho_B)=\outdeg(\rho_{B'})=2$, we obtain  
$\phi_{N_1}(\rho_B)=\omega(B)=\binom{m}{3}$ and $\phi_{N_2}(\rho_{B'})=\omega(B')=\binom{m+1}{3}$. 
This and $|L_{\widetilde N_1}(x)|<\ell$ together with Equation\ \eqref{EQ:XXX-1} implies that 
 \begin{align*}
          \Phi^{**}(N_1)-\Phi^{**}(N_2) &=
            \binom{|L_{\widetilde N_1}(x)|}{2}	+	\left(\binom{m}{3}-\binom{m+1}{3}\right)\binom{\ell}{2} \\
            &<\underbrace{\left(1+\binom{m}{3}-\binom{m+1}{3}\right)}_{<0\text{ for $m\geq3$}}\binom{\ell}{2}<0.
	\end{align*}
But then, $\Phi^{**}(N_1)<\Phi^{**}(N_2)$ must hold; a contradiction to the assumption that 
    $N_1$ maximizes the weighted total cophenetic index within the class $\BinLevelOneN$.

 Hence,    $x=\rho_{N_1}$ must hold, that is, $\rho_{N_1}$ is the $\preceq_{N_1}$-maximal true tree vertex. 
		  Note that $x$ has precisely two children $x_1$ and $x_2$ in $N_1$. 
		  There are now two cases: either  (a) at least one of $x_1$ and $x_2$ is the root of a non-trivial block 
		  of $N_1$
		  or (b) none of $x_1$ and $x_2$ are roots of a non-trivial block 
		  of $N_1$ (and thus, they are leaves or true tree vertices). 
		  
		  We start with considering Case (a) and assume w.l.o.g.\ that $x_1$ is the
		  root of a non-trivial block $B$ in $N_1$. By Lemma \ref{lem:one_block},
		  $B$ has to be a crescent consisting of the two internally vertex-disjoint
		  paths $(x_1=\rho_B,\ldots,v_m=\eta_B)$ with $m\geq 3$ and
		  $(\rho_B,\eta_B)$. 
 	   We now construct a network $N_2$, by deleting
 	   $x = \rho_{N_1}$ and its incident edges as well as the edge $(v_{m-1},v_m)$ and adding
 	   vertex $v_{m+1}$ and edges $(v_{m-1},v_{m+1}),(v_{m+1},v_m)$, and
 	   $(v_{m+1},x_2)$ (see Figure \ref{fig:no_true_tree_vertex} for an illustration).
 	   One easily verifies that $N_2\in \BinLevelOneN$.
 	   Let $B'$ be the extended block obtained from $B$ in $N_2$ in this manner. 
 	   Note that $N_1=N_1(x)=\partN_1(x)$ and $N_2=N_2(\rho_{B'})=\partN_2(\rho_{B'})$.

		The relevant neighbors of $x=\rho_{N_1}$ in $N_1$ are $x_1$ and  $x_2$ in case $x_2$ is not a leaf. 
		Thus, we define $V_X$ as the set that contains $x_1$ 
		as well as $x_2$ in case $x_2$ is relevant. 
		We collect all relevant neighbors of $\rho_B=x_1$ in
 	  $N_1$ in the set $V_B$ and note that $V_B$ might be empty. 
		By similar arguments as in the previous case, $v\in V_B$  has a unique parent in $B$
		and all $u\in	V_B$ remain relevant neighbors of	$x_1 = \rho_{N_2}$ in $N_2$. 
		Moreover, $V'_X\coloneqq V_X\setminus \{x_1\}$ contains $x_2$ if and only if $x_2$
		was a relevant neighbor of $x$ (in which case $x_2$ is also a relevant neighbor
		of  $x_1 = \rho_{N_2}$ in $N_2$).
		In summary, all relevant neighbors of $x=\rho_{N_1}$ in $N_1$ are contained in the set $V_X$
		and all relevant neighbors of $x_1 = \rho_{N_2}$ in $N_2$ are contained in the set 
		$V'_X\cup V_B$. Observe that $\ell \coloneqq |L_{N_1}(\rho_{N_1})| =  |L_{N_2}(\rho_{N_2})| $, 
		Theorem~\ref{thm:local-phi} and Corollary~\ref{fact:sum-up-part} imply 

 \begingroup
\allowdisplaybreaks
             \begin{align}
           \Phi^{**}(N_1)-\Phi^{**}(N_2)
           =& (\Phi^*( N_1)-\Phi^*( N_2))+(\phi_{N_1}(x)-\phi_{N_2}(\rho_{B'}))\binom{\ell}{2}   \notag\\
           = & \sum_{v\in V_X} \left(\Phi^*(\partial  N_1 (v))+\phi_{N_1}(v)\binom{|L_{ N_1}(v)|}{2}\right)
             	-\sum_{v\in V'_X \cupdot V_B} \left(\Phi^*(\partial  N_2 (v))+\phi_{ N_2}(v)\binom{|L_{N_2}(v)|}{2}\right) \notag\\
             &+	(\phi_{N_1}(x)-\phi_{N_2}(\rho_{B'}))\binom{\ell}{2}. \label{eq:asDf}\\
        \end{align}
\endgroup

By similar arguments as before, 
\begin{equation}
		\Phi^*(\partial N_1(x_1)) = \sum_{v\in V_B}\left(\Phi^*(\partial N_2(v))+\phi_{ N_2}(v)\binom{|L_{ N_2}(v)|}{2}\right).
\label{eq:alpha1XAB2}
\end{equation}

Taking Equations~\eqref{eq:asDf} and \eqref{eq:alpha1XAB2}   together with $|L_{N_1}(x_1)|<\ell$ and
$\phi_{N_1}(x)=1$, $\phi_{N_1}(x_1)=\omega(B)=\binom{m}{3}$ and $\phi_{N_2}(\rho_{B'})=\omega(B')=\binom{m+1}{3}$ implies
 \begingroup
\allowdisplaybreaks
             \begin{align*}
           \Phi^{**}(N_1)-\Phi^{**}(N_2)
           = & \left(\Phi^*(\partial  N_1 (x_{1}))+\phi_{N_1}(x_{1})\binom{|L_{ N_1}(x_{1})|}{2}\right)
             	-\sum_{v\in V_B} \left(\Phi^*(\partial N_2 (v))+\phi_{ N_2}(v)\binom{|L_{ N_2}(v)|}{2}\right) \\
             &+	(\phi_{N_1}(x)-\phi_{N_2}(\rho_{B'}))\binom{\ell}{2} \\
           = & \phi_{ N_1}(x_1)\binom{|L_{ N_1}(x_1)|}{2}
            	+	(\phi_{N_1}(x)-\phi_{N_2}(\rho_{B'}))\binom{\ell}{2} \\
           < & 	 (\phi_{ N_1}(x_1)
            	+	\phi_{N_1}(x)-\phi_{N_2}(\rho_{B'}))\binom{\ell}{2} \\
           = & \underbrace{\left(\binom{m}{3}+1-\binom{m+1}{3}\right)}_{<0\text{ for $m\geq3$}}\binom{\ell}{2}<0,
        \end{align*}
\endgroup
again contradicting the assumption that $N_1$ maximizes the weighted total cophenetic index within the class $\BinLevelOneN$.

For Case (b),  assume that neither child of $\rho_B$ is the root of a non-trivial block. Note that since $n\geq 3$, at least one child has to be a true tree vertex. Moreover, by Lemma~\ref{lem:one_block}, $N_1$ has to contain at least one crescent. 
Hence, there is a true tree vertex $w$ with children $w_1$ and $w_2$, such that at least one child, without loss of generality $w_1$, is the root of a block $B$ that is a crescent and is of size $m\geq 3$. We can now apply exactly the same arguments as in 
Case (a) and by interchanging the roles of $x$ and $w$ as well as of $w_1$ and $x_1$ and of $w_2$ and $x_2$ to obtain a contradiction. In summary, $N_1$ cannot contain any true tree vertex. 
\end{proof}

\begin{lemma}\label{lem:exactly_one_block}
Every binary phylogenetic level-$1$ network on $n$ leaves that achieves a maximum weighted total cophenetic index within the class $\BinLevelOneN$ contains precisely one non-trivial block.
\end{lemma}
\begin{proof}
By Observation\ \ref{obs:bounded-indegree2}, networks in $\BinLevelOneN$
	have bounded maximum weighted total cophenetic index. Thus, let $N_1\in
	\BinLevelOneN$ be a binary phylogenetic level-$1$ network, $n\geq 3$, that
	maximizes the weighted total cophenetic index within the class
	$\BinLevelOneN$. By Lemma \ref{lem:one_block}, $N_1$ has to contain at least
	one non-trivial block. Assume, for contradiction,
	that $N_1$ contains more than one non-trivial block.
	Lemma~\ref{lem:bin_no_true_vert} implies that $N_1$ cannot contain true
	tree vertices. This together with the fact that $N$ is binary implies that
	any two non-trivial blocks
	in $N_1$ are separated by exactly one edge. In $N_1$, there is in particular a block
	$B_2$ such that its root $\rho_{B_2}$ has no descendants that are also roots of any
	non-trivial block in $N_1$. By the latter arguments and, in particular, since $N_1$ has no
	true tree vertices, there is a non-trivial block $B_1$ in $N_1$ such that
	there are two vertices $v\in V(B_1)$ and $w\in V(B_2)$ with $(v,w)\in E(N_1)$.
	Since $w\prec_{N_1} v$, Lemma \ref{lem:unique-max-B} implies that $w=
	\rho_{B_2}$. Furthermore, since $N_1$ is binary, $v\neq \rho_{B_1}$. By Lemma
	\ref{lem:one_block}, $B_1$ and $B_2$ are crescents. Hence, $B_1$ consists of
	the two internally vertex-disjoint paths
	$(v_1=\rho_{B_1},\ldots,v_{m_1}=\eta_{B_1})$ with $m_1\geq 3$ and
	$(\rho_{B_1},\eta_{B_1})$. Analogously, $B_2$ consists of the two internally
	vertex-disjoint paths $(w_1=\rho_{B_2},\ldots,w_{m_2}=\eta_{B_2})$ with
	$m_2\geq 3$ and $(\rho_{B_2},\eta_{B_2})$. 
	 By  Lemma~\ref{lem:binary_swap}, we can assume that
    $v = par_{N_1}(w)=v_{2}$.
	Moreover, since there is no
	non-trivial block $B'$ in $N_1$ such that $\rho_{B'}\prec_{N_1} \rho_{B_2}$
	and since $N_1$ has no true tree vertices, all vertices $x\in V(N_1)\setminus V(B_2)$, 
	whose (unique) parent $par_{N_1}(x)$ is in $V(B_2)\setminus \{\rho_{B_2}\}$
  must be leaves. Since $N_1$ is binary, one
	easily observes that $\rho_{B_2}$ has no relevant neighbors in $N_1$ at all
	and that each leaf in $L_{N_1}(\rho_{B_2})$ has a unique parent in
	$V(B_2)\setminus \{\rho_{B_2}\}$, i.e., $|L_{N_1}(\rho_{B_2})|=m_2-1$. Let
	$L_{N_1}(\rho_{B_2}) = \{l_1,\dots, l_{m_2-1}\}$.
	We now construct a network $N_2$ from $N_1$ by
	removing $v_2$ and $w_1$ and all their incident edges  and adding edges $(v_1,w_2)$ and $(w_{m_2},v_3)$
   (cf.  Figure~\ref{fig:exactly_one_block}).
		  
\begin{SCfigure}[.4][ht]
	\caption{Sketches of Networks $N_1$ and $N_2$ with $\Phi^{**}(N_2)>\Phi^{**}(N_1)$ used in the proof of Lemma~\ref{lem:exactly_one_block}. }
   \includegraphics[width=0.7\textwidth]{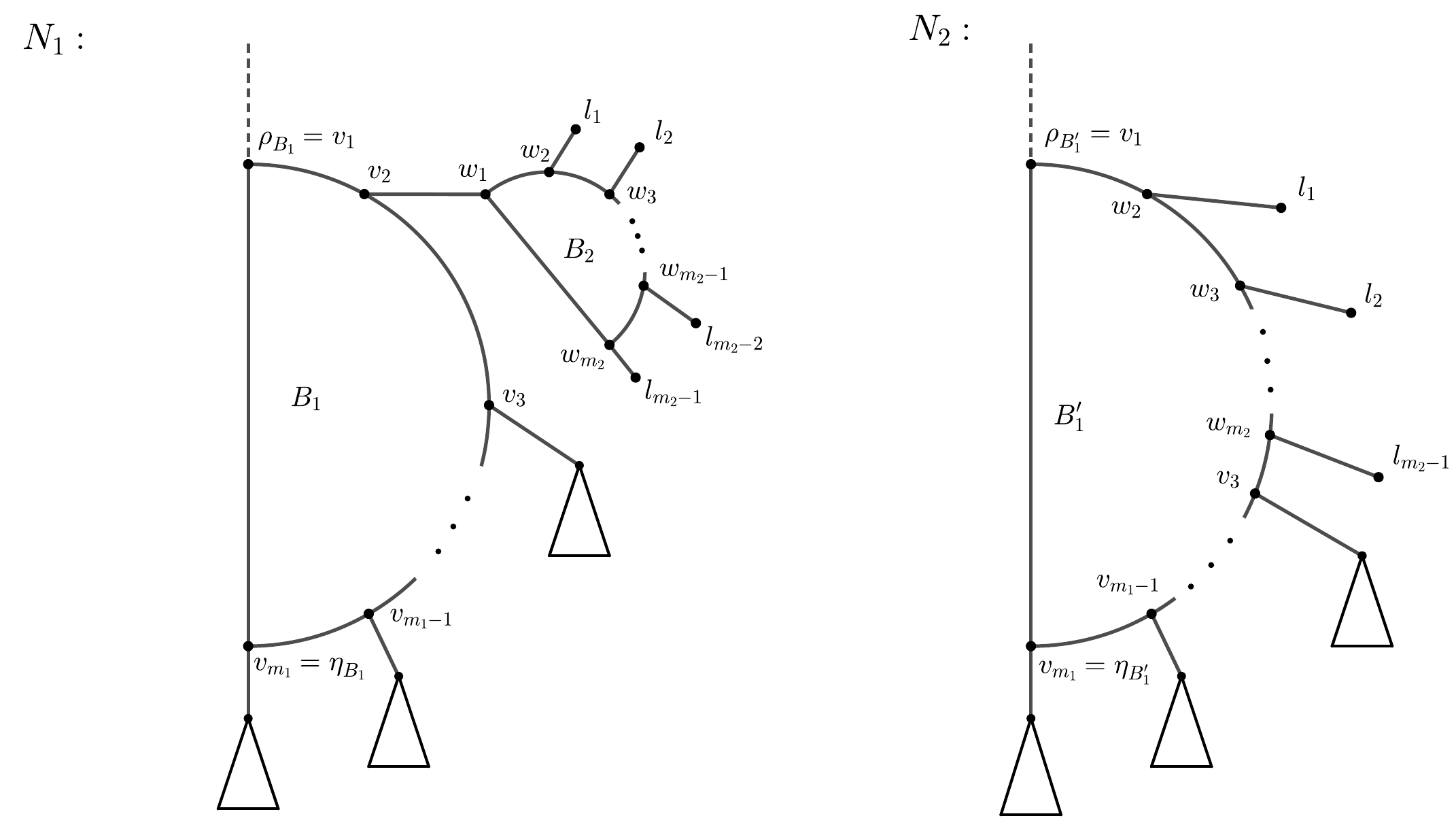}
    \label{fig:exactly_one_block}
\end{SCfigure}

   Let $B_1'$ be the extended block  obtained from $B_1$ in $N_2$, 
   $\widetilde N_1\coloneqq\partN_1(\rho_{B_1})=N_1(\rho_{B_1})$, and
 	 $\widetilde N_2\coloneqq\partN_2(\rho_{B_1'})=N_2(\rho_{B'_1})$.
   By definition and since none of the vertices in $V(B_1)\setminus \{\rho_{B_1}\}$
		are relevant, $\rho_{B_2}$ is a relevant neighbor of $\rho_{B_1}$ in $N_1$
		and we collect all other relevant neighbors of $\rho_{B_1}$ in $N_1$ (if there are any)
		in the set $V_{B}$. 
    It is a straightforward task to verify that $v\in V_{B}$ is a relevant neighbor of 
	  $\rho_{B_1}$ in $N_1$ if and only if $v$ is a relevant neighbor of 
	  $\rho_{B'_1}$ in $N_2$. 
    
    By similar arguments as in the proof of Case (3) in Lemma \ref{lem:max_is_binary} and by Theorem~\ref{thm:local-phi} and
    Corollary~\ref{fact:sum-up-part} as well as the facts that $\ell = |L_{N_1}(\rho_{B_1})| = |L_{N_2}(\rho_{B'_1})|$, 
    $\phi_{N_1}(\rho_{B_1})=\omega(B_1)+\epsilon=\binom{m_1}{3}+\epsilon$, 
    $\phi_{N_1}(\rho_{B_2})=\omega(B_2)+\epsilon=\binom{m_2}{3}+\epsilon$, and 
    $\phi_{N_2}(\rho_{B'_1})=\omega(B'_1)+\epsilon=\binom{m_1+m_2-2}{3}+\epsilon$ 
    we obtain

\begingroup
\allowdisplaybreaks
             \begin{align*}
             &\Phi^{**}(N_1)-\Phi^{**}(N_2)\\
            =&\Phi^*(\widetilde N_1)-\Phi^*(\widetilde N_2)+(\phi_{N_1}(\rho_{B_1})-\phi_{N_2}(\rho_{B_1'}))\binom{\ell}{2}\\
            =&\sum\limits_{v\in V_B}\left(\Phi^*(\partial \widetilde N_1(v)+\phi_{\widetilde N_1}(v)\binom{|L_{\widetilde N_1}(v)|}{2}\right)+\underbrace{\Phi^*(\partial \widetilde N_1(\rho_{B_2}))}_{=0,\text{ since }relV_{\partial \widetilde N_1(\rho_{B_2})}=\emptyset}+\phi_{\partial \widetilde N_1}(\rho_{B_2})\binom{|L_{\widetilde N_1}(\rho_{B_2})|}{2}\\
            &-\sum\limits_{v\in V_B}\left(\Phi^*(\partial \widetilde N_2(v))+\phi_{\partial \widetilde N_2}(v)\binom{|L_{\partial \widetilde N_2}(v)|}{2}\right)+\phi_{N_1}(\rho_{B_1})\binom{\ell}{2}-\phi_{N_2}(\rho_{B_1'})\binom{\ell}{2}\\
            =&\phi_{\widetilde N_1}(\rho_{B_2})\binom{|L_{\widetilde N_1}(\rho_{B_2})|}{2}+\phi_{N_1}(\rho_{B_1})\binom{\ell}{2}-\phi_{N_2}(\rho_{B_1'})\binom{\ell}{2}\geq 0.
            \end{align*}
\endgroup 
            Since $ \ell >|L_{\widetilde N_1}(\rho_{B_2})| $, we thus obtain
            $\phi_{\widetilde N_1}(\rho_{B_2})\binom{\ell}{2}+\phi_{N_1}(\rho_{B_1})\binom{\ell}{2}>\phi_{N_2}(\rho_{B_1'})\binom{\ell}{2}$  and, therefore,
            $\binom{m_2}{3}+\binom{m_1}{3}+\epsilon > \ \binom{m_1+m_2-2}{3}$. 
From here one we can reuse the exactly same arguments as in the proof of Lemma~\ref{lem:max_is_binary} starting at Equation\ \eqref{eq:fromHereOn} and by interchanging the role of $m_1$ and $m$
to obtain a contradiction. Hence, a binary phylogenetic level-$1$ network with maximum weighted total cophenetic index cannot contain more than one non-trivial block. 
\end{proof}

We are now in the position to prove Theorem \ref{thm:binary_galledtree_maximum}
\begin{proof}[Proof of Theorem \ref{thm:binary_galledtree_maximum}]
One easily verifies that the binary phylogenetic level-$1$ network $\NC_n$ ($n\geq 3$) is the only network within the class $\BinLevelOneN$
without true tree vertices and that contains exactly one block that is a crescent. This together with 
Lemmas~\ref{lem:one_block}, \ref{lem:bin_no_true_vert}, and \ref{lem:exactly_one_block} implies that 
$\NC_n$ is the unique binary phylogenetic level-$1$  network that maximizes the weighted total cophenetic index within the class $\BinLevelOneN$, $n\geq 3$.
Moreover, it is an easy task to verify that
 $\Phi^{**}(\NC_n)=\phi_{\NC_n}(\rho_{\NC_n})\binom{n}{2}=\left(\binom{n+1}{3}+\epsilon\right)\binom{n}{2}$.
  The latter arguments establish  Theorem \ref{thm:binary_galledtree_maximum} for the class $\BinLevelOneN$. 
  Moreover, by Lemma \ref{lem:max_is_binary}, each network $N\in \HybdidDegTwoN$, $n\geq 3$, that maximizes the 
  weighted total cophenetic index within the class $\HybdidDegTwoN$ must be binary. This together with 
  the latter arguments establishes 
  Theorem \ref{thm:binary_galledtree_maximum} for  the class $\HybdidDegTwoN$, $n\geq 3$. 
\end{proof}

\section*{Acknowledgements} 
The authors wish to thank Lina Herbst, Sophie Kersting, David Schaller and Mirko Wilde for
helpful discussions concerning an earlier version of this manuscript. Moreover,
MF wishes to thank the joint research project \textbf{\textit{DIG-IT!}}
supported by the European Social Fund (ESF), reference: ESF/14-BM-A55-0017/19,
and the Ministry of Education, Science and Culture of Mecklenburg-Vorpommerania,
Germany, under which parts of this project have been conducted. 
This work was supported by the Data-driven Life Science (DDLS) program funded by
the Knut and Alice Wallenberg Foundation, Sweden.

\bibliographystyle{plaindin}
\bibliography{NetB}

\end{document}